%% file: main.tex
\documentclass[11pt]{article}
\usepackage{setspace} %\onehalfspacing
\usepackage{fancyhdr} 
\usepackage[font=small,labelfont=bf]{caption}
\usepackage{booktabs}
\usepackage{enumitem}
\usepackage{tabularx} % extra features for tabular environment
\usepackage{amsmath}  % improve math presentation
\usepackage{graphicx, subcaption} % takes care of graphic including machinery
\usepackage{algorithm}% write algorithms
\usepackage{algorithmicx}% write algorithms
\usepackage{algpseudocode}% format algorithm environment
\algnewcommand{\LeftComment}[1]{\Statex \(\triangleright\) #1}

\usepackage[margin=1in,letterpaper]{geometry} % decreases margins
\usepackage{cite} % takes care of citations
\usepackage[final]{hyperref} % adds hyper links inside the generated pdf file
\usepackage [english]{babel}
\usepackage [autostyle, english = american]{csquotes}
\usepackage{amsfonts}
\usepackage{amsmath, amssymb, amsthm}
\MakeOuterQuote{"}

\usepackage{libertine}\usepackage[libertine]{newtxmath}
\usepackage[scaled=0.96]{zi4}

\usepackage{float}
\usepackage{physics}
\usepackage{empheq}
 
\newlength\dlf  % Define a new measure, dlf

\usepackage{gensymb}
\usepackage{mathtools}

\newtheorem{theorem}{Theorem}[section]

\newtheorem{conj}[theorem]{Conjecture}
\newtheorem{corollary}[theorem]{Corollary}

\newtheorem{proposition}[theorem]{Proposition}

\newtheorem{problem}[theorem]{Problem}

\theoremstyle{definition}

\newtheorem{defn}[theorem]{Definition}
\newtheorem*{remark}{Remark}
\let\realbfseries=\bfseries
\def\bfseries{\realbfseries\boldmath}

\let\epsilon=\varepsilon
\def\defn#1{\textbf{\textit{\boldmath #1}}}

\newcommand{\OO}{\texorpdfstring{\,\vcenter{\hbox{\includegraphics[scale=0.2]{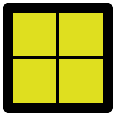}}}\,}O}
\newcommand{\TT}{\texorpdfstring{\,\vcenter{\hbox{\includegraphics[scale=0.2]{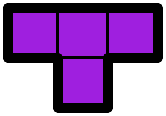}}}\,}T}
\newcommand{\LL}{\texorpdfstring{\,\vcenter{\hbox{\includegraphics[scale=0.2]{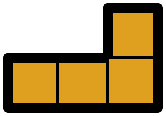}}}\,}L}
\newcommand{\JJ}{\texorpdfstring{\,\vcenter{\hbox{\includegraphics[scale=0.2]{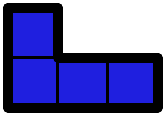}}}\,}J}
\renewcommand{\SS}{\texorpdfstring{\,\vcenter{\hbox{\includegraphics[scale=0.2]{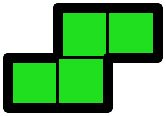}}}\,}S}
\newcommand{\ZZ}{\texorpdfstring{\,\vcenter{\hbox{\includegraphics[scale=0.2]{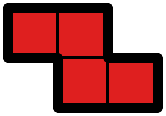}}}\,}Z}
\newcommand{\II}{\texorpdfstring{\,\vcenter{\hbox{\includegraphics[scale=0.2]{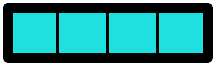}}}\,}I}
\newcommand{\ALL}{\II, \allowbreak \OO, \allowbreak \TT, \allowbreak \SS, \allowbreak \ZZ, \allowbreak \JJ, \allowbreak \LL}

\begin{document}
\input{report}

\bibliographystyle{alpha}
\bibliography{biblio}

\input{appendix}
%++++++++++++++++++++++++++++++++++++++++

\end{document}

%% file: report.tex
\title{Tetris with Few Piece Types}
\author{%
  MIT Hardness Group%
    \thanks{Artificial first author to highlight that the other authors (in
      alphabetical order) worked as an equal group. Please include all
      authors (including this one) in your bibliography, and refer to the
      authors as “MIT Hardness Group” (without “et al.”).}
\and
  Erik D. Demaine%
    \thanks{MIT Computer Science and Artificial Intelligence Laboratory,
      32 Vassar St., Cambridge, MA 02139, USA, \protect\url{{edemaine,hhall314,jeli}@mit.edu}}
\and
  Holden Hall\footnotemark[2]
\and
  Jeffery Li\footnotemark[2]
}
\date{}

\maketitle

\begin{abstract}
    We prove NP-hardness and \#P-hardness of Tetris clearing (clearing an initial board using a given sequence of pieces) with the Super Rotation System (SRS), even when the pieces are limited to \emph{any two} of the seven Tetris piece types.
    This result is the first advance on a question posed twenty years ago:
    %\cite{Tetris_IJCGA}:
    which piece sets are easy vs.\ hard?
    All previous Tetris NP-hardness proofs used five of the seven piece types.
    We also prove ASP-completeness of Tetris clearing, using three piece types,
    as well as versions of 3-Partition and Numerical 3-Dimensional Matching
    where all input integers are distinct.
    Finally, we prove NP-hardness of Tetris survival and clearing
    under the "hard drops only" and "20G" modes, using two piece types,
    improving on a previous "hard drops only" result that used five piece types.
\end{abstract}

\section{Introduction}\label{sec:intro}

Tetris is one of the oldest and most popular puzzle video games, originally created by Alexey Pajitnov in 1984.
Tetris has reached mainstream media many times,
most recently in the biopic \emph{Tetris} \cite{Tetris_movie}
and with the news of 13-year-old Willis Gibson being the first person to "beat" the NES version of Tetris by reaching a killscreen \cite{Tetris_Beat}.

The rules of Tetris are simple.  In each round, a tetromino piece
(one of $\ALL$)
spawns at the top of a grid and periodically moves down one unit, assuming the squares below the piece are empty.
The player can repeatedly move this piece one unit left, one unit right, or one unit down, or rotate the piece by $\pm 90^\circ$.
When any part of the piece rests on top of a filled square for long enough that it triggers an automatic downward move, the piece "locks" in place, and stops moving.
If a piece stops above a certain height or where the next piece would spawn, the player loses; otherwise, the next piece spawns at the top of the grid, and play continues.
Completely filling a row causes the row to clear,
and all squares above that row move downward by one unit.
For more detailed rules, see \cite{Tetris_Wikipedia}.

To study Tetris from a computational complexity perspective, we generally
assume that the player is given a sequence of pieces and an initial board state
of filled cells, making the game perfect information
(as introduced in \cite{Tetris_IJCGA}).
The two main objectives we consider here are ``clearing'' and ``survival''
(as introduced in \cite{TotalTetris_JIP}).
In \defn{Tetris clearing}, we want to determine whether we can clear the
entire board after placing all the given pieces.
In \defn{Tetris survival}, we want to determine whether
the player can avoid losing before placing all the given pieces.
Previous work shows that these problems are NP-complete,
even to approximate various metrics within $n^{1-\epsilon}$ \cite{Tetris_IJCGA},
or with only $8$ columns or $4$ rows \cite{ThinTetris_JCDCGGG2019},
or with additional constraints on drops \cite{HardDrops},
or with $k$-ominoes for $k \geq 3$ clearing or $k \geq 4$ survival
\cite{TotalTetris_JIP}.

\subsection{Our Results}\label{sec:results}

One of the open problems posed in the original paper proving Tetris NP-hard
twenty years ago \cite{Tetris_IJCGA} is to determine which subsets of the
seven Tetris piece types $\{\ALL\}$
suffice for NP-hardness, and which admit a polynomial-time algorithm.
All existing Tetris NP-hardness proofs
\cite{Tetris_IJCGA,ThinTetris_JCDCGGG2019,HardDrops}
use at least five of the seven piece types.
% \cite{Tetris_IJCGA} claims 5 -- {LG, RG, I, Sq,T} -- no S or Z
% Other 5-tuples on https://erikdemaine.org/papers/Tetris_IJCGA/paper.pdf#page=24
% \cite{ThinTetris_JCDCGGG2019} seems to use all 7
% \cite{HardDrops} does not seem to use S or O
In particular, \cite[Section~6.2]{Tetris_IJCGA}
mentions various sets of five piece types that suffice.
What about fewer
%than five
piece types?

\begin{table}
    \centering
    \small
    \tabcolsep=0.7\tabcolsep
    \begin{tabular}{c|c|c|c|c|c|c|c|}
         & $\II$ & $\OO$ & $\TT$ & $\SS$ & $\ZZ$ & $\JJ$ & $\LL$ \\\hline
         $\II$ & $-$ & Prop. \ref{prop:IL} (H) & Prop. \ref{prop:IL} (H) & Prop. \ref{prop:IL} & Prop. \ref{prop:IL} & Prop. \ref{prop:IL} (H, G) & Prop. \ref{prop:IL} (H, G) \\\hline
         $\OO$ & $-$ & $-$ & Prop. \ref{prop:OT} & Prop. \ref{prop:OS} & Prop. \ref{prop:OS} & Prop. \ref{prop:OJ} & Prop. \ref{prop:OJ} \\\hline
         $\TT$ & $-$ & $-$ & $-$ & Prop. \ref{prop:ST} & Prop. \ref{prop:ST} & Prop. \ref{prop:JZ} & Prop. \ref{prop:JZ} \\\hline
         $\SS$ & $-$ & $-$ & $-$ & $-$ & Prop. \ref{prop:SZ} & Prop. \ref{prop:JS} & Prop. \ref{prop:JZ} \\\hline
         $\ZZ$ & $-$ & $-$ & $-$ & $-$ & $-$ & Prop. \ref{prop:JZ} & Prop. \ref{prop:JS} \\\hline
         $\JJ$ & $-$ & $-$ & $-$ & $-$ & $-$ & $-$ & Prop. \ref{prop:JL} (G?) \\\hline
         $\LL$ & $-$ & $-$ & $-$ & $-$ & $-$ & $-$ & $-$ \\\hline
    \end{tabular}
    \caption{Our NP-hardness results for Tetris clearing assuming SRS. Each entry in a specific row and column corresponds to the proposition for the hardness of the two-element subset consisting of the row piece and column piece (for example, the entry "Prop. \ref{prop:IL}" in row $\II$ and column $\OO$ indicates that Proposition \ref{prop:IL} proves hardness for the subset $\{\II, \OO\}$). Letters in parentheses denote additional models ("H" for "hard drop only", "G" for "20G"); question mark indicates a conjecture for hardness under that additional model.}
    \label{tab:subset}
\end{table}

Our main results are the first to make progress on this question:
for any \emph{size-2} subset $A\subseteq \{\ALL\}$,
Tetris clearing is NP-complete with pieces restricted to~$A$.
Most pairs of piece types require different constructions for their reductions;
refer to Table~\ref{tab:subset}.
Our results require us to specify more details of the piece rotation model,
specifically what happens when the player rotates a piece in a way that
collides with a filled square.
We assume the \defn{Super Rotation System (SRS)} \cite{SRS_TetrisWiki},
first introduced in the 2001 game \emph{Tetris Worlds}
and as part of the Tetris Company's Tetris Guideline for how all modern (2001+)
Tetris games should behave \cite{Guideline_TetrisWiki}.
%which has emerged as the standard rotation system for most modern
%implementations of Tetris.

For every size-2 subset $A$ of piece types,
we also establish \#P-hardness for the corresponding problem of
counting the number of ways to clear the board.
Here we distinguish solutions by the final placement of each piece,
not the sequence of moves to make those placements
(as long as the placement is valid).
This definition lets us ignore e.g.\ the null effect of
moving a piece repeatedly left and right.

For certain size-3 subsets of piece types, we further establish
ASP-completeness for Tetris clearing.  Recall that an NP search problem
is \defn{ASP-complete} \cite{Yato-Seta-2003}
if there is a parsimonious reduction from all NP search problems
(including a polynomial-time bijection between solutions).
In particular, ASP-completeness implies NP-hardness of finding another solution
given $k$ solutions, for any $k \geq 0$, as well as \#P-completeness.
These results hold for piece types $\{\II, \TT, \LL\}$ and $\{\II, \TT, \JJ\}$.

We also study Tetris under two more restrictive models on piece moves:

\begin{itemize}
    \item \defn{Hard drops only}: In this model, pieces do not move downward on their own, and if the player moves a piece downward, the piece moves maximally downward before locking into place (a \defn{hard drop} maneuver). The player is still free to rotate or move the piece left or right before hard-dropping the piece.
    This model is motivated by most Tetris games awarding higher scores
    for hard drops, and was posed in \cite{Tetris_IJCGA}.
    \item \defn{20G}: In this model, instead of periodically moving down one unit, all pieces move maximally downward \emph{instantly and on their own}, and the player is not allowed to control how fast a piece moves downward.
    The player is still free to rotate or move the piece left or right before the piece locks.
    This model is motivated by levels with the maximum possible gravity, as in Level 20+ of regular Tetris with 20 rows \cite{20G_TetrisWiki}.
\end{itemize}

For certain size-2 subsets of piece types, we establish NP-hardness
of both Tetris survival and clearing under either of these models.
Table~\ref{tab:subset} labels which of our Tetris clearing results hold in which models.

Along the way, we prove new results about 3-Partition and Numerical
3-Dimensional Matching (3DM): both problems are strongly ASP-complete
even when all integers are assumed distinct.
These results are of independent interest for ASP-hardness reductions.
Previously, these problems were known to be ASP-complete with
multisets of integers \cite{PathPuzzles_GC},
and strongly NP-complete with distinct integers \cite{HULETT2008594}.

\subsection{Outline}\label{sec:outline}

The structure of the rest of the paper is as follows.
Section \ref{sec:srs} details the Super Rotation System (SRS),
an important aspect of modern Tetris and used in our constructions.
Section~\ref{sec:3part} proves ASP-completeness of 3-Partition with
Distinct Integers and Numerical 3-Dimensional Matching with Distinct Integers,
two problems we reduce from.
Section~\ref{sec:2piecetetris} discusses our hardness results for Tetris clearing with SRS with only two piece types.
Section~\ref{sec:survival} discusses some Tetris survival results under the "hard drops only" and "20G" models.
Section \ref{sec:asp} proves ASP-completeness of Tetris clearing with SRS.

\section{Super Rotation System (SRS)}\label{sec:srs}

Most previous Tetris results are not sensitive to exactly how Tetris pieces
rotation: most reasonable rotation models work \cite[Section~6.4]{Tetris_IJCGA}.
By contrast, many of the results in this paper focus specifically on
(and require) the \defn{Super Rotation System (SRS)} \cite{SRS_TetrisWiki},
defined as follows.

Each piece has a defined \defn{rotation center},
as indicated by dots in Figure \ref{allpieces},
except for $\II$ and $\OO$, whose rotation centers are the centers of the $4\times 4$ squares in Figure \ref{allpieces}. When unobstructed, all non-$\OO$ tetrominoes will rotate purely about the rotation center (note that $\OO$ pieces cannot rotate). The key feature about SRS is \defn{kicking}: if a tetromino is obstructed when a rotation is attempted, the game will attempt to "kick" the tetromino into one of four alternate positions, each tested sequentially; if all four positions do not work, then the rotation will fail. See Figure \ref{SRSrotate} for an example of this kicking process. The full data for wall kicks can be found in Tables \ref{tab:kickdataMain} and \ref{tab:kickdataI}, and at \cite{SRS_TetrisWiki}. Of note is that SRS wall kicks are vertically symmetric for all pieces or pairs of pieces (i.e., $\JJ\leftrightarrow \LL$ and $\SS\leftrightarrow \ZZ$) except for the $\II$ piece, so all rotations can be mirrored.

\begin{figure}[ht]
    \centering
    \includegraphics[width=160pt]{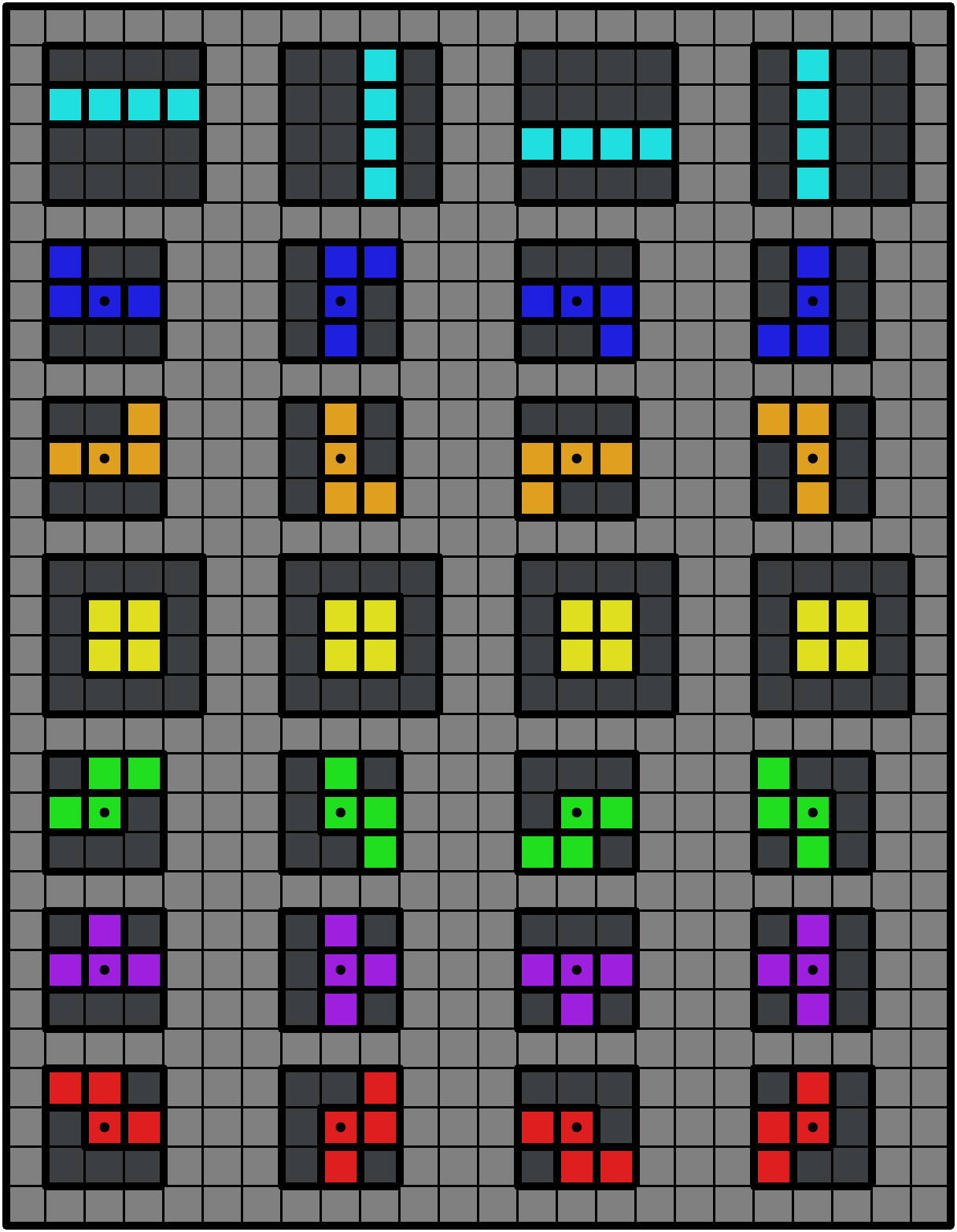}
    \caption{All tetromino pieces, in order from top to bottom: $\II$, $\JJ$, $\LL$, $\OO$, $\SS$, $\TT$, $\ZZ$. The first column is the default orientation of a piece upon spawning in; each column to the right indicates a $90^\circ$ rotation clockwise about the rotation center of the piece.}
    \label{allpieces}
\end{figure}

\begin{figure}[ht]
  \centering
  \begin{subfigure}[b]{0.3\textwidth}
    \centering
    \includegraphics[width=60pt]{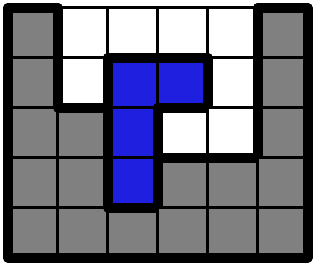}
    \caption{}
  \end{subfigure}
  \begin{subfigure}[b]{0.3\textwidth}
    \centering
    \includegraphics[width=60pt]{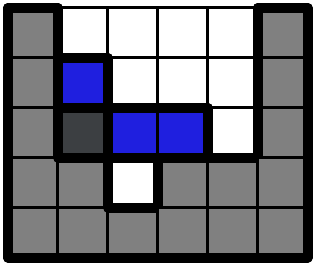}
    \caption{}
  \end{subfigure}
  \begin{subfigure}[b]{0.3\textwidth}
    \centering
    \includegraphics[width=60pt]{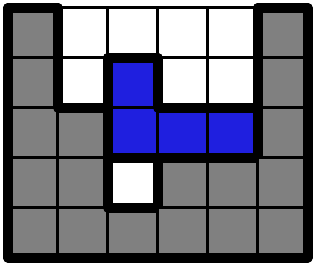}
    \caption{}
  \end{subfigure}
  \caption{An example of the SRS kick system. Suppose the $\JJ$ piece in (a) is being rotated $90^\circ$ counter-clockwise. Test 1 (which is $(0, 0)$) would fail, due to the dark gray square shown in (b). Test 2 (which is $(+1, 0)$) would succeed, as shown in (c), and so the $\JJ$ piece would rotate to the position in (c).}
  \label{SRSrotate}
\end{figure}

\begin{table}[ht]
    \centering
    \footnotesize
    \begin{tabular}{c|c|c|c|c|c|}
         & Test 1 & Test 2 & Test 3 & Test 4 & Test 5 \\\hline
         $0\to R$ & $(0, 0)$ & $(-1, 0)$ & $(-1, +1)$ & $(0, -2)$ & $(-1, -2)$ \\\hline
         $R\to 0$ & $(0, 0)$ & $(+1, 0)$ & $(+1, -1)$ & $(0, +2)$ & $(+1, +2)$ \\\hline
         $R\to 2$ & $(0, 0)$ & $(+1, 0)$ & $(+1, -1)$ & $(0, +2)$ & $(+1, +2)$ \\\hline
         $2\to R$ & $(0, 0)$ & $(-1, 0)$ & $(-1, +1)$ & $(0, -2)$ & $(-1, -2)$ \\\hline
         $2\to L$ & $(0, 0)$ & $(+1, 0)$ & $(+1, +1)$ & $(0, -2)$ & $(+1, -2)$ \\\hline
         $L\to 2$ & $(0, 0)$ & $(-1, 0)$ & $(-1, -1)$ & $(0, +2)$ & $(-1, +2)$ \\\hline
         $L\to 0$ & $(0, 0)$ & $(-1, 0)$ & $(-1, -1)$ & $(0, +2)$ & $(-1, +2)$ \\\hline
         $0\to L$ & $(0, 0)$ & $(+1, 0)$ & $(+1, +1)$ & $(0, -2)$ & $(+1, -2)$ \\\hline
    \end{tabular}
    \caption{Kick data for $\JJ$, $\LL$, $\SS$, $\TT$, and $\ZZ$ pieces. $0$ indicates the default orientation, and $R$, $2$, and $L$ indicate the orientation reached from a $90^\circ$, $180^\circ$, and $270^\circ$ rotation clockwise (respectively) from the default orientation. An ordered pair $(a, b)$ denotes a translation of the center by $a$ units in the $x$ direction and $b$ units in the $y$ direction. Positive $x$ direction is rightwards, and positive $y$ direction is upward.}
    \label{tab:kickdataMain}
\end{table}

\begin{table}[ht]
    \centering
    \footnotesize
    \begin{tabular}{c|c|c|c|c|c|}
         & Test 1 & Test 2 & Test 3 & Test 4 & Test 5 \\\hline
         $0\to R$ & $(0, 0)$ & $(-2, 0)$ & $(+1, 0)$ & $(-2, -1)$ & $(+1, +2)$ \\\hline
         $R\to 0$ & $(0, 0)$ & $(+2, 0)$ & $(-1, 0)$ & $(+2, +1)$ & $(-1, -2)$ \\\hline
         $R\to 2$ & $(0, 0)$ & $(-1, 0)$ & $(+2, 0)$ & $(-1, +2)$ & $(+2, -1)$ \\\hline
         $2\to R$ & $(0, 0)$ & $(+1, 0)$ & $(-2, 0)$ & $(+1, -2)$ & $(-2, +1)$ \\\hline
         $2\to L$ & $(0, 0)$ & $(+2, 0)$ & $(-1, 0)$ & $(+2, +1)$ & $(-1, -2)$ \\\hline
         $L\to 2$ & $(0, 0)$ & $(-2, 0)$ & $(+1, 0)$ & $(-2, -1)$ & $(+1, +2)$ \\\hline
         $L\to 0$ & $(0, 0)$ & $(+1, 0)$ & $(-2, 0)$ & $(+1, -2)$ & $(-2, +1)$ \\\hline
         $0\to L$ & $(0, 0)$ & $(-1, 0)$ & $(+2, 0)$ & $(-1, +2)$ & $(+2, -1)$ \\\hline
    \end{tabular}
    \caption{Kick data for $\II$ pieces, with same notation as Table \ref{tab:kickdataMain}.}
    \label{tab:kickdataI}
\end{table}

This system of kicking tetrominoes during rotations allows for moves which are often called \defn{twists} or \defn{spins}. All the spins that we utilize are detailed in Appendix \ref{appendix:spins}.

\section{3-Partition and Numerical 3DM with Distinct Integers}\label{sec:3part}

Our reductions to Tetris are all from one of the following two problems,
which are strengthenings of two standard strongly NP-complete problems:

\begin{problem}[\textbf{3-Partition with Distinct Integers}]
    Given a \textbf{set} $A = \{a_1, a_2, \ldots, a_n\}$ of $n$
    \textbf{distinct} positive integers
    such that $\frac t4 < a_i < \frac t2$ for each $i$,
    where $t = \frac 3n \sum_{i=1}^n a_i$,
    determine whether there is a partition of $A$ into $\frac n3$ groups
    $D_1, \ldots, D_{n/3}$
    (each necessarily of size $3$)
    having the same sum $\sum_{x\in D_j}x = t$.
\end{problem}

\begin{problem}[\textbf{Numerical 3-Dimensional Matching (3DM) with Distinct Integers}]
    Given three \textbf{sets} $$A = \{a_1, a_2, \ldots , a_n\}, B = \{b_1, b_2, \ldots , b_n\},\text{ and }C = \{c_1, c_2, \ldots , c_n\}$$ of $n$ positive integers, where all $3n$ integers are \textbf{distinct}, and a target sum $t = \frac 1n \sum_{i=1}^n (a_i+b_i+c_i)$, determine whether there is a partition of $A\cup B\cup C$ into $n$ groups $D_1, \ldots, D_n$, each with exactly one element from each of $A$, $B$, and $C$, and $\sum_{x\in D_j}x = t$ for all $j$.
\end{problem}

Without the "distinct" and "set" conditions, both problems are well-known to be
\defn{strongly NP-complete}, meaning that the problem is NP-hard even if the
$a_i$ integers are bounded by a polynomial in~$n$.
This property makes it feasible to represent each integer $a_i$ (and $t$)
in unary, which is the approach taken by all past Tetris NP-hardness proofs
\cite{Tetris_IJCGA,TotalTetris_JIP,ThinTetris_JCDCGGG2019,HardDrops},
as then the total reduction size is still polynomial in~$n$.

We want to ensure all integers are distinct in order to have more control
over our reductions' blowup in the number of solutions,
as needed for \#P- and ASP-hardness.
Bosboom et al.~\cite{PathPuzzles_GC} proved that numerical 3DM is
strongly ASP-complete when $A$ is restricted to be a set,
but allowed for $B$ and $C$ to be multisets as usual,
and did not forbid repeated integers between $A,B,C$.
Hulett, Will, and Woeginger \cite{HULETT2008594} proved that
both 3-Partition and Numerical 3DM remain strongly NP-hard
with distinct integers.
We extend their proof to obtain ASP-completeness:

\begin{theorem}\label{3partasp}
    3-Partition with Distinct Integers,
    and Numerical 3-Dimensional Matching with Distinct Integers,
    are strongly ASP-complete.
\end{theorem}

To prove this result, we use the following intermediate problems
(which are thus also ASP-complete):

\begin{problem}[\textbf{Positive 1-in-3SAT}]
    Given a boolean formula in 3CNF (i.e., an \textsc{and} of clauses consisting of 3 literals), where all literals are positive, does there exist an assignment of the variables to either true or false such that each clause has exactly one literal set to true?
\end{problem}

\begin{problem}[\textbf{Tripartite Edge-Disjoint Triangle Partition}]
    Given an undirected tripartite graph $G = (V, E)$, can we partition $E$ into disjoint triangles?
\end{problem}

\begin{proof}[Proof of Theorem \ref{3partasp}]
    We give a chain of parsimonious reductions from 3SAT,
    which is known to be ASP-complete \cite{Yato-Seta-2003}:

    \begin{enumerate}
        \item \textbf{3SAT $\to$ Positive 1-in-3SAT}: 
          Hunt, Marathe, Radhakrishnan, and Stearns~\cite[Theorem~3.8]{Hunt1in3SAT}
          gave such a parsimonious reduction.%
          \footnote{Their problem ``1-\textsc{Ex3MonoSat}'' is
            Positive 1-in-3SAT with the additional constraint that
            every clause has exactly three literals.
            Their reduction is also planarity preserving,
            so chaining with their parsimonious reduction from 3SAT to
            Planar 3SAT, we obtain that Planar 1-in-3SAT is also ASP-complete.}
          See also \cite[Lemma 2.1]{PathPuzzles_GC}.

        \item \textbf{Positive 1-in-3SAT $\to$ Tripartite Edge-Disjoint Triangle Partition}:
        We follow a simplification of a reduction from FCP 1-in-3SAT
        to FCP Tripartite Edge-Disjoint Triangle Partition
        \cite[Theorem~12]{FewestClues_TCS},
        which in turn is based on a reduction from 3SAT to
        Tripartite Edge-Disjoint Triangle Partition \cite{Holyer-1981}.
        By reducing from Positive 1-in-3SAT, we simplify the
        reduction of \cite{FewestClues_TCS} by avoiding negative literals.

        We represent each variable by a sufficiently large
        triangular grid of vertices,
        with opposite sides of a parallelogram identified to form a flat torus,
        as shown in Figure~\ref{fig:triangle partition T}.
        This grid has exactly two solutions, corresponding to
        true (the triangles in Figure~\ref{fig:triangle partition T})
        and false (the triangles in Figure~\ref{fig:triangle partition F});
        note that the two solutions consist of exactly the same edges,
        and cover each exactly once.
        For each clause $(x,y,z)$, we pick one triangle of positive
        orientation, remove its edges, and unify the corresponding vertices
        of these triangles and of the three neighboring triangles
        of negative orientation,
        as shown in Figure~\ref{fig:triangle partition clause}.
        Exactly one variable must choose the true state so as to cover
        the edges surrounding the unified hole exactly once.
        By choosing the variable gadgets large enough,
        we can ensure that the clause gadgets are disjoint from each other.
        Each gadget has a unique way to implement a given assignment,
        so this reduction is parsimonious.

        \begin{figure}
          \centering
          \subcaptionbox{\label{fig:triangle partition T} Variable gadget in true state. A/B denote unification to form torus.}
            {\includegraphics[page=1,scale=0.28]{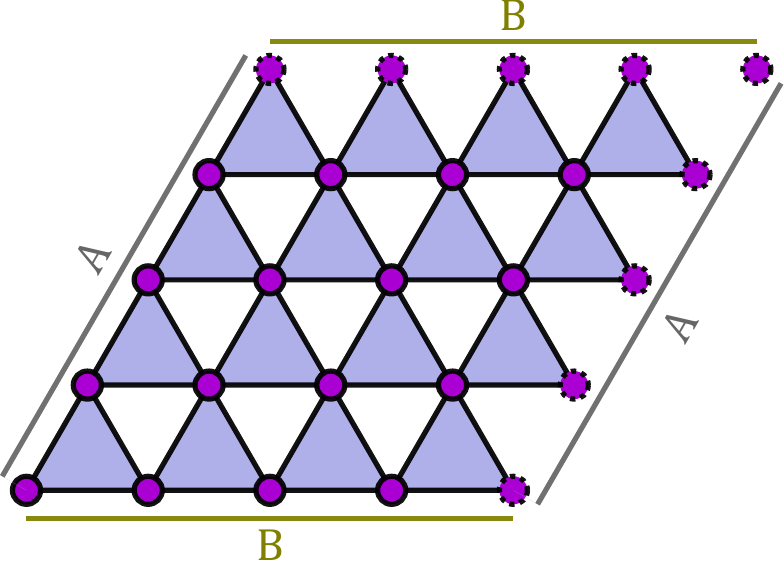}}%
          \hfill
          \subcaptionbox{\label{fig:triangle partition F} Variable gadget in false state. A/B denote unification to form torus.}
            {\includegraphics[page=2,scale=0.28]{figs/triangle-partition.pdf}}%
          \hfill
          \subcaptionbox{\label{fig:triangle partition clause} Clause gadget bringing together three variable gadgets. Vertices connected by dashes are unified.}
            {\includegraphics[page=3,scale=0.28]{figs/triangle-partition.pdf}}%
          \caption{Reduction from Positive 1-in-3SAT to Tripartite Edge-Disjoint Triangle Partition.}
          \label{fig:triangle partition}
        \end{figure}

        \item \textbf{Tripartite Edge-Disjoint Triangle Partition $\to$ Numerical 3DM with Distinct Integers}:
        We combine a chain of reductions, from
        Triangle Edge-Disjoint Triangle Partition to
        Latin Square Completion \cite{Colbourn-1984},
        and from Latin Square Completion
        to Numerical 3DM with Distinct Integers \cite{HULETT2008594}.

        If $U = \{u_1,u_2,\ldots\}, V=\{v_1,v_2,\ldots\}, W=\{w_1,w_2,\ldots\}$ is the vertex tripartition, then we do the following:
        \begin{itemize}
            \item Let $q = 2\max\{|U|, |V|, |W|\}$, and let the target sum be $t=19q^6$.
            \item Map each edge $(u_i,w_k)$ to $2q^6 + iq - k \in A$.
            \item Map each edge $(v_j,w_k)$ to $7q^6 + jq^2 + k \in B$.
            \item Map each edge $(u_i,v_j)$ to $t - (9q^6 + jq^2 + iq) = 10q^6 - jq^2 - iq \in C$.
        \end{itemize}
        The lemmas in \cite{HULETT2008594} show that all the integers
        in $A$, $B$, and $C$ are distinct (i.e., we have a valid instance
        of Numerical 3DM with Distinct Integers);
        and that any triple summing to $t$ consists of one element each
        from $A$, $B$, and $C$, with the elements corresponding to a
        triangle in the graph.
        Thus we obtain a bijection between triangle partitions and
        Numerical 3DM solutions, i.e., the reduction is parsimonious.
        \item \textbf{Numerical 3DM with Distinct Integers $\to$ 3-Partition with Distinct Integers}:
        We use standard techniques to relate these problems.
        Convert each integer $a_i$, $b_i$, and $c_i$ in Numerical 3DM
        to integers $8 a_i+1$, $8 b_i+2$, and $8 c_i-3$, respectively,
        in 3-Partition; and convert $t$ to $8t$.
        In particular, all integers are still distinct, because we scale up
        by a factor of $8$ and then shift values by less than~$4$.
        Furthermore, working modulo $8$, every triple of integers summing to
        $t$ must take exactly one $a_i$, one $b_j$, and one $c_k$.
        Therefore we have a parsimonious reduction.
    \end{enumerate}

    Composing these reductions, we obtain that 3-Partition with
    Distinct Integers, and Numerical 3DM with Distinct Integers, are ASP-hard.
    Both problems are NP search problems, so they are ASP-complete.
\end{proof}

%Note also that previous proofs for hardness of clearing and survival Tetris \cite{Tetris_IJCGA, TotalTetris_JIP, ThinTetris_JCDCGGG2019} reduce from 3-Partition, but with \emph{multisets} (i.e., the $a_i$ need not be distinct). Here, we use the version of the problem with distinct integers so that, for our \#P-hardness proofs, no additional multiplicative factors arise due to some of the $a_i$ being the same integer.

%For the purposes of counting problems, we also consider the following problem:

%\begin{problem}[\textbf{Numerical 3-Dimensional Matching with Distinct Integers}]
%    Given three \textbf{sets} $$A = \{a_1, a_2, \ldots , a_n\}, B = \{b_1, b_2, \ldots , b_n\},\text{ and }C = \{c_1, c_2, \ldots , c_n\}$$ of positive integers (i.e., all integers are distinct) and a target sum $t = \frac 1n \sum_{i=1}^n (a_i+b_i+c_i)$, determine whether there is a partition of $A\cup B\cup C$ into $n$ groups $D_1, \ldots, D_n$, each with $3$ elements, 1 from each of $A$, $B$, and $C$, and $\sum_{x\in D_j}x = t$ for all $j$.
%\end{problem}

%%%%%%INCLUDE PROOF HERE%%%%%
%. Furthermore, there is a parsimonious reduction from Numerical 3-Dimensional Matching to 3-Partition --- convert all $a_i$, $b_i$, and $c_i$ to $16a_i+1$, $16b_i+4$, and $16c_i-5$, respectively. This ensures that each group in the corresponding 3-Partition instance must take exactly one $a_i$, one $b_j$, and one $c_k$, so the solutions to both instances are in bijection with each other. Thus, 3-Partition is also ASP-complete.

\section{Tetris with Two Piece Types}\label{sec:2piecetetris}

In this section, we will prove that for any size-$2$ subset $A\subseteq \{\ALL\}$, Tetris clearing with SRS is NP-hard, and the corresponding counting problem is \#P-hard, even if the sequence of pieces given to the player only contains the piece types in $A$. We will also show that some of the reductions work under the "hard drop only" model and the "20G" model. Refer to Table \ref{tab:subset} for a table of all of our results.

All of our reductions are from 3-Partition with Distinct Integers and are in the same flavor as the reduction for clearing 3-tris with rotation as given in the Total Tetris paper \cite{TotalTetris_JIP}, which we will use some terminology from. In particular, the reductions will involve a starting board involving $\frac n3$ structures, which we will call "\defn{bottles}", of equal height of $\Theta(t\cdot \text{poly}(n))$, spaced sufficiently far apart so that bottles do not interact with each other except for line clears, and possibly along with an additional structure, which we will call a "\defn{finisher}", to the right of the rightmost bottle. 

Each bottle consists of a neck portion with $n$ constant-sized "\defn{top segments}", a body portion with $t$ $\text{poly}(n)$-sized "\defn{units}", and possibly $O(n)$ extra lines either above the neck portion, between the "top segments", between the neck portion and the body portion, and/or below the body portion that get cleaned up after the rest of the lines. To simplify our arguments, we make the size of each unit larger than the size of the neck portion.

The finisher will be a structure that specifically prevents the rows in the body portion from clearing before all of the top segments are cleared, and is located in same rows as the body portion of the bottles when required. We will use three types of finishers, a $\JJ$ finisher, an $\LL$ finisher (which is a vertically symmetric version of the $\JJ$ finisher), and a $\TT$ finisher, shown in Figure \ref{Finishers}. Note that the finishers can be adapted to any number of rows larger than $4$, and there is exactly one way to clear each type of finisher.

\begin{figure}[ht]
  \centering
  \begin{subfigure}[b]{0.3\textwidth}
    \centering
    \includegraphics[width=120pt]{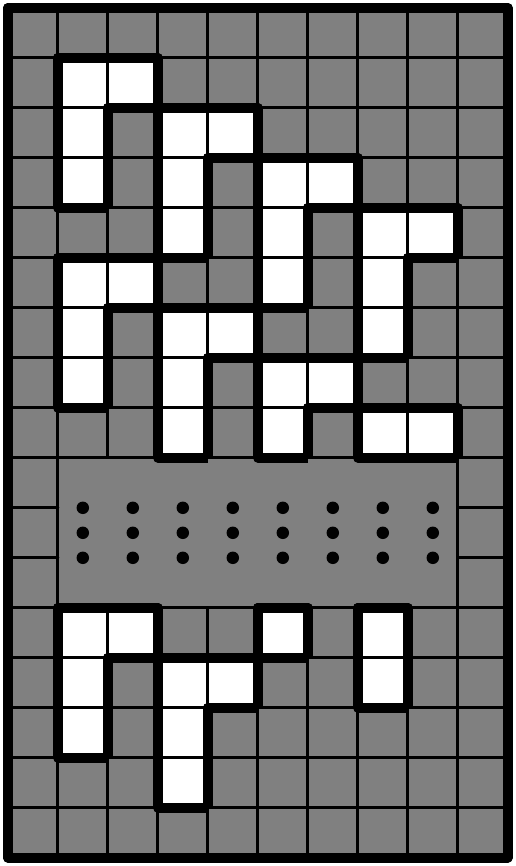}
    \caption{}
  \end{subfigure}
  \begin{subfigure}[b]{0.3\textwidth}
    \centering
    \includegraphics[width=120pt]{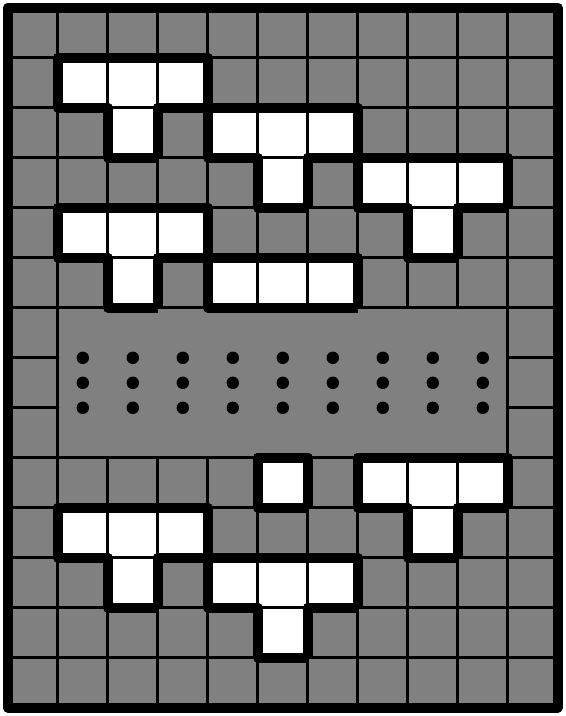}
    \caption{}
  \end{subfigure}
  \begin{subfigure}[b]{0.3\textwidth}
    \centering
    \includegraphics[width=120pt]{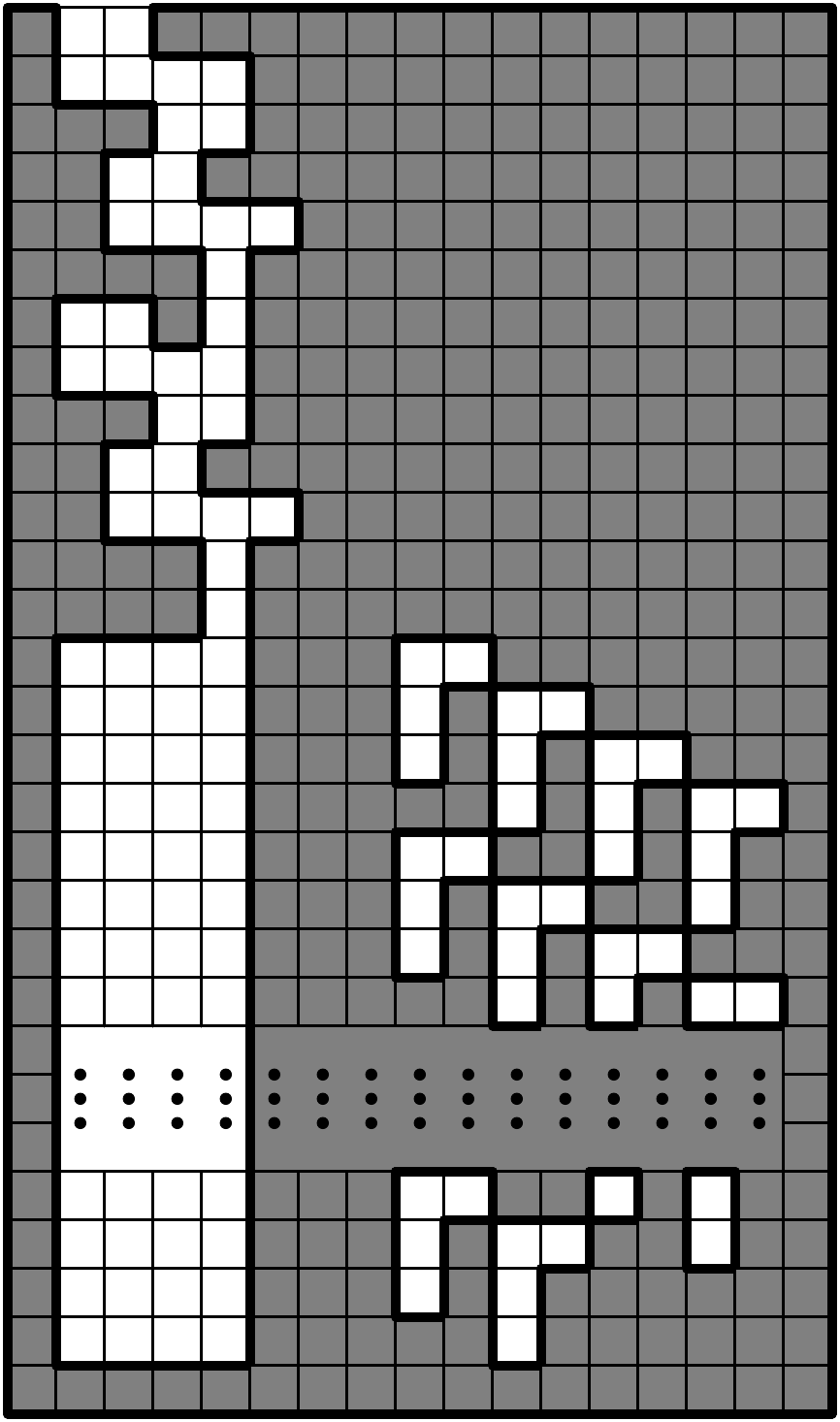}
    \caption{}
  \end{subfigure}
  \caption{The $\JJ$ and $\TT$ finishers (the $\LL$ finisher can be obtained by reflecting the $\JJ$ finisher through a vertical line), and an example of the $\JJ$ finisher next to a bottle (in the ${\OO, \JJ}$ setup).}
  \label{Finishers}
\end{figure}

For each element $a_i\in A$, we create a sequence of pieces $S_i$, which can be decomposed into three subsequences:
\begin{itemize}
    \item \defn{Priming sequence}: A piece sequence that, if used correctly, properly blocks all bottles but one in the same "top segment", and if used incorrectly, either directly "overflows" the bottle (i.e., puts blocks above the line under which all of our pieces must go) or "clogs" the bottle (i.e., improperly blocks the bottle and prevents the player from being able to clear the lines necessary to re-open the bottle). For all of the bottle structures except for the one for $\{\OO, \TT\}$, the pieces in the priming sequence cannot rotate or translate below the topmost "top segment" under SRS, and any piece placed into a "top segment" of a bottle prevents any piece in the filling sequence from reaching the body portion of that bottle.
    \item \defn{Filling sequence}: A piece sequence of length $\Theta(a_i)$ that "fills" $a_i$ units in the body portion of the unblocked bottle. If there are not enough units left to fill, then the pieces corresponding to one of the units will cause an overflow due to there not being enough empty space in the neck portion for all of the pieces (using the fact that the size of each unit is larger than the size of the neck portion).
    \item \defn{Closing sequence}: A piece sequence that properly clears the lines corresponding to the "top segment" blocked by the priming sequence and resets the states of the neck portion of the bottles (albeit with one less "top segment").
\end{itemize}

We also have a \defn{finale sequence} $F$, possibly the empty sequence, which helps clear any finishers on the board and the remaining lines on the board after the lines corresponding to the neck and body portions have been cleared.

In this section, when we write a sequence of pieces, we will use parentheses around sequences, commas between different piece types, and exponentiation to denote repeated pieces of the same piece type. For example, a sequence written as $(\TT^2, \OO^3, \SS)$ consists of 2 $\TT$s, 3 $\OO$s, and an $\SS$, in that order. The sequence of pieces given to the player will be of the form $(S_1, S_2, \ldots, S_n, F)$.

\subsection{General Argument}\label{sec:genarg}

We provide a very general argument for why these reductions work. If there exists a valid 3-partition $(D_1, \ldots, D_{n/3})$ for $\{a_1, a_2, \ldots, a_n\}$, then for each $S_i$, determine the corresponding $j_i$ such that $a_i\in D_{j_i}$, then use the priming sequence to block all bottles properly except for the $(j_i)$th one, the filling sequence to fill $a_i$ units in the body portion of the $(j_i)$th bottle, and the closing sequence to reset the states of the neck portion of the bottles. After all the $S_i$ are used in this way, all lines corresponding to the "top segments" will be cleared as there are $n$ such "top segments" with each $S_i$ clearing exactly one of them, and each bottle will be filled to exactly $t$ units. Thus, in the case where there are no pieces in the finale sequence, the lines corresponding to the body portions of the bottles will be cleared, meaning that no lines remain and we have cleared the board, and in the case where there are pieces in the finale sequence, the only lines that remain are those that can be cleared by the finale sequence. Thus, the sequence $(S_1, S_2, \ldots, S_n, F)$ can clear the board.

Conversely, if the sequence $(S_1, S_2, \ldots, S_n, F)$ can clear the board, then we claim that there is a corresponding 3-partition for $\{a_1, a_2, \ldots, a_n\}$. In particular, for each $S_i$, the priming sequence must properly block all but one bottle, say the $(j_i)$th bottle, forcing all the pieces in the filling sequence into the $(j_i)$th bottle. The filling sequence must then fill exactly $a_i$ units in the $(j_i)$th bottle before the closing sequence, and it must do so without overfilling the body portion of the bottle, as otherwise there will be an overflow in that bottle. In particular, this means that, for each $1\leq j\leq \frac n3$, the sum of the $a_i$ corresponding to the $S_j$ that filled some units in the $j$th bottle must be at most $t$. However, since $\sum a_i$ is exactly $t\left(\frac n3\right)$, the sum of the $a_i$ corresponding to the $S_j$ that filled some units in the $j$th bottle must actually be exactly $t$. In other words, there is a way to partition the $a_i$ into $\frac n3$ subsets $D_1, \ldots, D_{n/3}$ such that the sum of the elements in each subset is $t$. Thus, there is a corresponding 3-partition for $\{a_1, a_2, \ldots, a_n\}$.

This general argument shows how \textsc{yes} instances of the two problems (3-Partition with Distinct Integers, Tetris clearing with SRS and restricted piece types) are equivalent, and hence that this reduction works. The rest of the subsections in this section show how to set up the bottles for each size-$2$ subset of piece types.

\subsection{Subsets with $\II$ Pieces}\label{sec:2piecewithi}

First we show how the reduction in the Total Tetris paper \cite{TotalTetris_JIP} can be easily adapted to any subset of pieces with $\II$ pieces plus an additional piece type:

\begin{proposition}\label{prop:IL}
    Tetris clearing with SRS is NP-hard, and the corresponding counting problem is \#P-hard, even if:
    \begin{itemize}
        \item The type of pieces in the sequence given to the player is restricted to any of $\{\II, \JJ\}$, $\{\II, \LL\}$, $\{\II, \OO\}$, $\{\II, \SS\}$, $\{\II, \TT\}$, or $\{\II, \ZZ\}$;
        \item The model being considered is "hard drops only" and type of pieces in the sequence given to the player is restricted to any of $\{\II, \JJ\}$, $\{\II, \LL\}$, $\{\II, \OO\}$, or $\{\II, \TT\}$; or
        \item The model being considered is "20G" and the type of pieces in the sequence given to the player is restricted to either $\{\II, \JJ\}$ or $\{\II, \LL\}$.
    \end{itemize}
\end{proposition}

\begin{proof}
    Refer to Figures \ref{ILSetup} (which shows the bottle structure for $\{\II, \LL\}$) and \ref{I+OtherSetup} (which shows the bottle structure for the rest of the subsets). We focus on the reduction for $\{\II, \LL\}$; the rest of the reductions are similar. No finishers are required for any of the setups.

    % IL Setup Images
\begin{figure}[ht]
  \centering
  \begin{subfigure}[b]{0.1\textwidth}
    \centering
    \includegraphics[width=40pt]{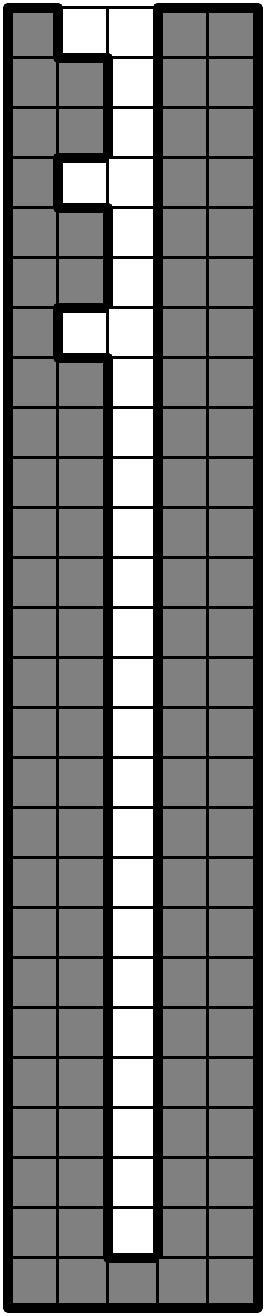}
    \caption{}
  \end{subfigure}
  \begin{subfigure}[b]{0.1\textwidth}
    \centering
    \includegraphics[width=40pt]{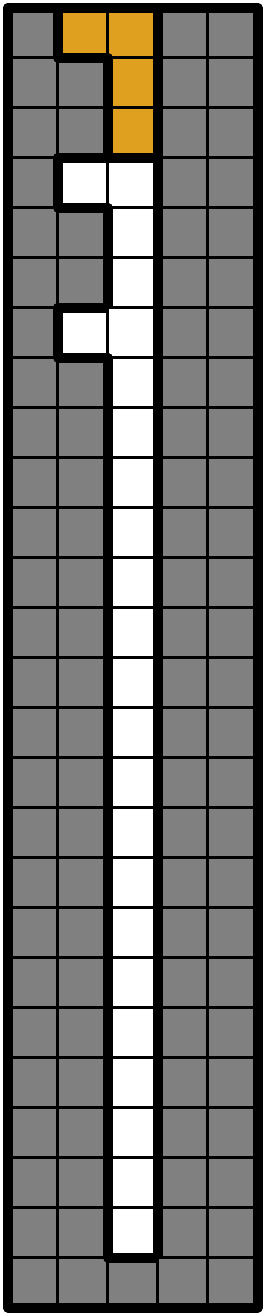}
    \caption{}
  \end{subfigure}
  \begin{subfigure}[b]{0.1\textwidth}
    \centering
    \includegraphics[width=40pt]{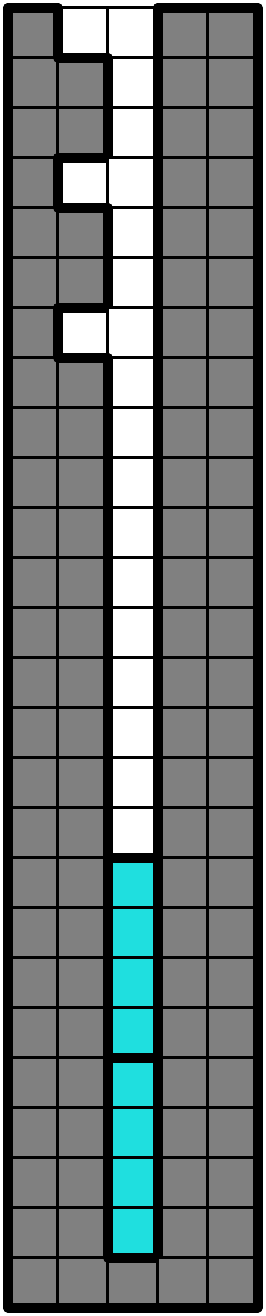}
    \caption{}
  \end{subfigure}
  \begin{subfigure}[b]{0.1\textwidth}
    \centering
    \includegraphics[width=40pt]{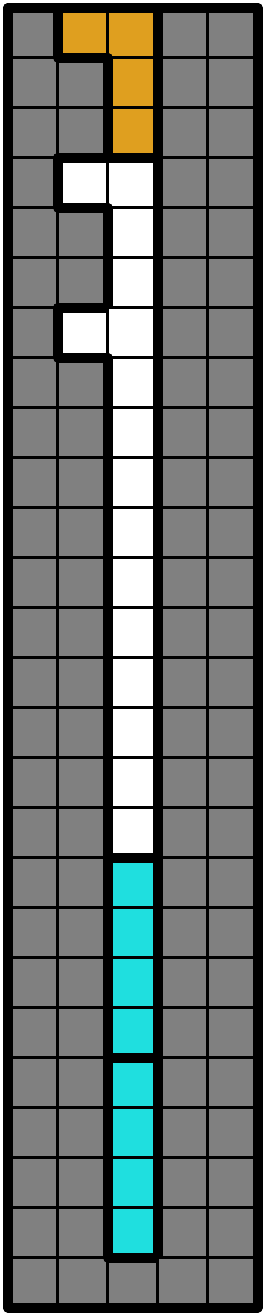}
    \caption{}
  \end{subfigure}
  \begin{subfigure}[b]{0.1\textwidth}
    \centering
    \includegraphics[width=40pt]{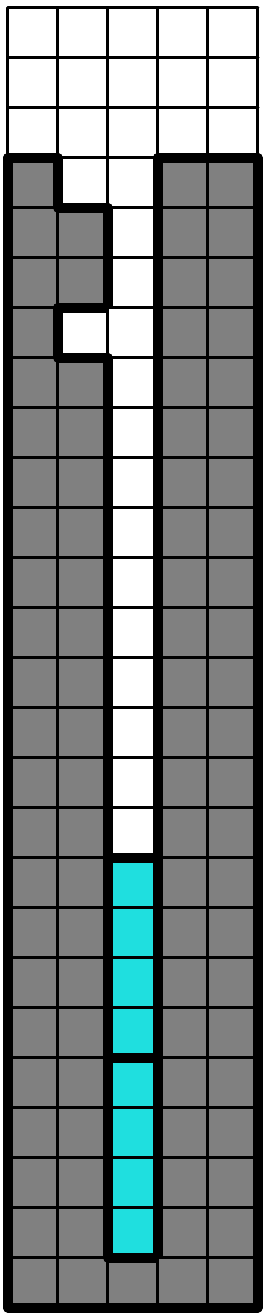}
    \caption{}
  \end{subfigure}
  \caption{The bottle structure for $\{\II, \LL\}$. (b) shows how an $\LL$ piece must block a bottle during the priming sequence. (c) shows the result of a filling sequence where $a_i=2$. (d) shows where the $\LL$ piece in the closing sequence must go. (e) shows the result of the closing sequence.}
  \label{ILSetup}
\end{figure}

    Here, the "top segments" are the $\LL$-shaped holes at the top of the bottle, and a "unit" is one $8n\times 1$ rectangle towards the bottom of the bottle (note that the size of the neck portion is $4n$, which is smaller than $8n$, the size of a "unit"); no extra lines are necessary. Each $a_i$ is encoded by the sequence of pieces $(\LL^{n/3 - 1}, \II^{2na_i}, \LL)$. The priming sequence is $(\LL^{n/3 - 1})$, the filling sequence is $(\II^{2na_i})$, and the closing sequence is $(\LL)$.
    
    Our finale sequence is empty.

    Since we have all of the required components, we can apply the general argument from Section \ref{sec:genarg} to conclude NP-hardness. 
    
    The proofs for the other piece types are similar: replace $\LL$ with $\JJ$, $\OO$, $\TT$, $\SS$, or $\ZZ$, respectively, and use the corresponding bottle structure in Figure \ref{I+OtherSetup}. Notice that the $\SS$ and $\ZZ$ reductions require the $\SS/\ZZ$ spin in Figure \ref{SSpin2} to rotate the piece into the uppermost "top segment".

    These reductions are $((\frac n3)!(\frac n3-1)!^n)$-monious:%
    \footnote{We call a reduction \defn{$c$-monious} if it blows up the number of solutions by exactly a multiplicative factor of~$c$.}
    for each solution to the 3-Partition with Distinct Integers instance, there are $(\frac n3)!$ ways to "permute" which subsets correspond to which bottles, and for each $a_i$ sequence, there are $(\frac n3-1)!$ ways to permute which piece in the priming sequence blocks which bottle. As such, we can also conclude \#P-hardness.

    % I+other Setup Images
\begin{figure}[ht]
  \centering
  \begin{subfigure}[b]{0.18\textwidth}
    \centering
    \includegraphics[width=40pt]{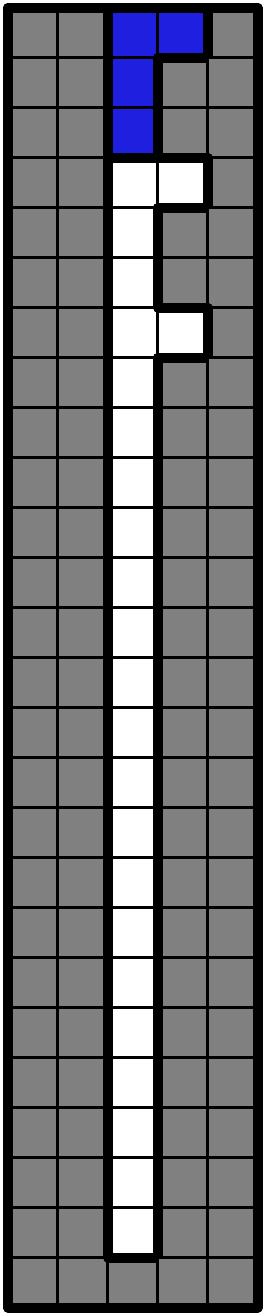}
    \caption{$\{\II, \JJ\}$}
  \end{subfigure}
  \begin{subfigure}[b]{0.18\textwidth}
    \centering
    \includegraphics[width=40pt]{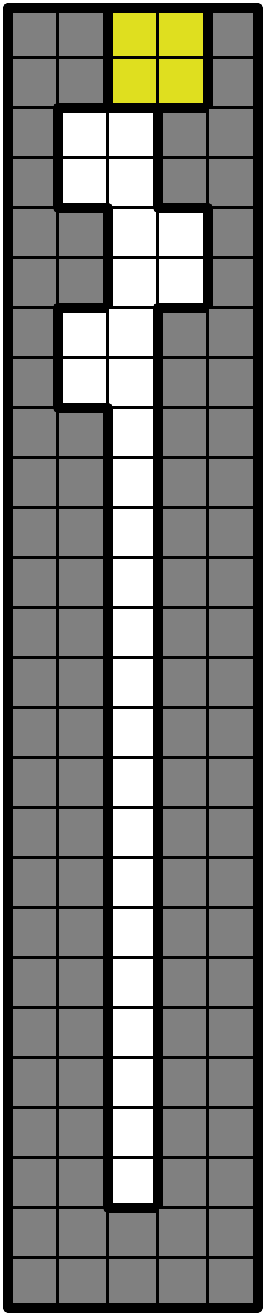}
    \caption{$\{\II, \OO\}$}
  \end{subfigure}
  \begin{subfigure}[b]{0.18\textwidth}
    \centering
    \includegraphics[width=40pt]{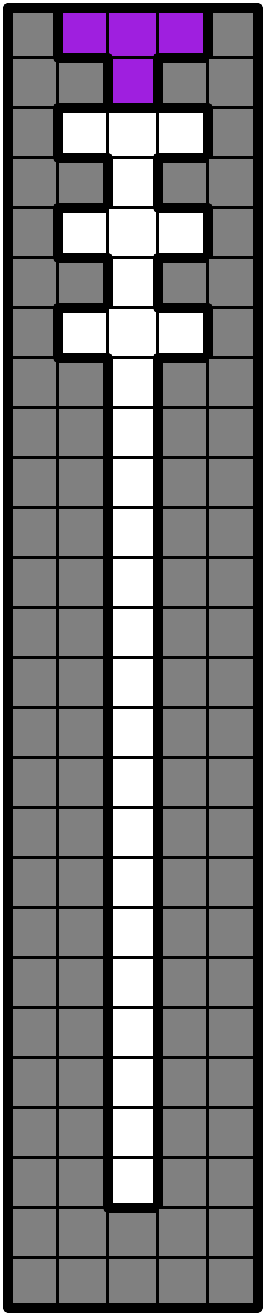}
    \caption{$\{\II, \TT\}$}
  \end{subfigure}
  \begin{subfigure}[b]{0.18\textwidth}
    \centering
    \includegraphics[width=40pt]{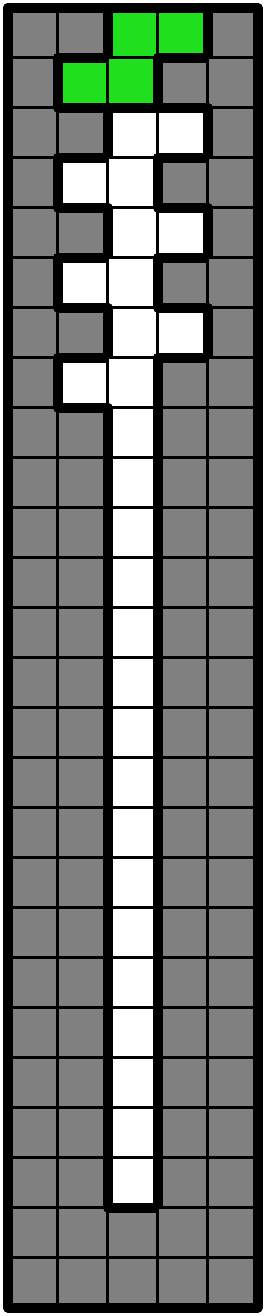}
    \caption{$\{\II, \SS\}$}
  \end{subfigure}
  \begin{subfigure}[b]{0.18\textwidth}
    \centering
    \includegraphics[width=40pt]{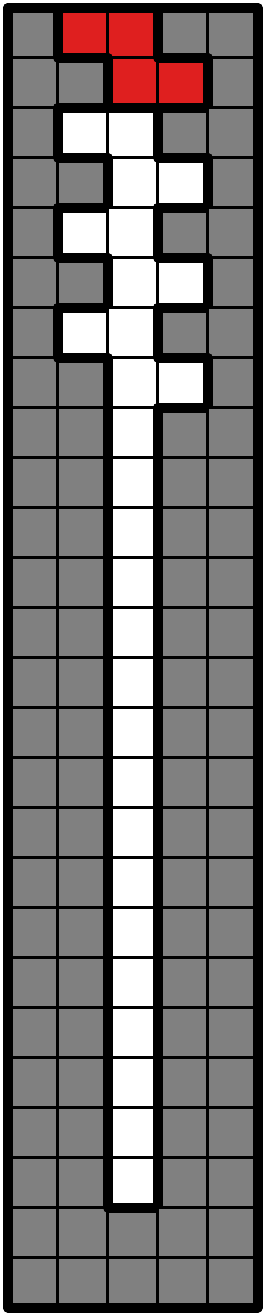}
    \caption{$\{\II, \ZZ\}$}
  \end{subfigure}
  \caption{The bottle structures for the other subsets containing $\II$, including how the non-$\II$ piece must block a bottle during the priming and closing sequence.}
  \label{I+OtherSetup}
\end{figure}

    Notice that the $\{\II, \JJ\}$, $\{\II, \LL\}$, $\{\II, \OO\}$, and $\{\II, \TT\}$ set-ups do not require any last-second rotations (and in fact do not require the kicking system of SRS), and so all pieces can be hard-dropped from the top of the board for the $\{\II, \JJ\}$, $\{\II, \LL\}$, $\{\II, \OO\}$, and $\{\II, \TT\}$ set-ups. 
    
    In addition, because all buckets are two units wide in the $\{\II, \JJ\}$ and $\{\II, \LL\}$ set-ups and the $\II$, $\JJ$, and $\LL$ pieces are at least three units wide when oriented horizontally, the $\{\II, \JJ\}$ and $\{\II, \LL\}$ set-ups work even if pieces experience 20G gravity: when any piece drops, orient the piece horizontally, move it left or right to the desired bucket, and rotate accordingly.
\end{proof}

\subsection{Other Subsets with $\OO$ Pieces}\label{sec:2piecewitho}

We now move on to all the remaining subsets which contain $\OO$ pieces.

\begin{proposition}\label{prop:OJ}
    Tetris clearing with SRS is NP-hard, and the corresponding counting problem is \#P-hard, even if the type of pieces in the sequence given to the player is restricted to either $\{\OO, \JJ\}$ or $\{\OO, \LL\}$.
\end{proposition}

\begin{proof}
    First we discuss the $\{\OO, \JJ\}$ case. Refer to Figure \ref{OJSetup}(a), which shows the bottle structure for $\{\OO, \JJ\}$. We will also use a $\JJ$ finisher in our setup to prevent rows in the body portion from clearing early.

    % OJ Setup Images
\begin{figure}[ht]
  \centering
  \begin{subfigure}[b]{0.1\textwidth}
    \centering
    \includegraphics[width=40pt]{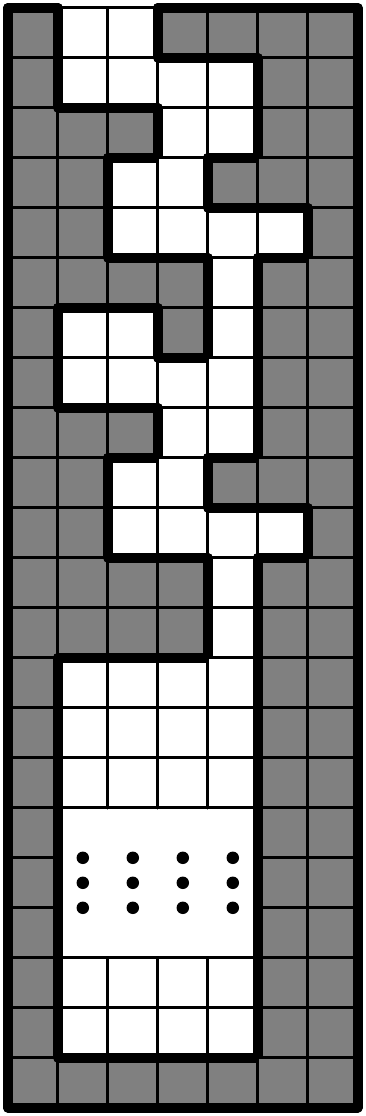}
    \caption{}
  \end{subfigure}
  \begin{subfigure}[b]{0.1\textwidth}
    \centering
    \includegraphics[width=40pt]{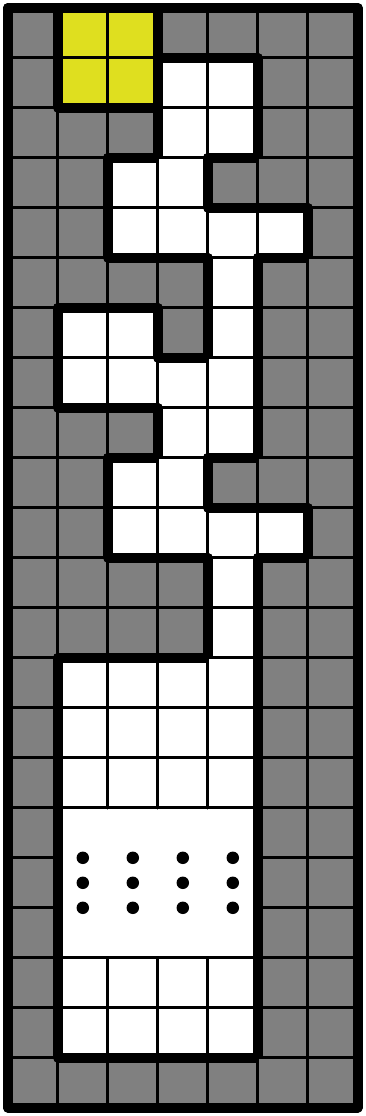}
    \caption{}
  \end{subfigure}
  \begin{subfigure}[b]{0.1\textwidth}
    \centering
    \includegraphics[width=40pt]{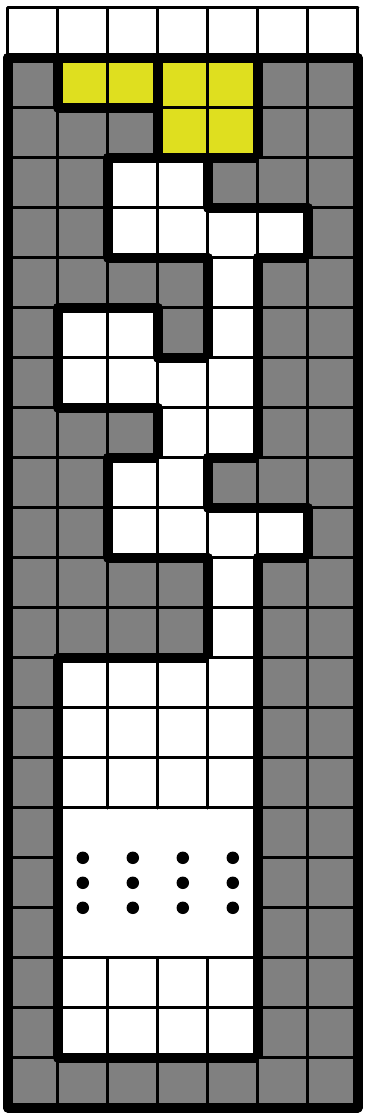}
    \caption{}
  \end{subfigure}
  \begin{subfigure}[b]{0.1\textwidth}
    \centering
    \includegraphics[width=40pt]{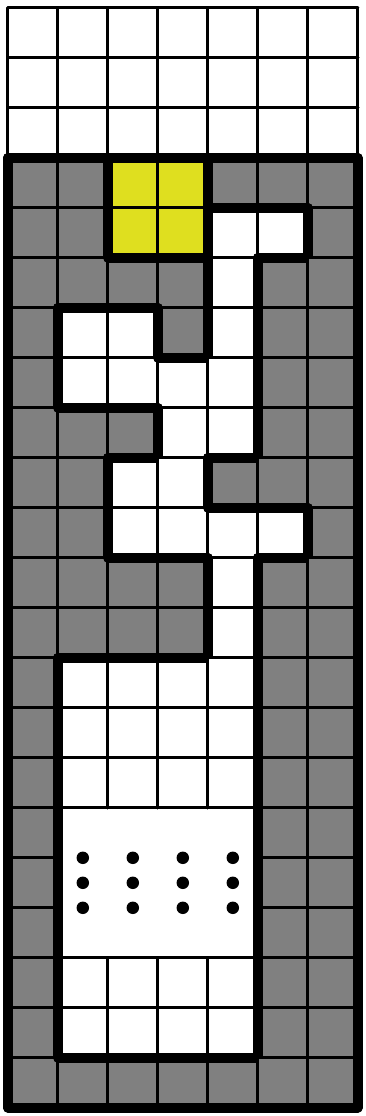}
    \caption{}
  \end{subfigure}
  \begin{subfigure}[b]{0.1\textwidth}
    \centering
    \includegraphics[width=40pt]{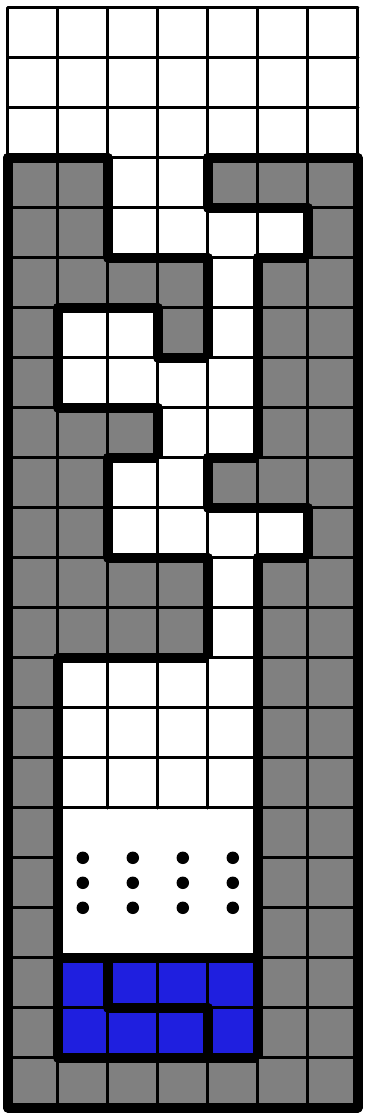}
    \caption{}
  \end{subfigure}
  \begin{subfigure}[b]{0.1\textwidth}
    \centering
    \includegraphics[width=40pt]{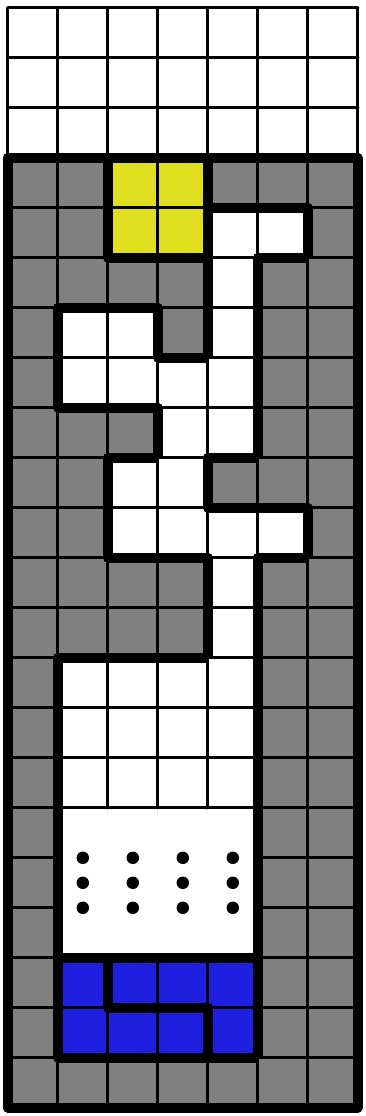}
    \caption{}
  \end{subfigure}
  \begin{subfigure}[b]{0.1\textwidth}
    \centering
    \includegraphics[width=40pt]{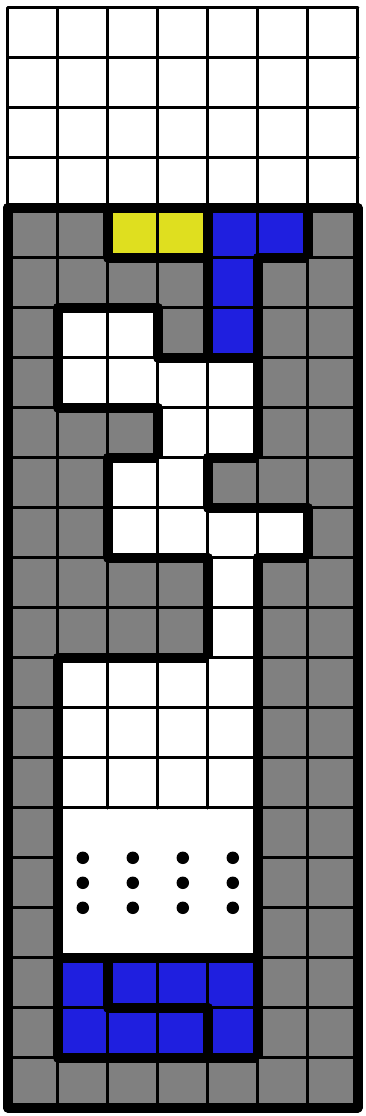}
    \caption{}
  \end{subfigure}
  \begin{subfigure}[b]{0.1\textwidth}
    \centering
    \includegraphics[width=40pt]{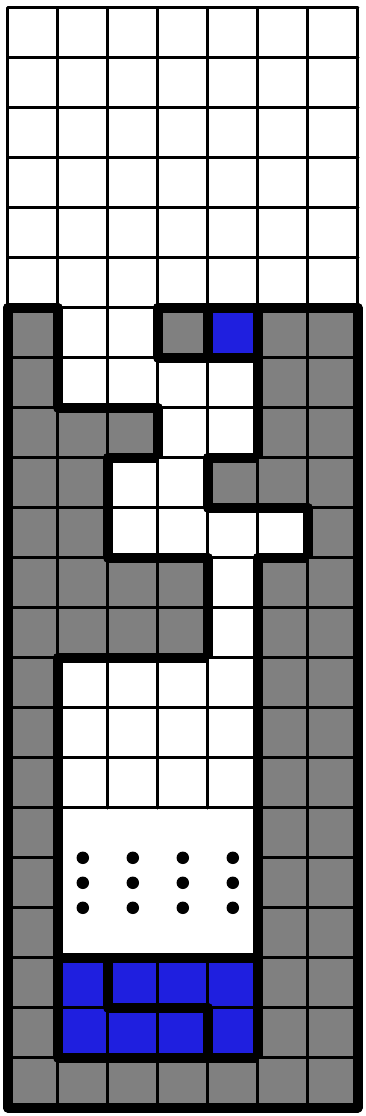}
    \caption{}
  \end{subfigure}
  \caption{The bottle structure for $\{\OO, \JJ\}$. (b--c) show how the first two $\OO$ pieces per bottle must fit during the priming sequence. (d) shows how the $\OO$ pieces must block a bottle during the priming sequence. (e) shows the result of a filling sequence (requires the $\JJ$-spin sequence in Figure \ref{JSpin-long1} to get $\JJ$ pieces through the neck portion and the $\JJ$ spin in Figure \ref{JSpin4} to get $\JJ$ pieces into the body portion). (f--g) show where the pieces in the closing sequence must go. (h) shows the result of the closing sequence.}
  \label{OJSetup}
\end{figure}

    Here, a "top segment" is a pattern of $16$ squares repeating every $7$ lines at the top of the bottle, and a "unit" is one $8n\times 4$ rectangle towards the bottom of the bottle (note that the size of the neck portion is $16n$, which is smaller than $32n$, the size of a "unit"). 
    
    Each $a_i$ is encoded by the sequence of pieces $(\OO^{n-1}, \JJ^{8na_i}, \OO, \JJ^{n/3})$. The priming sequence is $(\OO^{n-1})$, the filling sequence is $(\JJ^{8na_i})$, and the closing sequence is $(\OO, \JJ^{n/3})$, with the ways to block, fill, and re-open a bottle properly shown in Figure \ref{OJSetup}. Discussion on improper piece placements can be found in Appendix \ref{appendix:ojclogs}.

    The finale sequence is $(\JJ^{8nt-2})$ to clear the $\JJ$ finisher.

    Since we have all of the required components, we can apply the general argument from Section \ref{sec:genarg} to conclude NP-hardness.

    Now, for our \#P-hardness analysis, let $f_4(m)$ denote the number of ways to place $m$ $\JJ$ pieces in an rectangle of width $4$ and height $m$ to fill the rectangle exactly, using Tetris rules. This reduction is $\left((\frac n3)!\left(((\frac n3)!)^3(\frac n3-1)!\right)^n(f_4(8nt))^{n/3}\right)$-monious, as for each solution to the 3-Partition with Distinct Integers instance, there are $(\frac n3)!$ ways to "permute" which subsets correspond to which bottles, $((\frac n3)!)^2(\frac n3-1)!$ ways to permute how the pieces in the priming sequence get placed for each $a_i$ sequence (the first $\frac n3$ pieces must be placed as in Figure \ref{OJSetup}(b), then the next $\frac n3$ pieces must be placed as in Figure \ref{OJSetup}(c), before we can reach Figures \ref{OJSetup}(d--f)), $(\frac n3)!$ ways to permute how the $\JJ$ pieces in the closing sequence get placed for each $a_i$ sequence, and $f_4(8nt)$ ways to permute how the $\JJ$ pieces get placed within the body portion of each bottle. As such, we can also conclude \#P-hardness.

    We get a similar argument for $\{\OO, \LL\}$ by vertical symmetry.
\end{proof}

A demo of the $\{\OO, \JJ\}$ bottle structure can be found at \url{https://jstris.jezevec10.com/map/80188}.

\begin{proposition}\label{prop:OS}
    Tetris clearing with SRS is NP-hard, and the corresponding counting problem is \#P-hard, even if the type of pieces in the sequence given to the player is restricted to either $\{\OO, \SS\}$ or $\{\OO, \ZZ\}$.
\end{proposition}

\begin{proof}

    First we discuss the $\{\OO, \SS\}$ case. Refer to Figure \ref{OSSetup}(a), which shows the bottle structure for $\{\OO, \SS\}$. We do not use a finisher in our setup.

    % OS Setup Images
\begin{figure}[ht]
  \centering
  \begin{subfigure}[b]{0.09\textwidth}
    \centering
    \includegraphics[width=40pt]{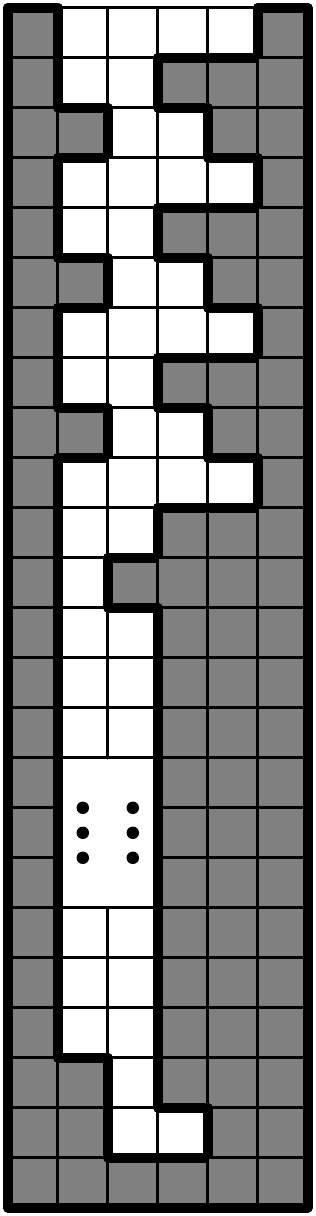}
    \caption{}
  \end{subfigure}
  \begin{subfigure}[b]{0.09\textwidth}
    \centering
    \includegraphics[width=40pt]{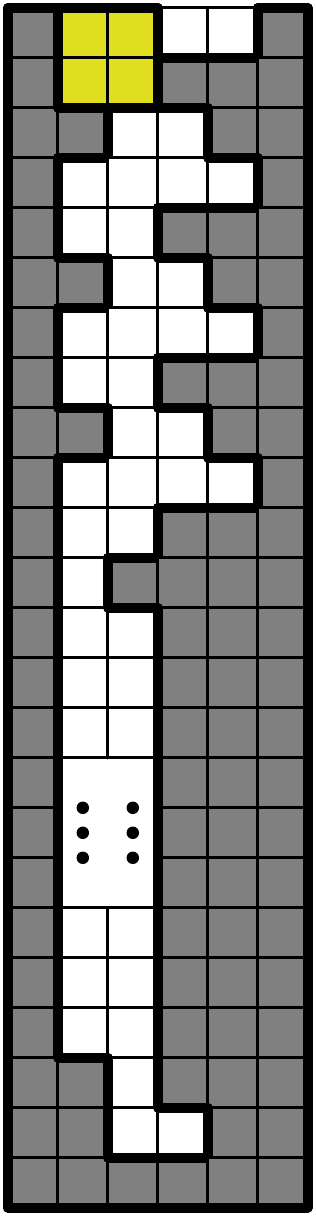}
    \caption{}
  \end{subfigure}
  \begin{subfigure}[b]{0.09\textwidth}
    \centering
    \includegraphics[width=40pt]{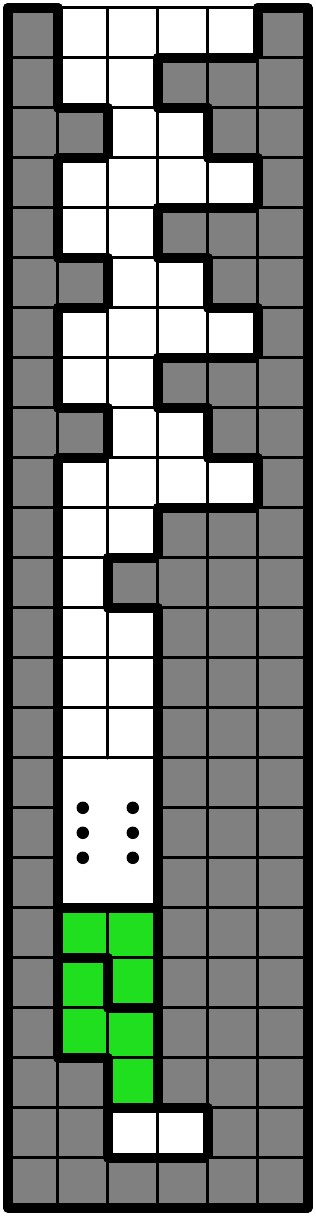}
    \caption{}
  \end{subfigure}
  \begin{subfigure}[b]{0.09\textwidth}
    \centering
    \includegraphics[width=40pt]{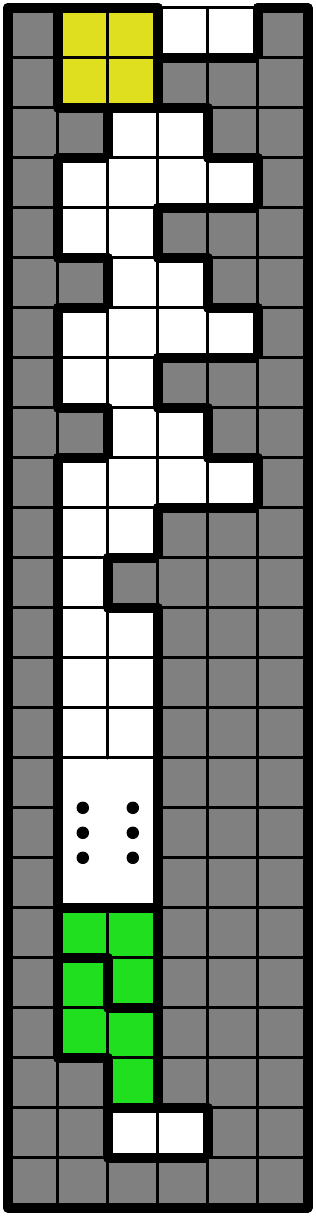}
    \caption{}
  \end{subfigure}
  \begin{subfigure}[b]{0.09\textwidth}
    \centering
    \includegraphics[width=40pt]{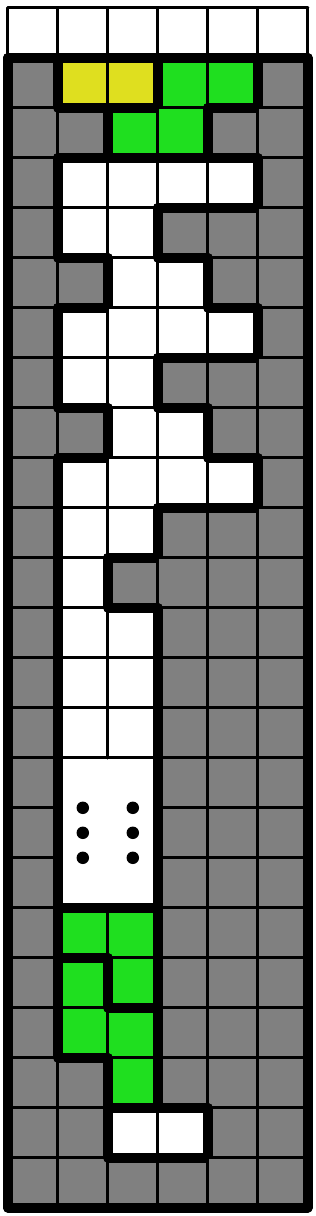}
    \caption{}
  \end{subfigure}
  \begin{subfigure}[b]{0.09\textwidth}
    \centering
    \includegraphics[width=40pt]{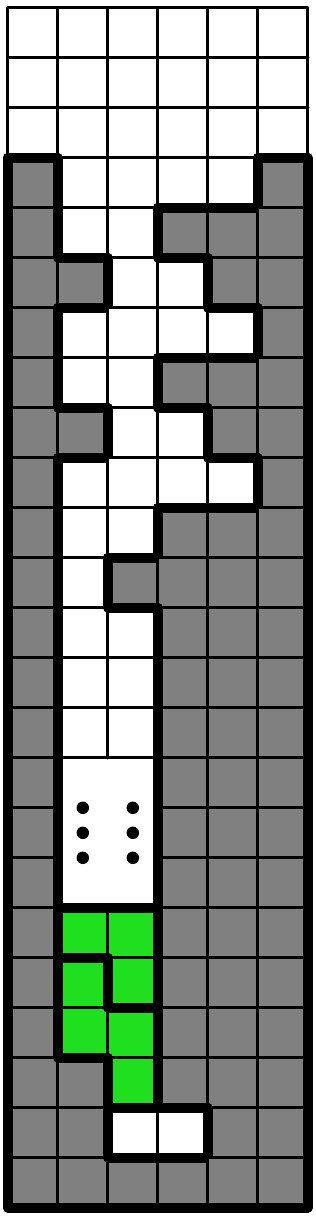}
    \caption{}
  \end{subfigure}
  \begin{subfigure}[b]{0.09\textwidth}
    \centering
    \includegraphics[width=40pt]{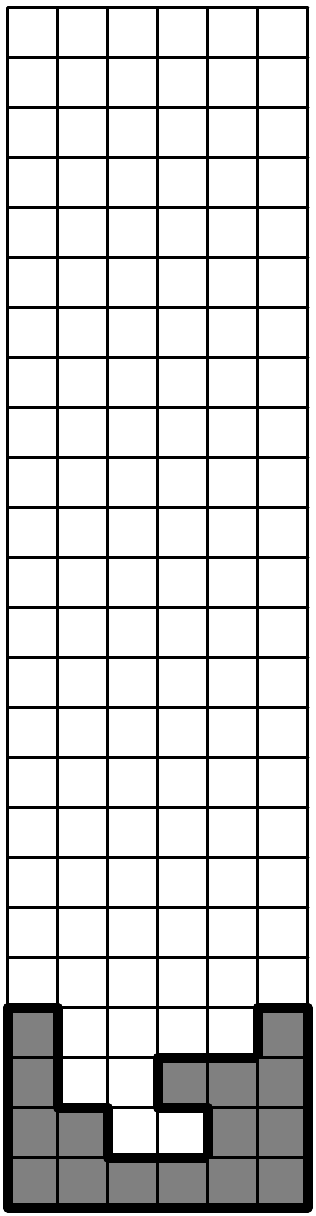}
    \caption{}
  \end{subfigure}
  \caption{The bottle structure for $\{\OO, \SS\}$. (b) shows how the $\OO$ pieces must block a bottle during the priming sequence. (c) shows the result of a filling sequence (requires the $\SS$-spin sequence in Figure \ref{SSpin-long1} to get $\SS$ pieces through the neck portion and the $\SS$ spin in Figure \ref{SSpin1} to get $\SS$ pieces into the body portion). (d--e) show where the pieces in the closing sequence must go (requires the $\SS$ spin in Figure \ref{SSpin2}). (f) shows the result of the closing sequence. (g) shows the state of the board immediately before the finale sequence.}
  \label{OSSetup}
\end{figure}

    Here, a "top segment" is the pattern that repeats every 3 lines at the top of the bottle, and a "unit" is a region of area $16n$ consisting of $4n$ $\SS$-shaped pieces stacked on top of each other. Three extra lines are necessary, corresponding to an extra "top segment". Note that the size of the neck portion plus the area between the neck and body portion is $8n+6$, which is smaller than $16n$, the size of a "unit".
    
    Each $a_i$ is encoded by the sequence of pieces $(\OO^{n/3 - 1}, \SS^{4na_i}, \OO, \SS^{n/3})$. The priming sequence is $(\OO^{n/3 - 1})$, the filling sequence is $(\SS^{4na_i})$, and the closing sequence is $(\OO, \SS^{n/3})$, with the ways to block, fill, and re-open a bottle properly shown in Figure \ref{OSSetup}. Discussion on improper piece placements can be found in Appendix \ref{appendix:osclogs}.

    The finale sequence is $(\OO^{n/3}, \SS^{n/3})$; the remaining lines collapse down to another copy of the "top segment" (see Figure \ref{OSSetup}(g)) and can be cleared in the same fashion as in the priming sequence and the closing sequence.

    Since we have all of the required components, we can apply the general argument from Section \ref{sec:genarg} to conclude NP-hardness.

    This reduction is $\left((\frac n3)!\left((\frac n3-1)!(\frac n3)!\right)^n((\frac n3)!)^2\right)$-monious: for each solution to the 3-Partition with Distinct Integers instance, there are $(\frac n3)!$ ways to "permute" which subsets correspond to which bottles, $(\frac n3-1)!$ ways to permute how the pieces in the priming sequence get placed for each $a_i$ sequence, $(\frac n3)!$ ways to permute how the $\SS$ pieces in the closing sequence get placed for each $a_i$ sequence, and $((\frac n3)!)^2$ ways to permute how the pieces in the finale sequence get placed. As such, we can also conclude \#P-hardness.

    We get a similar argument for $\{\OO, \ZZ\}$ by vertical symmetry.
\end{proof}

A demo of the $\{\OO, \SS\}$ bottle structure can be found at \url{https://jstris.jezevec10.com/map/81818}.

\begin{proposition}\label{prop:OT}
    Tetris clearing with SRS is NP-hard, and the corresponding counting problem is \#P-hard, even if the type of pieces in the sequence given to the player is restricted to $\{\OO, \TT\}$.
\end{proposition}

\begin{proof}
    Refer to Figure \ref{OTSetup}(a), which shows the bottle structure for $\{\OO, \TT\}$. We will also use a $\TT$ finisher in our setup to prevent rows in the body portion from clearing early.

    % OT Setup Images
\begin{figure}[ht]
  \centering
  \begin{subfigure}[b]{0.1\textwidth}
    \centering
    \includegraphics[width=40pt]{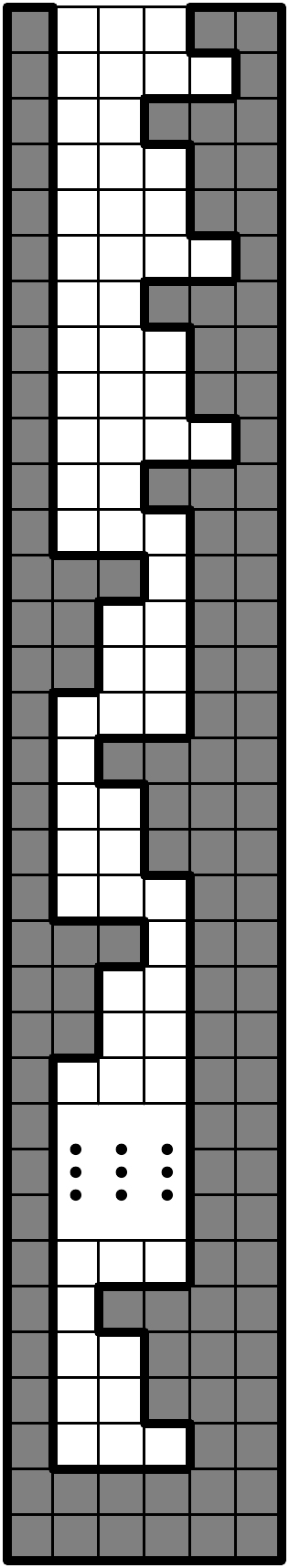}
    \caption{}
  \end{subfigure}
  \begin{subfigure}[b]{0.1\textwidth}
    \centering
    \includegraphics[width=40pt]{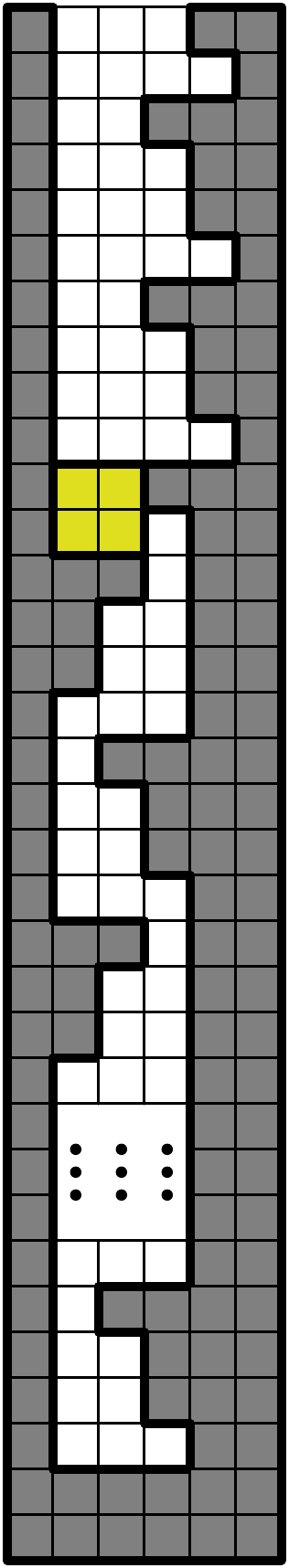}
    \caption{}
  \end{subfigure}
  \begin{subfigure}[b]{0.1\textwidth}
    \centering
    \includegraphics[width=40pt]{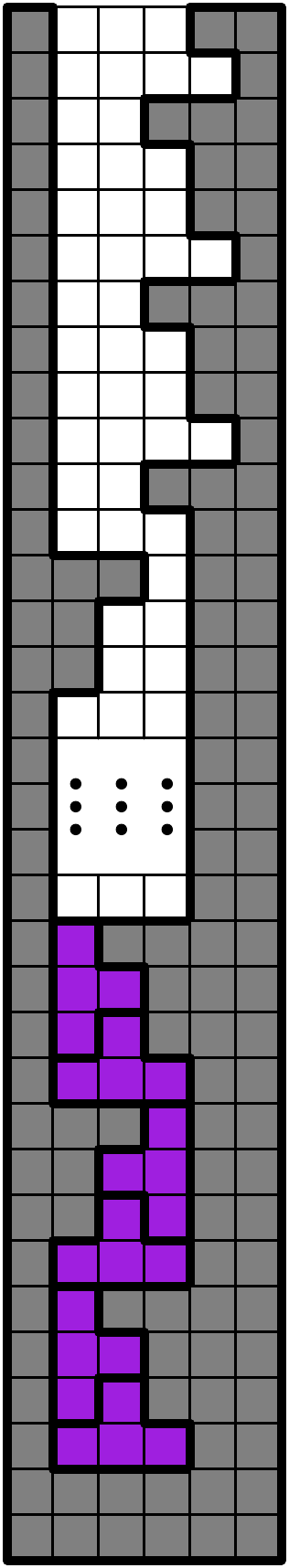}
    \caption{}
  \end{subfigure}
  \begin{subfigure}[b]{0.1\textwidth}
    \centering
    \includegraphics[width=40pt]{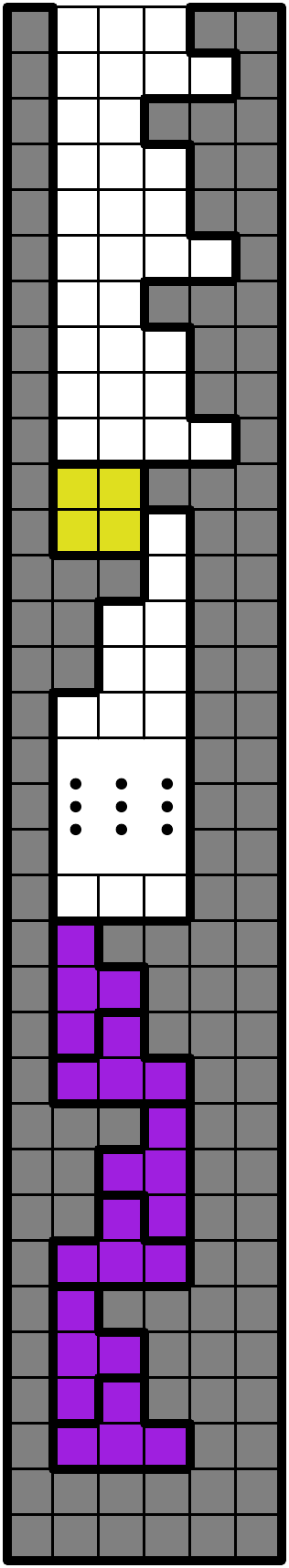}
    \caption{}
  \end{subfigure}
  \begin{subfigure}[b]{0.1\textwidth}
    \centering
    \includegraphics[width=40pt]{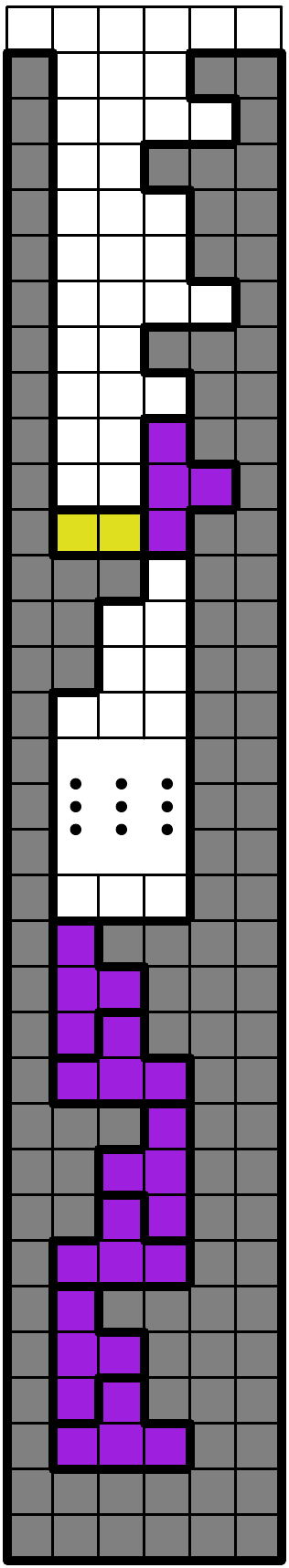}
    \caption{}
  \end{subfigure}
  \begin{subfigure}[b]{0.1\textwidth}
    \centering
    \includegraphics[width=40pt]{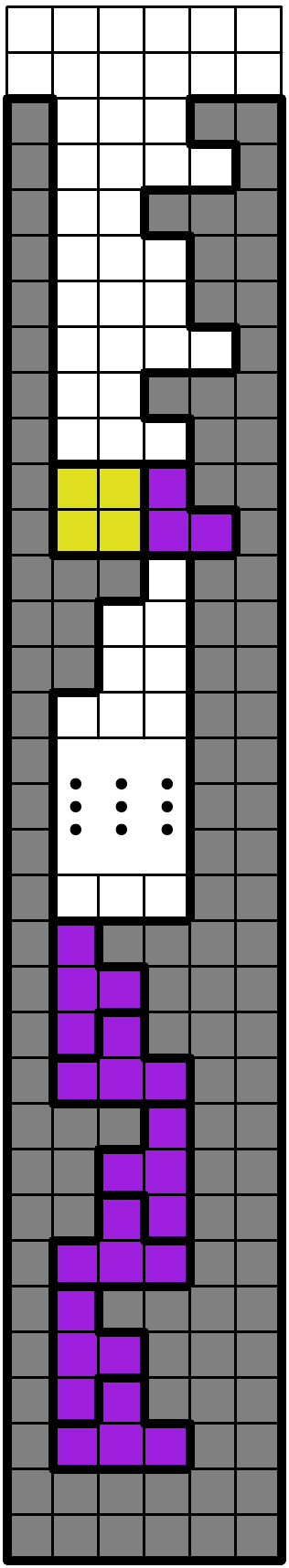}
    \caption{}
  \end{subfigure}
  \begin{subfigure}[b]{0.1\textwidth}
    \centering
    \includegraphics[width=40pt]{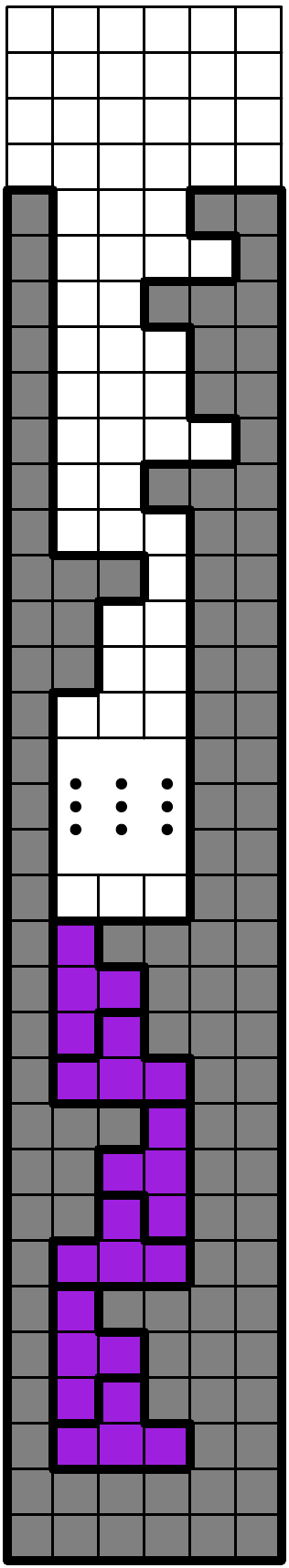}
    \caption{}
  \end{subfigure}
  \caption{The bottle structure for $\{\OO, \TT\}$. (b) shows how the $\OO$ piece must block a bottle during the priming sequence. (c) shows the result of a filling sequence (requires the $\TT$-spin sequence in Figure \ref{TSpin} to get $\TT$ pieces through the body portion). (d--f) show where the pieces in the closing sequence must go. (g) shows the result of the closing sequence.}
  \label{OTSetup}
\end{figure}

    Here, a "top segment" is the pattern that repeats every 4 lines at the top of the bottle, and a "unit" is a region of area $16n$ consisting of the $8$-line pattern right underneath the lowermost "top segment" in Figure \ref{OTSetup}(a) repeated $n$ times (note that the size of the neck portion is $12n$, which is smaller than $16n$, the size of a "unit"). No extra lines are necessary. 
    
    Each $a_i$ is encoded by the sequence of pieces $(\OO^{n/3 - 1}, \TT^{4na_i}, \OO, \TT^{n/3}, \OO^{n/3})$. The priming sequence is $(\OO^{n/3 - 1})$, the filling sequence is $(\TT^{4na_i})$, and the closing sequence is $(\OO, \TT^{n/3}, \OO^{n/3})$, with the ways to block, fill, and re-open a bottle properly shown in Figure \ref{OTSetup}. Discussion on improper piece placements can be found in Appendix \ref{appendix:otclogs}.

    The finale sequence is $(\TT^{8nt-2})$ to clear the $\TT$ finisher.

    Since we have all of the required components, we can apply the general argument from Section \ref{sec:genarg} to conclude NP-hardness.

    This reduction is $\left((\frac n3)!((\frac n3-1)!((\frac n3)!)^2)^n\right)$-monious: for each solution to the 3-Partition with Distinct Integers instance, there are $(\frac n3)!$ ways to "permute" which subsets correspond to which bottles, $(\frac n3-1)!$ ways to permute how the pieces in the priming sequence get placed for each $a_i$ sequence, and $((\frac n3)!)^2$ ways to permute how the $\TT$ pieces and $\OO$ pieces in the closing sequence get placed for each $a_i$ sequence. As such, we can also conclude \#P-hardness.
\end{proof}

A demo of the $\{\OO, \TT\}$ bottle structure can be found at \url{https://jstris.jezevec10.com/map/80169}.

\subsection{Two-Element Subsets of $\{\SS, \TT, \ZZ\}$}\label{sec:stz}

\begin{proposition}\label{prop:ST}
    Tetris clearing with SRS is NP-hard, and the corresponding counting problem is \#P-hard, even if the type of pieces in the sequence given to the player is restricted to either $\{\SS, \TT\}$ or $\{\TT, \ZZ\}$.
\end{proposition}

\begin{proof}
    First we discuss the $\{\SS, \TT\}$ case. Refer to Figure \ref{STSetup}(a), which shows the bottle structure for $\{\SS, \TT\}$. We do not use a finisher in our setup.

    % ST Setup Images
\begin{figure}[ht]
  \centering
  \begin{subfigure}[b]{0.09\textwidth}
    \centering
    \includegraphics[width=40pt]{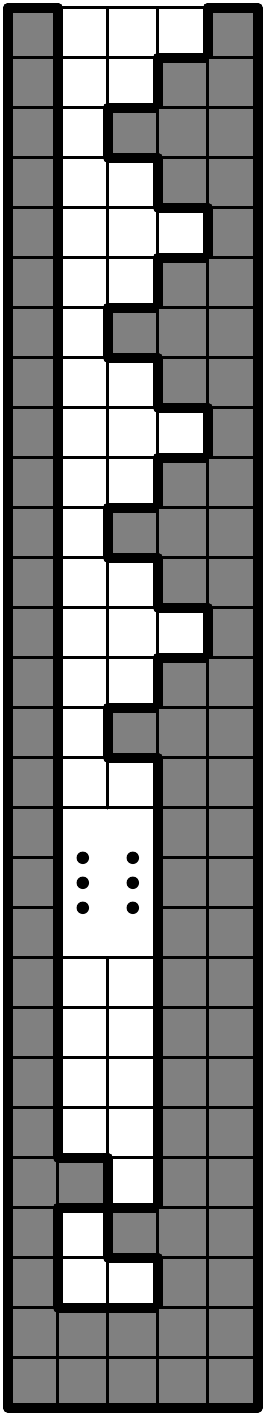}
    \caption{}
  \end{subfigure}
  \begin{subfigure}[b]{0.09\textwidth}
    \centering
    \includegraphics[width=40pt]{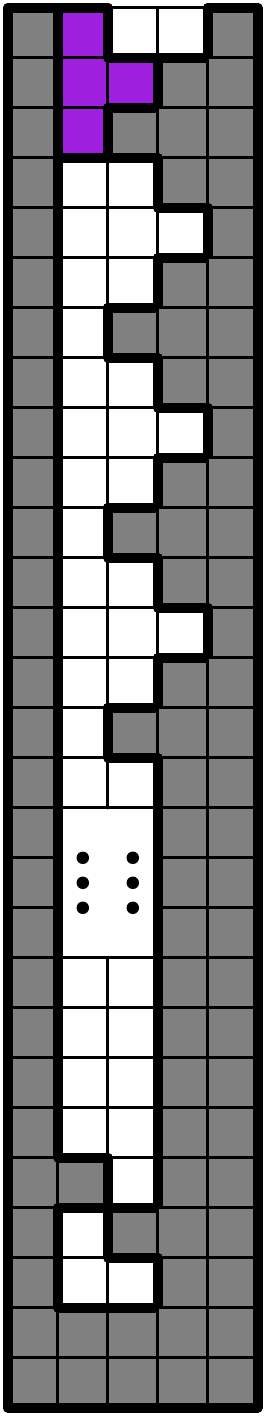}
    \caption{}
  \end{subfigure}
  \begin{subfigure}[b]{0.09\textwidth}
    \centering
    \includegraphics[width=40pt]{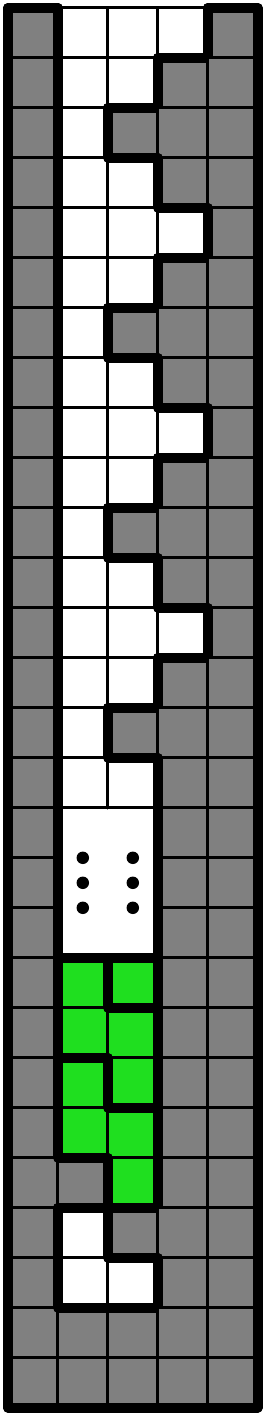}
    \caption{}
  \end{subfigure}
  \begin{subfigure}[b]{0.09\textwidth}
    \centering
    \includegraphics[width=40pt]{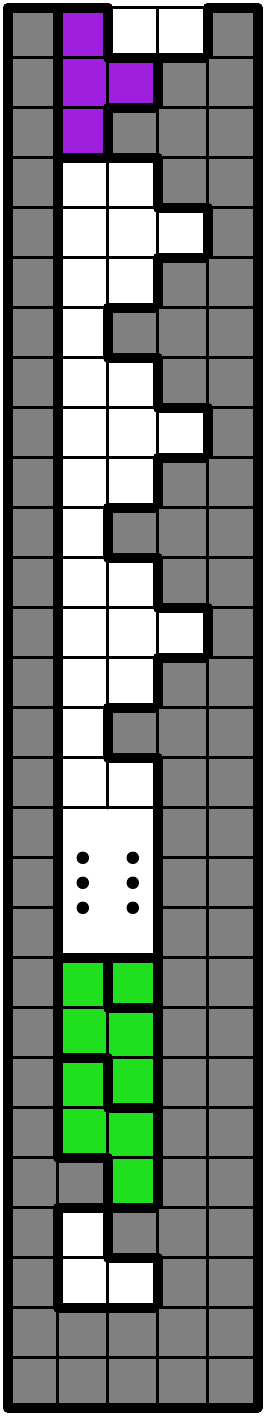}
    \caption{}
  \end{subfigure}
  \begin{subfigure}[b]{0.09\textwidth}
    \centering
    \includegraphics[width=40pt]{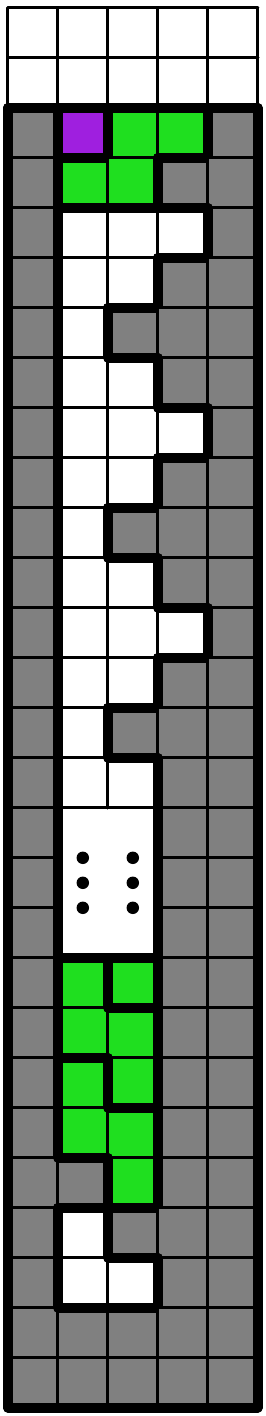}
    \caption{}
  \end{subfigure}
  \begin{subfigure}[b]{0.09\textwidth}
    \centering
    \includegraphics[width=40pt]{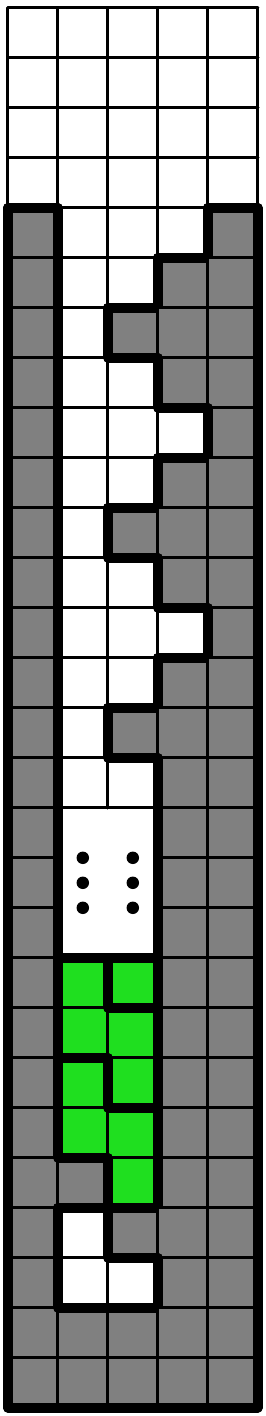}
    \caption{}
  \end{subfigure}
  \begin{subfigure}[b]{0.09\textwidth}
    \centering
    \includegraphics[width=40pt]{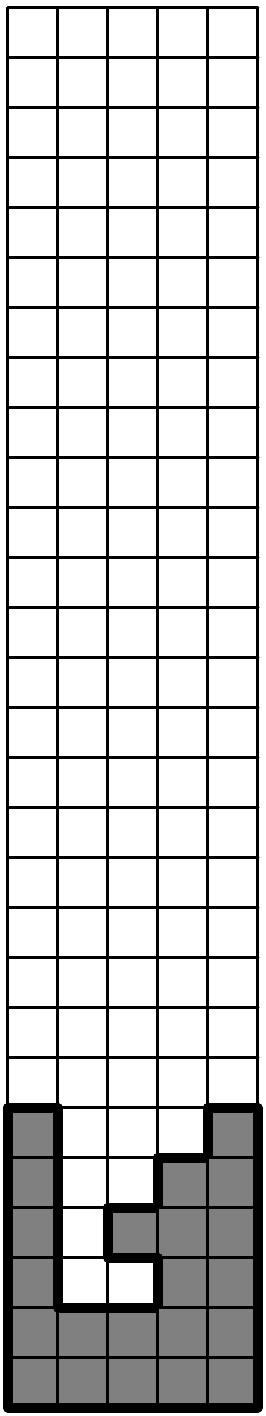}
    \caption{}
  \end{subfigure}
  \caption{The bottle structure for $\{\SS, \TT\}$. (b) shows how the $\TT$ pieces must block a bottle during the priming sequence. (c) shows the result of a filling sequence (requires the $\SS$-spin sequence in Figure \ref{SSpin-long2} to get $\SS$ pieces through the neck portion and the $\SS$ spin in Figure \ref{SSpin1} to get $\SS$ pieces into the body portion). (d--e) show where the pieces in the closing sequence must go (requires the $\SS$ spin in Figure \ref{SSpin2}). (f) shows the result of the closing sequence. (g) shows the state of the board immediately before the finale sequence.}
  \label{STSetup}
\end{figure}

    Here, a "top segment" is the pattern that repeats every 4 lines at the top of the bottle, and a "unit" is a region of area $16n$ consisting of $4n$ $\SS$-shaped pieces stacked on top of each other. Four extra lines are used: two between the neck portion and the body portion, and two under the body portion. Note that the size of the neck portion plus the area between the neck and body portion is $8n+5$, which is smaller than $16n$, the size of a "unit".
    
    Each $a_i$ is encoded by the sequence of pieces $(\TT^{n/3 - 1}, \SS^{4na_i}, \TT, \SS^{n/3})$. The priming sequence is $(\TT^{n/3 - 1})$, the filling sequence is $(\SS^{4na_i})$, and the closing sequence is $(\TT, \SS^{n/3})$, with the ways to block, fill, and re-open a bottle properly shown in Figure \ref{STSetup}. Discussion on improper piece placements can be found in Appendix \ref{appendix:stclogs}.

    The finale sequence is $(\TT^{n/3}, \SS^{n/3})$; the remaining lines collapse down to another copy of the "top segment" (see Figure \ref{STSetup}(g)) and can be cleared in the same fashion as described in the priming sequence and the closing sequence.

    Since we have all of the required components, we can apply the general argument from Section \ref{sec:genarg} to conclude NP-hardness.

    This reduction is $\left((\frac n3)!((\frac n3-1)!)^n((\frac n3)!)^{n+2}\right)$-monious: for each solution to the 3-Partition with Distinct Integers instance, there are $(\frac n3)!$ ways to "permute" which subsets correspond to which bottles, $(\frac n3-1)!$ ways to permute how the pieces in the priming sequence get placed for each $a_i$ sequence, $(\frac n3)!$ ways to permute how the $\SS$ pieces in the closing sequence get placed for each $a_i$ sequence, and $((\frac n3)!)^2$ ways to permute how the pieces in the finale sequence get placed. As such, we can also conclude \#P-hardness.

    We get a similar argument for $\{\TT, \ZZ\}$ by vertical symmetry.
\end{proof}

A demo of the $\{\SS, \TT\}$ bottle structure can be found at \url{https://jstris.jezevec10.com/map/80184}.

\begin{proposition}\label{prop:SZ}
    Tetris clearing with SRS is NP-hard, and the corresponding counting problem is \#P-hard, even if the type of pieces in the sequence given to the player is restricted to $\{\SS, \ZZ\}$.
\end{proposition}

\begin{proof}

    Refer to Figure \ref{SZSetup}(a), which shows the bottle structure for $\{\SS, \ZZ\}$. We do not use a finisher in our setup.

    % SZ Setup Images
\begin{figure}[ht]
  \centering
  \begin{subfigure}[b]{0.09\textwidth}
    \centering
    \includegraphics[width=40pt]{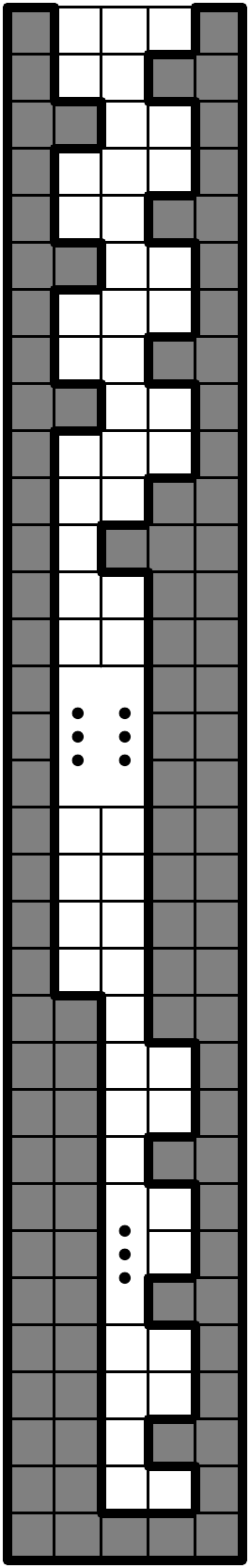}
    \caption{}
  \end{subfigure}
  \begin{subfigure}[b]{0.09\textwidth}
    \centering
    \includegraphics[width=40pt]{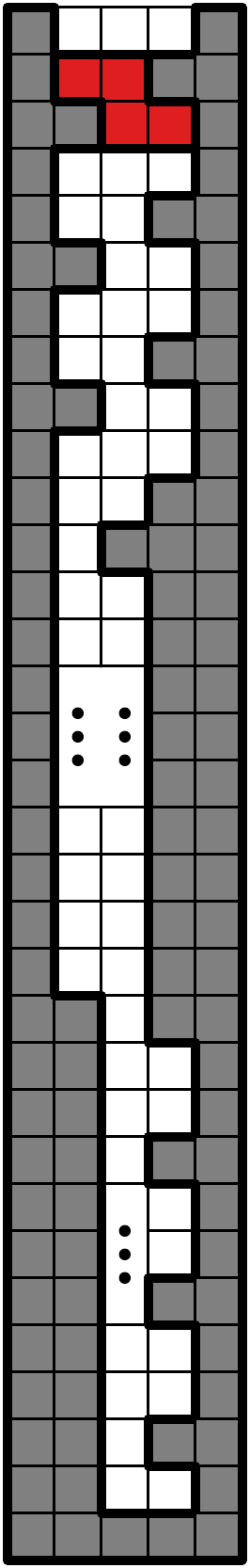}
    \caption{}
  \end{subfigure}
  \begin{subfigure}[b]{0.09\textwidth}
    \centering
    \includegraphics[width=40pt]{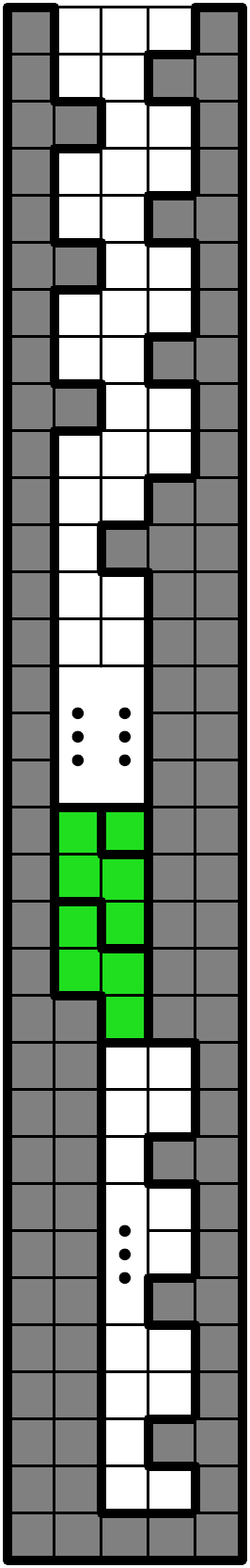}
    \caption{}
  \end{subfigure}
  \begin{subfigure}[b]{0.09\textwidth}
    \centering
    \includegraphics[width=40pt]{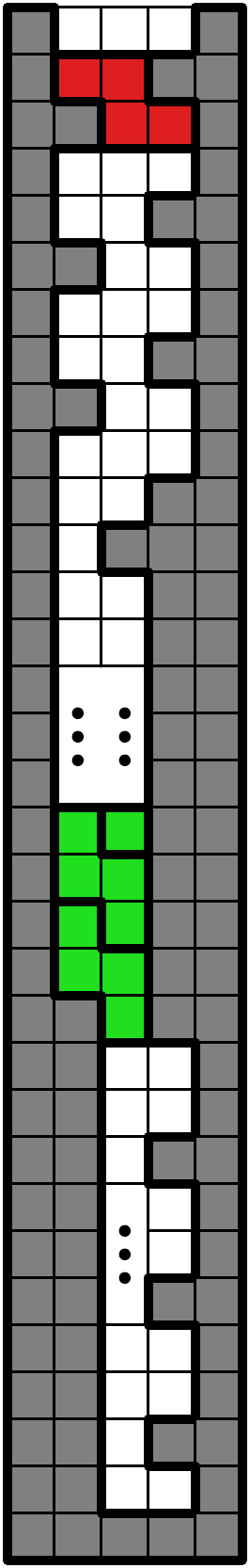}
    \caption{}
  \end{subfigure}
  \begin{subfigure}[b]{0.09\textwidth}
    \centering
    \includegraphics[width=40pt]{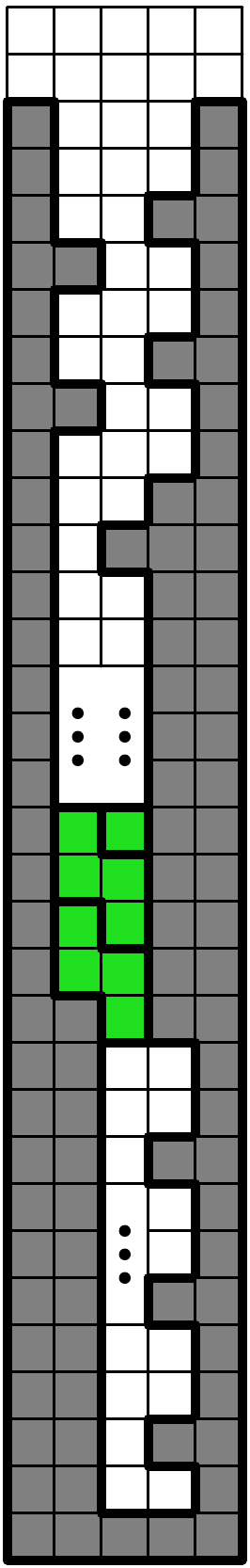}
    \caption{}
  \end{subfigure}
  \caption{The bottle structure for $\{\SS, \ZZ\}$. (b) shows how the $\ZZ$ piece must block a bottle during the priming sequence (requires the $\ZZ$-piece version of the $\SS$ spin in Figure \ref{SSpin2}). (c) shows the result of a filling sequence (requires the $\SS$-spin sequence in Figure \ref{SSpin-long1} to get $\SS$ pieces through the neck portion and the $\SS$ spin in Figure \ref{SSpin1} to get $\SS$ pieces into the body portion). (d) shows where the piece in the closing sequence must go. (e) shows the result of the closing sequence.}
  \label{SZSetup}
\end{figure}

    Here, a "top segment" is the $\ZZ$-shaped region of area $4$, and a "unit" is a region of area $16n$ consisting of $4n$ $\SS$-shaped pieces stacked on top of each other. We use $O(n)$ extra rows: one row with a $1\times 3$ hole above each top segment, two rows between the bottommost "top segment" and the body portion, and $3(n+1)+1$ other rows below the body portion. Note that the size of the region in the bottle above the body portion is $7n+5$, which is smaller than $16n$, the size of a "unit".
    
    Each $a_i$ is encoded by the sequence of pieces $(\ZZ^{n/3 - 1}, \SS^{a_i}, \ZZ)$. The priming sequence is $(\ZZ^{n/3 - 1})$, the filling sequence is $(\SS^{a_i})$, and the closing sequence is $(\ZZ)$, with the ways to block, fill, and re-open a bottle properly shown in Figure \ref{SZSetup}. Discussion on improper piece placements can be found in Appendix \ref{appendix:szclogs}.

    The finale sequence is $(\ZZ^{n/3 + 2n(n+1)/3})$. The remaining lines will collapse down to another structure as indicated in Figure \ref{SZSetupFinale}(a). The only way to clear the structure is to use the first $\frac n3$ pieces to clear two rows as indicated in Figure \ref{SZSetupFinale}(b), and then repeatedly use $2n/3$ $\ZZ$ pieces to clear four rows per iteration as indicated in Figure \ref{SZSetupFinale}(c--e) until the board is cleared.

     % SZ Finale Images
\begin{figure}[ht]
  \centering
  \begin{subfigure}[b]{0.1\textwidth}
    \centering
    \includegraphics[width=40pt]{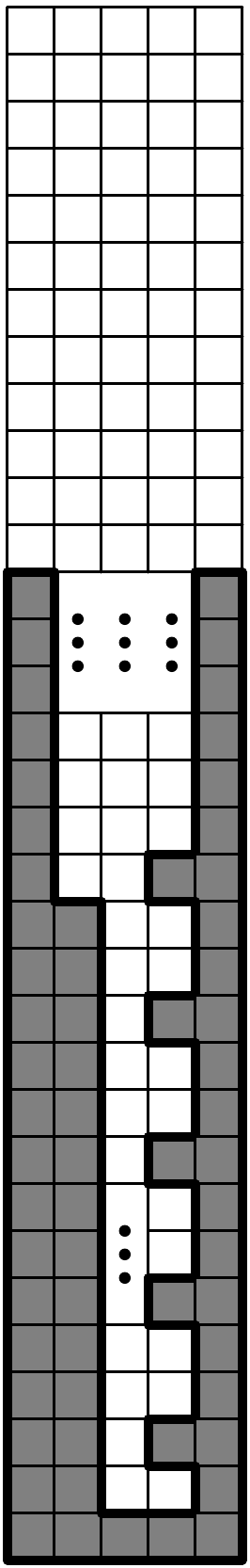}
    \caption{}
  \end{subfigure}
  \begin{subfigure}[b]{0.1\textwidth}
    \centering
    \includegraphics[width=40pt]{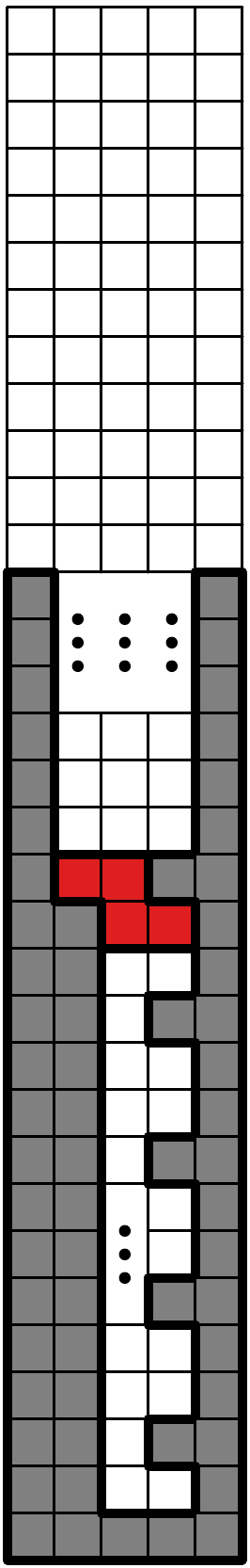}
    \caption{}
  \end{subfigure}
  \begin{subfigure}[b]{0.1\textwidth}
    \centering
    \includegraphics[width=40pt]{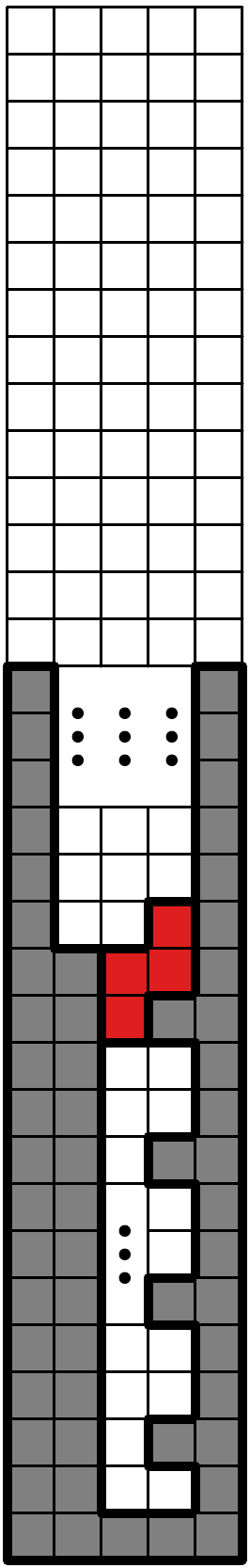}
    \caption{}
  \end{subfigure}
  \begin{subfigure}[b]{0.1\textwidth}
    \centering
    \includegraphics[width=40pt]{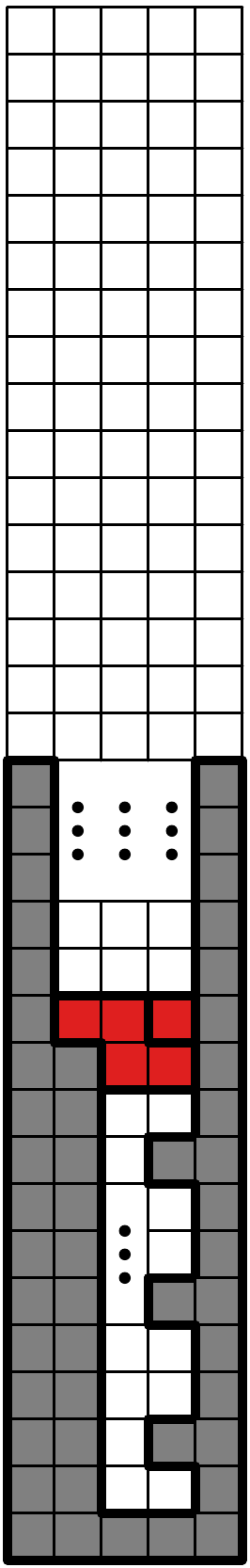}
    \caption{}
  \end{subfigure}
  \begin{subfigure}[b]{0.1\textwidth}
    \centering
    \includegraphics[width=40pt]{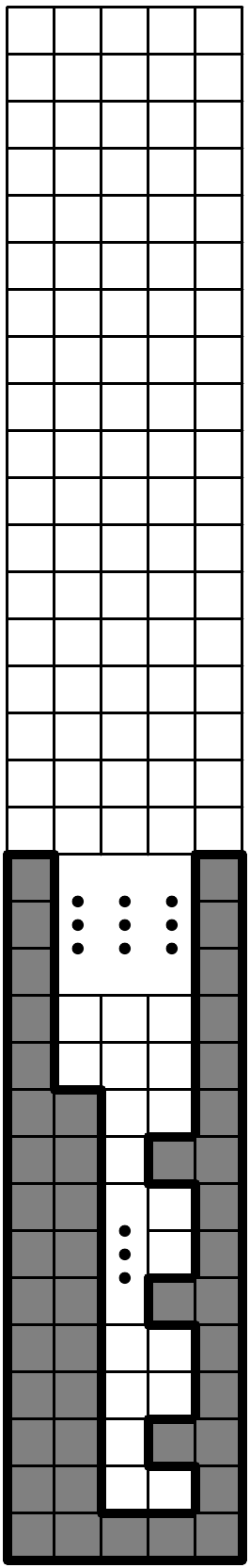}
    \caption{}
  \end{subfigure}
  \caption{The finale structure for $\{\SS, \ZZ\}$, including how to clear rows using $\ZZ$ pieces (requires the $\ZZ$-piece version of the $\SS$ spin in Figure \ref{SSpin2}).}
  \label{SZSetupFinale}
\end{figure}

    Since we have all of the required components, we can apply the general argument from Section \ref{sec:genarg} to conclude NP-hardness.

    This reduction is $\left((\frac n3)!((\frac n3-1)!)^n((\frac n3)!)^{2n+3}\right)$-monious: for each solution to the 3-Partition with Distinct Integers instance, there are $(\frac n3)!$ ways to "permute" which subsets correspond to which bottles, $(\frac n3-1)!$ ways to permute how the pieces in the priming sequence get placed for each $a_i$ sequence and $((\frac n3)!)^{2n+3}$ ways to permute how the pieces in the finale sequence get placed (the only re-ordering allowed is within each group of $\frac n3$ $\ZZ$ pieces). As such, we can also conclude \#P-hardness. 
\end{proof}

A demo of the $\{\SS, \ZZ\}$ bottle structure can be found at \url{https://jstris.jezevec10.com/map/80198}.

\subsection{Remaining Subsets with More Complex Structures}\label{sec:remainingsubsets}

\begin{proposition}\label{prop:JZ}
    Tetris clearing with SRS is NP-hard, and the corresponding counting problem is \#P-hard, even if the type of pieces in the sequence given to the player is restricted to any of $\{\JJ, \TT\}$, $\{\JJ, \ZZ\}$, $\{\LL, \SS\}$, or $\{\LL, \TT\}$.
\end{proposition}

\begin{proof}
    First we discuss the $\{\JJ, \ZZ\}$ and $\{\JJ, \TT\}$ cases. Refer to Figure \ref{JZSetup}(a), which shows the bottle structure for $\{\JJ, \ZZ\}$ and $\{\JJ, \TT\}$. We will also use a $\JJ$ finisher in our setup to prevent rows in the body portion from clearing early.

    % JZ Setup Images
\begin{figure}[ht]
  \centering
  \begin{subfigure}[b]{0.1\textwidth}
    \centering
    \includegraphics[width=40pt]{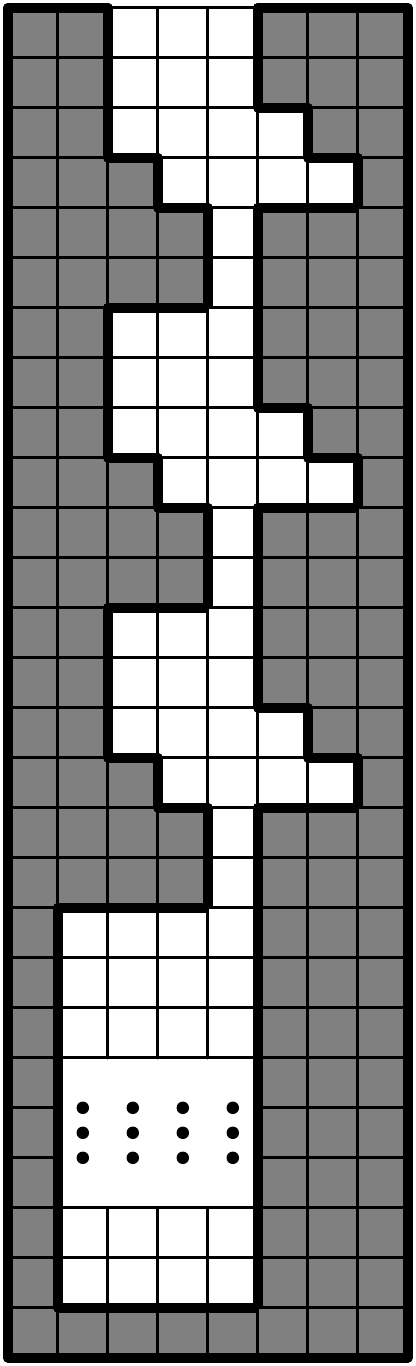}
    \caption{}
  \end{subfigure}
  \begin{subfigure}[b]{0.1\textwidth}
    \centering
    \includegraphics[width=40pt]{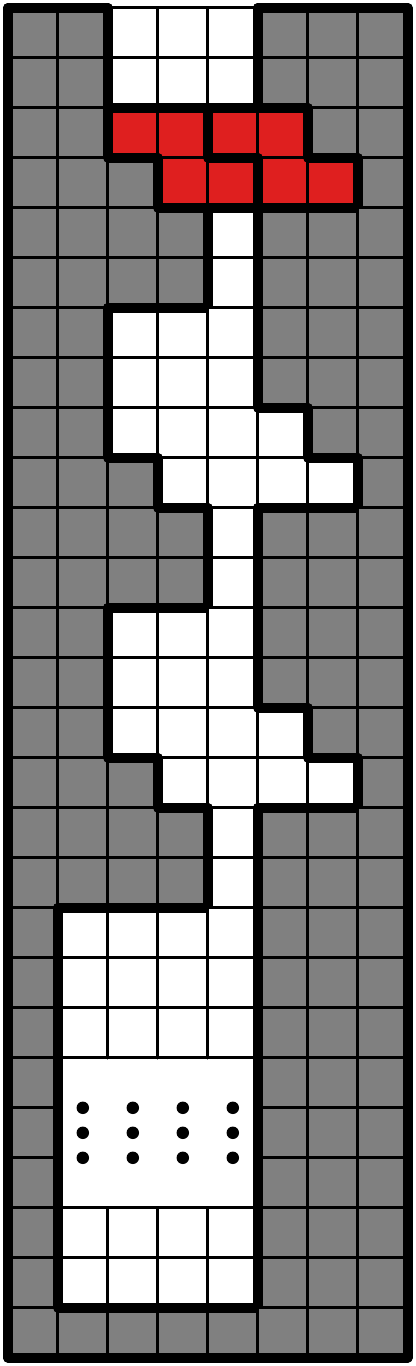}
    \caption{}
  \end{subfigure}
  \begin{subfigure}[b]{0.1\textwidth}
    \centering
    \includegraphics[width=40pt]{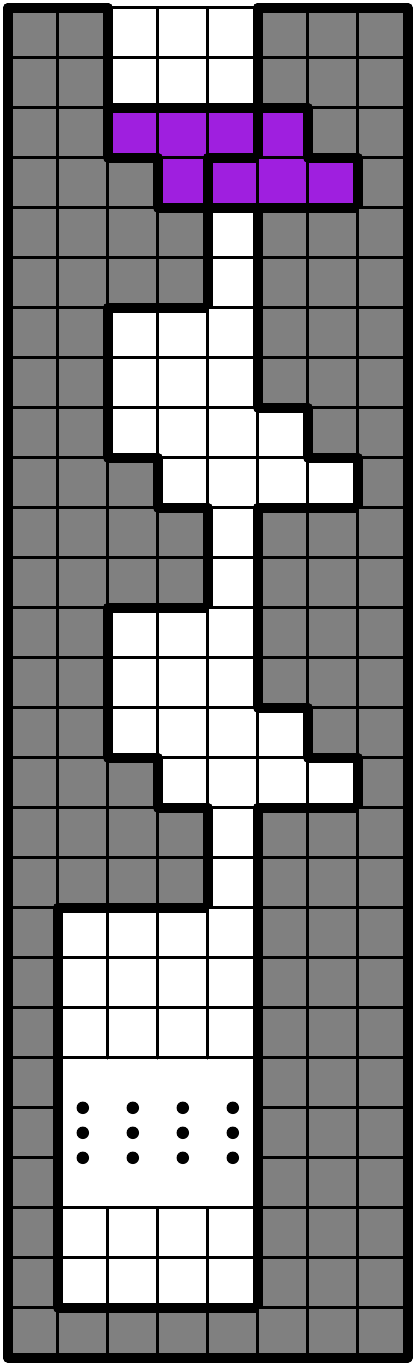}
    \caption{}
  \end{subfigure}
  \begin{subfigure}[b]{0.1\textwidth}
    \centering
    \includegraphics[width=40pt]{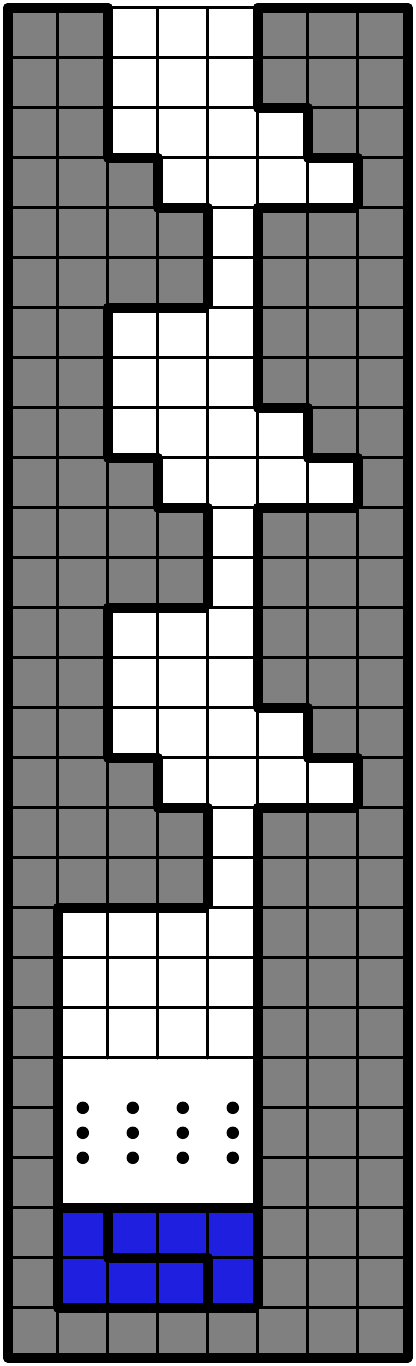}
    \caption{}
  \end{subfigure}
  \begin{subfigure}[b]{0.1\textwidth}
    \centering
    \includegraphics[width=40pt]{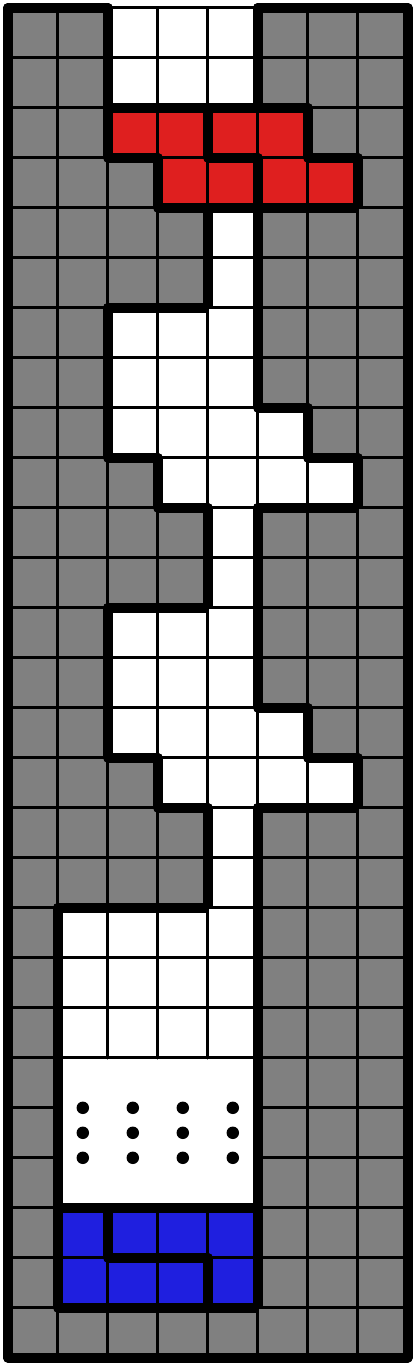}
    \caption{}
  \end{subfigure}
  \begin{subfigure}[b]{0.1\textwidth}
    \centering
    \includegraphics[width=40pt]{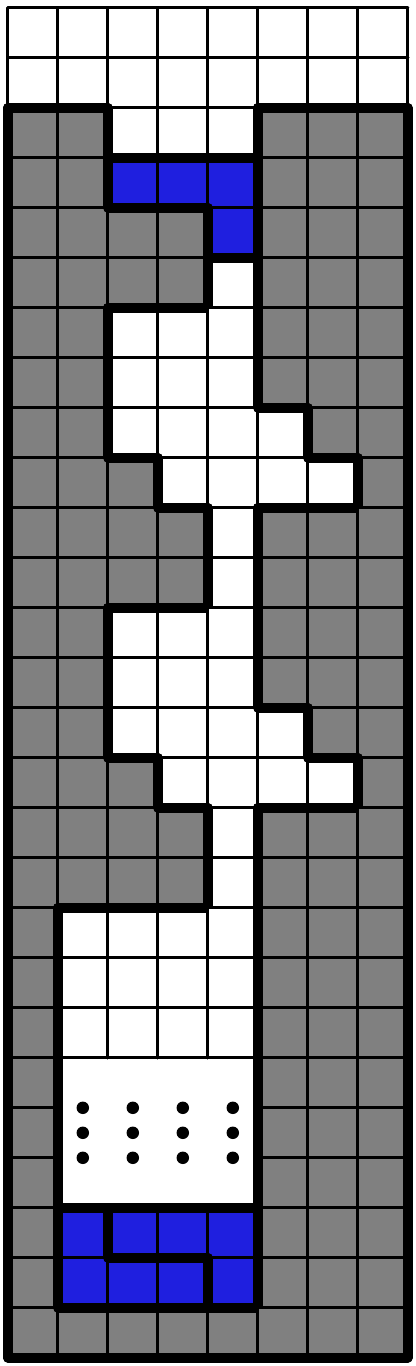}
    \caption{}
  \end{subfigure}
  \begin{subfigure}[b]{0.1\textwidth}
    \centering
    \includegraphics[width=40pt]{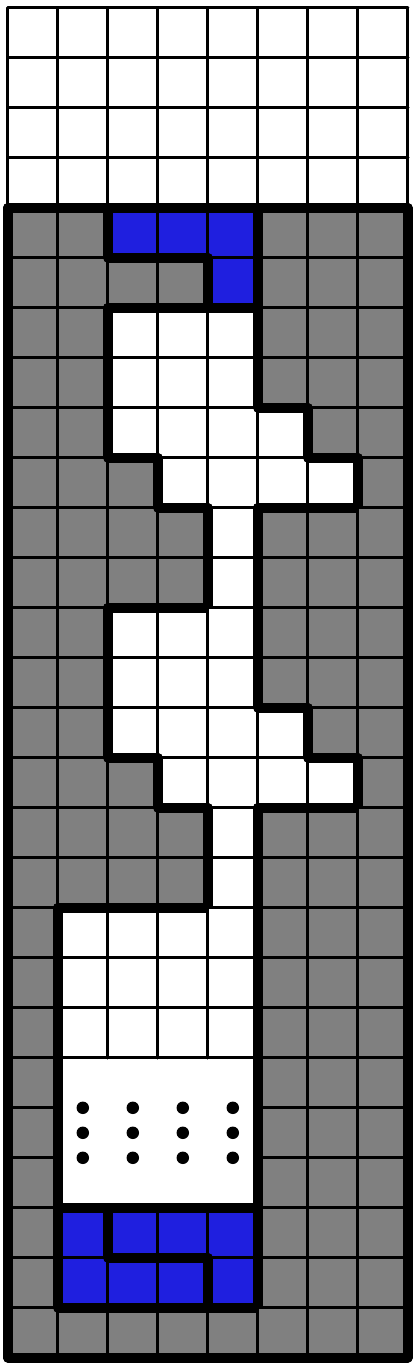}
    \caption{}
  \end{subfigure}
  \begin{subfigure}[b]{0.1\textwidth}
    \centering
    \includegraphics[width=40pt]{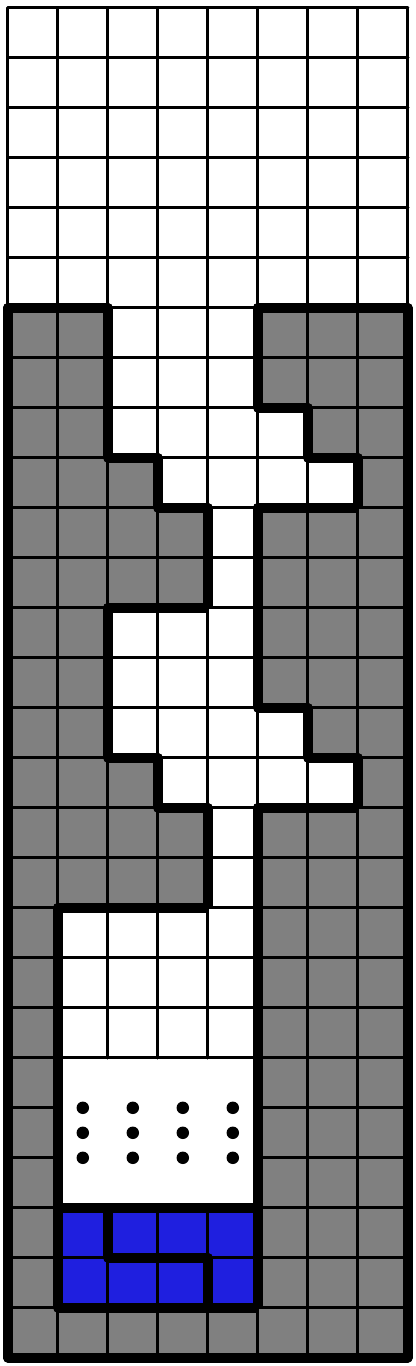}
    \caption{}
  \end{subfigure}
  \caption{The bottle structure for $\{\JJ, \ZZ\}$ and $\{\JJ, \TT\}$. (b) and (c) show how the $\ZZ$ pieces or $\TT$ pieces must block a bottle during the priming sequence (requires the $\ZZ$-piece version of the $\SS$ spin in Figure \ref{SSpin2} for the second $\ZZ$ piece). (d) shows the result of a filling sequence (requires the $\JJ$ spin in Figure \ref{JSpin4} to get $\JJ$ pieces through the neck portion and into the body portion). (e--g) show where the pieces in the closing sequence must go. (h) shows the result of the closing sequence.}
  \label{JZSetup}
\end{figure}

    Here, a "top segment" is the pattern that repeats every 6 lines at the top of the bottle, and a "unit" is one $6n\times 4$ rectangle (note that the size of the neck portion is $16n$, which is smaller than $24n$, the size of a "unit"). No extra lines are necessary. 
    
    For $\{\JJ, \ZZ\}$, each $a_i$ is encoded by the sequence of pieces $(\ZZ^{2n/3 - 2}, \JJ^{6na_i}, \ZZ^2, \allowbreak \JJ^{2n/3})$. The priming sequence is $(\ZZ^{2n/3 - 2})$, the filling sequence is $(\JJ^{6na_i})$, and the closing sequence is $(\ZZ^2, \JJ^{2n/3})$, with the ways to block, fill, and re-open a bottle properly shown in Figure \ref{JZSetup}. Discussion on improper piece placements can be found in Appendix \ref{appendix:jzclogs}.

    The finale sequence is $(\JJ^{6nt-2})$ to clear the $\JJ$ finisher.

    Since we have all of the required components, we can apply the general argument from Section \ref{sec:genarg} to conclude NP-hardness.

    \begin{sloppypar}Now, for our \#P-hardness analysis, let $f_4(m)$ denote the number of ways to place $m$ $\JJ$ pieces in an rectangle of width $4$ and height $m$ to fill the rectangle exactly. This reduction is $\left((\frac n3)!\left(\frac{(2n/3-2)!((n/3)!)^2}{2^{(n/3-1)}}\right)^n(f_4(6nt))^{n/3}\right)$-monious: for each solution to the 3-Partition with Distinct Integers instance, there are $(\frac n3)!$ ways to "permute" which subsets correspond to which bottles, $\frac{(2n/3-2)!}{2^{(n/3-1)}}$ ways to permute how the pieces in the priming sequence get placed for each $a_i$ sequence (the only restriction is on how the two $\ZZ$ pieces blocking the same bottle get placed), $((\frac n3)!)^2$ ways to permute how the $\JJ$ pieces in the closing sequence get placed for each $a_i$ sequence, and $f_4(6nt)$ ways to permute how the $\JJ$ pieces get placed within the body portion of each bottle. As such, we can also conclude \#P-hardness.\end{sloppypar}

    For $\{\JJ, \TT\}$, we use the exact same bottle structure, and have each $a_i$ be encoded by the sequence of pieces $(\TT^{2n/3 - 2}, \JJ^{6na_i}, \TT^2, \JJ^{2n/3})$. The way to block a bottle is slightly different and is shown in Figure \ref{JZSetup}(c) and the cases to check for the priming sequence are slightly different (see Appendix \ref{appendix:jzclogs}), but the proof is overall the same.

    We get a similar argument for $\{\LL, \SS\}$ and $\{\LL, \TT\}$ by vertical symmetry.
\end{proof}

A demo of the $\{\JJ, \ZZ\}$ bottle structure can be found at \url{https://jstris.jezevec10.com/map/80205}.

\begin{proposition}\label{prop:JS}
    Tetris clearing with SRS is NP-hard, and the corresponding counting problem is \#P-hard, even if the type of pieces in the sequence given to the player is restricted to either $\{\JJ, \SS\}$ or $\{\LL, \ZZ\}$.
\end{proposition}

\begin{proof}
    First we discuss the $\{\JJ, \SS\}$ case. Refer to Figure \ref{JSSetup}(a), which shows the bottle structure for $\{\JJ, \SS\}$. We will also use a $\JJ$ finisher in our setup to prevent rows in the body portion from clearing early.

    % JS Setup Images
\begin{figure}[ht]
  \centering
  \begin{subfigure}[b]{0.1\textwidth}
    \centering
    \includegraphics[width=40pt]{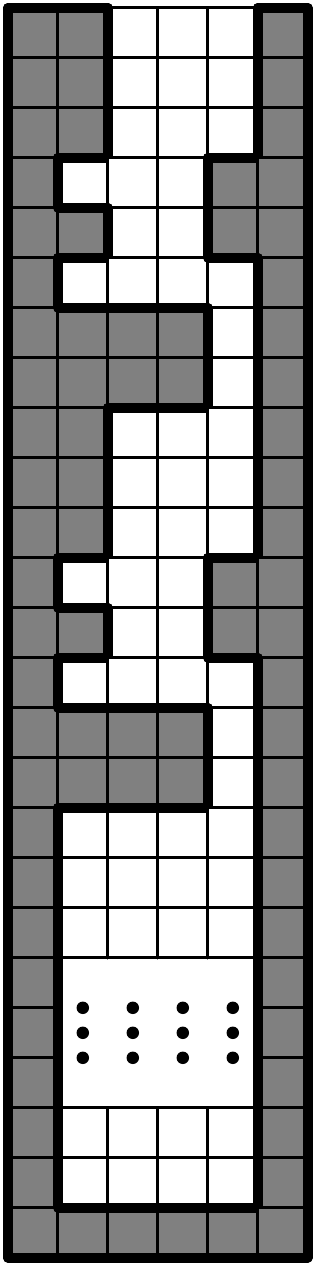}
    \caption{}
  \end{subfigure}
  \begin{subfigure}[b]{0.1\textwidth}
    \centering
    \includegraphics[width=40pt]{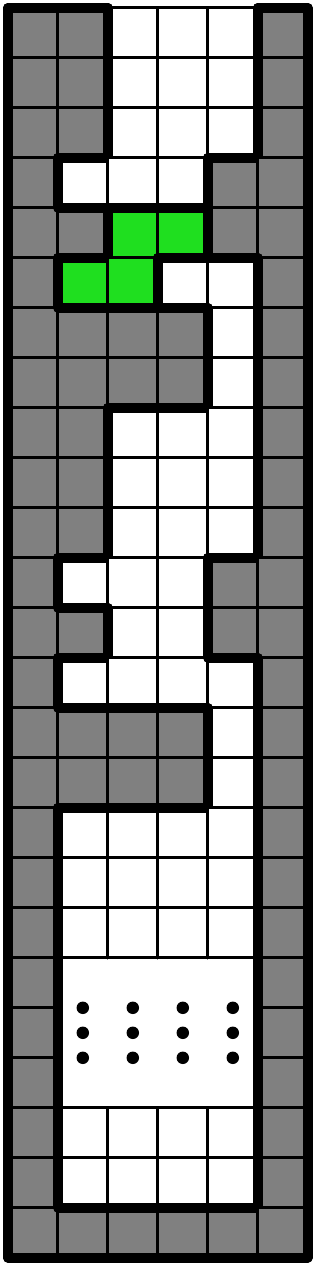}
    \caption{}
  \end{subfigure}
  \begin{subfigure}[b]{0.1\textwidth}
    \centering
    \includegraphics[width=40pt]{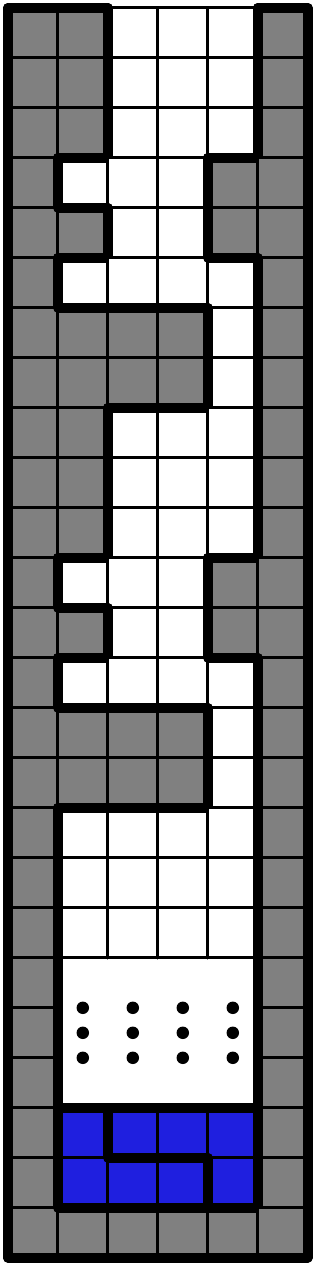}
    \caption{}
  \end{subfigure}
  \begin{subfigure}[b]{0.1\textwidth}
    \centering
    \includegraphics[width=40pt]{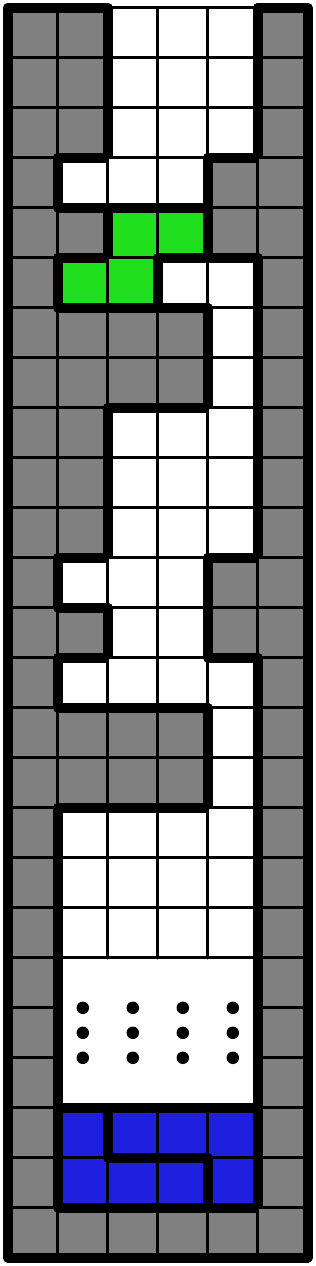}
    \caption{}
  \end{subfigure}
  \begin{subfigure}[b]{0.1\textwidth}
    \centering
    \includegraphics[width=40pt]{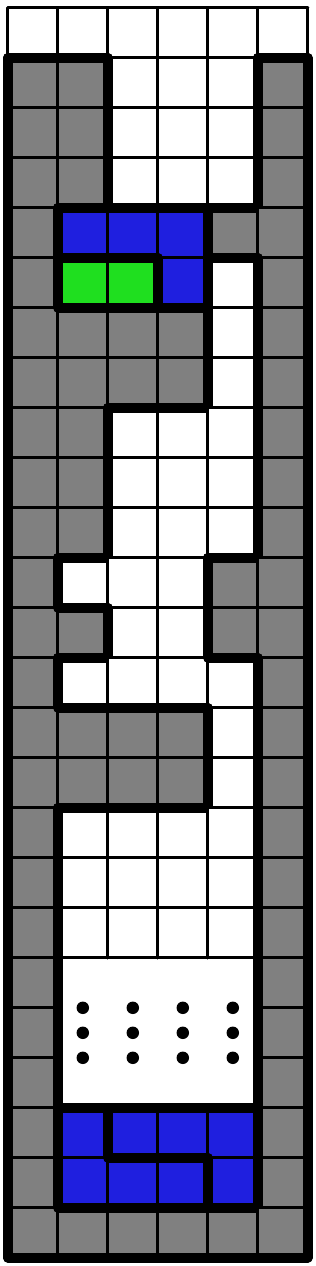}
    \caption{}
  \end{subfigure}
  \begin{subfigure}[b]{0.1\textwidth}
    \centering
    \includegraphics[width=40pt]{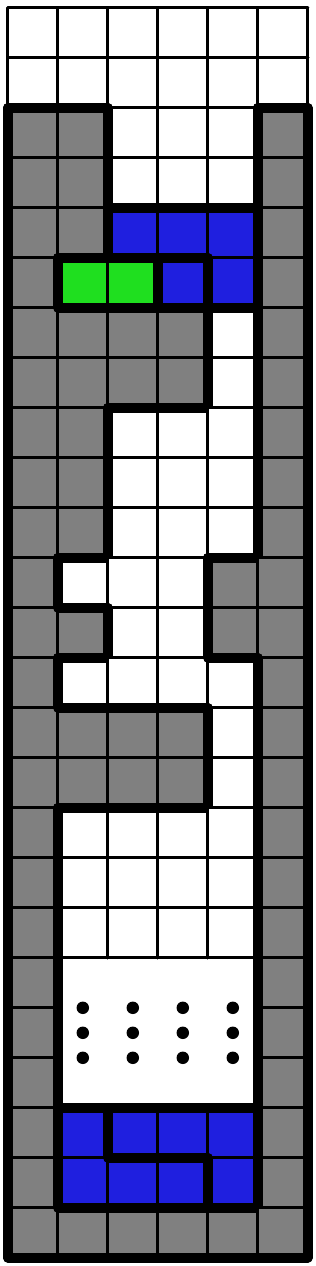}
    \caption{}
  \end{subfigure}
  \begin{subfigure}[b]{0.1\textwidth}
    \centering
    \includegraphics[width=40pt]{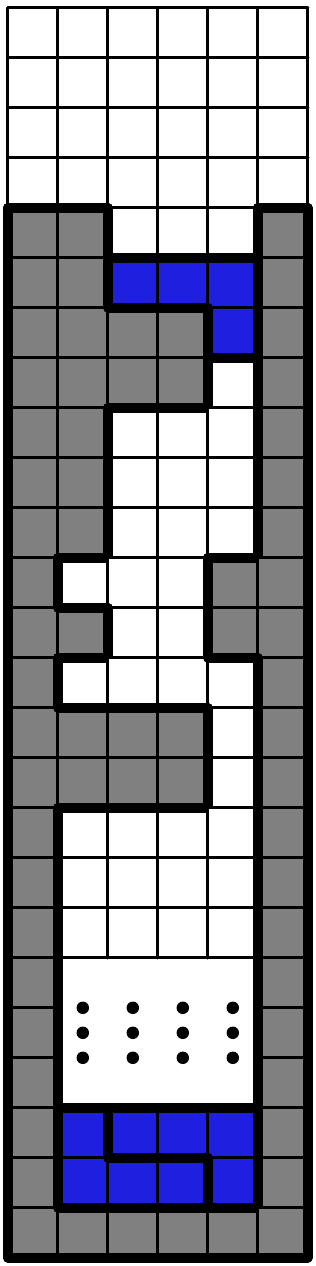}
    \caption{}
  \end{subfigure}
  \begin{subfigure}[b]{0.1\textwidth}
    \centering
    \includegraphics[width=40pt]{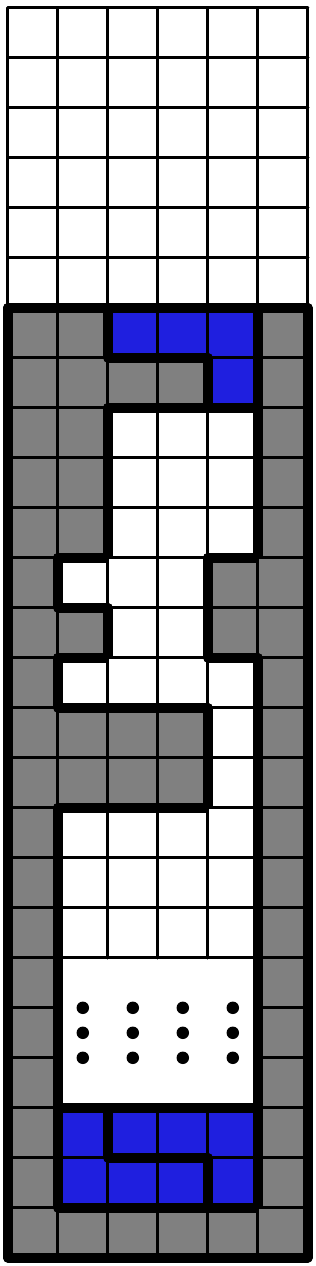}
    \caption{}
  \end{subfigure}
  \begin{subfigure}[b]{0.1\textwidth}
    \centering
    \includegraphics[width=40pt]{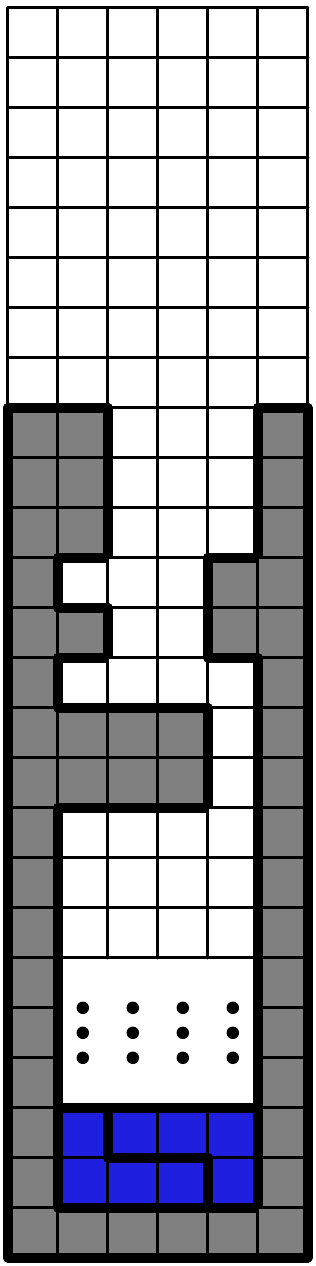}
    \caption{}
  \end{subfigure}
  \caption{The bottle structure for $\{\JJ, \SS\}$. (b) shows how the $\SS$ piece must block a bottle during the priming sequence. (c) shows the result of a filling sequence (requires the $\JJ$ spins in Figures \ref{JSpin2} and \ref{JSpin3} to get $\JJ$ pieces through the neck portion and into the body portion). (d--h) show where the pieces in the closing sequence must go (requires the $\JJ$ spin in Figure \ref{JSpin1}). (i) shows the result of the closing sequence.}
  \label{JSSetup}
\end{figure}

    Here, a "top segment" is the pattern that repeats every 8 lines at the top of the bottle, and a "unit" is one $6n\times 4$ rectangle (note that the size of the neck portion is $20n$, which is smaller than $24n$, the size of a "unit"). No extra lines are necessary. 
    
    Each $a_i$ is encoded by the sequence of pieces $(\SS^{n/3 - 1}, \JJ^{6na_i}, \SS, \JJ^{4n/3})$. The priming sequence is $(\SS^{n/3 - 1})$, the filling sequence is $(\JJ^{6na_i})$, and the closing sequence is $(\SS, \JJ^{4n/3})$, with the ways to block, fill, and re-open a bottle properly shown in Figure \ref{JSSetup}. Discussion on improper piece placements can be found in Appendix \ref{appendix:jsclogs}.

    The finale sequence is $(\JJ^{6nt-2})$ to clear the $\JJ$ finisher.

    Since we have all of the required components, we can apply the general argument from Section \ref{sec:genarg} to conclude NP-hardness.

    \begin{sloppypar}Now, for our \#P-hardness analysis, let $f_4(m)$ denote the number of ways to place $m$ $\JJ$ pieces in an rectangle of width $4$ and height $m$ to fill the rectangle exactly. This reduction is $\left((\frac n3)!\left((\frac{n}{3}-1)!(\frac{n}{3})!^4\right)^n(f_4(6nt))^{n/3}\right)$-monious: for each solution to the 3-Partition with Distinct Integers instance, there are $(\frac n3)!$ ways to "permute" which subsets correspond to which bottles, $(\frac{n}{3}-1)!$ ways to permute how the $\SS$ pieces in the priming sequence get placed for each $a_i$ sequence, $((\frac n3)!)^4$ ways to permute how the $\JJ$ pieces in the closing sequence get placed for each $a_i$ sequence, and $f_4(6nt)$ ways to permute how the $\JJ$ pieces get placed within the body portion of each bottle. As such, we can also conclude \#P-hardness.\end{sloppypar}

    We get a similar argument for $\{\LL, \ZZ\}$ by vertical symmetry.
\end{proof}

A demo of the $\{\JJ, \SS\}$ bottle structure can be found at \url{https://jstris.jezevec10.com/map/83069}.

\begin{proposition}\label{prop:JL}
    Tetris clearing with SRS is NP-hard, and the corresponding counting problem is \#P-hard, even if the type of pieces in the sequence given to the player is restricted to $\{\JJ, \LL\}$.
\end{proposition}

\begin{proof}
    Refer to Figure \ref{JLSetup}(a), which shows the bottle structure for $\{\JJ, \LL\}$. We will also use a $\JJ$ finisher in our setup to prevent rows in the body portion from clearing early.

    % JL Setup Images
\begin{figure}[ht]
  \centering
  \begin{subfigure}[b]{0.11\textwidth}
    \centering
    \includegraphics[width=50pt]{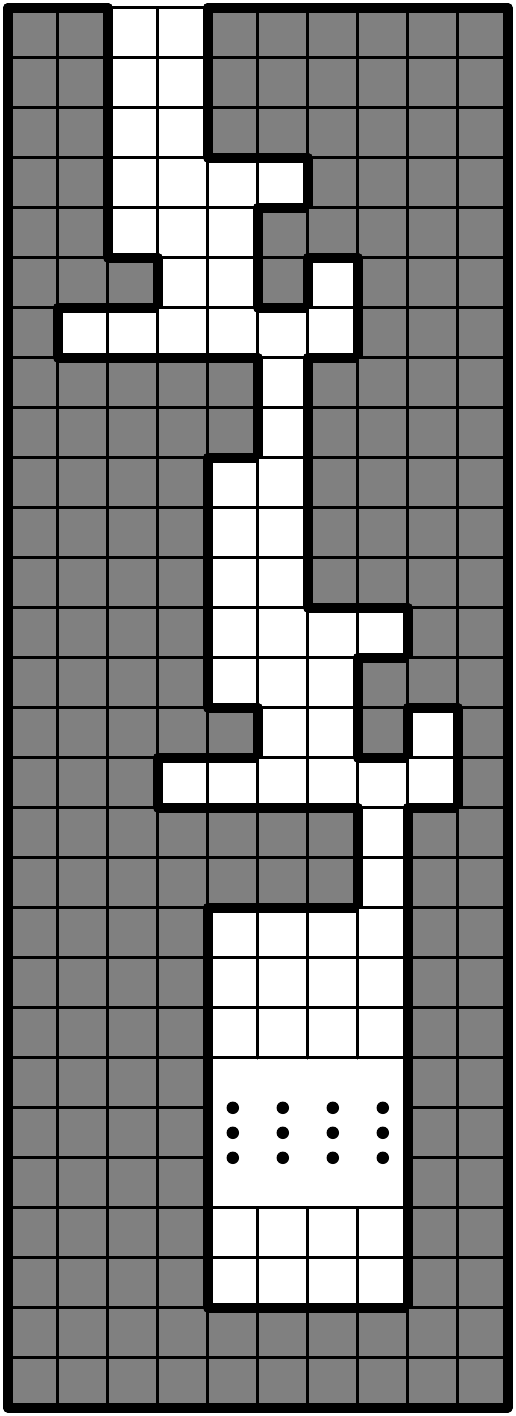}
    \caption{}
  \end{subfigure}
  \begin{subfigure}[b]{0.11\textwidth}
    \centering
    \includegraphics[width=50pt]{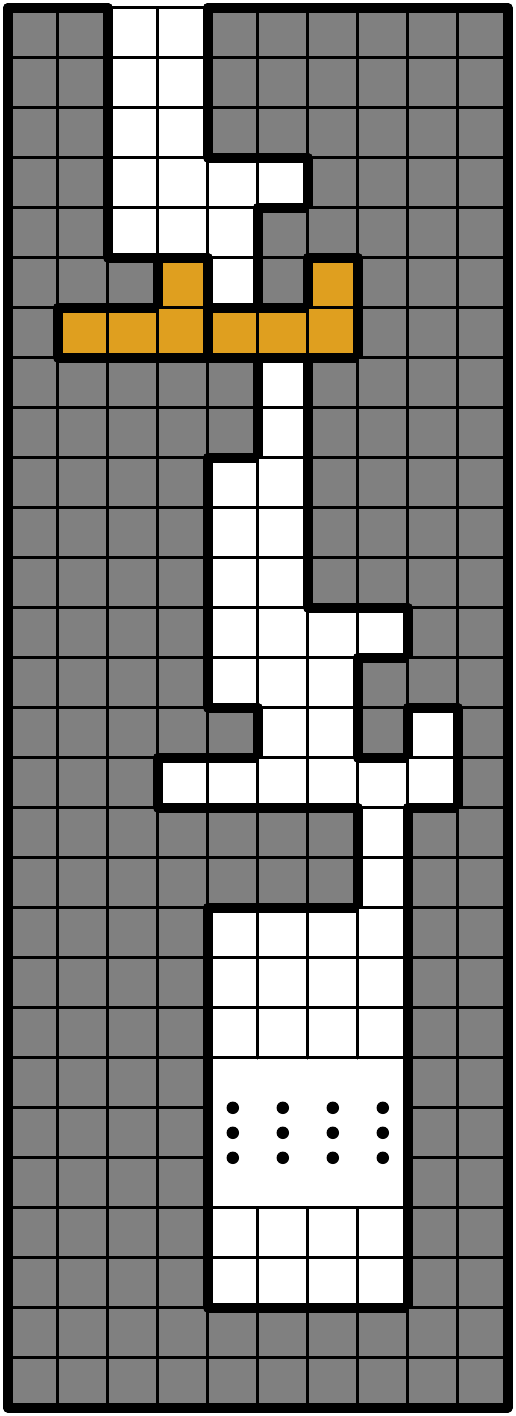}
    \caption{}
  \end{subfigure}
  \begin{subfigure}[b]{0.11\textwidth}
    \centering
    \includegraphics[width=50pt]{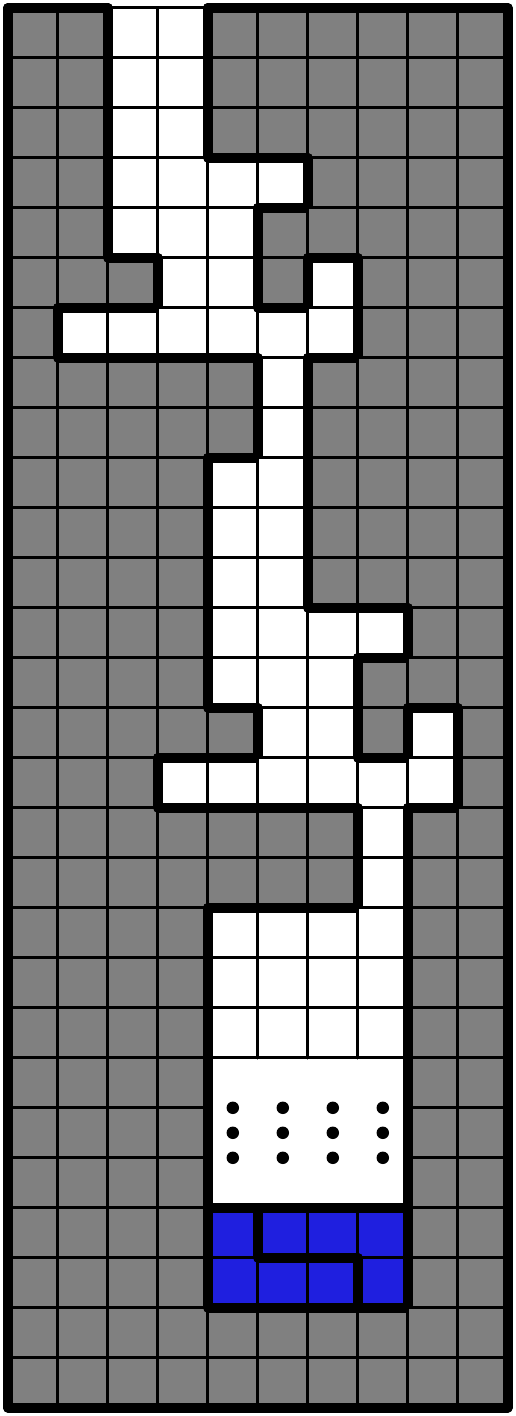}
    \caption{}
  \end{subfigure}
  \begin{subfigure}[b]{0.11\textwidth}
    \centering
    \includegraphics[width=50pt]{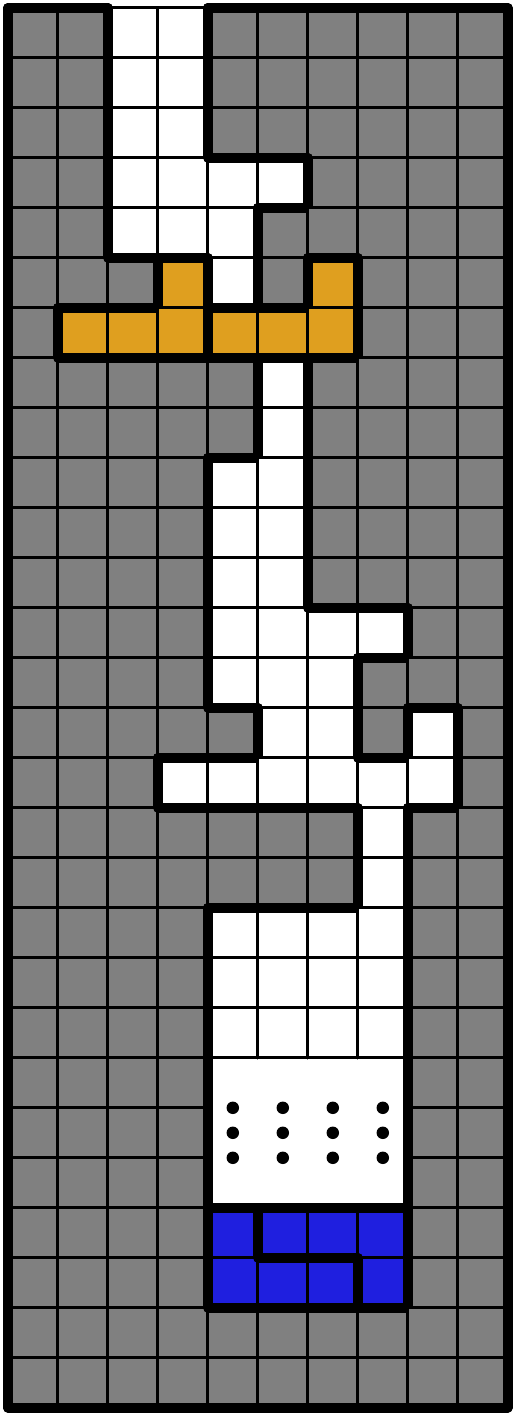}
    \caption{}
  \end{subfigure}
  \begin{subfigure}[b]{0.11\textwidth}
    \centering
    \includegraphics[width=50pt]{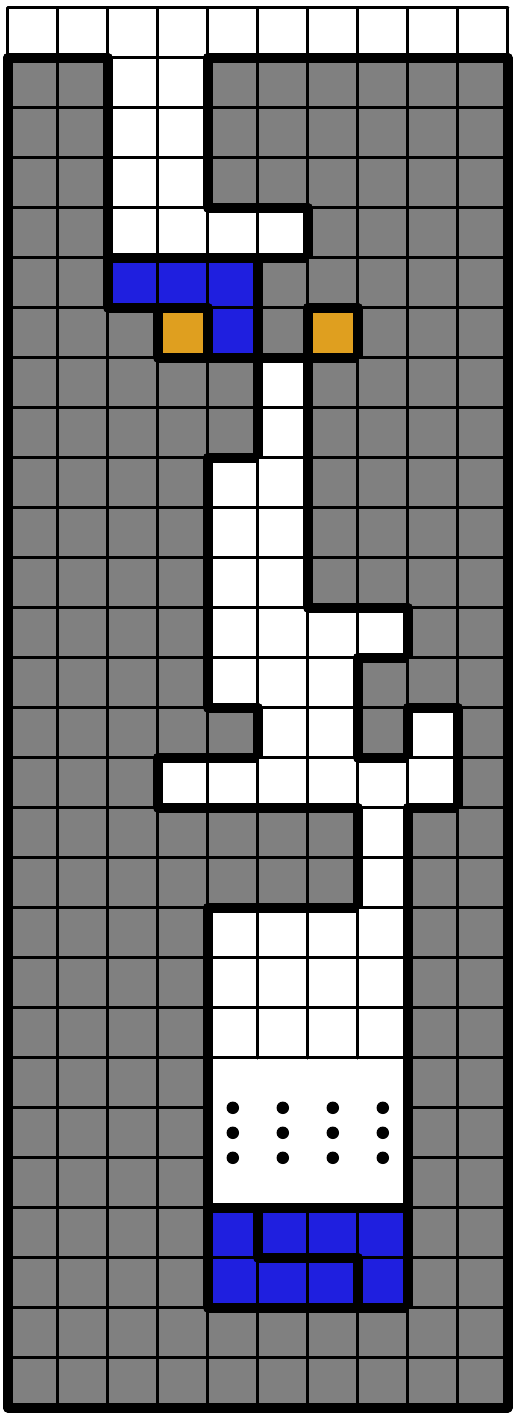}
    \caption{}
  \end{subfigure}
  \begin{subfigure}[b]{0.11\textwidth}
    \centering
    \includegraphics[width=50pt]{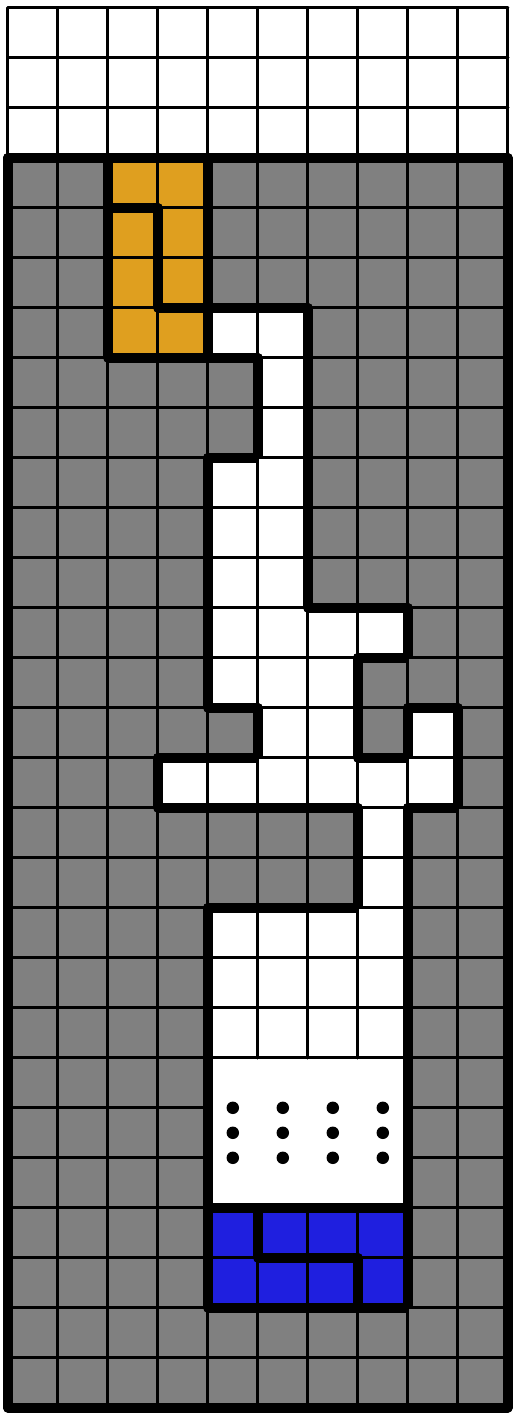}
    \caption{}
  \end{subfigure}
  \begin{subfigure}[b]{0.11\textwidth}
    \centering
    \includegraphics[width=50pt]{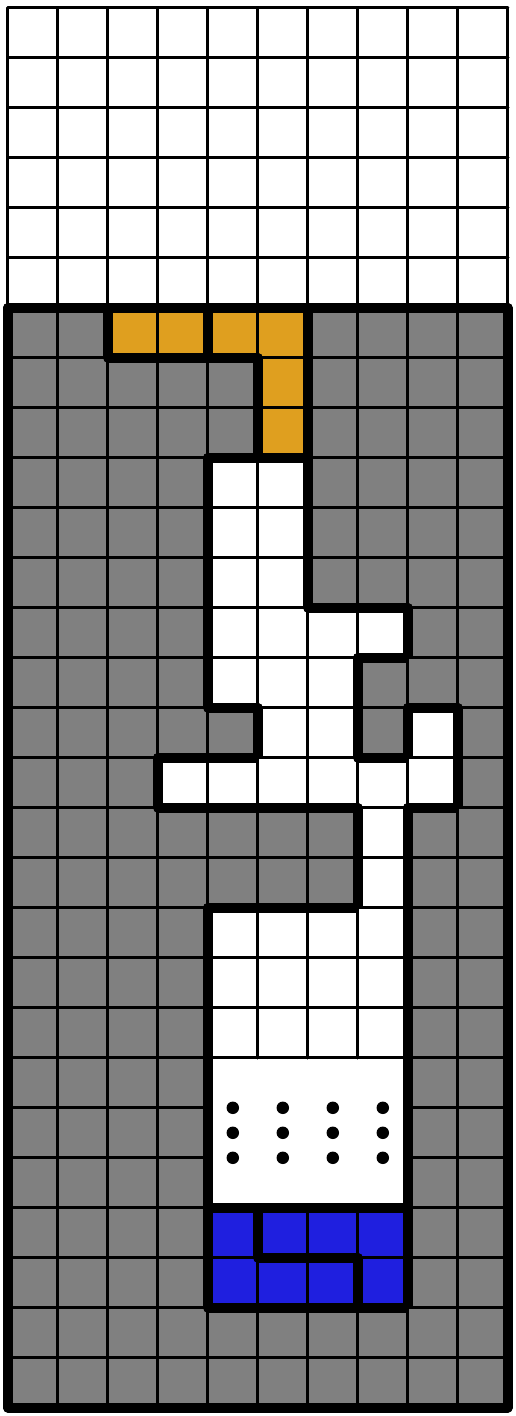}
    \caption{}
  \end{subfigure}
  \begin{subfigure}[b]{0.11\textwidth}
    \centering
    \includegraphics[width=50pt]{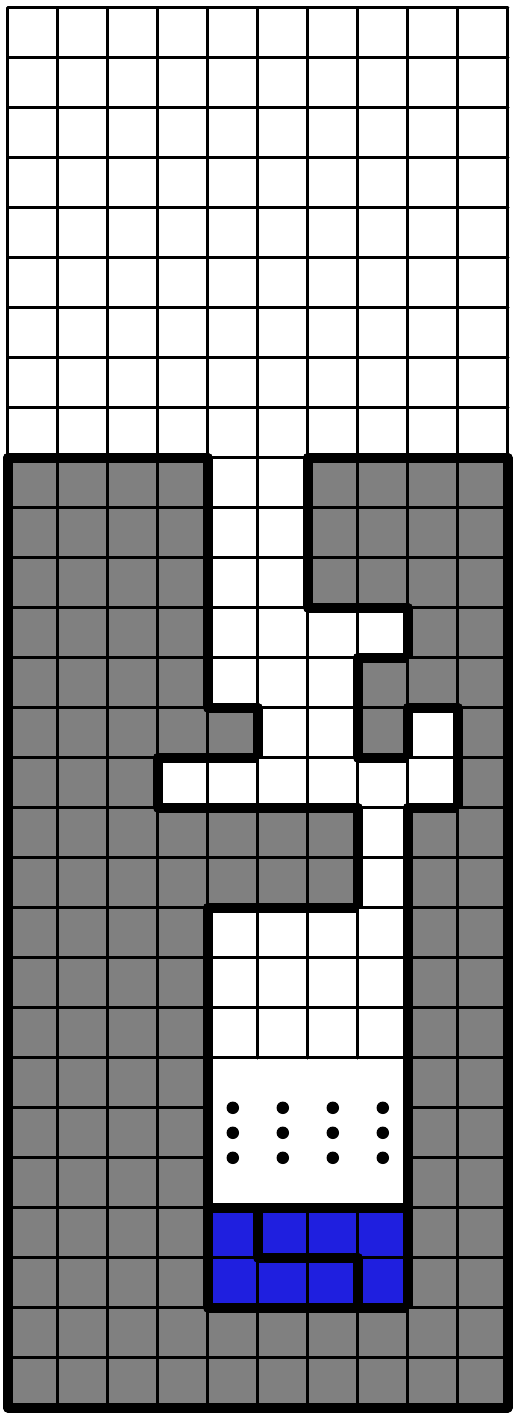}
    \caption{}
  \end{subfigure}
  \caption{The bottle structure for $\{\JJ, \LL\}$. (b) shows how the $\LL$ pieces must block a bottle during the priming sequence (requires the $\LL$-piece versions of the $\JJ$ spins in Figures \ref{JSpin3} and \ref{JSpin5})). (c) shows the result of a filling sequence (requires the $\JJ$ spins in Figures \ref{JSpin3} and \ref{JSpin4} to get $\JJ$ pieces through the neck portion and into the body portion). (d--g) show where the pieces in the closing sequence must go. (h) shows the result of the closing sequence.}
  \label{JLSetup}
\end{figure}

    Here, a "top segment" is the pattern that repeats every 9 lines at the top of the bottle, and a "unit" is one $8n\times 4$ rectangle (note that the size of the neck portion is $24n$, which is smaller than $32n$, the size of a "unit"). Notice that the "top segments" are horizontally offset from each other in order to align the bottom lines of a "top segment" to the top lines of the next "top segment". We will also use a $\JJ$ finisher in our setup to prevent rows in the body portion from clearing early.
    
    Each $a_i$ is encoded by the sequence of pieces $(\LL^{2n/3-2}, \JJ^{8na_i}, \LL^2, \JJ^{n/3}, \LL^n)$. The priming sequence is $(\LL^{2n/3 - 2})$, the filling sequence is $(\JJ^{8na_i})$, and the closing sequence is $(\LL^2, \JJ^{n/3}, \LL^n)$, with the ways to block, fill, and re-open a bottle properly shown in Figure \ref{JLSetup}. Discussion on improper piece placements can be found in Appendix \ref{appendix:jlclogs}.

    The finale sequence is $(\JJ^{8nt-2})$ to clear the $\JJ$ finisher.

    Since we have all of the required components, we can apply the general argument from Section \ref{sec:genarg} to conclude NP-hardness.

    \begin{sloppypar}Now, for our \#P-hardness analysis, let $f_4(m)$ denote the number of ways to place $m$ $\JJ$ pieces in an rectangle of width $4$ and height $m$ to fill the rectangle exactly. This reduction is $\left((\frac n3)!\left(\frac{(2n/3-2)!(2n/3)!((n/3)!)^2}{2^{(2n/3-1)}}\right)^n(f_4(8nt))^{n/3}\right)$-monious: for each solution to the 3-Partition with Distinct Integers instance, there are $(\frac n3)!$ ways to "permute" which subsets correspond to which bottles, $\frac{(2n/3-2)!}{2^{(n/3-1)}}$ ways to permute how the pieces in the priming sequence get placed for each $a_i$ sequence (the only restriction is on how the two $\LL$ pieces blocking the same bottle get placed), $(\frac n3)!\times \frac{(2n/3)!(n/3)!}{2^{n/3}}$ ways to permute how the $\JJ$ pieces and $\LL$ pieces in the closing sequence get placed for each $a_i$ sequence, and $f_4(8nt)$ ways to permute how the $\JJ$ pieces get placed within the body portion of each bottle. As such, we can also conclude \#P-hardness.\end{sloppypar}
\end{proof}

A demo of the $\{\JJ, \LL\}$ bottle structure can be found at \url{https://jstris.jezevec10.com/map/80195}.

We also note the following:

\begin{conj}
    The reduction in the proof of Proposition \ref{prop:JL} works even if pieces experience 20G gravity.
\end{conj}

The reasoning is similar to the reasoning given for the proof of Proposition \ref{prop:IL} (the holes at the top of the bottles are always at most $2$ units wide); however, some more verification is required to make sure that the necessary rotations can be done in 20G gravity.

\subsection{Putting It All Together}\label{sec:alltogether}

Combining all of these results, we get the following result:

\begin{theorem}
    For any size-$2$ subset $A\subseteq\{\ALL\}$, Tetris clearing with SRS is NP-hard, and the corresponding counting problem is \#P-hard, even if the type of pieces in the sequence given to the player is restricted to the piece types in $A$.
\end{theorem}

\begin{proof}
    Propositions \ref{prop:IL}, \ref{prop:OJ}, \ref{prop:OS}, \ref{prop:OT}, \ref{prop:ST}, \ref{prop:SZ}, \ref{prop:JZ}, \ref{prop:JS}, and~\ref{prop:JL} cover all size-$2$ subsets of piece types, as shown in Table~\ref{tab:subset}; combining all of the reductions, we obtain the desired result.
\end{proof}

\begin{remark}
    All of our reductions involve a linear-factor blowup in the $a_i$ for the filling sequences (i.e., we use $\Theta(na_i)$ pieces in the filling sequences); this makes it easier to argue about what happens when an overfill happens and makes the bottle analogy more fitting (since the neck portion is smaller while the body portion is much larger) but makes our reductions somewhat inefficient. Perhaps it is possible to reduce the blowup to a constant factor, though the argument may be a bit more complex.
\end{remark}

% I'm adding this here because it seems like the most natural place to add it, but feel free to move it. Sorry it took so long -Holden
\section{Tetris Survival: Hard Drops Only and 20G}\label{sec:survival}

%In "hard drops only" Tetris, there is no gravity, and the only way to play pieces is to "hard drop" them, instantly forcing them to the bottom. In "20G" Tetris, gravity is at a maximum and all pieces fall instantly to the bottom. 

The previous section mentions NP-hardness of Tetris clearing under the "hard drops only" and "20G" Tetris models. Previous results about general Tetris \cite{TotalTetris_JIP, ThinTetris_JCDCGGG2019} have also proven NP-hardness of Tetris survival, so in this section, we prove that, in both of these variants, Tetris survival is NP-hard using $\{\II,\OO\}$ pieces. This improves upon a result by Temprano \cite{HardDrops} which proves hardness for "hard drops only" mode using the piece subset $\{\LL, \JJ, \II, \ZZ, \TT\}$.

\begin{theorem}
    Tetris survival is NP-complete in the "hard drops only" and "20G" game modes, even if the type of pieces in the sequence given to the player is restricted to $\{\II, \OO\}$.
\end{theorem}

\begin{proof}
    We reduce from 3-Partition with Distinct Integers using a similar bottle structure to other proofs in this paper. From a 3-Partition with Distinct Integers instance, we create a setup consisting of $\frac{n}{3}$ width 1 buckets, each of height $4t$, separated by width 1 columns. In addition, we create a bucket of height $t-1$ on the left which is blocked by a single square and has an open square on the upper left diagonal. We add one additional column on the left to obtain an even width board. We leave two rows at the top of the board empty. See Figure \ref{20GHD}(a) for details. Each $a_i$ is encoded by the sequence $(\OO^{n/3 + 1}, \II^{a_i}, \OO)$.

% Survival Pictures
\begin{figure}[ht]
  \centering
  \begin{subfigure}[b]{0.25\textwidth}
    \centering
    \includegraphics[width=100pt]{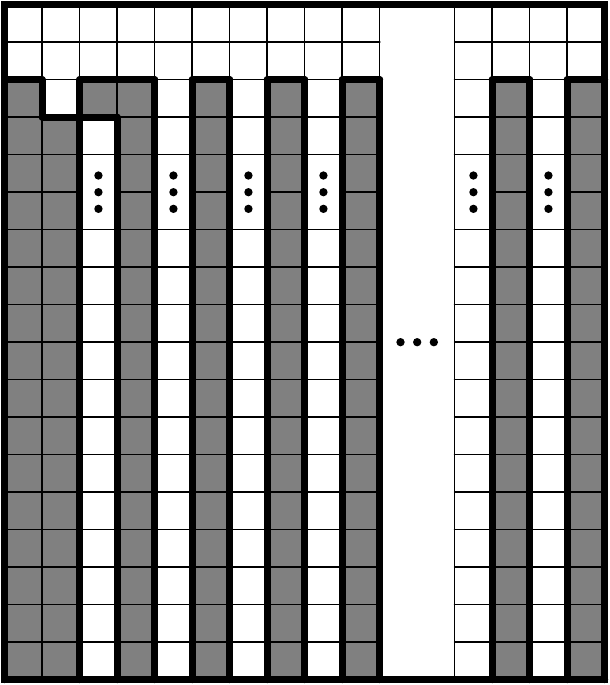}
    \caption{}
  \end{subfigure}
  \begin{subfigure}[b]{0.25\textwidth}
    \centering
    \includegraphics[width=100pt]{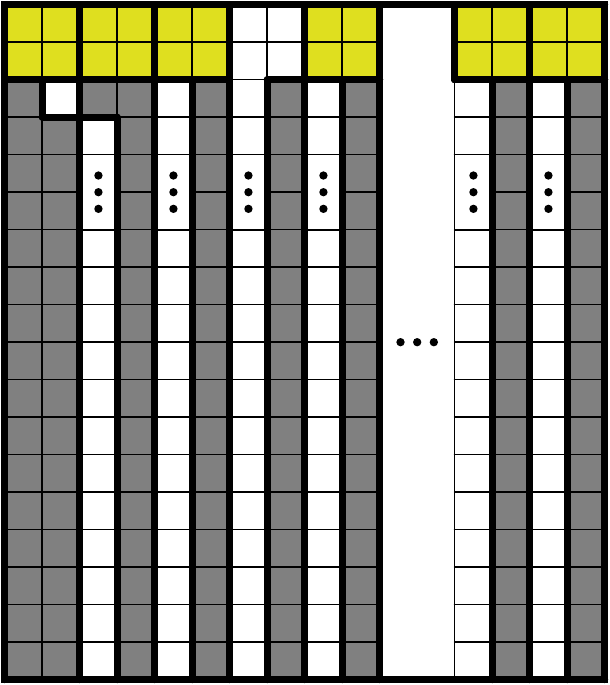}
    \caption{}
  \end{subfigure}
  \begin{subfigure}[b]{0.25\textwidth}
    \centering
    \includegraphics[width=100pt]{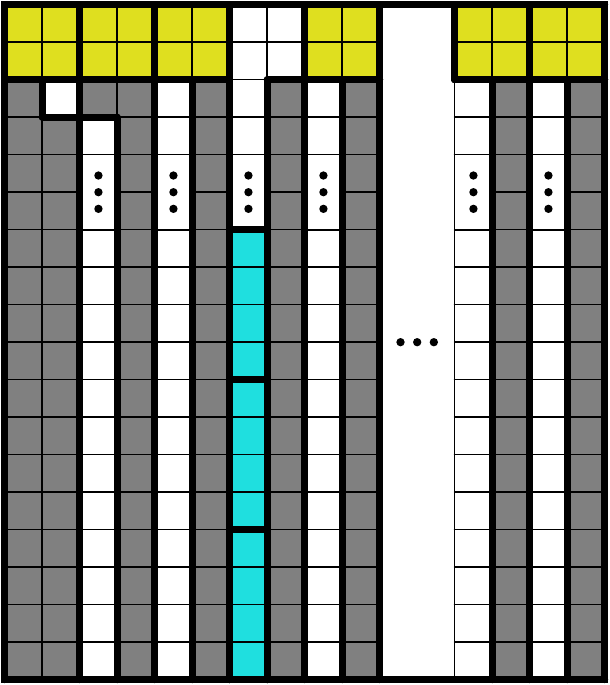}
    \caption{}
  \end{subfigure}
  \begin{subfigure}[b]{0.25\textwidth}
    \centering
    \includegraphics[width=100pt]{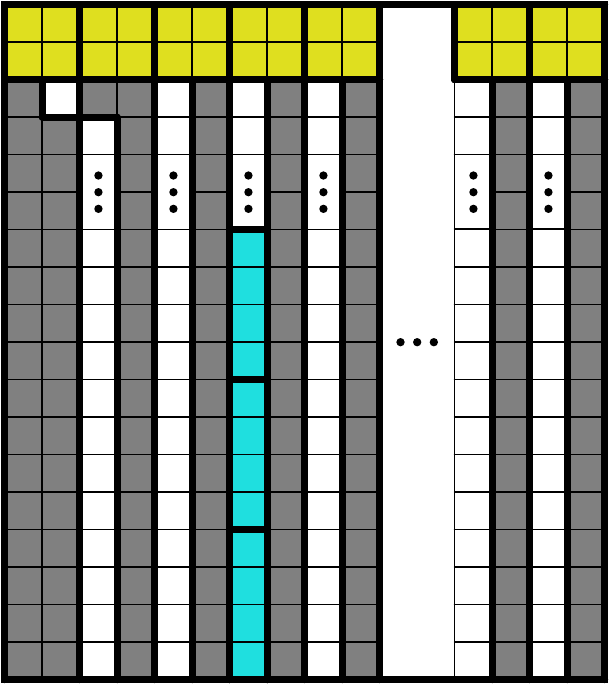}
    \caption{}
  \end{subfigure}
  \begin{subfigure}[b]{0.25\textwidth}
    \centering
    \includegraphics[width=100pt]{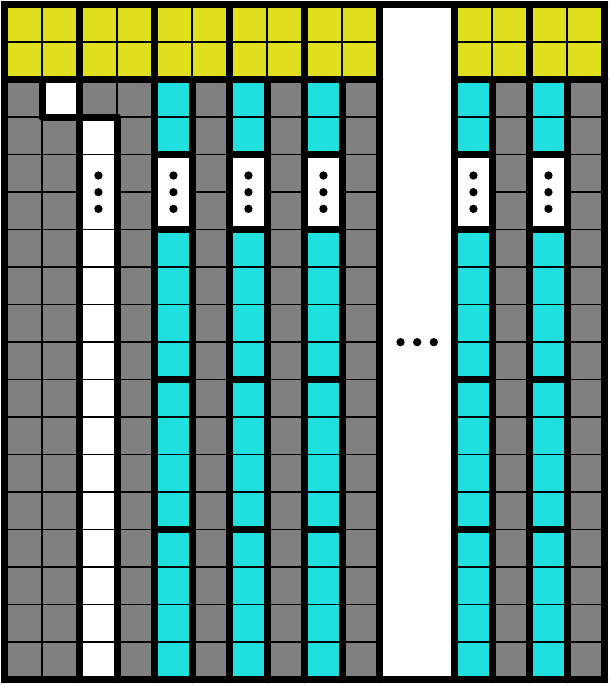}
    \caption{}
  \end{subfigure}
  \begin{subfigure}[b]{0.25\textwidth}
    \centering
    \includegraphics[width=100pt]{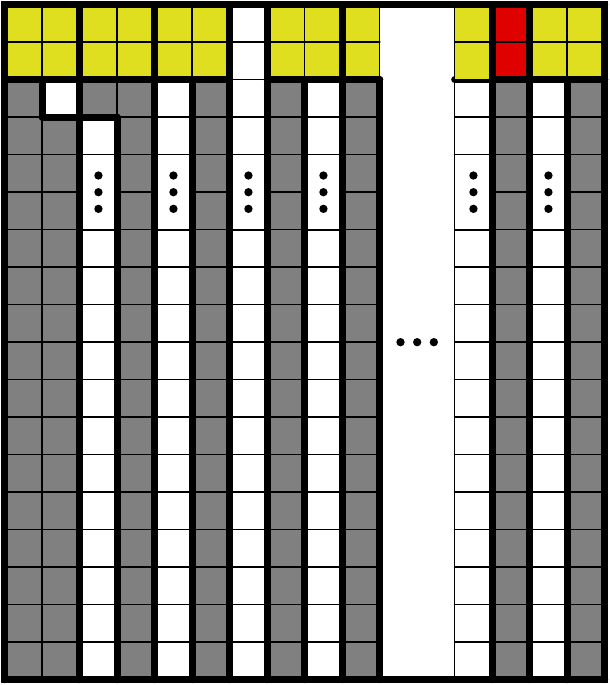}
    \caption{}
  \end{subfigure}
  \caption{The bucket structure for "hard drops only" and "20G" Tetris modes. (b) shows how beginning sequence must block all but one bucket. (c) shows the result of a filling sequence. (d) shows how the closing sequence clears the top two rows. (e) shows a full board. (f) shows improper handling during the priming sequence.}
  \label{20GHD}
\end{figure}
    
    The priming sequence is $(\OO^{n/3 + 1})$. Due to parity constraints, the player is forced to block all but one of the buckets using these pieces. They can choose to leave either one two-block wide gap or two one-block wide gaps in the top two rows. However, if they choose to leave one-block gaps, they will create spaces that can never be filled without overflowing the screen and causing a game loss (see the squares highlighted in red in Figure \ref{20GHD}(f)), so we can assume that they will leave a $2\times 2$ gap, as shown in Figure \ref{20GHD}(b).
    
    The filling sequence is $(\II^{a_i})$; the $\II$ pieces must be placed in the pre-selected bucket, as shown in Figure \ref{20GHD}(c). If pieces are not placed in accordance with a correct partition, then there will at some point be an extra $\II$ piece which does not fit in the open bucket. This will cause the player to lose as an $\II$ piece has length $4$ and cannot fit in the $2\times 2$ gap in the top $2$ rows (note that no $\II$ piece will stick partially out of a bucket since each bucket has a height that is a multiple of $4$).

    The closing sequence is $(\OO)$. The $\OO$ piece can only fit in the $2\times 2$ square formed during the priming sequence (see Figure \ref{20GHD}(e)), and it clears the top two rows, resetting the board to the initial state, except with a somewhat more filled bucket. 
    
    Once the player has received the entire sequence corresponding to each number in the 3-partition instance, including the final closing piece, the entire board will be full (see Figure \ref{20GHD}(e)), with the exception of the extra inaccessible bucket. Because of this, the player must have filled each of the buckets to exactly the right height, solving the 3-partition instance, otherwise they would have lost.

    For the entire duration of this sequence, every piece can be hard dropped into place as there are no overhangs in any of the buckets. Furthermore, each piece can be successfully maneuvered under 20G conditions. The $\OO$ pieces cannot fall down any buckets, so they can be safely slid to any location in the top $2$ rows, and the $\II$ pieces can move over the placed $\OO$ pieces to the only open bucket.
\end{proof}

\begin{corollary}
    The above reduction can be extended to NP-hardness for board clearing under "hard drops only" or "20G" conditions if we also allow for $\TT$ pieces.
\end{corollary}

\begin{remark}
    We have already proved that Tetris clearing under the "hard drops only" and "20G" models is hard with just two types of pieces. The primary reason for including hardness with a larger set of pieces is that this shows that, in both game modes, there is a board configuration and subset of pieces where the problems of surviving and clearing the board are both hard at the same time. In addition, our Tetris clearing results under the "20G" model does not include any proper subsets of $\{\II, \OO, \TT\}$, so this result is interesting in its own right.
\end{remark}

\begin{proof}

We begin with the above proof that survivability with $\{\II,\OO\}$ is hard. We extend the sequence with the finale sequence $(\TT, \II^t, \OO^{2n/3}, \TT^3)$. Figure \ref{20GHDClear} shows the clearing process. We begin by using the first $\TT$ piece to open the inaccessible bucket by clearing a row. We then use the $\II$ pieces to fill the previously inaccessible bucket, clearing all but the final two rows in the process. The final $\II$ piece protrudes one square from its bucket because of the row cleared by the first $\TT$. We use the $\OO$ pieces to fill the $\frac{2n}{3}\times {4}$ space on the right of the board. Finally, we place the remaining three $\TT$ pieces to fill the remaining space and completely clear the board. All of these pieces can be placed in both special game modes.
\end{proof}

% Hard Drops + 20G Board Clear
\begin{figure}[ht]
  \centering
  \begin{subfigure}[b]{0.25\textwidth}
    \centering
    \includegraphics[width=100pt]{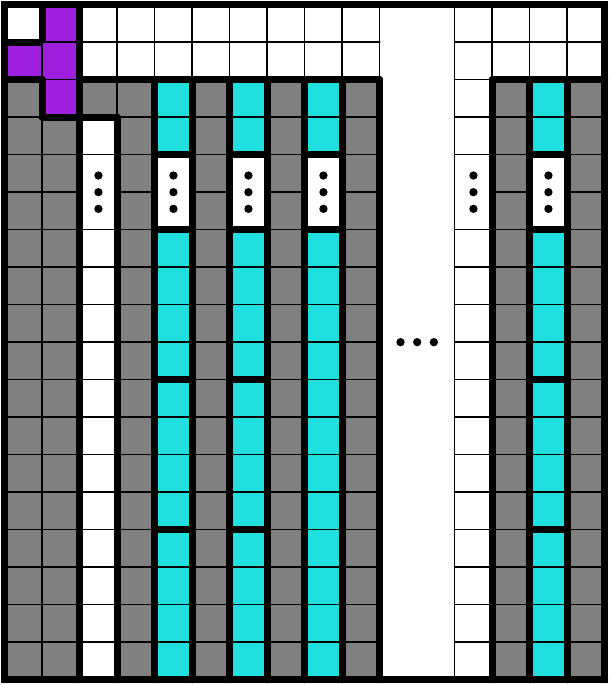}
    \caption{}
  \end{subfigure}
  \begin{subfigure}[b]{0.25\textwidth}
    \centering
    \includegraphics[width=100pt]{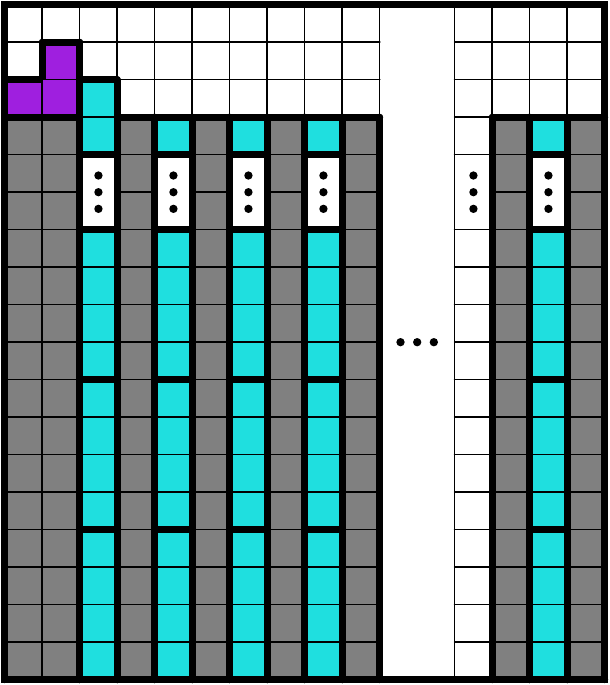}
    \caption{}
  \end{subfigure}
  \begin{subfigure}[b]{0.25\textwidth}
    \centering
    \includegraphics[width=100pt]{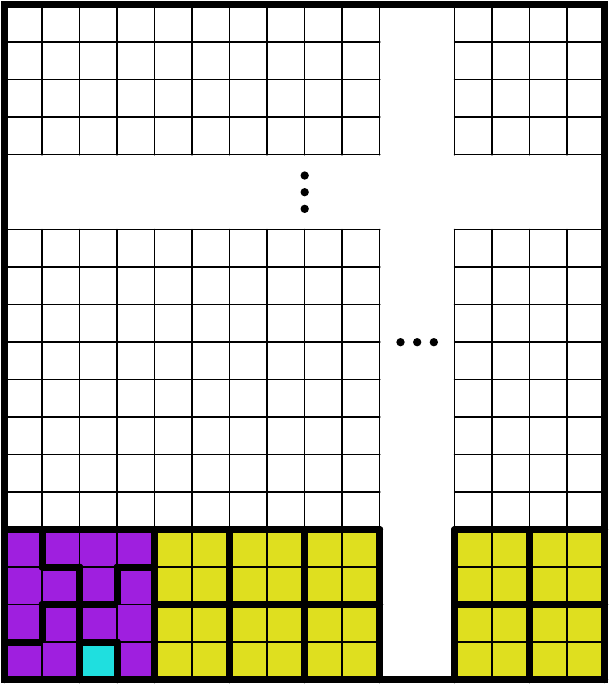}
    \caption{}
  \end{subfigure}
  \caption{The clearing procedure for "hard drops only" and "20G" using $\{\II, \OO, \TT\}$}
  \label{20GHDClear}
\end{figure}

\section{ASP-Completeness of Tetris}\label{sec:asp}

Even though the reductions in Section~\ref{sec:2piecetetris} are sufficient to prove \#P-hardness, the reductions are not parsimonious, so they cannot be used to prove ASP-completeness. Indeed, the "blocking" bottles paradigm likely cannot be used to show ASP-completeness as there are many ways to permute the pieces that block all but one bottle. Thus, for ASP-completeness, we turn back to the "priming" buckets paradigm in \cite{Tetris_IJCGA}:

\begin{theorem}
    Tetris clearing with SRS is ASP-complete even if the type of pieces in the sequence given to the player is restricted to either $\{\II, \TT, \LL\}$ or $\{\II, \TT, \JJ\}$.
\end{theorem}

\begin{proof}

    First we discuss the $\{\II, \TT, \LL\}$ case. We give a parsimonious reduction from Numerical 3-Dimensional Matching with Distinct Integers; refer to Figure \ref{ASPSetup}(a), which shows the bucket structure plus rightmost columns for $\{\II, \TT, \LL\}$.

% ASP Setup Images
\begin{figure}[ht]
  \centering
  \begin{subfigure}[b]{0.25\textwidth}
    \centering
    \includegraphics[width=100pt]{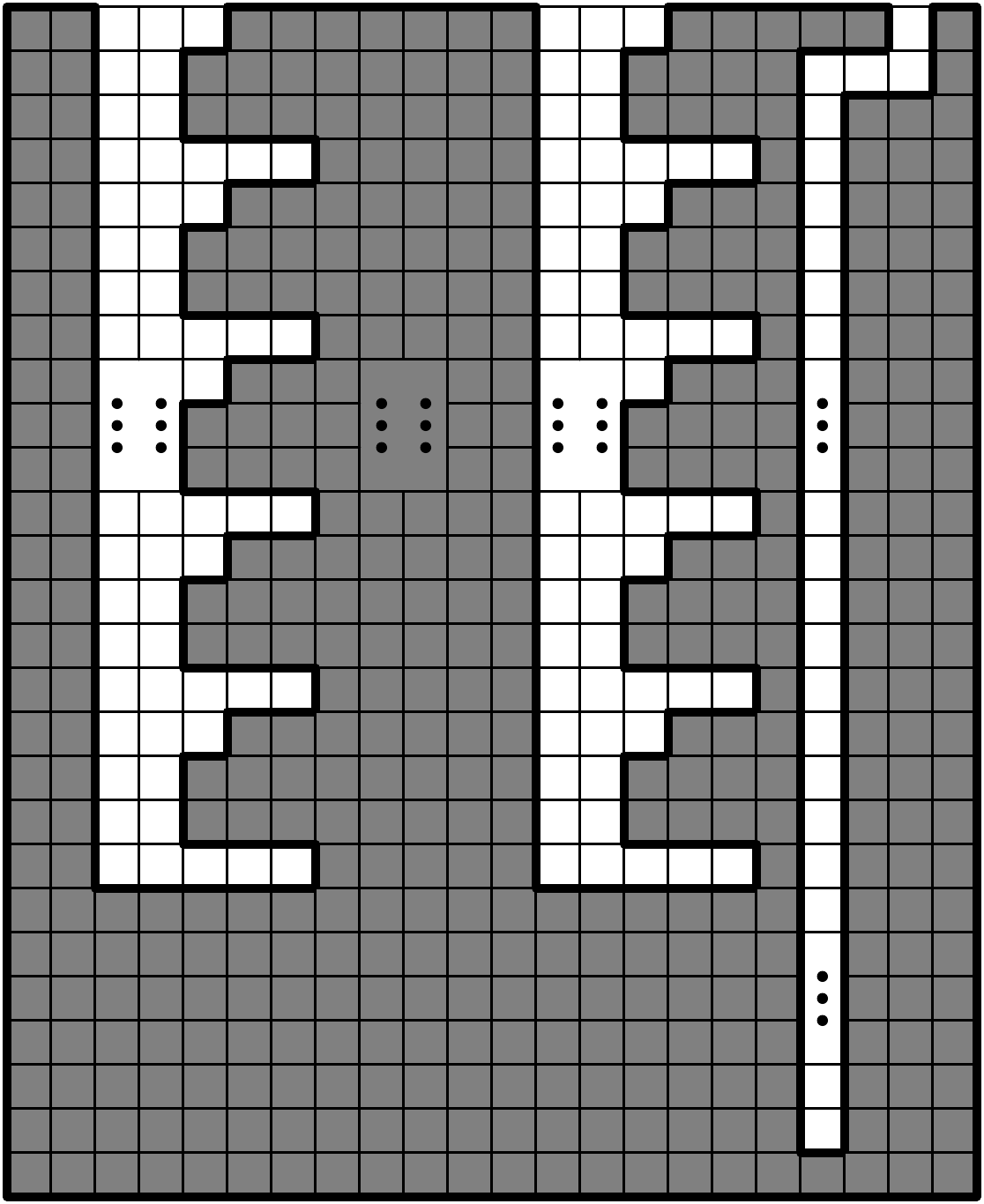}
    \caption{}
  \end{subfigure}
  \begin{subfigure}[b]{0.25\textwidth}
    \centering
    \includegraphics[width=100pt]{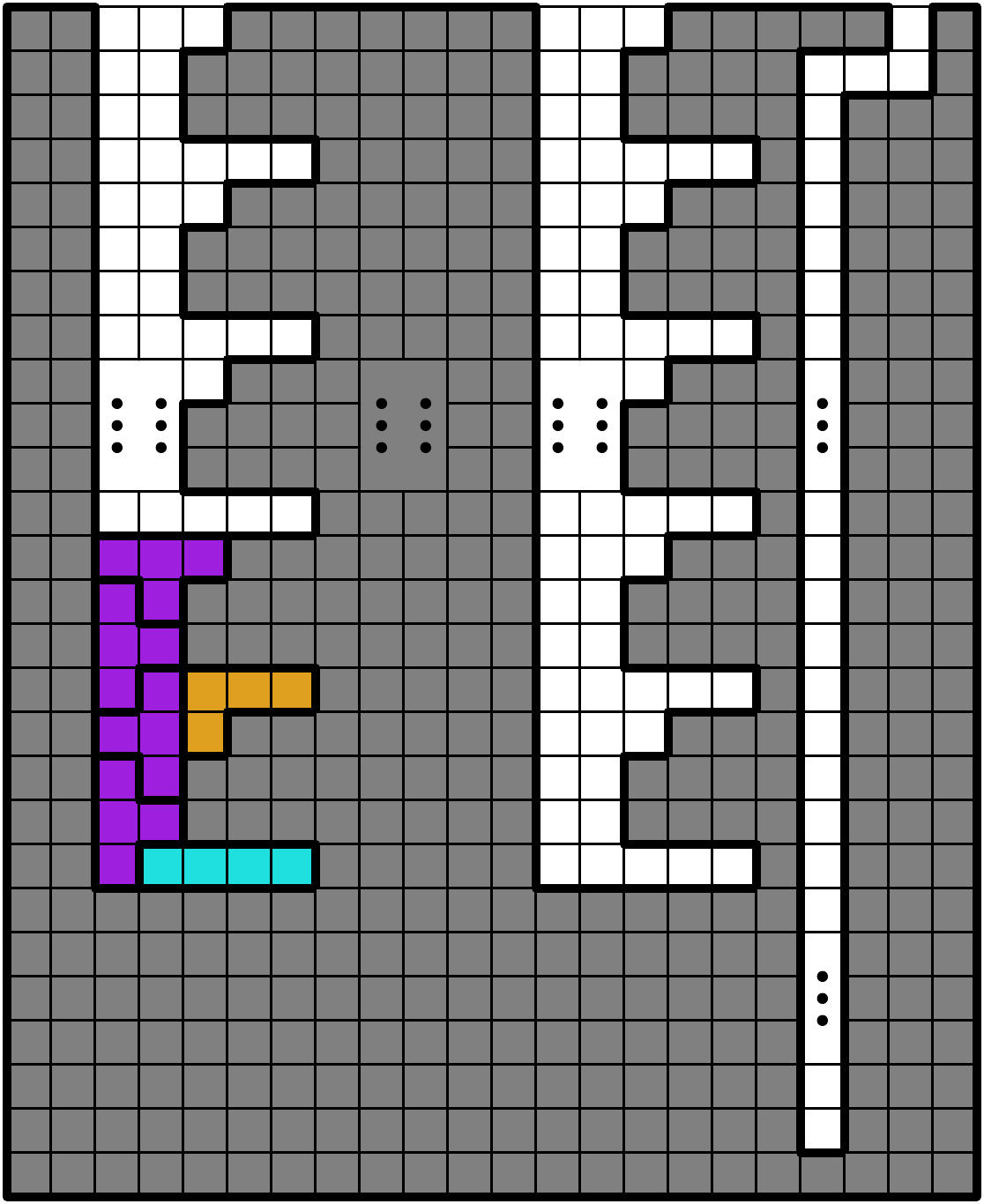}
    \caption{}
  \end{subfigure}
  \begin{subfigure}[b]{0.25\textwidth}
    \centering
    \includegraphics[width=100pt]{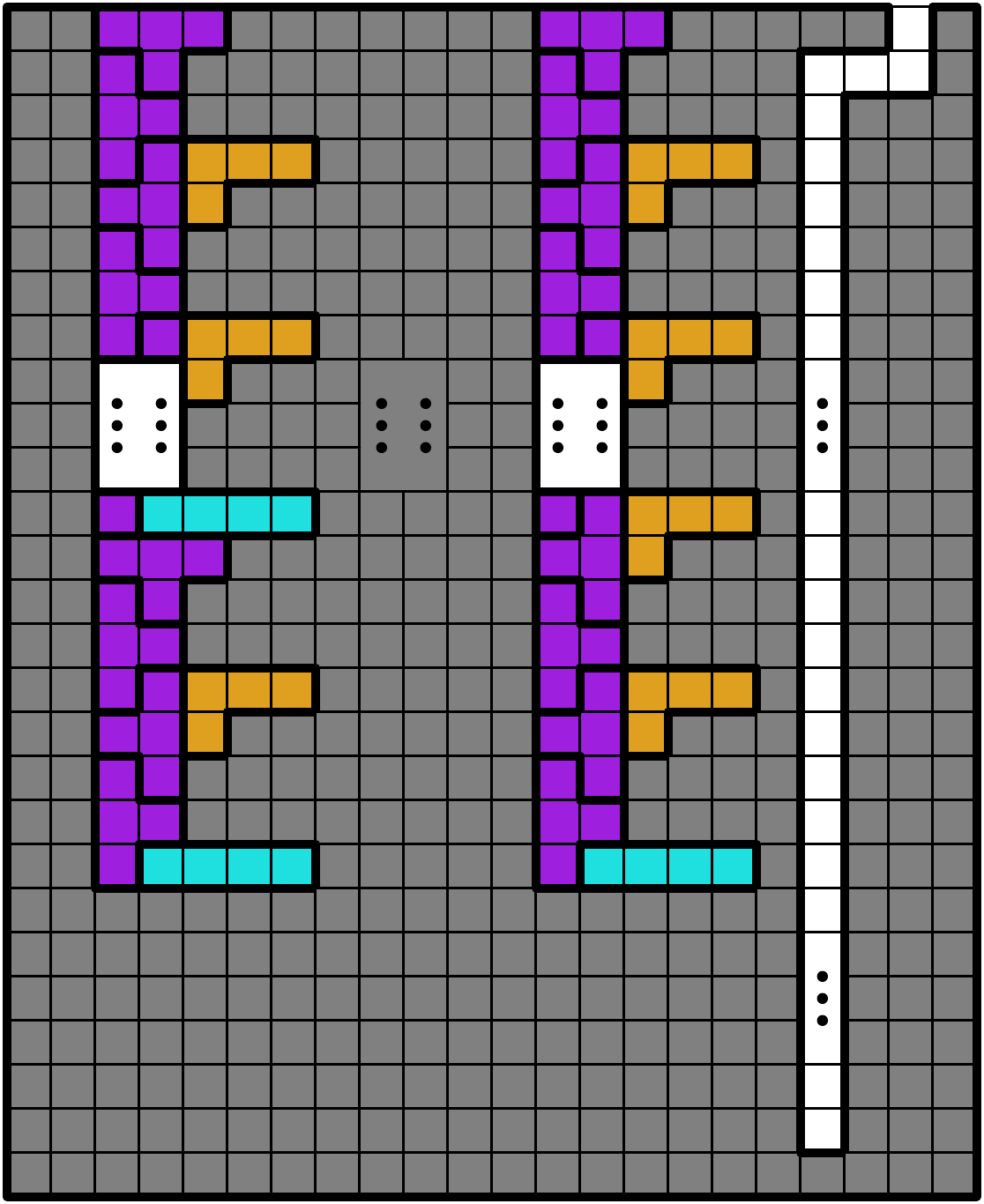}
    \caption{}
  \end{subfigure}
  \begin{subfigure}[b]{0.25\textwidth}
    \centering
    \includegraphics[width=100pt]{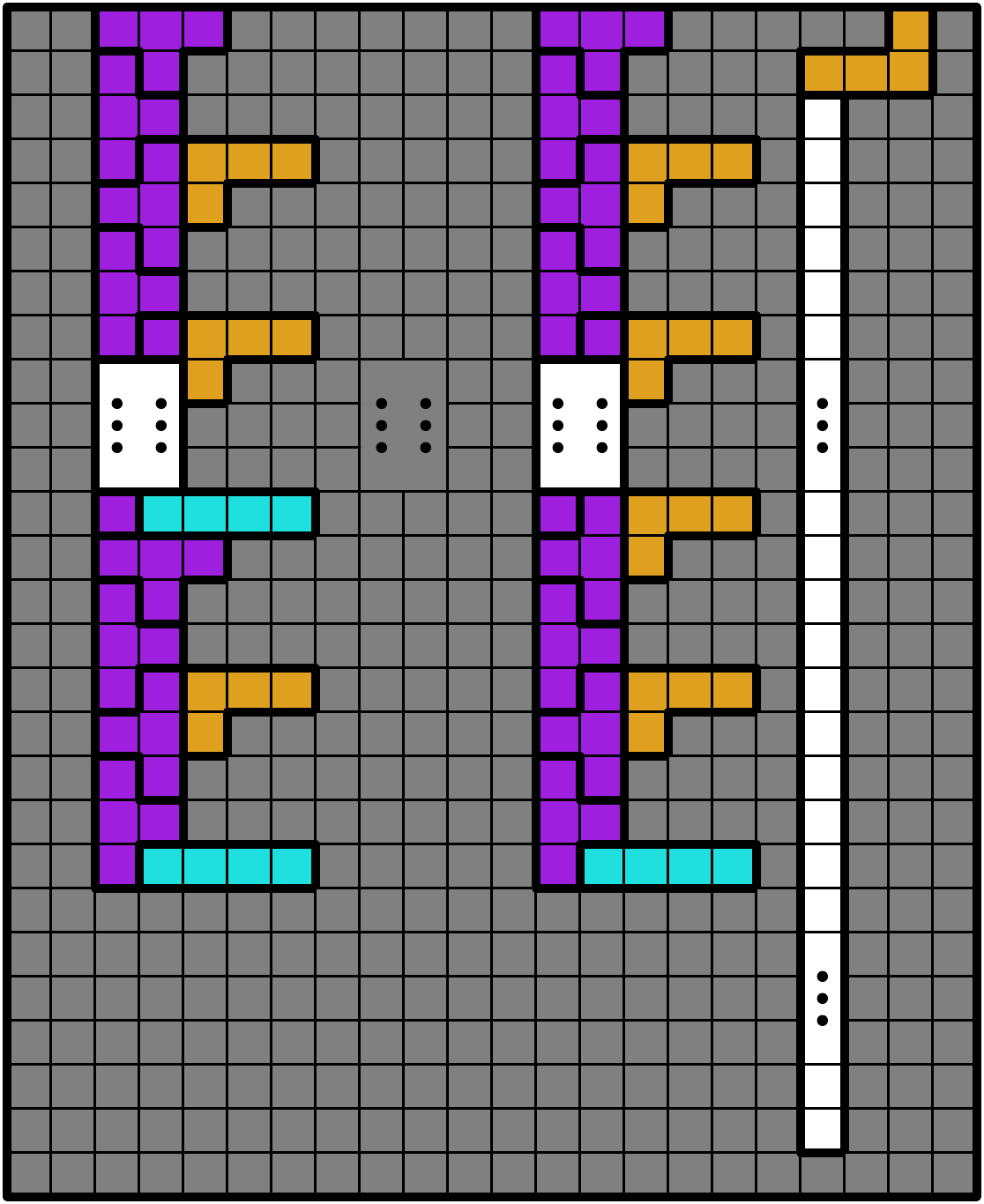}
    \caption{}
  \end{subfigure}
  \begin{subfigure}[b]{0.25\textwidth}
    \centering
    \includegraphics[width=100pt]{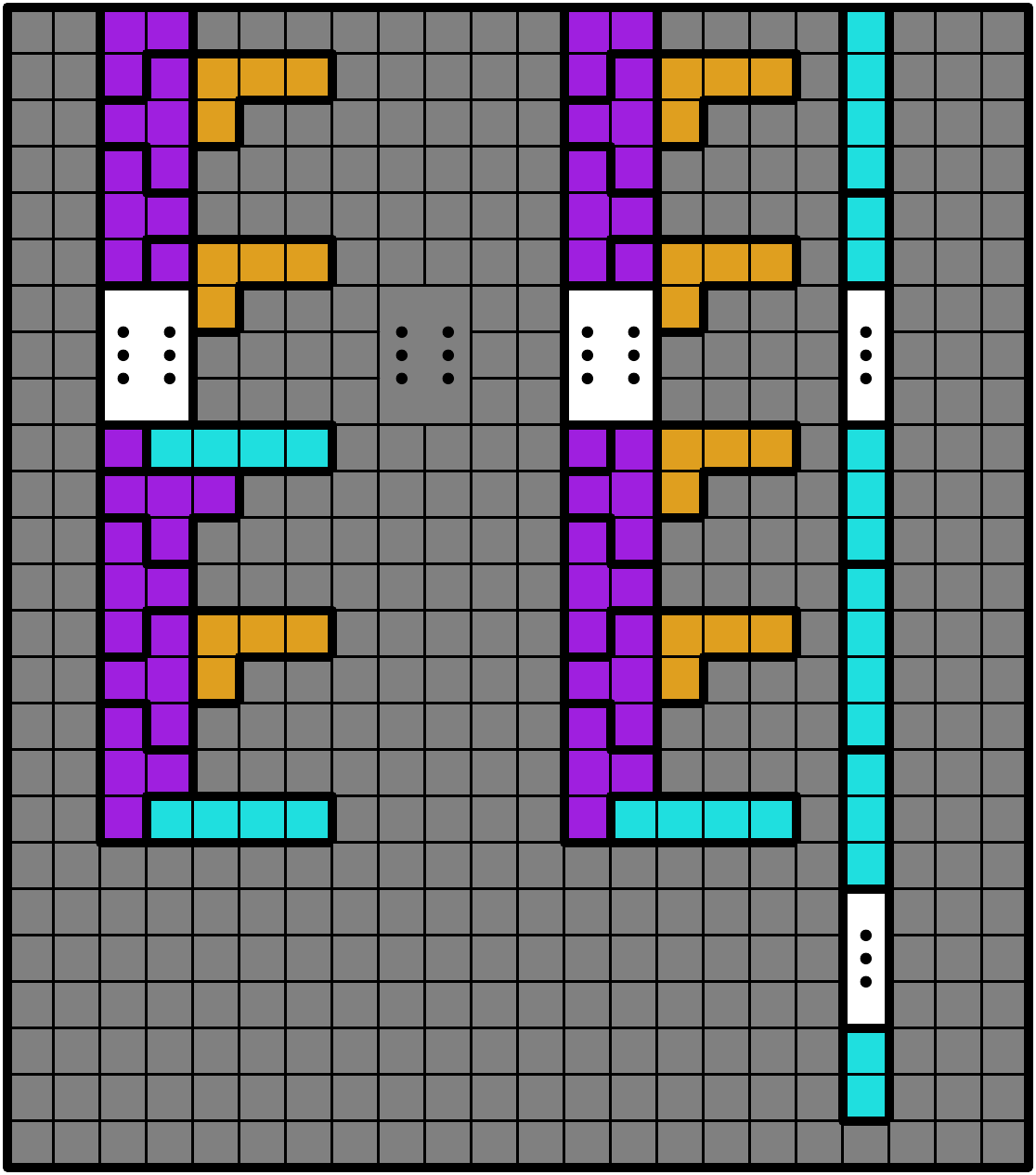}
    \caption{}
  \end{subfigure}
  \caption{The bucket structure plus rightmost columns for $\{\II, \TT, \LL\}$. (b) shows how the $\II$ piece must prime a bucket (requires the $\II$ spin in Figure \ref{ISpin}) and how the remainder of the pieces in a sequence must fit in the bucket ($m=2$ here). (c) shows how the board looks like before the finale sequence. (d--e) show how the pieces in the finale sequence must be placed (requires the $\LL$-piece version of the $\JJ$ spin in Figure \ref{JSpin3}).}
  \label{ASPSetup}
\end{figure}

    Here, a "unit" is the pattern that repeats every four rows. A positive integer $m$ is encoded by the sequence of pieces $(\II, (\TT, \LL, \TT)^{m-1}, \TT, \TT)$. The finale sequence is $(\LL, \II^N)$, where $N=\text{poly}(nt)$ is much larger than the height of the buckets, and is used to clear the rightmost columns after the buckets have been filled.

    In this case, the $\II$ piece serves as the "primer", and must be placed as indicated in Figure \ref{ASPSetup}(b) (the placement is possible due to $\II$ spins). The $\TT$ pieces cannot be placed in a non-primed bucket without blocking off certain holes, particularly the squares in which an $\II$ piece or an $\LL$ piece must be placed, and misplacing an $\LL$ piece (i.e., putting it in a different, non-primed bucket, putting it where the $\LL$ piece goes during the finale sequence, or putting it too high in the bucket) causes the next two $\TT$ pieces to block off squares in which an $\II$ piece or an $\LL$ piece must be placed. Thus, once the $\II$ piece is placed in a bucket, the placements of the rest of the pieces encoding the positive integer $m$ are forced. Further discussion on and figures for improper piece placements can be found in Appendix \ref{appendix:aspclogs}.

    Lastly, to make the reduction parsimonious, from the instance $$A = \{a_1, a_2, \ldots , a_n\}, B = \{b_1, b_2, \ldots , b_n\},\text{ and }C = \{c_1, c_2, \ldots , c_n\}$$ of Numerical 3-Dimensional Matching with Distinct Integers, we scale the $a_i$, $b_i$, and $c_i$ and add/subtract constants such that each group must consist of exactly one $a_i$, one $b_j$, and one $c_k$ (i.e., no group with at least two elements from any of $A$, $B$, or $C$ can sum to the target sum). In addition, we "pre-fill" all of the buckets so that the $i$th bucket from the left is already filled up to $a_i$ units, and the sequence of pieces given to the player consist only of the pieces in the $b_i$ sequences, the pieces in the $c_i$ sequences, and the pieces in the finale sequence. This ensures that the subsets have a fixed ordering and must consist of exactly one $a_i$, one $b_j$, and one $c_k$. Combined with the fact that all the piece placements after each $\II$ piece in a sequence encoding a positive integer are forced, this means that each solution to the Numerical 3-Dimensional Matching instance corresponds to exactly one way to clear the Tetris board, and vice versa.

    Therefore, we have a parsimonious reduction from Numerical 3-Dimensional Matching with Distinct Integers to Tetris clearing with $\{\II, \TT, \LL\}$, meaning that Tetris clearing with SRS is ASP-complete even if the type of pieces in the sequence given to the player is restricted to $\{\II, \TT, \LL\}$.

    We get a similar argument for $\{\II, \TT, \JJ\}$ by vertical symmetry; the $\II$ spin in Figure \ref{ISpin} still works when mirrored even though the kick tests for $\II$ are not vertically symmetric.
\end{proof}

Demos of the bucket structure can be found at \url{https://jstris.jezevec10.com/map/80170} for $\{\II, \TT, \LL\}$ and at \url{https://jstris.jezevec10.com/map/83325} for $\{\II, \TT, \JJ\}$.

%\section{Lumines}

%Lumines, like Tetris, is a similar block-dropping game. Lumines consists of black and white $2 \times 2$ blocks that fall, which you can rotate. Every so often a `Time Line' passes and erases all $2 \times 2$ squares of the same color.  

%[Define and introduce Lumines]

%\subsection{NP-hardness of Lumines}

%We prove that Lumines is NP-hard by doing a reduction from 3-partition, similarly to section 3.

%Specifically, we will use bottles that are opened and closed. 

%[Proof]

%\section{$\leq$n-tris Hardness from an Empty Board}

%DESCRIBE LOSS CONDITION NOT PARTIAL TOP OUT

%\subsection{$\leq$ 4c-tris is NP-Hard}

%\subsection{$\leq$ $c+1$-tris is NP-}

%\begin{theorem}
%    Odd-COLUMN Empty ($\leq c+1$)-TRIS is NP-complete for any $c $
%\end{theorem}

\section{Open Problems}\label{sec:openprobs}

One big open problem that still remains is the computational complexity of Tetris clearing with SRS if the player is only given one piece type (for example, if the sequence consists of entirely $\TT$ pieces). In this case, the "blocking" bottles paradigm no longer works, because the same piece type cannot be used both to block and to fill bottles without the reduction breaking, so a proof of hardness would involve an entirely different setup. It is also possible that Tetris clearing with SRS and with only one piece type is in P.
For example, \cite{Tetris_IJCGA} conjectures polynomial time for the $\II$ piece type.

Similarly, it is open whether or not Tetris clearing with SRS is ASP-complete for subsets of pieces that are not supersets of $\{\II, \TT, \JJ\}$ or $\{\II, \TT, \LL\}$. One could likely construct similar structures for other 3-piece subsets, but arguing whether Tetris clearing with SRS is ASP-complete for 1- or 2-piece subsets may require different ideas.

Some open questions arise regarding whether our results can be extended if we consider different objectives or add additional features. For example, can we establish results for 2-piece subsets, similar to those in Section \ref{sec:2piecetetris}, for Tetris survival? If we add a "holding" function, where the player can put one piece aside for later use, can get similar results?

Modern variants of Tetris also use different random generators to ensure that the player does not receive the same piece arbitrarily many times in a row. One of the simplest random generators is called a \defn{7-bag randomizer}. For this randomizer, the sequence of pieces is divided into groups (or "bags") of 7, each group containing one of each tetromino in one of $7! = 5{,}040$ possible orderings. Can we show NP-hardness even if the sequence of pieces has to be able to be generated from this randomizer?

\section*{Acknowledgments}

This paper was initiated during open problem solving in the MIT class
on Algorithmic Lower Bounds: Fun with Hardness Proofs (6.5440)
taught by Erik Demaine in Fall 2023.
We thank the other participants of that class for helpful discussions
and providing an inspiring atmosphere.
Figures drawn with SVG Tiler
(\url{https://github.com/edemaine/svgtiler}).

%% file: appendix.tex
\appendix
\clearpage

\section{Spins}\label{appendix:spins}

In this appendix, we note the important spins used in our constructions. If the diagram has dark gray blocks indicated, then the dark gray blocks indicate the blocks required to perform the spin (as those blocks cause all the kick tests before the successful test to fail).

\subsection{Spins involving $\II$ Pieces}\label{appendix:ispins}

\begin{figure}[!ht]
  \centering
  \begin{subfigure}[b]{0.3\textwidth}
    \centering
    \includegraphics[width=60pt]{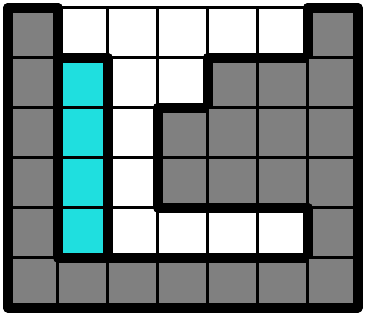}
    \caption{}
  \end{subfigure}
  \begin{subfigure}[b]{0.3\textwidth}
    \centering
    \includegraphics[width=60pt]{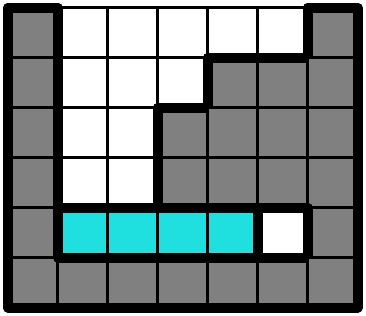}
    \caption{}
  \end{subfigure}
  \begin{subfigure}[b]{0.3\textwidth}
    \centering
    \includegraphics[width=60pt]{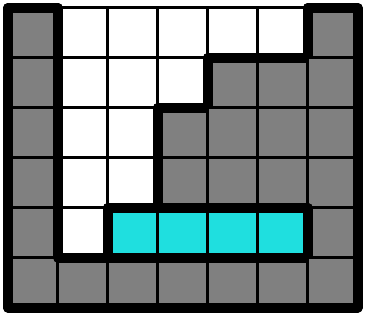}
    \caption{}
  \end{subfigure}
  \caption{An $\II$ spin used in the setup for ASP-completeness for the subset $\{\II, \TT, \LL\}$. Depending on the direction of rotation and the location of the rotation center (i.e. if the center is on the left side of the $\II$ piece or the right side of the $\II$ piece), the rotation will send the piece to either (b) or (c); however, no rotation will send the center of the piece upwards.}
  \label{ISpin}
\end{figure}

\subsection{Spins involving $\JJ/\LL$ Pieces}\label{appendix:jspins}

\begin{figure}[!ht]
  \centering
  \begin{subfigure}[b]{0.3\textwidth}
    \centering
    \includegraphics[width=60pt]{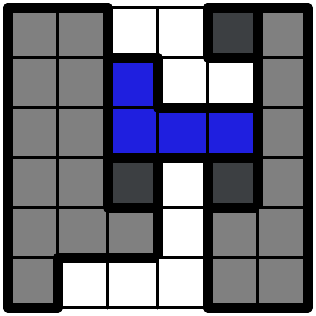}
    \caption{}
  \end{subfigure}
  \begin{subfigure}[b]{0.3\textwidth}
    \centering
    \includegraphics[width=60pt]{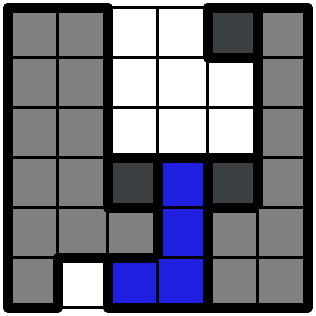}
    \caption{}
  \end{subfigure}
  \caption{A $\JJ$ spin used in the setups for $\{\OO, \JJ\}$, $\{\JJ, \ZZ\}$, and $\{\JJ, \LL\}$}
  \label{JSpin4}
\end{figure}

\begin{figure}[!ht]
  \centering
  \begin{subfigure}[b]{0.3\textwidth}
    \centering
    \includegraphics[width=60pt]{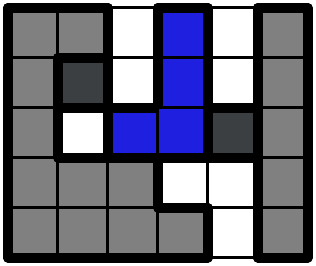}
    \caption{}
  \end{subfigure}
  \begin{subfigure}[b]{0.3\textwidth}
    \centering
    \includegraphics[width=60pt]{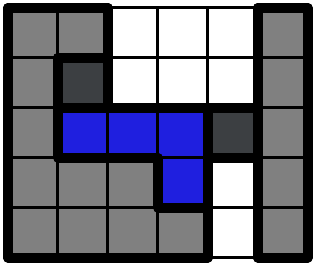}
    \caption{}
  \end{subfigure}
  \caption{A $\JJ$ spin used in the setup for $\{\JJ, \SS\}$}
  \label{JSpin1}
\end{figure}

\begin{figure}[!ht]
  \centering
  \begin{subfigure}[b]{0.3\textwidth}
    \centering
    \includegraphics[width=60pt]{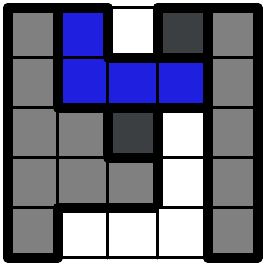}
    \caption{}
  \end{subfigure}
  \begin{subfigure}[b]{0.3\textwidth}
    \centering
    \includegraphics[width=60pt]{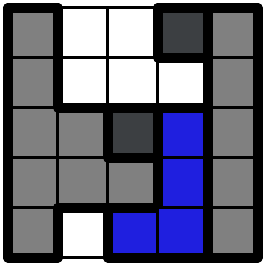}
    \caption{}
  \end{subfigure}
  \caption{A $\JJ$ spin used in the setup for $\{\JJ, \SS\}$}
  \label{JSpin2}
\end{figure}

\begin{figure}[!ht]
  \centering
  \begin{subfigure}[b]{0.3\textwidth}
    \centering
    \includegraphics[width=60pt]{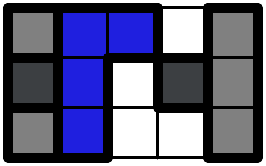}
    \caption{}
  \end{subfigure}
  \begin{subfigure}[b]{0.3\textwidth}
    \centering
    \includegraphics[width=60pt]{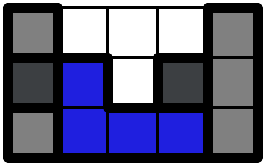}
    \caption{}
  \end{subfigure}
  \caption{A $\JJ$ spin used in the setups for $\{\JJ, \SS\}$, $\{\JJ, \LL\}$, and the ASP-completeness for $\{\II, \TT, \LL\}$}
  \label{JSpin3}
\end{figure}

\begin{figure}[!ht]
  \centering
  \begin{subfigure}[b]{0.3\textwidth}
    \centering
    \includegraphics[width=60pt]{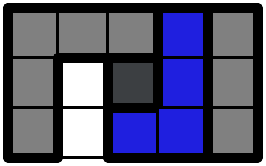}
    \caption{}
  \end{subfigure}
  \begin{subfigure}[b]{0.3\textwidth}
    \centering
    \includegraphics[width=60pt]{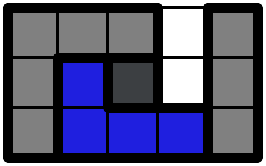}
    \caption{}
  \end{subfigure}
  \caption{A $\JJ$ spin used in the setup for $\{\JJ, \LL\}$ (although the version of the spin used is the $\LL$-spin version)}
  \label{JSpin5}
\end{figure}

\begin{figure}[!ht]
  \centering
  \begin{subfigure}[b]{0.16\textwidth}
    \centering
    \includegraphics[width=60pt]{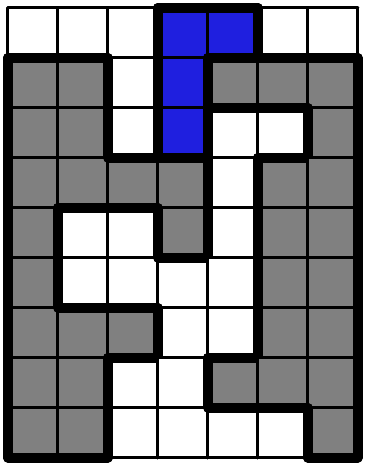}
    \caption{}
  \end{subfigure}
  \begin{subfigure}[b]{0.16\textwidth}
    \centering
    \includegraphics[width=60pt]{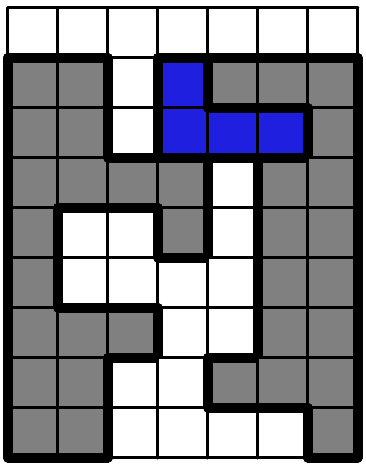}
    \caption{}
  \end{subfigure}
  \begin{subfigure}[b]{0.16\textwidth}
    \centering
    \includegraphics[width=60pt]{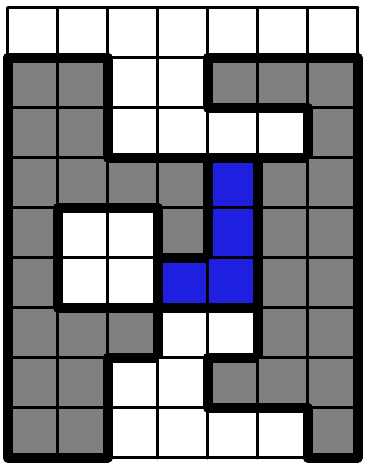}
    \caption{}
  \end{subfigure}
  \begin{subfigure}[b]{0.16\textwidth}
    \centering
    \includegraphics[width=60pt]{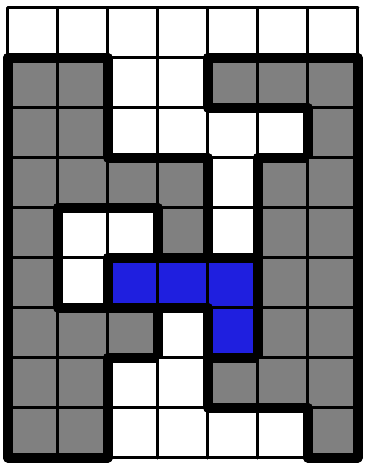}
    \caption{}
  \end{subfigure}
  \begin{subfigure}[b]{0.16\textwidth}
    \centering
    \includegraphics[width=60pt]{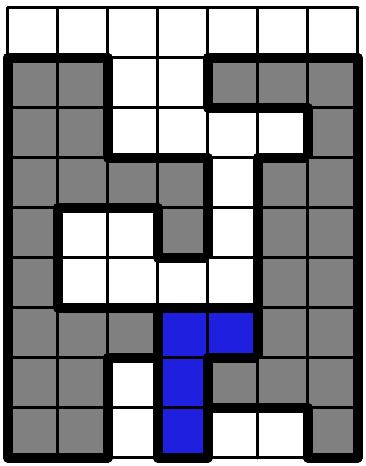}
    \caption{}
  \end{subfigure}
  \caption{A sequence of $\JJ$ spins used in the setup for $\{\OO, \JJ\}$}
  \label{JSpin-long1}
\end{figure}

$$ $$

$$ $$

\subsection{Spins involving $\SS/\ZZ$ Pieces}\label{appendix:sspins}

For these spins, the rotation center is important, so the rotation center is marked with a dot.

\begin{figure}[!ht]
  \centering
  \begin{subfigure}[b]{0.3\textwidth}
    \centering
    \includegraphics[width=60pt]{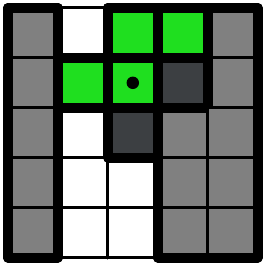}
    \caption{}
  \end{subfigure}
  \begin{subfigure}[b]{0.3\textwidth}
    \centering
    \includegraphics[width=60pt]{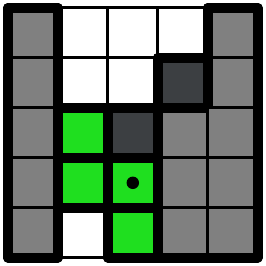}
    \caption{}
  \end{subfigure}
  \caption{An $\SS$ spin used in the setups for $\{\OO, \SS\}$, $\{\SS, \TT\}$, and $\{\SS, \ZZ\}$, used for getting $\SS$ pieces into the body portions of the bottles}
  \label{SSpin1}
\end{figure}

\begin{figure}[!ht]
  \centering
  \begin{subfigure}[b]{0.3\textwidth}
    \centering
    \includegraphics[width=60pt]{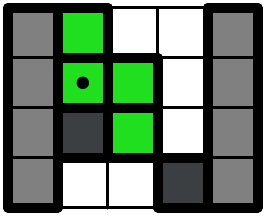}
    \caption{}
  \end{subfigure}
  \begin{subfigure}[b]{0.3\textwidth}
    \centering
    \includegraphics[width=60pt]{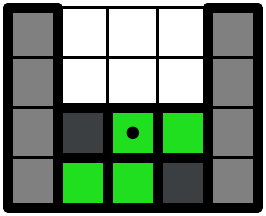}
    \caption{}
  \end{subfigure}
  \begin{subfigure}[b]{0.3\textwidth}
    \centering
    \includegraphics[width=60pt]{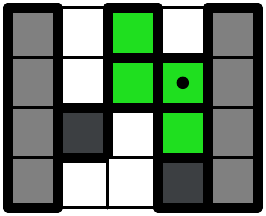}
    \caption{}
  \end{subfigure}
  \caption{An $\SS$ spin ((a) to (b)) used in the setups for $\{\II, \SS\}$, $\{\OO, \SS\}$, $\{\SS, \TT\}$, $\{\SS, \ZZ\}$, and $\{\JJ, \ZZ\}$ (and similar), used either for clogging a bottle during the priming sequence or for re-opening a bottle during the closing sequence (or both). (c) indicates what happens if one more rotation is applied when the $\SS$ piece is tucked in as in (b).}
  \label{SSpin2}
\end{figure}

\begin{figure}[!ht]
  \centering
  \begin{subfigure}[b]{0.15\textwidth}
    \centering
    \includegraphics[width=60pt]{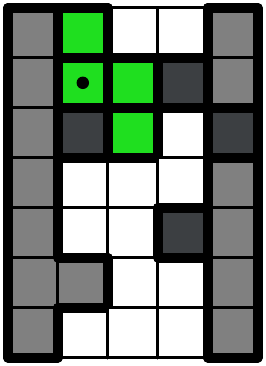}
    \caption{}
  \end{subfigure}
  \begin{subfigure}[b]{0.15\textwidth}
    \centering
    \includegraphics[width=60pt]{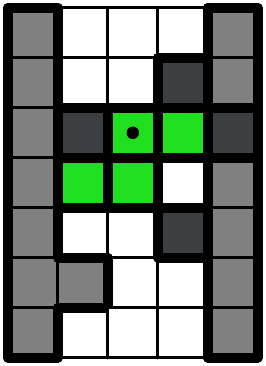}
    \caption{}
  \end{subfigure}
  \begin{subfigure}[b]{0.15\textwidth}
    \centering
    \includegraphics[width=60pt]{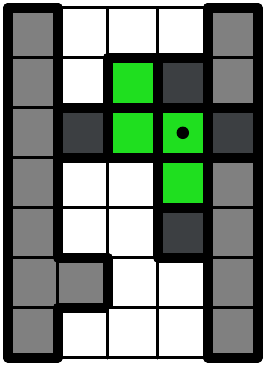}
    \caption{}
  \end{subfigure}
  \begin{subfigure}[b]{0.15\textwidth}
    \centering
    \includegraphics[width=60pt]{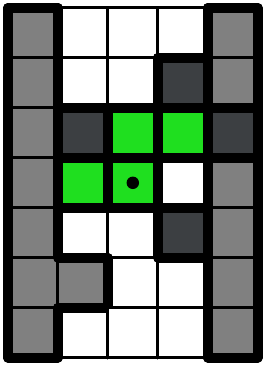}
    \caption{}
  \end{subfigure}
  \begin{subfigure}[b]{0.15\textwidth}
    \centering
    \includegraphics[width=60pt]{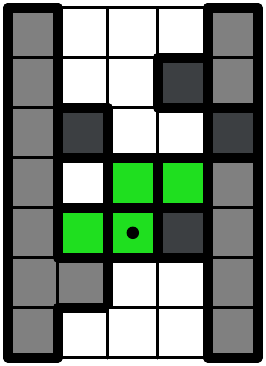}
    \caption{}
  \end{subfigure}
  \begin{subfigure}[b]{0.15\textwidth}
    \centering
    \includegraphics[width=60pt]{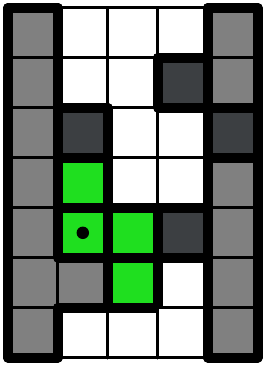}
    \caption{}
  \end{subfigure}
  \caption{A sequence of $\SS$ spins used in the setups for $\{\OO, \SS\}$ and $\{\SS, \ZZ\}$}
  \label{SSpin-long1}
\end{figure}

\begin{figure}[!ht]
  \centering
  \begin{subfigure}[b]{0.15\textwidth}
    \centering
    \includegraphics[width=60pt]{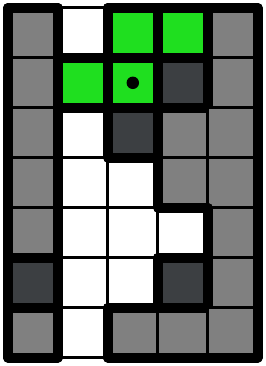}
    \caption{}
  \end{subfigure}
  \begin{subfigure}[b]{0.15\textwidth}
    \centering
    \includegraphics[width=60pt]{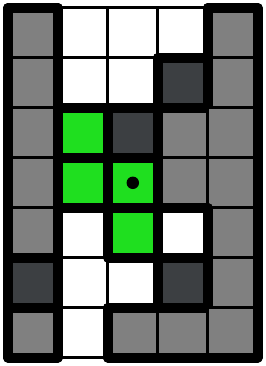}
    \caption{}
  \end{subfigure}
  \begin{subfigure}[b]{0.15\textwidth}
    \centering
    \includegraphics[width=60pt]{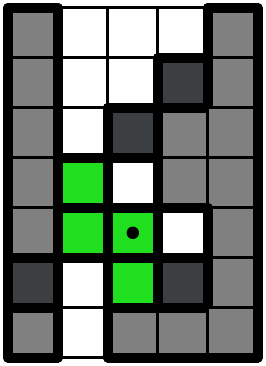}
    \caption{}
  \end{subfigure}
  \begin{subfigure}[b]{0.15\textwidth}
    \centering
    \includegraphics[width=60pt]{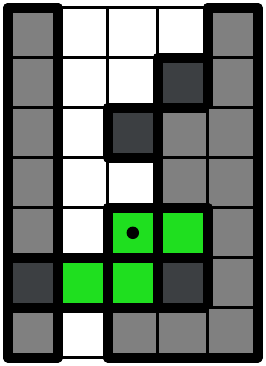}
    \caption{}
  \end{subfigure}
  \begin{subfigure}[b]{0.15\textwidth}
    \centering
    \includegraphics[width=60pt]{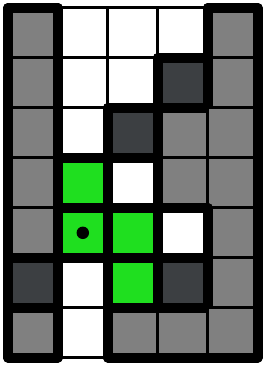}
    \caption{}
  \end{subfigure}
  \begin{subfigure}[b]{0.15\textwidth}
    \centering
    \includegraphics[width=60pt]{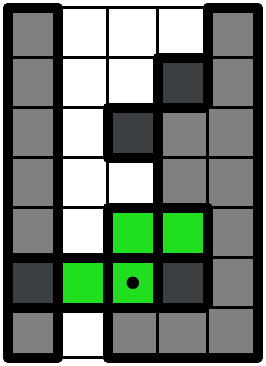}
    \caption{}
  \end{subfigure}
  \caption{A sequence of $\SS$ spins used in the setup for $\{\SS, \TT\}$}
  \label{SSpin-long2}
\end{figure}

$$ $$

$$ $$

\subsection{Spins involving $\TT$ Pieces}\label{appendix:tspins}

\begin{figure}[!ht]
  \centering
  \begin{subfigure}[b]{0.3\textwidth}
    \centering
    \includegraphics[width=60pt]{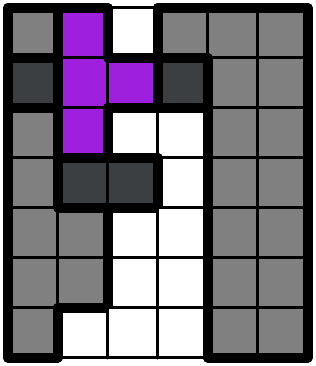}
    \caption{}
  \end{subfigure}
  \begin{subfigure}[b]{0.3\textwidth}
    \centering
    \includegraphics[width=60pt]{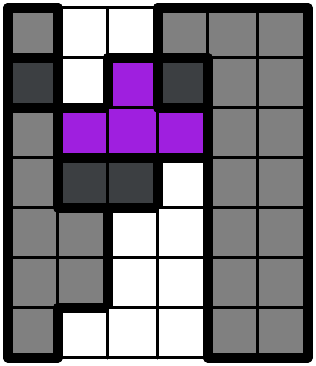}
    \caption{}
  \end{subfigure}
  \begin{subfigure}[b]{0.3\textwidth}
    \centering
    \includegraphics[width=60pt]{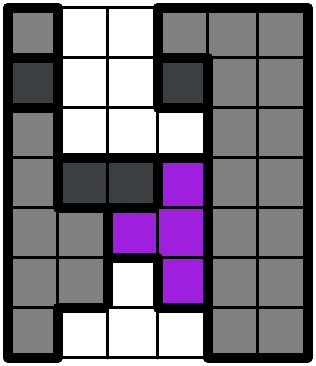}
    \caption{}
  \end{subfigure}
  \caption{A $\TT$-spin sequence used in the setup for $\{\OO, \TT\}$}
  \label{TSpin}
\end{figure}

$$ $$

$$ $$

\section{Clogs}

In this appendix, we illustrate improper and impossible piece placements that do not directly overflow the bottle and show how the improper piece placements break the bottle setup (or indirectly cause overflow). In all of our analyses (except for $\{\OO, \TT\}$ and the last clog for $\{\OO, \SS\}$), we assume that all top segments above where the clog happens, except for possibly the top segment directly above where the clog happens, get cleared properly (as we can always consider where the "highest" clog happens). Some of the piece placements that are not stated as being impossible may still not be possible, but we include them so that our analysis is thorough.

Note that the bottle structures in Figures \ref{ILSetup} and \ref{I+OtherSetup} do not have nontrivial improper piece placements (i.e. piece placements that do not cause overflow), so we focus on other setups.

\subsection{Clogs for $\{\OO, \JJ\}$}\label{appendix:ojclogs}

There are no improper ways to place the pieces in the priming and closing sequences that do not directly overflow the bottle, so we only look at the improper ways to place $\JJ$ pieces in the filling sequences into the neck portion of a bottle, as in Figure \ref{OJJClogs}.

\begin{figure}[!ht]
  \centering
  \begin{subfigure}[b]{0.12\textwidth}
    \centering
    \includegraphics[width=45pt]{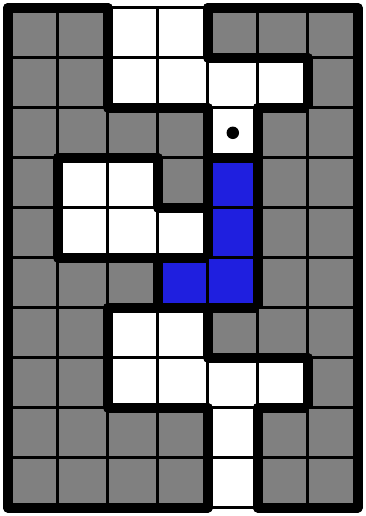}
    \caption{}
  \end{subfigure}
  \begin{subfigure}[b]{0.12\textwidth}
    \centering
    \includegraphics[width=45pt]{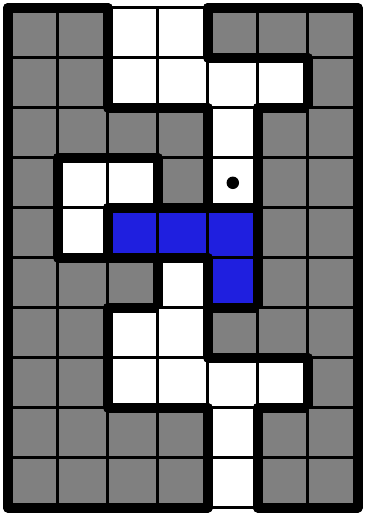}
    \caption{}
  \end{subfigure}
  \begin{subfigure}[b]{0.12\textwidth}
    \centering
    \includegraphics[width=45pt]{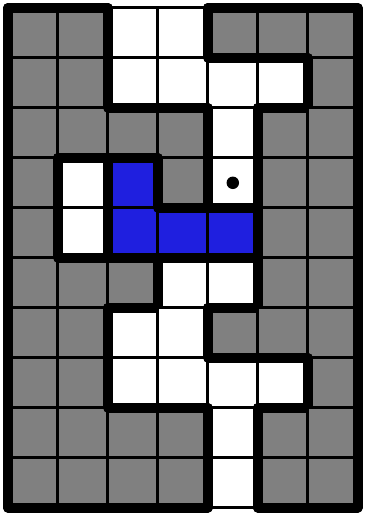}
    \caption{}
  \end{subfigure}
  \begin{subfigure}[b]{0.12\textwidth}
    \centering
    \includegraphics[width=45pt]{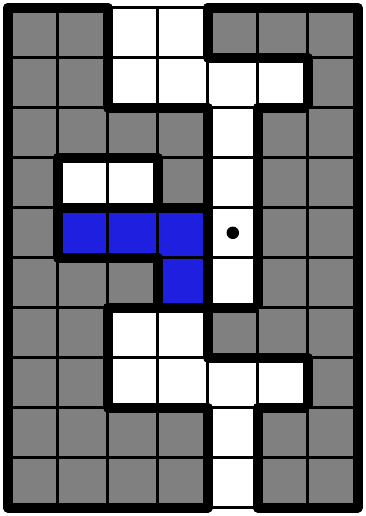}
    \caption{}
  \end{subfigure}
  \begin{subfigure}[b]{0.12\textwidth}
    \centering
    \includegraphics[width=45pt]{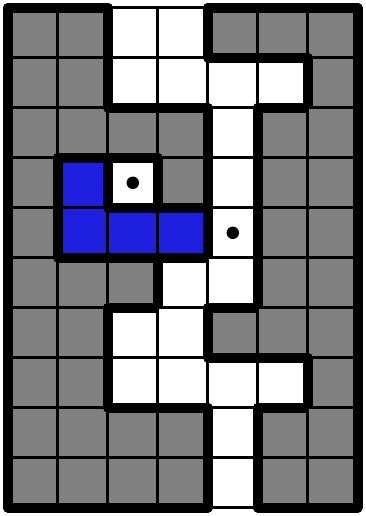}
    \caption{}
  \end{subfigure}
  \begin{subfigure}[b]{0.12\textwidth}
    \centering
    \includegraphics[width=45pt]{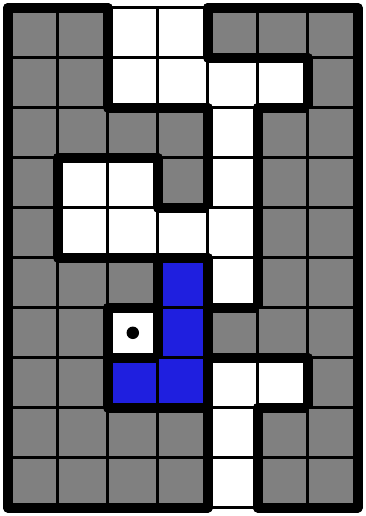}
    \caption{}
  \end{subfigure}
  \begin{subfigure}[b]{0.12\textwidth}
    \centering
    \includegraphics[width=45pt]{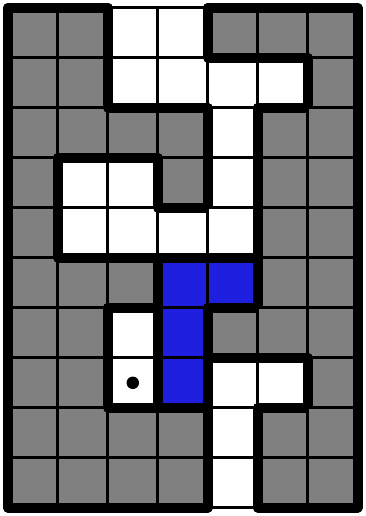}
    \caption{}
  \end{subfigure}
  \begin{subfigure}[b]{0.12\textwidth}
    \centering
    \includegraphics[width=45pt]{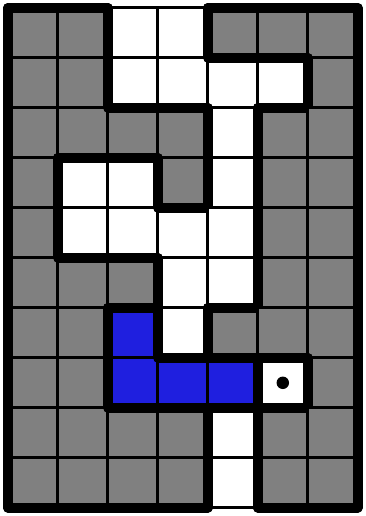}
    \caption{}
  \end{subfigure}
  \begin{subfigure}[b]{0.12\textwidth}
    \centering
    \includegraphics[width=45pt]{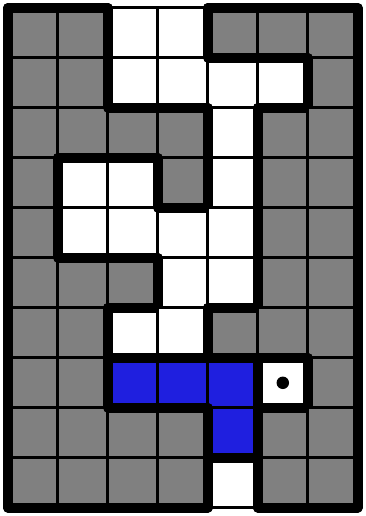}
    \caption{}
  \end{subfigure}
  \begin{subfigure}[b]{0.12\textwidth}
    \centering
    \includegraphics[width=45pt]{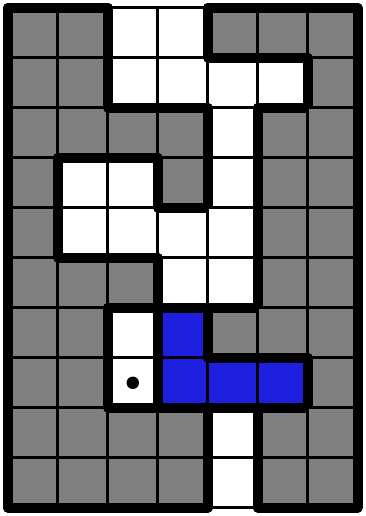}
    \caption{}
  \end{subfigure}
  \begin{subfigure}[b]{0.12\textwidth}
    \centering
    \includegraphics[width=45pt]{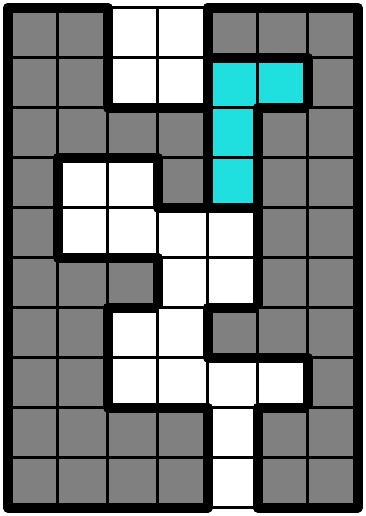}
    \caption{}
  \end{subfigure}
  \caption{Clogs or impossible piece placements (in the case of (k)) involving $\JJ$ pieces in the $\{\OO, \JJ\}$ setup}
  \label{OJJClogs}
\end{figure}

The first five improper piece placements ((a) through (e)) prevent additional $\JJ$ pieces from rotating into the top segment the $\JJ$ piece is placed in, with the blue squares next to empty squares with dots causing issues and preventing the necessary spins from working, and attempting to fix the top segment via line clearing leads to overflow. The next five improper piece placements ((f) through (j)) block off one to two empty squares (indicated with dots) from being accessible, and attempting to clear lines and fit a piece covering those empty squares leads to overflow.

It is also not possible to rotate a $\JJ$ piece into the required placement in the top segment during a filling sequence (as in Figure \ref{OJJClogs}(k)) because all of the rotations into that orientation change the $y$-coordinate of the center of the $\JJ$ piece by $0$, $+1$, or $-2$, and of the possible $y$-coordinates the center of the $\JJ$ piece can start at, there are no valid ways to place a $\JJ$ piece in one of the two possible orientations pre-rotation without colliding with one of the gray squares.

\subsection{Clogs for $\{\OO, \SS\}$}\label{appendix:osclogs}

There are no improper ways to place the pieces in the priming and closing sequences that do not directly overflow the bottle, so we only look at the improper ways to place $\SS$ pieces in the filling sequences into the neck portion of a bottle, as in Figure \ref{OSSClogs} ((d)-(e) indicates a sequence of placements right above the body portion).

\begin{figure}[!ht]
  \centering
  \begin{subfigure}[b]{0.12\textwidth}
    \centering
    \includegraphics[width=45pt]{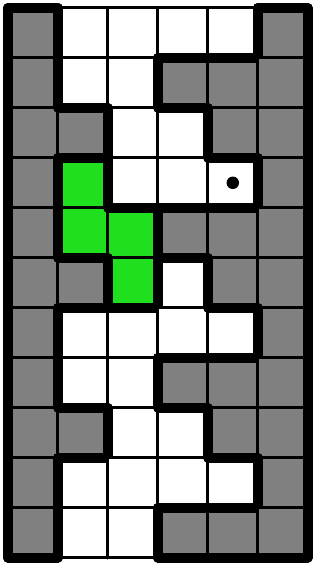}
    \caption{}
  \end{subfigure}
  \begin{subfigure}[b]{0.12\textwidth}
    \centering
    \includegraphics[width=45pt]{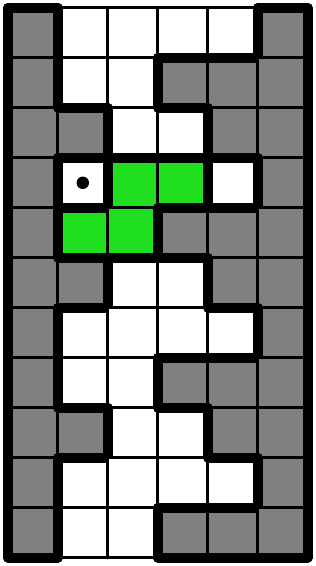}
    \caption{}
  \end{subfigure}
  \begin{subfigure}[b]{0.12\textwidth}
    \centering
    \includegraphics[width=45pt]{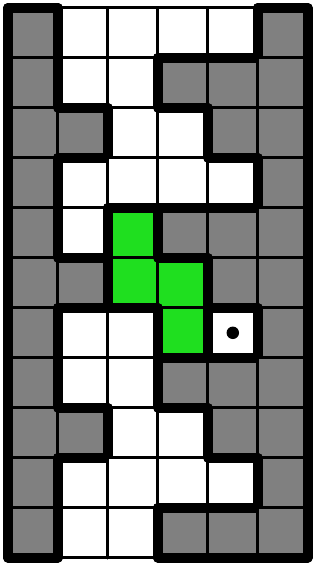}
    \caption{}
  \end{subfigure}
  \begin{subfigure}[b]{0.12\textwidth}
    \centering
    \includegraphics[width=45pt]{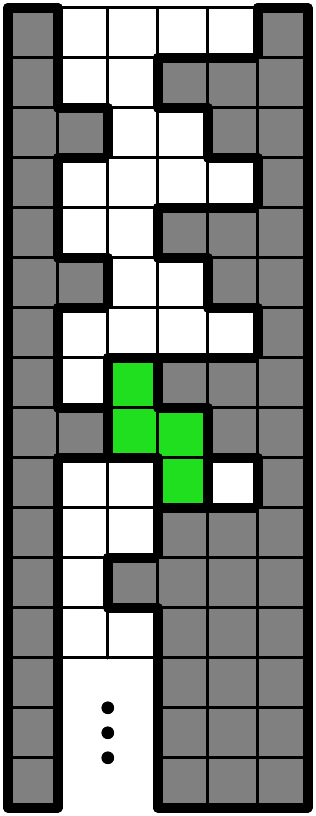}
    \caption{}
  \end{subfigure}
  \begin{subfigure}[b]{0.12\textwidth}
    \centering
    \includegraphics[width=45pt]{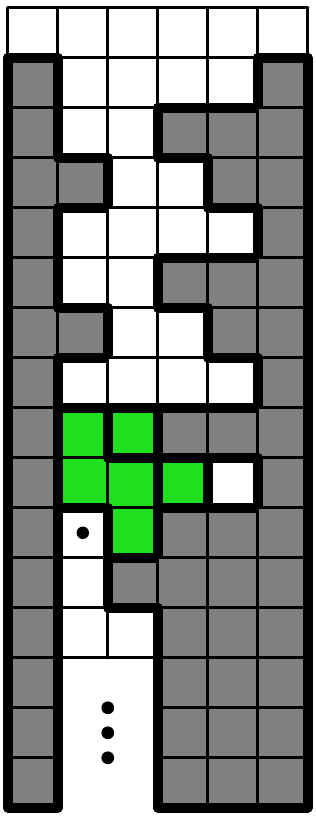}
    \caption{}
  \end{subfigure}
  \caption{Clogs involving $\SS$ pieces in the $\{\OO, \SS\}$ setup}
  \label{OSSClogs}
\end{figure}

For all of the improper piece placements in Figure \ref{OSSClogs}, the empty square indicated with a dot no longer becomes fill-able without overflow.

\subsection{Clogs for $\{\OO, \TT\}$}\label{appendix:otclogs}

First we discuss improper piece placements involving $\OO$ pieces in the priming sequence, which can be broken into two main cases: one indicated in Figure \ref{OTOClogs}(a) and one indicated in Figure \ref{OTOClogs}(b--c). The former shows that the $\OO$ pieces cannot be placed outside of the first two columns without causing a empty square to become inaccessible, and the latter shows that we cannot place a second $\OO$ piece from a priming or closing sequence into a bottle before a $\TT$ piece, as no $\TT$ piece can rotate in and cover the three empty squares with dots in Figure \ref{OTOClogs}(c). Thus, all of the $\OO$ pieces in the priming sequence must be placed in different bottles.

\begin{figure}[!ht]
  \centering
  \begin{subfigure}[b]{0.12\textwidth}
    \centering
    \includegraphics[width=45pt]{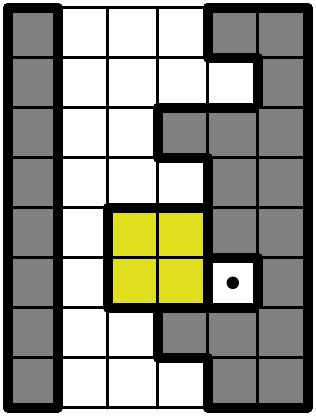}
    \caption{}
  \end{subfigure}
  \begin{subfigure}[b]{0.12\textwidth}
    \centering
    \includegraphics[width=45pt]{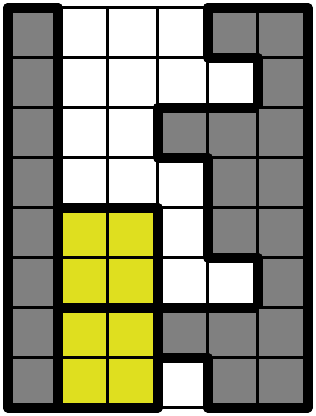}
    \caption{}
  \end{subfigure}
  \begin{subfigure}[b]{0.12\textwidth}
    \centering
    \includegraphics[width=45pt]{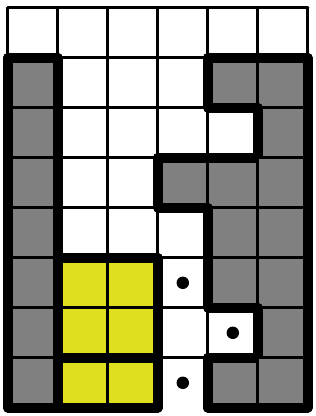}
    \caption{}
  \end{subfigure}
  \caption{Clogs involving $\OO$ pieces in the neck portion of a bottle in the $\{\OO, \TT\}$ setup}
  \label{OTOClogs}
\end{figure}

Next, we discuss improper piece placements involving $\TT$ pieces in the filling sequence. Here, we first note the improper piece placements within the body portion of the bottle in Figure \ref{OTBodyClogs}, where (c) is considered improper if there are any empty squares below the $\TT$ piece (as the presence of the $\TT$-finisher means that the rest of the bottle becomes inaccessible until the finale sequence, thus causing an overflow). In particular, this means that no $\OO$ pieces can be placed in the body portions of the bottles, meaning that all $\OO$ pieces must be placed in the neck portions. Furthermore, it is impossible to place an $\OO$ piece with its rotation center between the rightmost two columns of the bottle, even with line clearing with $\OO$s and $\TT$s, and placing an $\OO$ piece with its rotation center in between the middle two columns of the bottle, like in Figure \ref{OTOClogs}(a) and after Figure \ref{OTTClogs}(i), causes at least one empty square (indicated with a dot) to be inaccessible without causing an overflow. This means that all $\OO$ pieces must be placed within the leftmost two columns of the bottle, and by a counting argument, the first two columns must be covered by only $\OO$ pieces.

Therefore, all the $\TT$-piece placements in Figure \ref{OTTClogs} are improper and cause clogs, as they take up space required for $\OO$ pieces, in addition to potentially causing certain squares (marked with dots) to become inaccessible without causing an overflow.

\begin{figure}[!ht]
  \centering
  \begin{subfigure}[b]{0.12\textwidth}
    \centering
    \includegraphics[width=45pt]{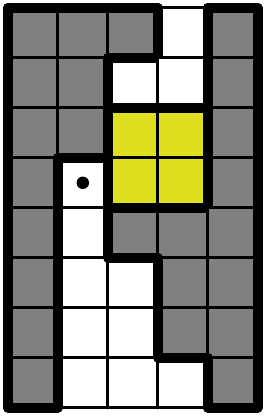}
    \caption{}
  \end{subfigure}
  \begin{subfigure}[b]{0.12\textwidth}
    \centering
    \includegraphics[width=45pt]{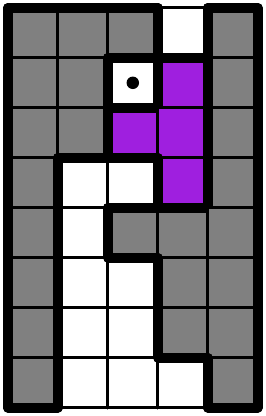}
    \caption{}
  \end{subfigure}
  \begin{subfigure}[b]{0.12\textwidth}
    \centering
    \includegraphics[width=45pt]{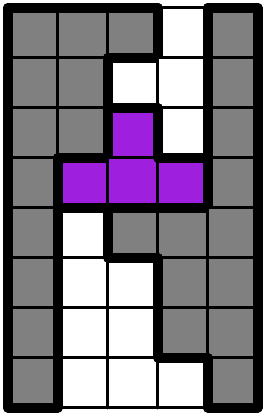}
    \caption{}
  \end{subfigure}
  \caption{Clogs in the body portion of a bottle in the $\{\OO, \TT\}$ setup}
  \label{OTBodyClogs}
\end{figure}

\begin{figure}[!ht]
  \centering
  \begin{subfigure}[b]{0.12\textwidth}
    \centering
    \includegraphics[width=45pt]{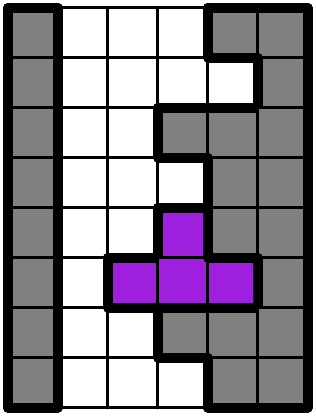}
    \caption{}
  \end{subfigure}
  \begin{subfigure}[b]{0.12\textwidth}
    \centering
    \includegraphics[width=45pt]{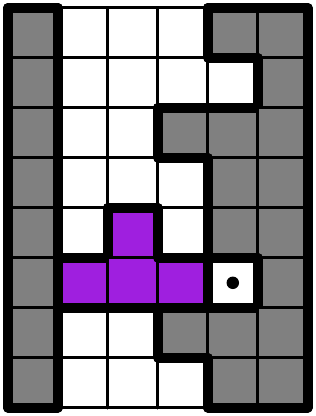}
    \caption{}
  \end{subfigure}
  \begin{subfigure}[b]{0.12\textwidth}
    \centering
    \includegraphics[width=45pt]{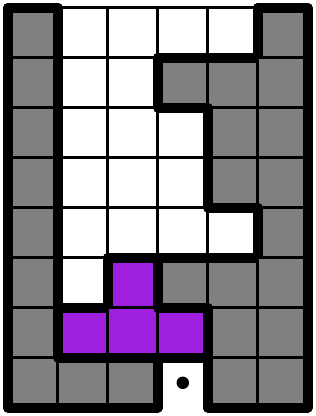}
    \caption{}
  \end{subfigure}
  \begin{subfigure}[b]{0.12\textwidth}
    \centering
    \includegraphics[width=45pt]{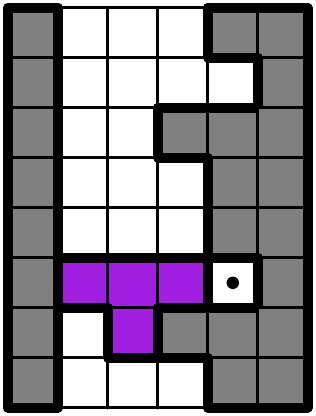}
    \caption{}
  \end{subfigure}
  \begin{subfigure}[b]{0.12\textwidth}
    \centering
    \includegraphics[width=45pt]{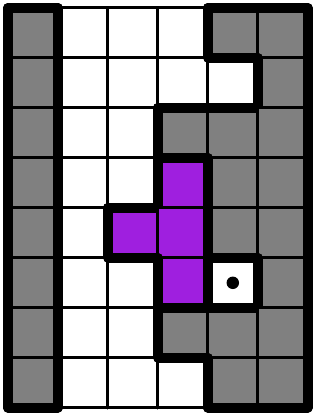}
    \caption{}
  \end{subfigure}
  \begin{subfigure}[b]{0.12\textwidth}
    \centering
    \includegraphics[width=45pt]{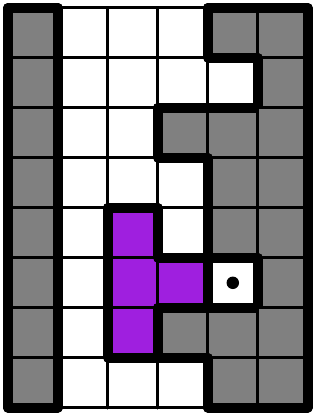}
    \caption{}
  \end{subfigure}
  \begin{subfigure}[b]{0.12\textwidth}
    \centering
    \includegraphics[width=45pt]{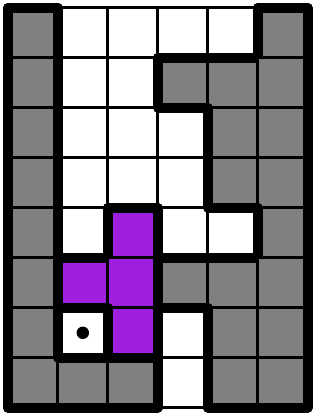}
    \caption{}
  \end{subfigure}
  \begin{subfigure}[b]{0.12\textwidth}
    \centering
    \includegraphics[width=45pt]{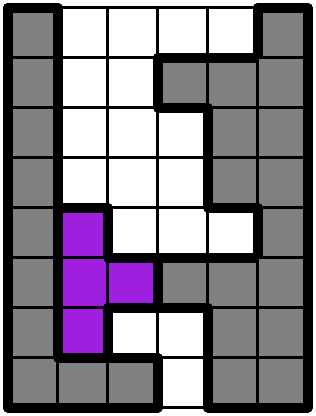}
    \caption{}
  \end{subfigure}
  \begin{subfigure}[b]{0.12\textwidth}
    \centering
    \includegraphics[width=45pt]{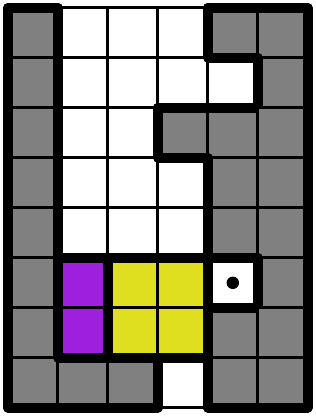}
    \caption{}
  \end{subfigure}
  \caption{Clogs involving $\TT$ pieces in the neck portion of a bottle in the $\{\OO, \TT\}$ setup}
  \label{OTTClogs}
\end{figure}

Improper piece placements involving pieces in the closing sequence are similar to the above; for $\TT$ pieces, any other placement takes up space required for $\OO$ pieces, and for $\OO$ pieces, the only other possible placements are those similar to Figure \ref{OTOClogs}(a).

\subsection{Clogs for $\{\SS, \TT\}$}\label{appendix:stclogs}

For the $\TT$ pieces in the priming sequence, the only improper piece placement that does not directly overflow the bottle is the one in Figure \ref{STClogs}(a); however, after the top line clears, no $\SS$ or $\TT$ piece can rotate into the rest of the bottle, and the empty square indicated with a dot cannot be filled, without overflowing the bottle.

For the $\SS$ pieces in the filling sequence, there are three different types of improper piece placements, as shown in Figure \ref{STClogs}(b--d); all of them cause specific empty squares (indicated with dots) to become inaccessible or unable to be filled without overflowing the bottle.

\begin{figure}[!ht]
  \centering
  \begin{subfigure}[b]{0.12\textwidth}
    \centering
    \includegraphics[width=45pt]{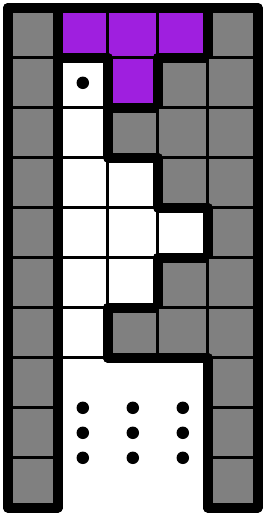}
    \caption{}
  \end{subfigure}
  \begin{subfigure}[b]{0.12\textwidth}
    \centering
    \includegraphics[width=45pt]{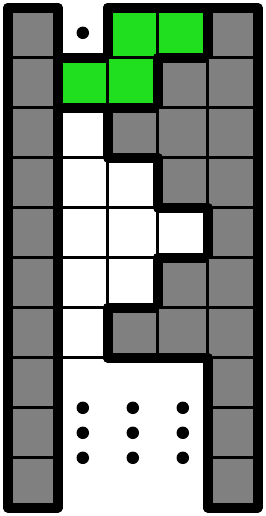}
    \caption{}
  \end{subfigure}
  \begin{subfigure}[b]{0.12\textwidth}
    \centering
    \includegraphics[width=45pt]{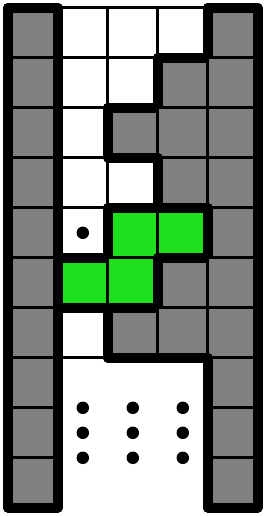}
    \caption{}
  \end{subfigure}
  \begin{subfigure}[b]{0.12\textwidth}
    \centering
    \includegraphics[width=45pt]{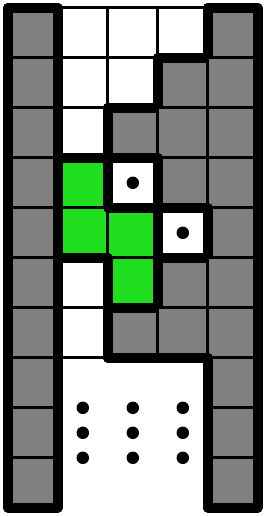}
    \caption{}
  \end{subfigure}
  \caption{Clogs in the $\{\SS, \TT\}$ setup}
  \label{STClogs}
\end{figure}

There are no improper piece placements among the pieces in the closing sequence that do not overflow the bottle (except for the first $\TT$ piece, which is covered by the discussion on $\TT$ pieces in the priming sequence).

\subsection{Clogs for $\{\SS, \ZZ\}$}\label{appendix:szclogs}

The improper piece placements for $\ZZ$ pieces in the priming and closing sequences are shown in Figure \ref{SZZClogs}; all of them cause at least one empty square (indicated with dots) to become inaccessible or unable to be filled without overflowing the bottle. 

Improper piece placements for $\ZZ$ pieces in the finale sequence cause similar issues with certain empty squares becoming inaccessible or unable to be filled without overflow.

\begin{figure}[!ht]
  \centering
  \begin{subfigure}[b]{0.12\textwidth}
    \centering
    \includegraphics[width=45pt]{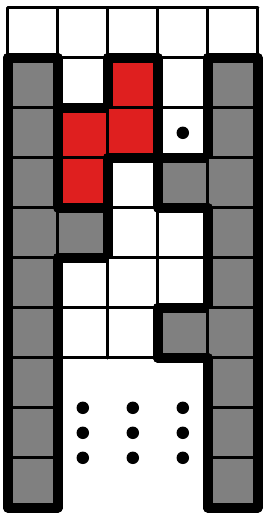}
    \caption{}
  \end{subfigure}
  \begin{subfigure}[b]{0.12\textwidth}
    \centering
    \includegraphics[width=45pt]{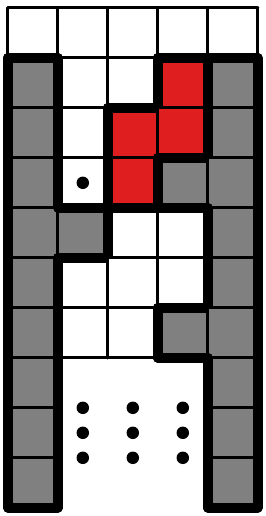}
    \caption{}
  \end{subfigure}
  \begin{subfigure}[b]{0.12\textwidth}
    \centering
    \includegraphics[width=45pt]{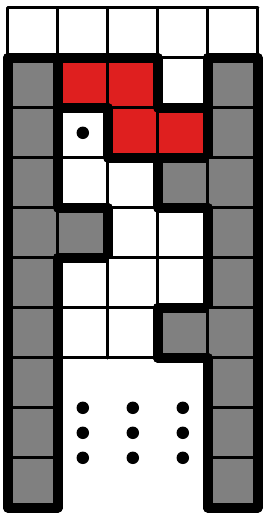}
    \caption{}
  \end{subfigure}
  \begin{subfigure}[b]{0.12\textwidth}
    \centering
    \includegraphics[width=45pt]{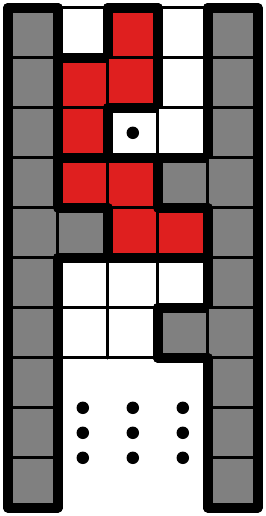}
    \caption{}
  \end{subfigure}
  \begin{subfigure}[b]{0.12\textwidth}
    \centering
    \includegraphics[width=45pt]{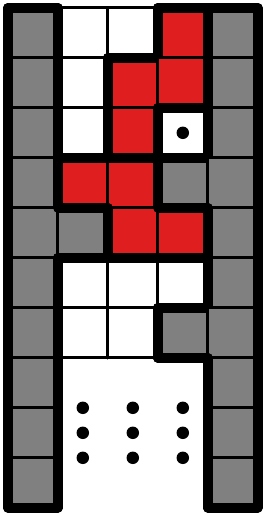}
    \caption{}
  \end{subfigure}
  \caption{Clogs involving $\ZZ$ pieces in the $\{\SS, \ZZ\}$ setup}
  \label{SZZClogs}
\end{figure}

The improper piece placements for $\SS$ pieces in the filling sequences are shown in Figure \ref{SZSClogs}. The piece placements in (a) and (b) force a $\ZZ$ piece to cover the empty square indicated with a dot, which causes at least one square in the rightmost column to become unable to be filled without overflow. The piece placement in (c) causes similar issues involving the empty square indicated with a dot. The piece placement in (d) leads to the other piece placements in (e)-(i), where it becomes impossible to reach the empty squares with dots without overflowing the bottle.

\begin{figure}[!ht]
  \centering
  \begin{subfigure}[b]{0.12\textwidth}
    \centering
    \includegraphics[width=45pt]{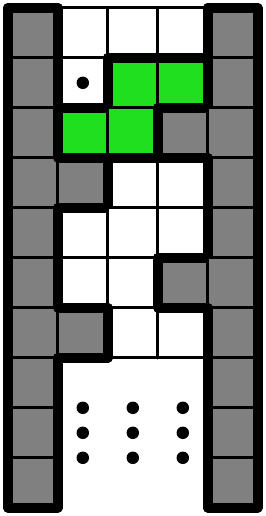}
    \caption{}
  \end{subfigure}
  \begin{subfigure}[b]{0.12\textwidth}
    \centering
    \includegraphics[width=45pt]{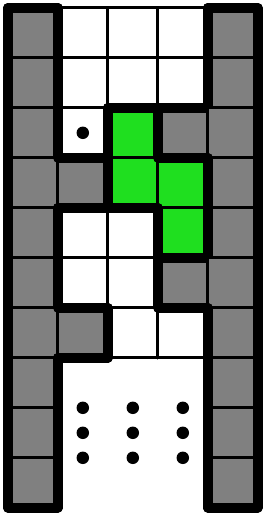}
    \caption{}
  \end{subfigure}
  \begin{subfigure}[b]{0.12\textwidth}
    \centering
    \includegraphics[width=45pt]{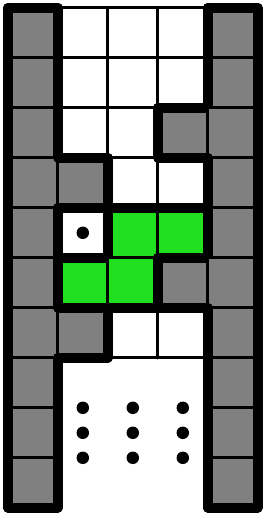}
    \caption{}
  \end{subfigure}
  \begin{subfigure}[b]{0.12\textwidth}
    \centering
    \includegraphics[width=45pt]{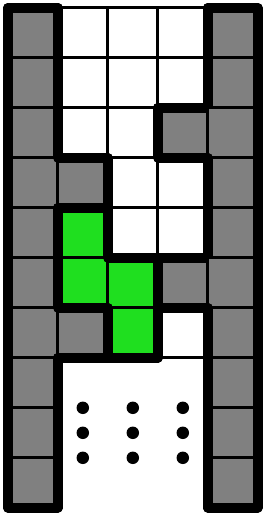}
    \caption{}
  \end{subfigure}
  \begin{subfigure}[b]{0.12\textwidth}
    \centering
    \includegraphics[width=45pt]{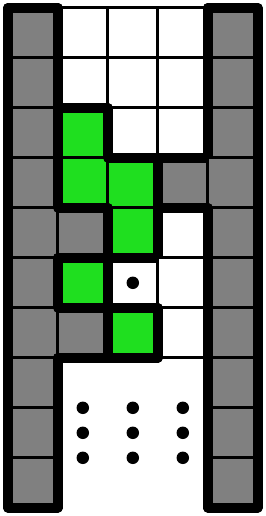}
    \caption{}
  \end{subfigure}
  \begin{subfigure}[b]{0.12\textwidth}
    \centering
    \includegraphics[width=45pt]{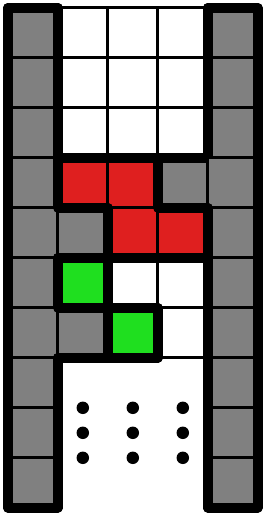}
    \caption{}
  \end{subfigure}
  \begin{subfigure}[b]{0.12\textwidth}
    \centering
    \includegraphics[width=45pt]{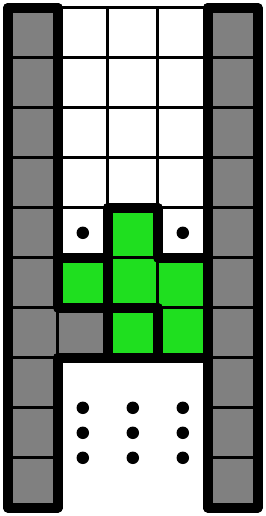}
    \caption{}
  \end{subfigure}
  \begin{subfigure}[b]{0.12\textwidth}
    \centering
    \includegraphics[width=45pt]{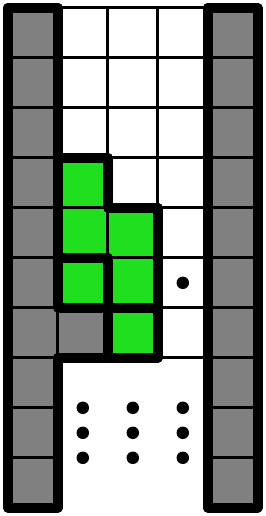}
    \caption{}
  \end{subfigure}
  \begin{subfigure}[b]{0.12\textwidth}
    \centering
    \includegraphics[width=45pt]{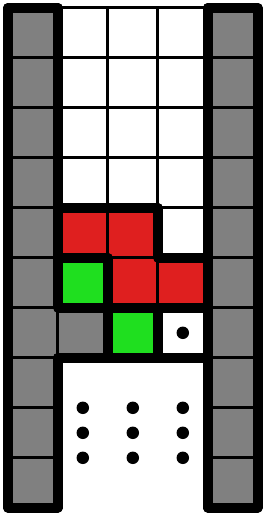}
    \caption{}
  \end{subfigure}
  \caption{Clogs involving $\SS$ pieces in the $\{\SS, \ZZ\}$ setup}
  \label{SZSClogs}
\end{figure}

\subsection{Clogs for $\{\JJ, \ZZ\}, \{\JJ, \TT\}$}\label{appendix:jzclogs}

The improper piece placements for $\ZZ$ pieces and $\TT$ pieces in the priming and closing sequences for $\{\JJ, \ZZ\}$ and $\{\JJ, \TT\}$, respectively, are shown in Figures \ref{JZZClogs} and \ref{JTTClogs}; all of them cause at least one empty square (indicated with dots) to become inaccessible or unable to be filled without overflowing the bottle, or cause the rest of the bottle to become inaccessible without overflow due to pieces not being able to move or rotate into the bottle. In particular, putting too many $\ZZ$ or $\TT$ pieces into a top segment during a priming sequence causes clogs similar to Figure \ref{JZZClogs}(g) and \ref{JTTClogs}(h--i).

\begin{figure}[!ht]
  \centering
  \begin{subfigure}[b]{0.12\textwidth}
    \centering
    \includegraphics[width=45pt]{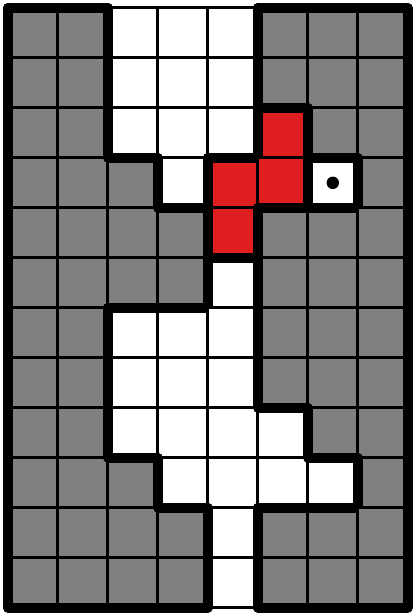}
    \caption{}
  \end{subfigure}
  \begin{subfigure}[b]{0.12\textwidth}
    \centering
    \includegraphics[width=45pt]{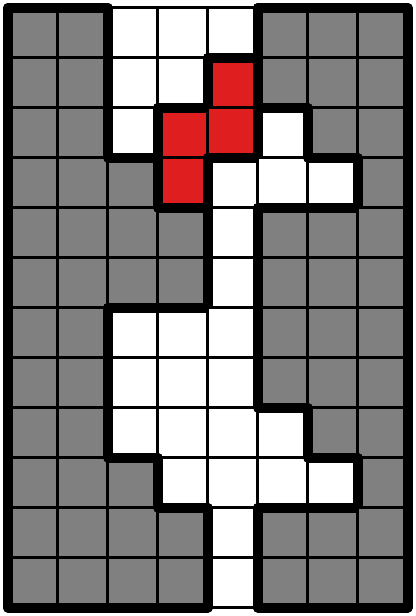}
    \caption{}
  \end{subfigure}
  \begin{subfigure}[b]{0.12\textwidth}
    \centering
    \includegraphics[width=45pt]{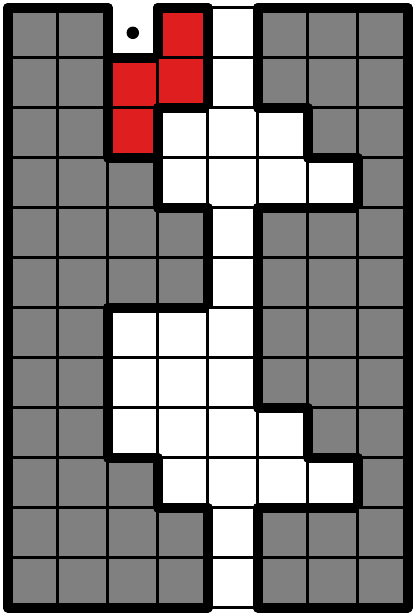}
    \caption{}
  \end{subfigure}
  \begin{subfigure}[b]{0.12\textwidth}
    \centering
    \includegraphics[width=45pt]{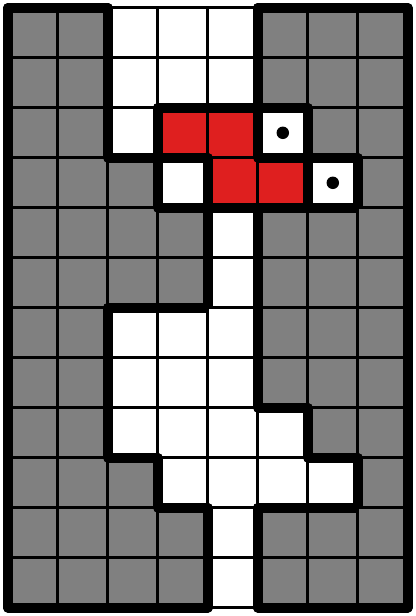}
    \caption{}
  \end{subfigure}
  \begin{subfigure}[b]{0.12\textwidth}
    \centering
    \includegraphics[width=45pt]{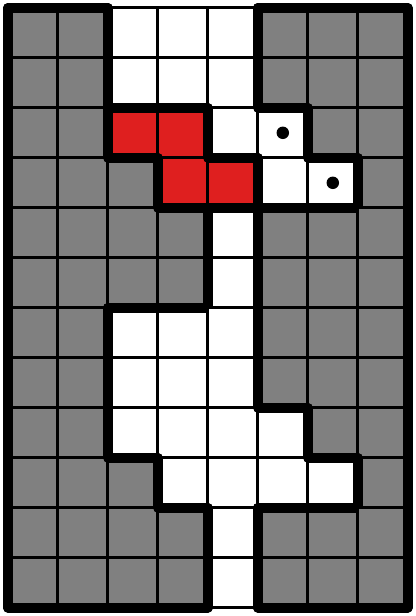}
    \caption{}
  \end{subfigure}
  \begin{subfigure}[b]{0.12\textwidth}
    \centering
    \includegraphics[width=45pt]{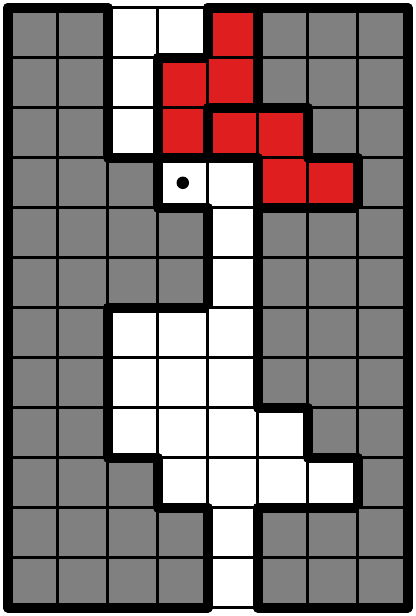}
    \caption{}
  \end{subfigure}
  \begin{subfigure}[b]{0.12\textwidth}
    \centering
    \includegraphics[width=45pt]{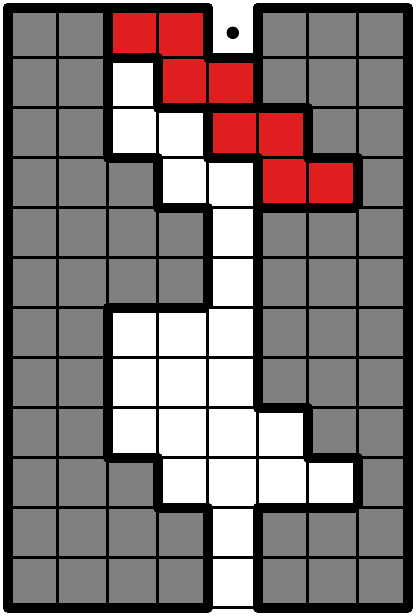}
    \caption{}
  \end{subfigure}
  \caption{Clogs involving $\ZZ$ pieces in the $\{\JJ, \ZZ\}$ setup}
  \label{JZZClogs}
\end{figure}

\begin{figure}[!ht]
  \centering
  \begin{subfigure}[b]{0.12\textwidth}
    \centering
    \includegraphics[width=45pt]{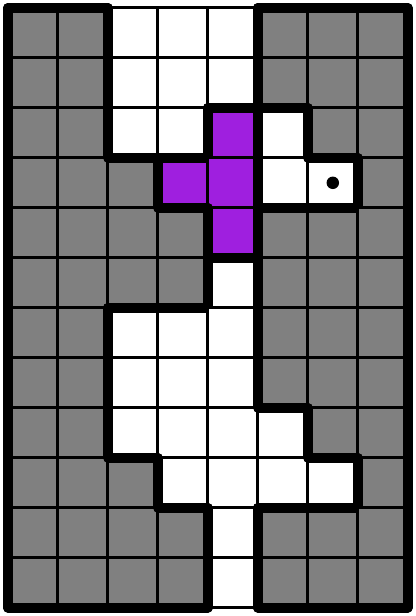}
    \caption{}
  \end{subfigure}
  \begin{subfigure}[b]{0.12\textwidth}
    \centering
    \includegraphics[width=45pt]{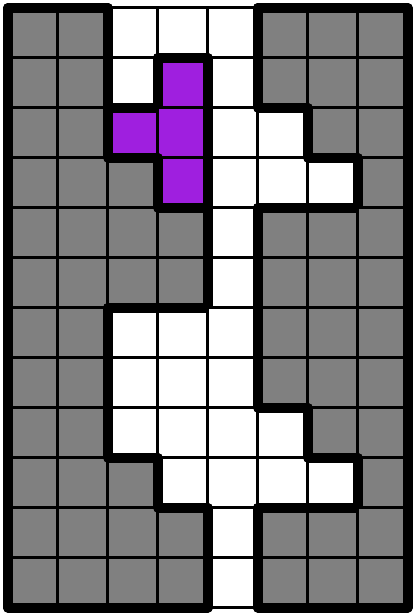}
    \caption{}
  \end{subfigure}
  \begin{subfigure}[b]{0.12\textwidth}
    \centering
    \includegraphics[width=45pt]{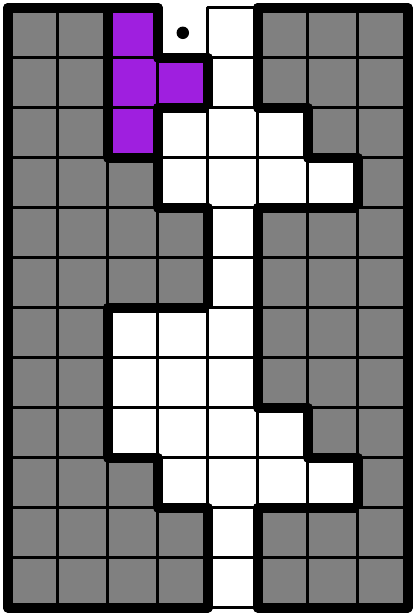}
    \caption{}
  \end{subfigure}
  \begin{subfigure}[b]{0.12\textwidth}
    \centering
    \includegraphics[width=45pt]{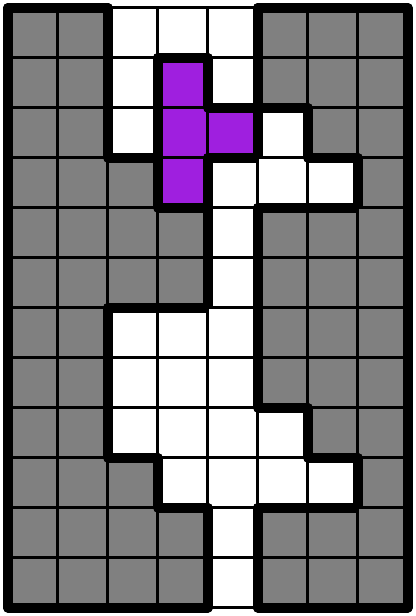}
    \caption{}
  \end{subfigure}
  \begin{subfigure}[b]{0.12\textwidth}
    \centering
    \includegraphics[width=45pt]{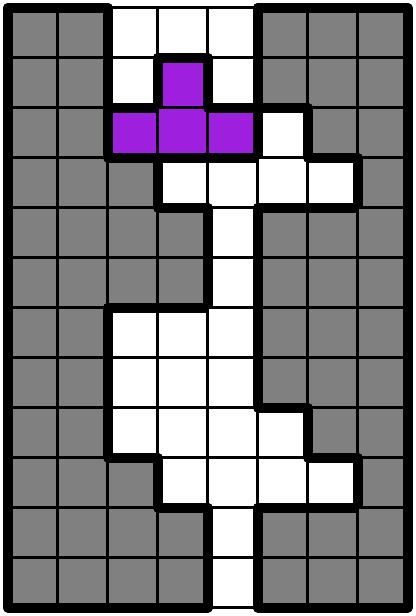}
    \caption{}
  \end{subfigure}
  \begin{subfigure}[b]{0.12\textwidth}
    \centering
    \includegraphics[width=45pt]{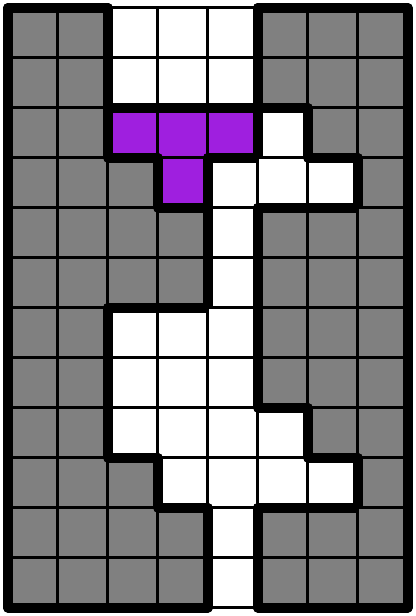}
    \caption{}
  \end{subfigure}
  \begin{subfigure}[b]{0.12\textwidth}
    \centering
    \includegraphics[width=45pt]{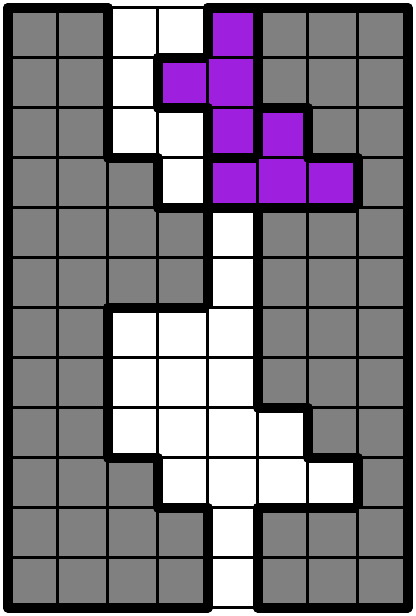}
    \caption{}
  \end{subfigure}
  \begin{subfigure}[b]{0.12\textwidth}
    \centering
    \includegraphics[width=45pt]{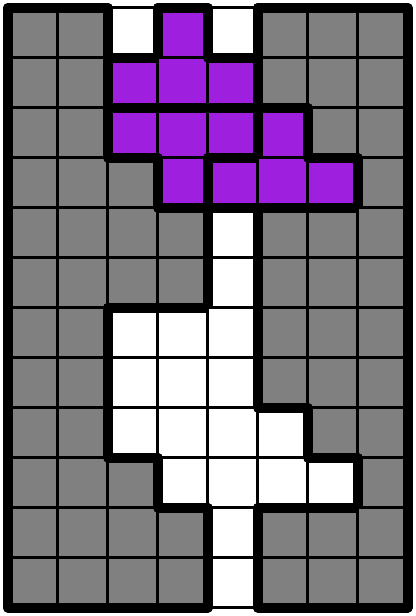}
    \caption{}
  \end{subfigure}
  \begin{subfigure}[b]{0.12\textwidth}
    \centering
    \includegraphics[width=45pt]{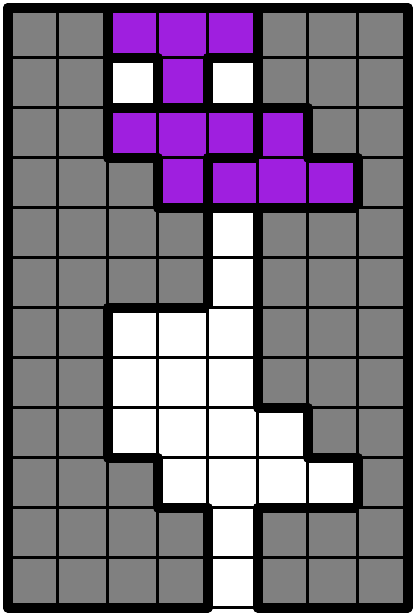}
    \caption{}
  \end{subfigure}
  \caption{Clogs involving $\TT$ pieces in the $\{\JJ, \TT\}$ setup}
  \label{JTTClogs}
\end{figure}

The improper piece placements for $\JJ$ pieces in the filling and closing sequences are shown in Figure \ref{JZJClogs}. The piece placements in (a)-(e) and (j) cause at least one empty square (indicated with a dot) to become inaccessible without overflowing the bottle. The piece placements in (f)-(i) either prevent more $\JJ$ pieces from rotating into the top segment the $\JJ$ piece is placed in or allow only one more $\JJ$ piece to rotate in and be stuck at the top of the top segment (in the case of (g) and (i)). Due to the structure of a top segment, the $\JJ$ pieces in the closing sequence must be placed as in Figure \ref{JZSetup}; (k)-(l) shows the result of the only improper piece placement involving a $\JJ$ piece in the closing sequence that does not directly overflow the bottle, but causes the empty square indicated with a dot to become inaccessible without overflowing the bottle.

\begin{figure}[!ht]
  \centering
  \begin{subfigure}[b]{0.12\textwidth}
    \centering
    \includegraphics[width=45pt]{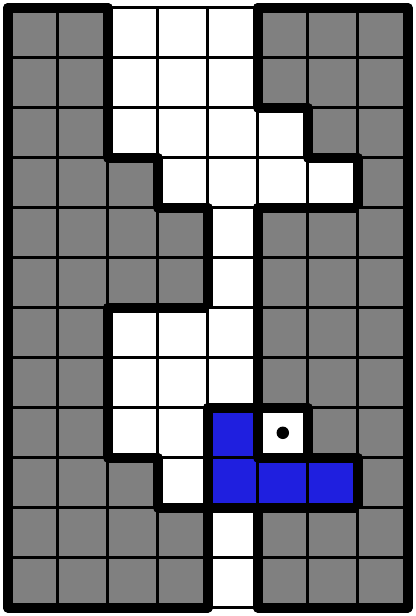}
    \caption{}
  \end{subfigure}
  \begin{subfigure}[b]{0.12\textwidth}
    \centering
    \includegraphics[width=45pt]{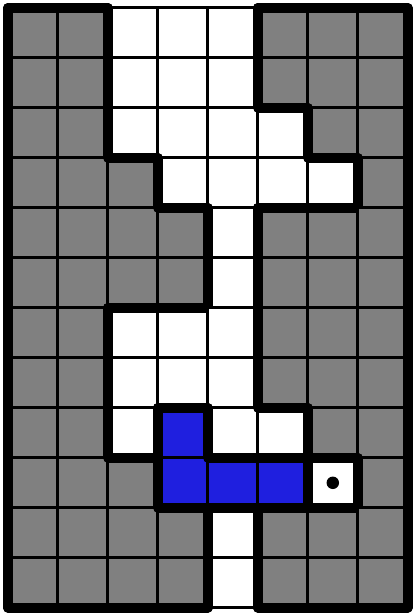}
    \caption{}
  \end{subfigure}
  \begin{subfigure}[b]{0.12\textwidth}
    \centering
    \includegraphics[width=45pt]{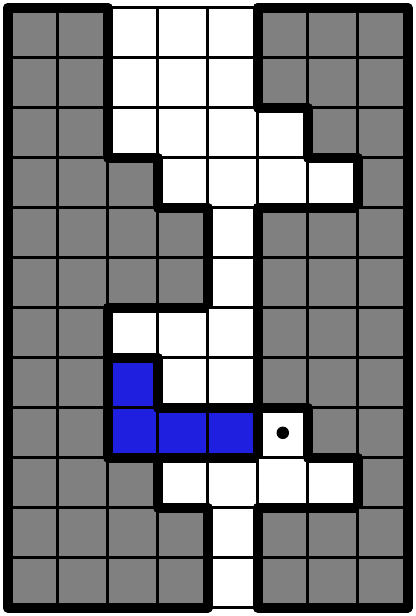}
    \caption{}
  \end{subfigure}
  \begin{subfigure}[b]{0.12\textwidth}
    \centering
    \includegraphics[width=45pt]{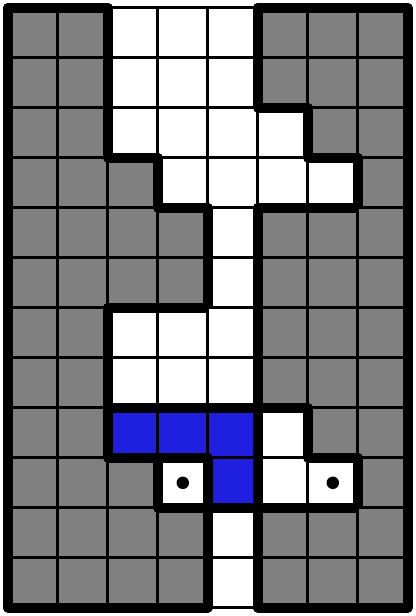}
    \caption{}
  \end{subfigure}
  \begin{subfigure}[b]{0.12\textwidth}
    \centering
    \includegraphics[width=45pt]{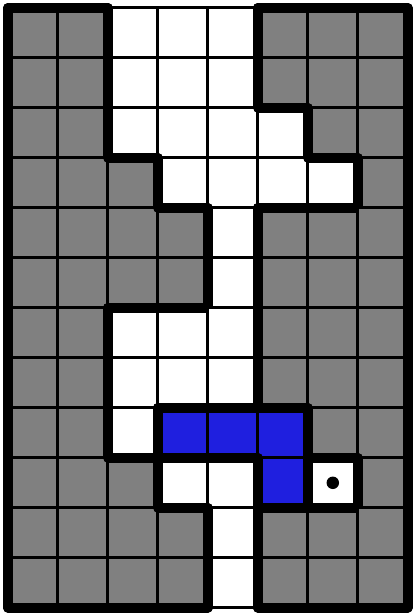}
    \caption{}
  \end{subfigure}
  \begin{subfigure}[b]{0.12\textwidth}
    \centering
    \includegraphics[width=45pt]{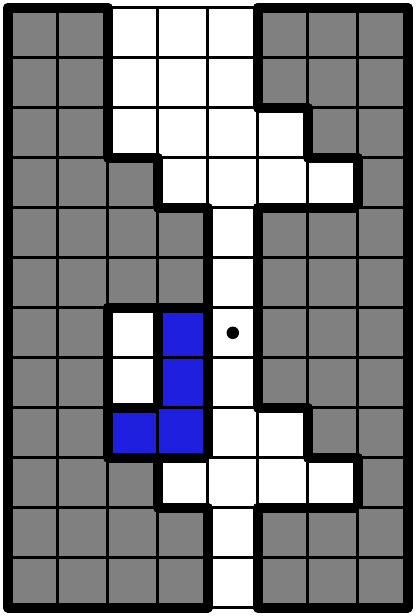}
    \caption{}
  \end{subfigure}
  \begin{subfigure}[b]{0.12\textwidth}
    \centering
    \includegraphics[width=45pt]{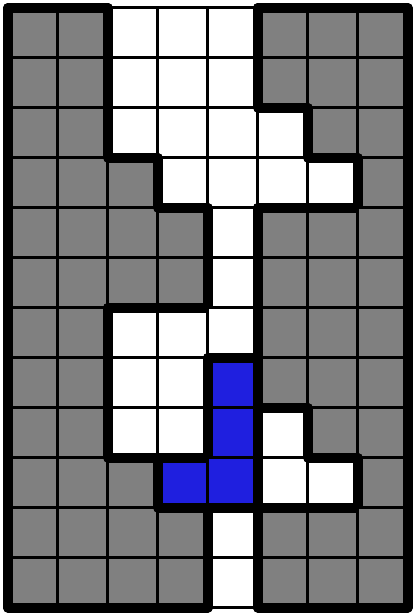}
    \caption{}
  \end{subfigure}
  \begin{subfigure}[b]{0.12\textwidth}
    \centering
    \includegraphics[width=45pt]{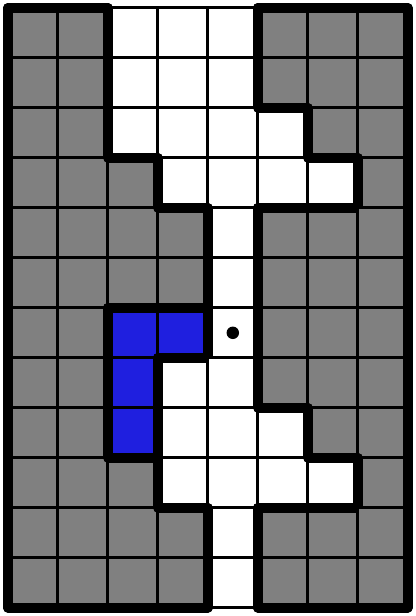}
    \caption{}
  \end{subfigure}
  \begin{subfigure}[b]{0.12\textwidth}
    \centering
    \includegraphics[width=45pt]{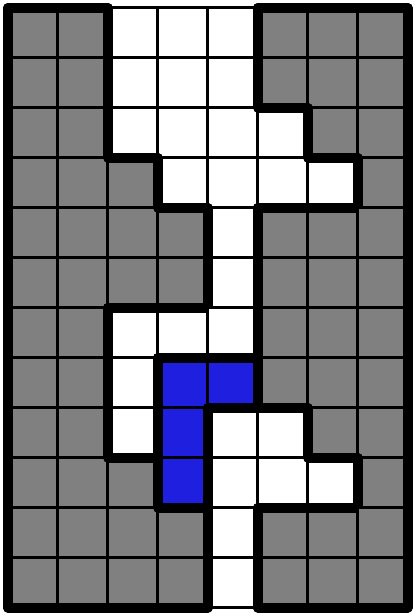}
    \caption{}
  \end{subfigure}
  \begin{subfigure}[b]{0.12\textwidth}
    \centering
    \includegraphics[width=45pt]{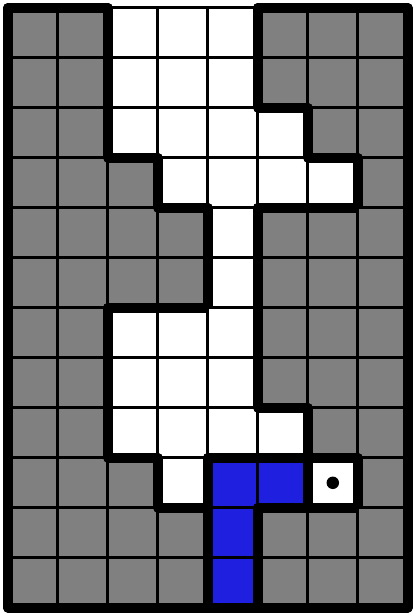}
    \caption{}
  \end{subfigure}
  \begin{subfigure}[b]{0.12\textwidth}
    \centering
    \includegraphics[width=45pt]{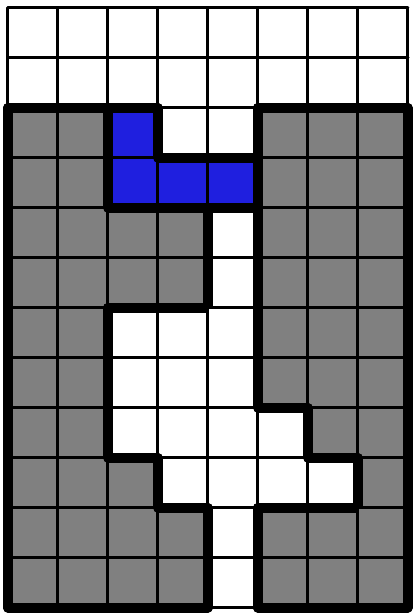}
    \caption{}
  \end{subfigure}
  \begin{subfigure}[b]{0.12\textwidth}
    \centering
    \includegraphics[width=45pt]{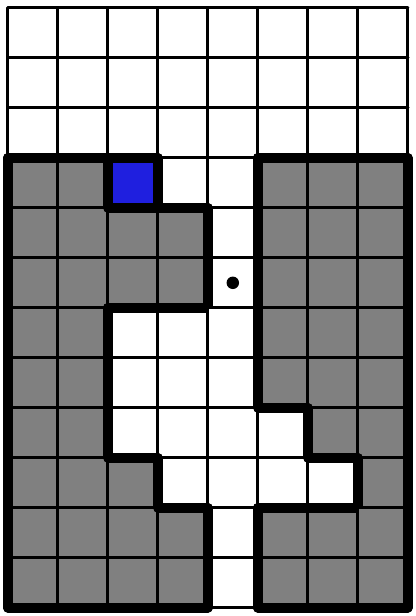}
    \caption{}
  \end{subfigure}
  \caption{Clogs involving $\JJ$ pieces in the $\{\JJ, \ZZ\}$ and $\{\JJ, \TT\}$ setups}
  \label{JZJClogs}
\end{figure}

\subsection{Clogs for $\{\JJ, \SS\}$}\label{appendix:jsclogs}

The improper piece placements for $\SS$ pieces in the priming and closing sequences are shown in Figure \ref{JSSClogs}; all of them cause at least one empty square (indicated with dots) to become inaccessible or unable to be filled without overflowing the bottle, or cause the rest of the bottle to become inaccessible without overflow due to pieces not being able to move or rotate into the bottle, or are impossible under SRS (as in (f)). In particular, putting more than one $\SS$ piece into a top segment during a priming sequence causes one of the $\SS$ pieces to be placed in one of the five improper ways shown in Figure \ref{JSSClogs}.

\begin{figure}[!ht]
  \centering
  \begin{subfigure}[b]{0.12\textwidth}
    \centering
    \includegraphics[width=45pt]{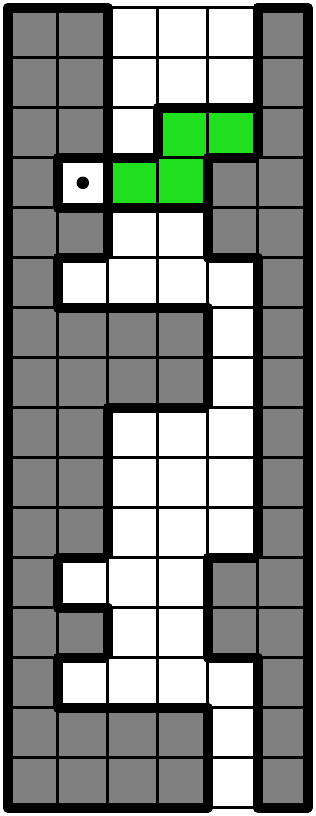}
    \caption{}
  \end{subfigure}
  \begin{subfigure}[b]{0.12\textwidth}
    \centering
    \includegraphics[width=45pt]{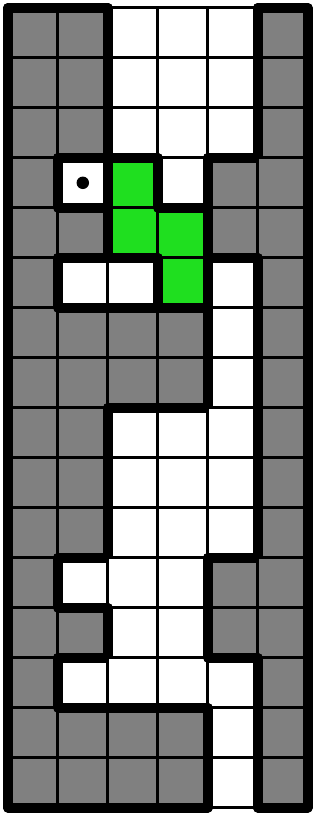}
    \caption{}
  \end{subfigure}
  \begin{subfigure}[b]{0.12\textwidth}
    \centering
    \includegraphics[width=45pt]{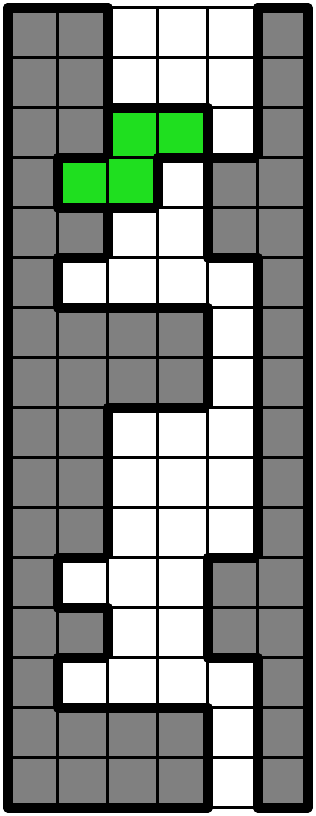}
    \caption{}
  \end{subfigure}
  \begin{subfigure}[b]{0.12\textwidth}
    \centering
    \includegraphics[width=45pt]{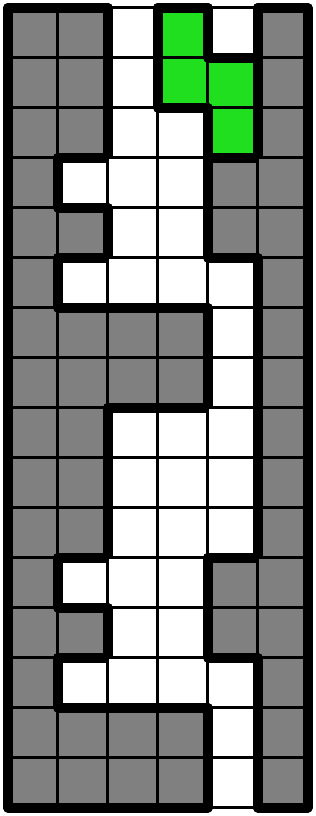}
    \caption{}
  \end{subfigure}
  \begin{subfigure}[b]{0.12\textwidth}
    \centering
    \includegraphics[width=45pt]{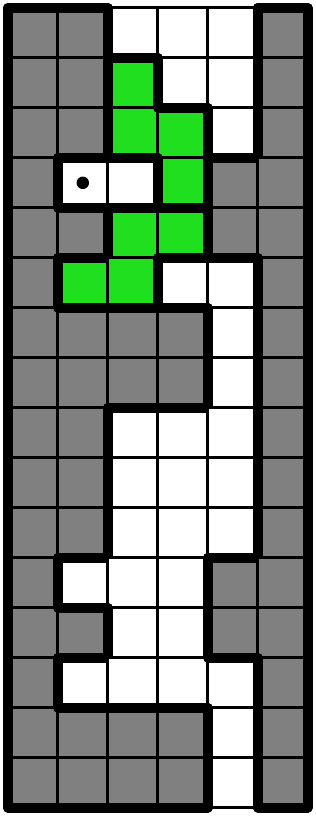}
    \caption{}
  \end{subfigure}
  \begin{subfigure}[b]{0.12\textwidth}
    \centering
    \includegraphics[width=45pt]{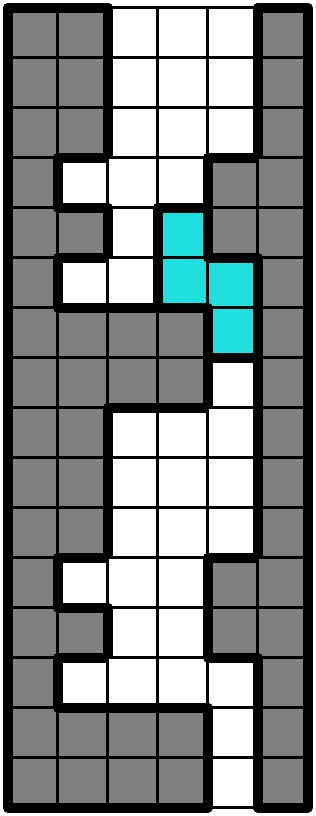}
    \caption{}
  \end{subfigure}
  \caption{Clogs or impossible piece placements (in the case of (f)) involving $\SS$ pieces in the $\{\JJ, \ZZ\}$ setup}
  \label{JSSClogs}
\end{figure}

The improper piece placements for $\JJ$ pieces in the filling sequences are shown in Figure \ref{JSJClogs}. The piece placements in (a), (b), (d)-(g), and (i) cause at least one empty square (indicated with a dot) to become inaccessible without overflowing the bottle, either due to being blocked off or due to pieces not being able to rotate in to fill the square. The piece placements in (c), (h), and (j) cause issues with $\JJ$ pieces attempting to rotate into the top segment the $\JJ$ piece is placed in; (k)-(m) show possible piece placements after (j) (with (m) being after the lines above the top segment have been cleared), with certain empty squares (indicated with dots) becoming inaccessible without overflowing the bottle.

\begin{figure}[!ht]
  \centering
  \begin{subfigure}[b]{0.12\textwidth}
    \centering
    \includegraphics[width=45pt]{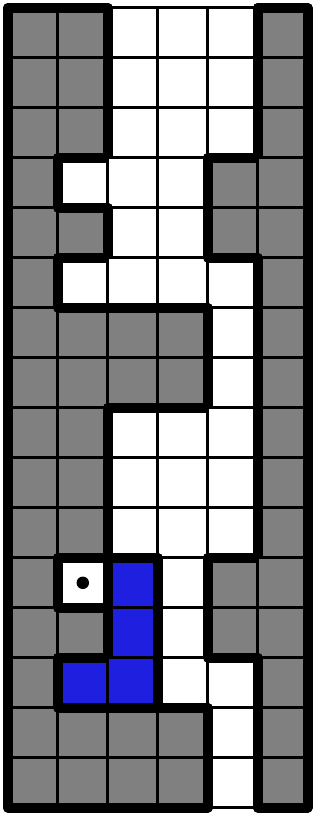}
    \caption{}
  \end{subfigure}
  \begin{subfigure}[b]{0.12\textwidth}
    \centering
    \includegraphics[width=45pt]{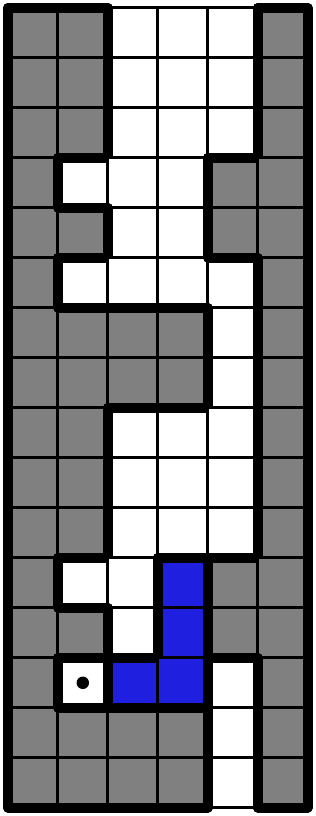}
    \caption{}
  \end{subfigure}
  \begin{subfigure}[b]{0.12\textwidth}
    \centering
    \includegraphics[width=45pt]{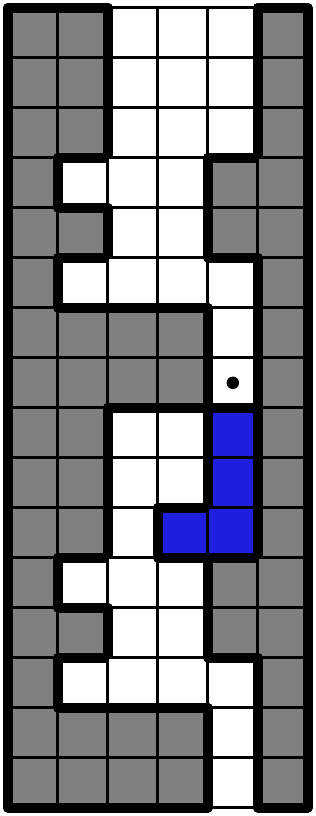}
    \caption{}
  \end{subfigure}
  \begin{subfigure}[b]{0.12\textwidth}
    \centering
    \includegraphics[width=45pt]{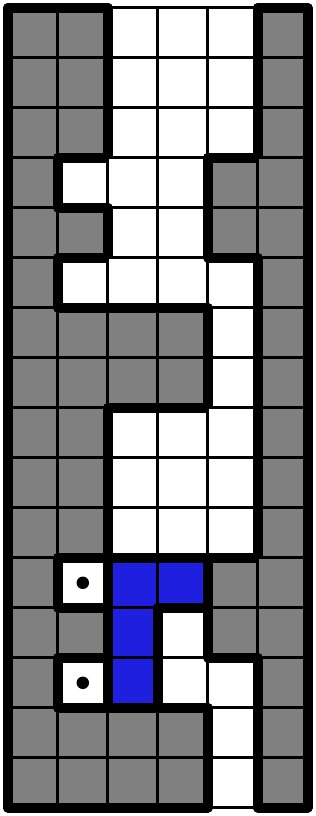}
    \caption{}
  \end{subfigure}
  \begin{subfigure}[b]{0.12\textwidth}
    \centering
    \includegraphics[width=45pt]{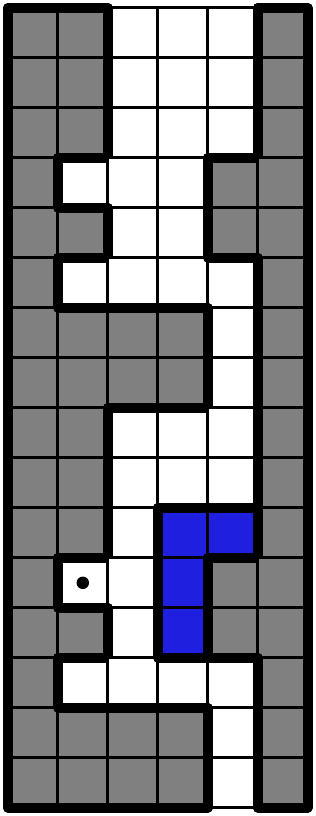}
    \caption{}
  \end{subfigure}
  \begin{subfigure}[b]{0.12\textwidth}
    \centering
    \includegraphics[width=45pt]{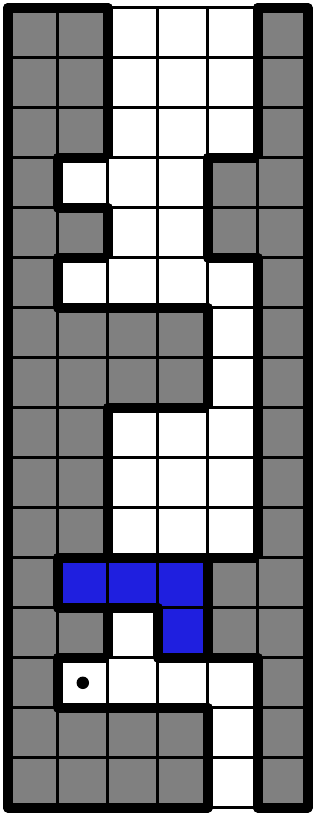}
    \caption{}
  \end{subfigure}
  \begin{subfigure}[b]{0.12\textwidth}
    \centering
    \includegraphics[width=45pt]{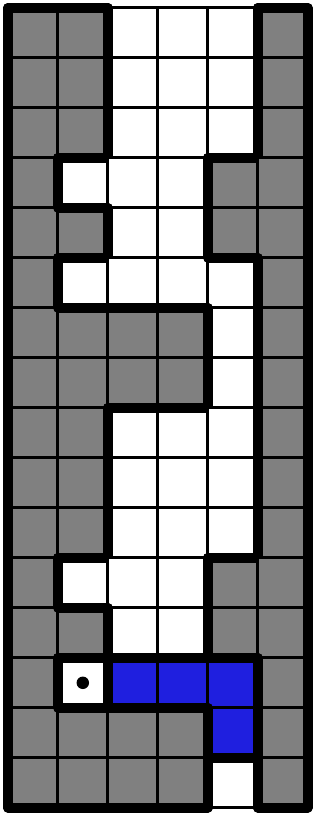}
    \caption{}
  \end{subfigure}
  \begin{subfigure}[b]{0.12\textwidth}
    \centering
    \includegraphics[width=45pt]{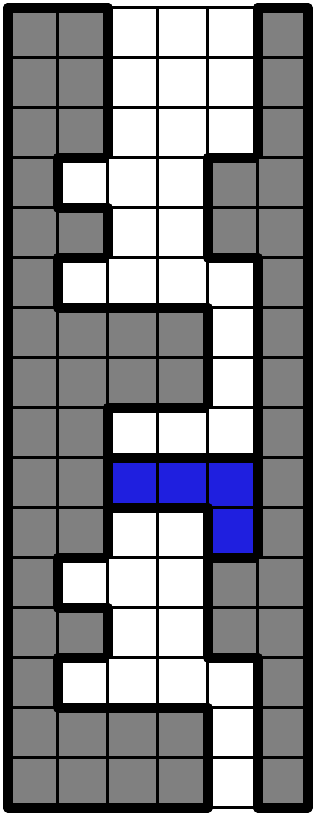}
    \caption{}
  \end{subfigure}
  \begin{subfigure}[b]{0.12\textwidth}
    \centering
    \includegraphics[width=45pt]{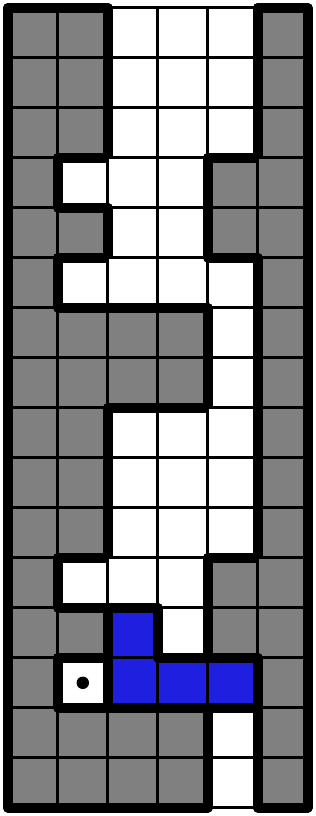}
    \caption{}
  \end{subfigure}
  \begin{subfigure}[b]{0.12\textwidth}
    \centering
    \includegraphics[width=45pt]{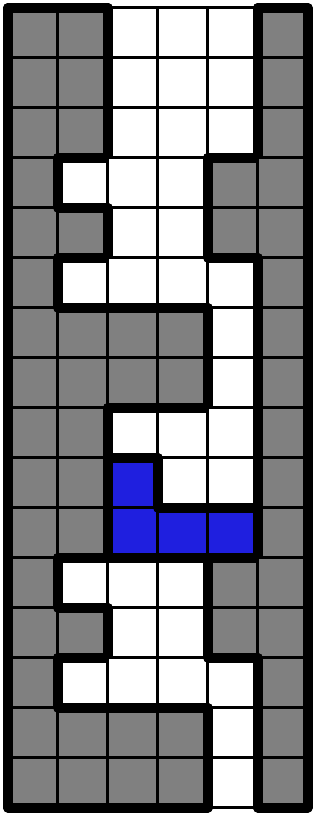}
    \caption{}
  \end{subfigure}
  \begin{subfigure}[b]{0.12\textwidth}
    \centering
    \includegraphics[width=45pt]{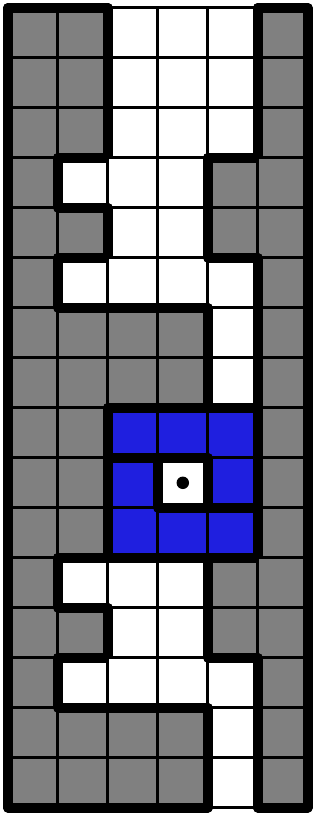}
    \caption{}
  \end{subfigure}
  \begin{subfigure}[b]{0.12\textwidth}
    \centering
    \includegraphics[width=45pt]{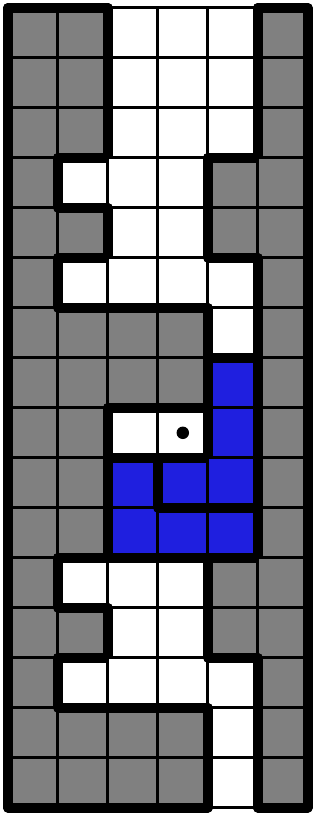}
    \caption{}
  \end{subfigure}
  \begin{subfigure}[b]{0.12\textwidth}
    \centering
    \includegraphics[width=45pt]{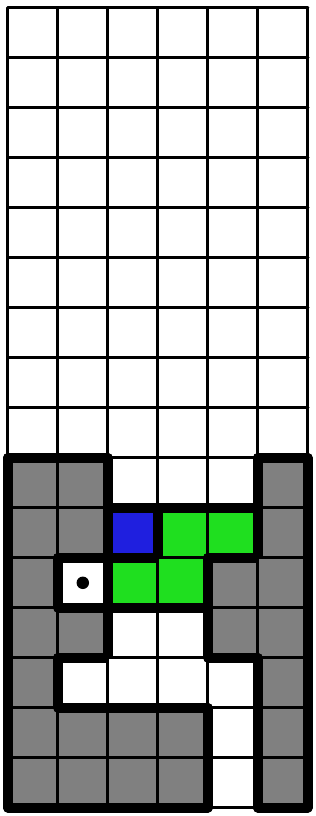}
    \caption{}
  \end{subfigure}
  \caption{Clogs involving $\JJ$ pieces in the $\{\JJ, \SS\}$ setup}
  \label{JSJClogs}
\end{figure}

Due to the structure of a top segment, the $\JJ$ pieces in the closing sequence must be placed as in Figure \ref{JSSetup}; improper piece placements for $\JJ$ pieces in the closing sequence here follow a similar analysis as improper piece placements for $\JJ$ pieces in the closing sequence for the $\{\JJ, \ZZ\}$ setup or for $\JJ$ pieces in the filling sequence for this setup.

\subsection{Clogs for $\{\JJ, \LL\}$}\label{appendix:jlclogs}

The improper piece placements for $\LL$ pieces in the priming sequences and the first two $\LL$ pieces in the closing sequences are shown in Figure \ref{JLLClogs}; all of them cause at least one empty square (indicated with dots) to become inaccessible or unable to be filled without overflowing the bottle, or cause the rest of the bottle to become inaccessible without overflow due to pieces not being able to move or rotate into the bottle. In particular, putting more than two $\LL$ pieces into a top segment during a priming sequence, or having more than two $\LL$ pieces in a top segment before a $\JJ$ piece gets placed in a closing sequence, causes one of the $\LL$ pieces to be placed in one of the improper ways shown in Figure \ref{JLLClogs}, particularly in (n) and (o).

\begin{figure}[!ht]
  \centering
  \begin{subfigure}[b]{0.11\textwidth}
    \centering
    \includegraphics[width=44pt]{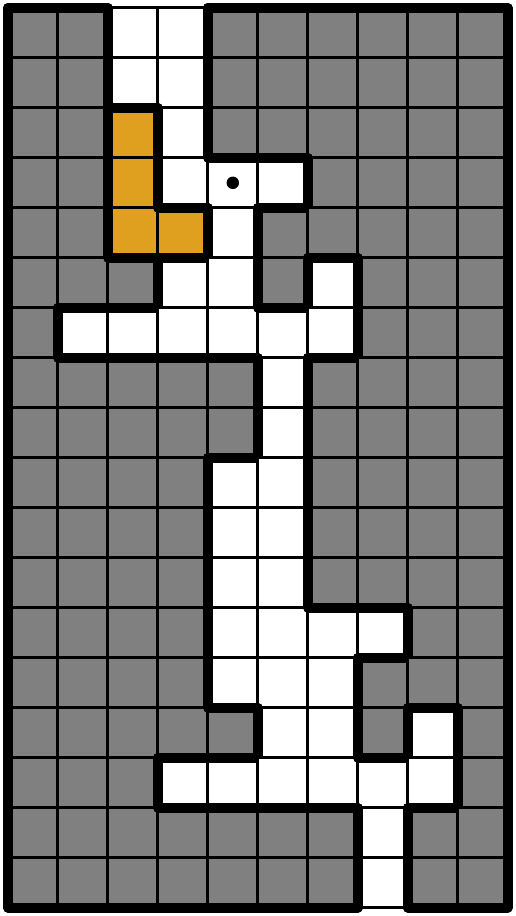}
    \caption{}
  \end{subfigure}
  \begin{subfigure}[b]{0.11\textwidth}
    \centering
    \includegraphics[width=44pt]{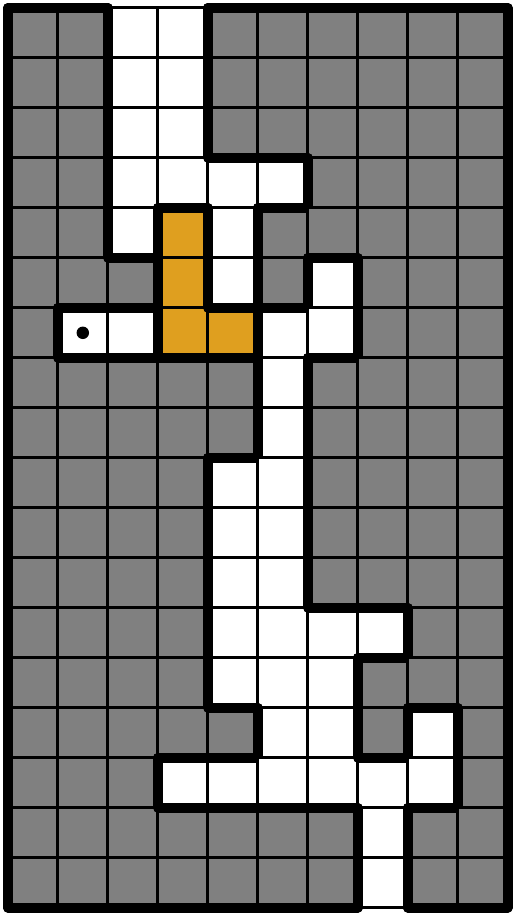}
    \caption{}
  \end{subfigure}
  \begin{subfigure}[b]{0.11\textwidth}
    \centering
    \includegraphics[width=44pt]{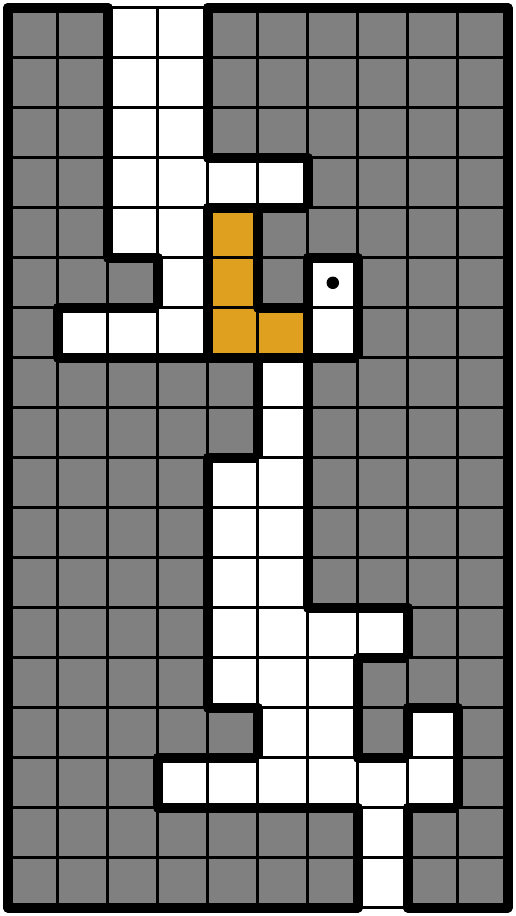}
    \caption{}
  \end{subfigure}
  \begin{subfigure}[b]{0.11\textwidth}
    \centering
    \includegraphics[width=44pt]{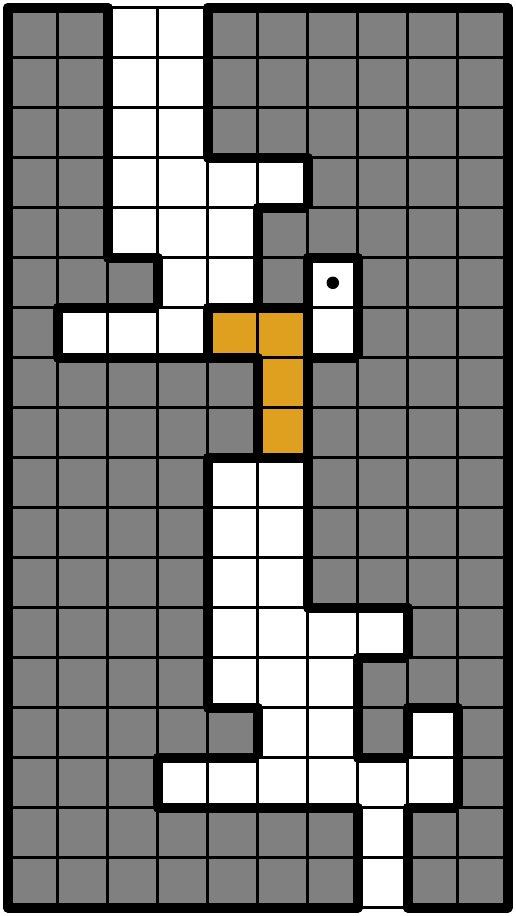}
    \caption{}
  \end{subfigure}
  \begin{subfigure}[b]{0.11\textwidth}
    \centering
    \includegraphics[width=44pt]{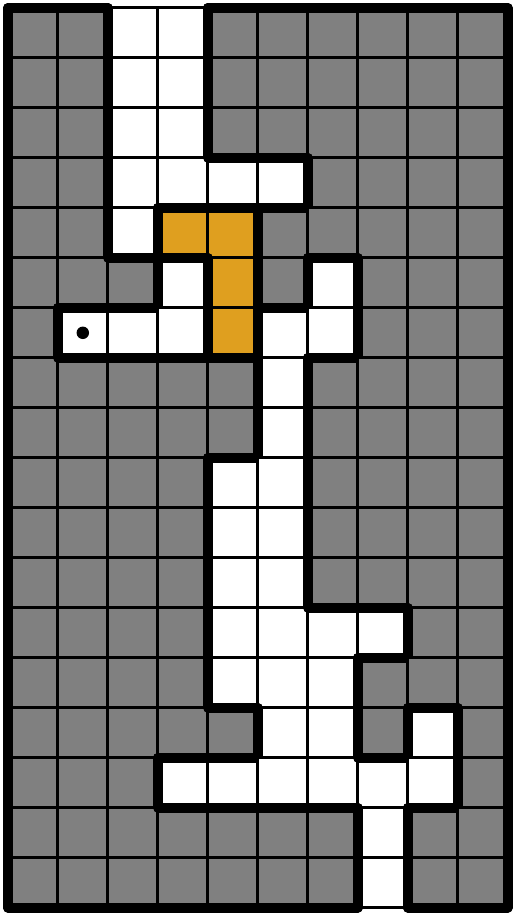}
    \caption{}
  \end{subfigure}
  \begin{subfigure}[b]{0.11\textwidth}
    \centering
    \includegraphics[width=44pt]{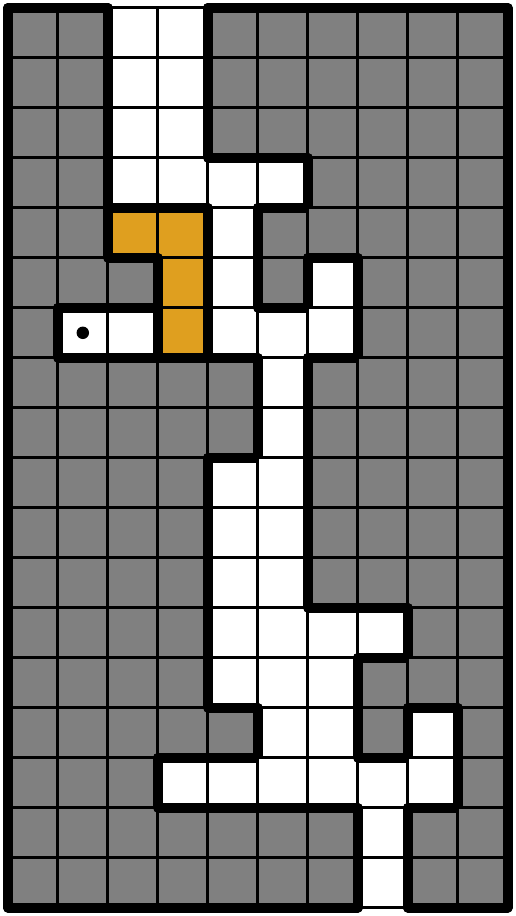}
    \caption{}
  \end{subfigure}
  \begin{subfigure}[b]{0.11\textwidth}
    \centering
    \includegraphics[width=44pt]{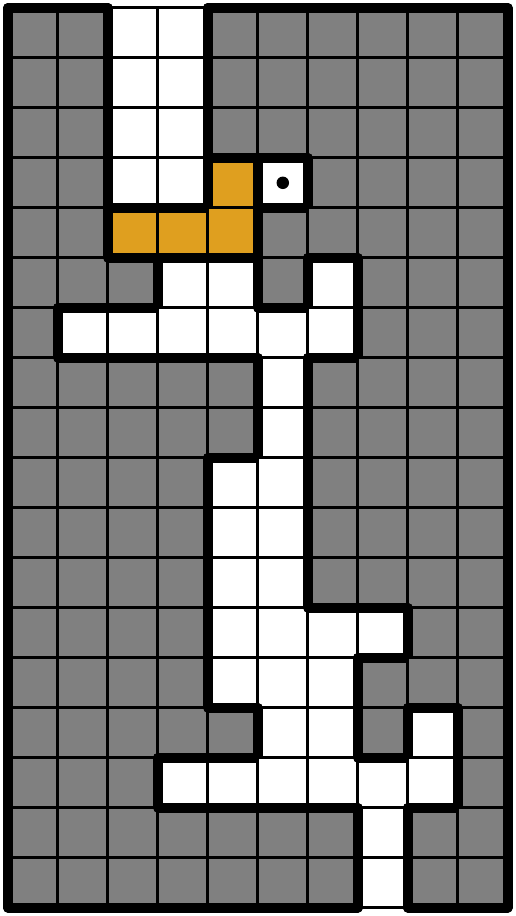}
    \caption{}
  \end{subfigure}
  \begin{subfigure}[b]{0.11\textwidth}
    \centering
    \includegraphics[width=44pt]{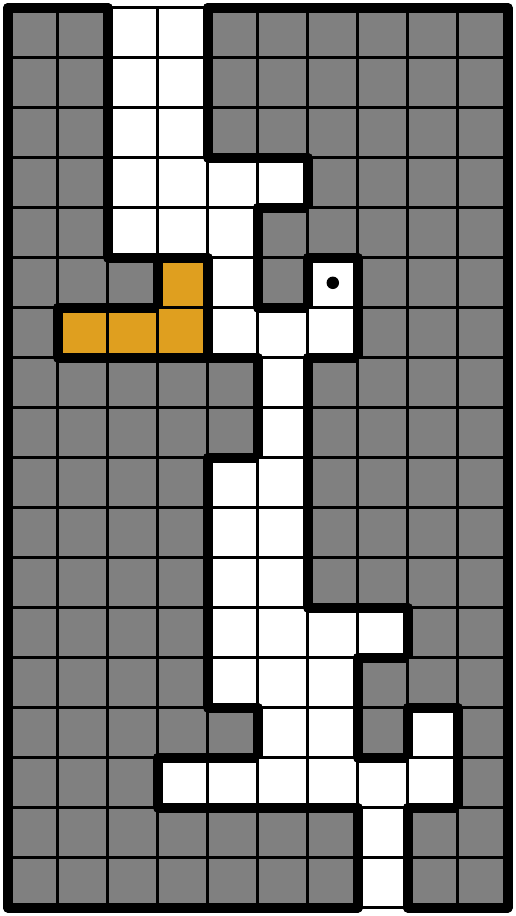}
    \caption{}
  \end{subfigure}
  \begin{subfigure}[b]{0.11\textwidth}
    \centering
    \includegraphics[width=44pt]{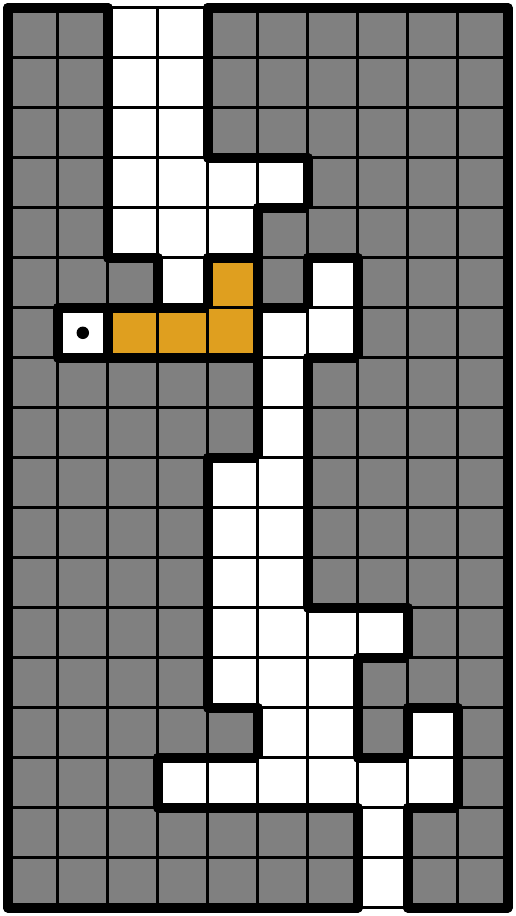}
    \caption{}
  \end{subfigure}
  \begin{subfigure}[b]{0.11\textwidth}
    \centering
    \includegraphics[width=44pt]{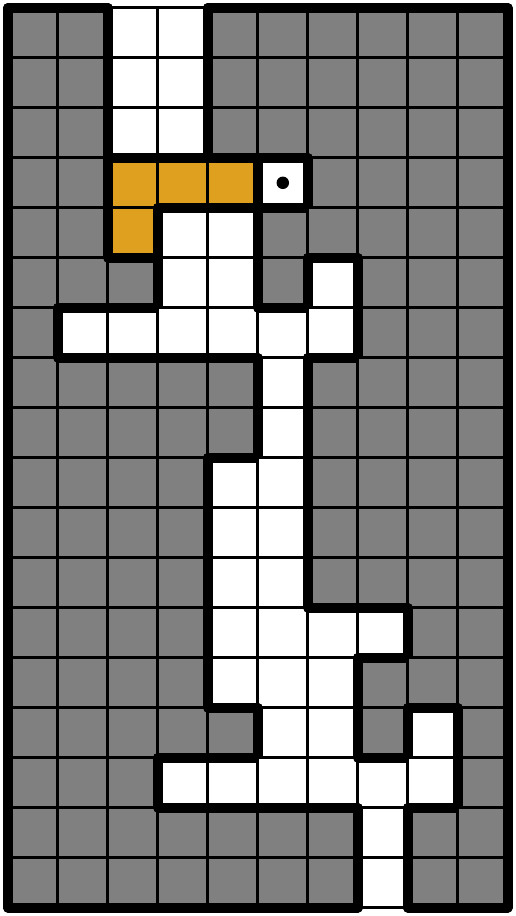}
    \caption{}
  \end{subfigure}
  \begin{subfigure}[b]{0.11\textwidth}
    \centering
    \includegraphics[width=44pt]{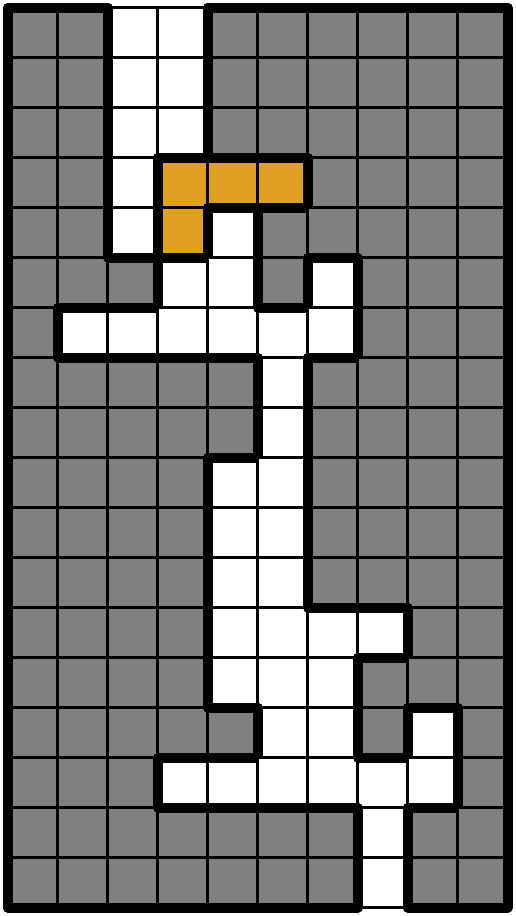}
    \caption{}
  \end{subfigure}
  \begin{subfigure}[b]{0.11\textwidth}
    \centering
    \includegraphics[width=44pt]{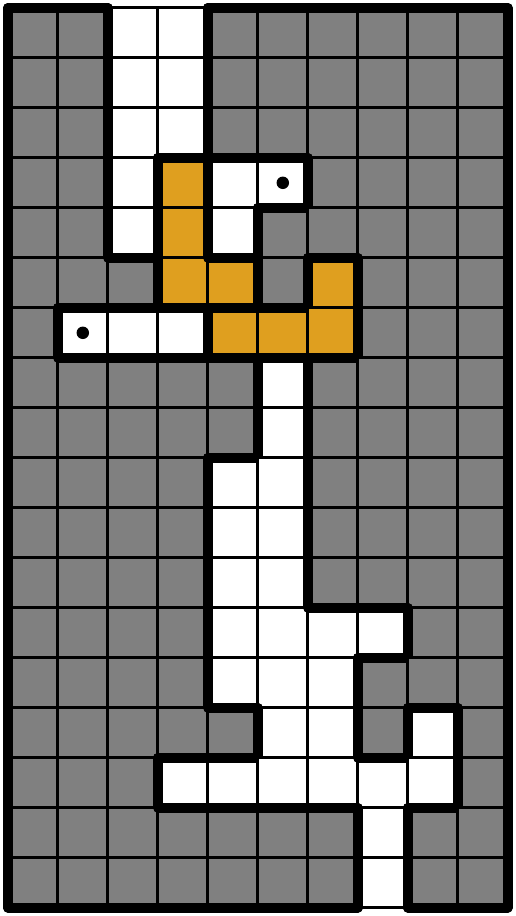}
    \caption{}
  \end{subfigure}
  \begin{subfigure}[b]{0.11\textwidth}
    \centering
    \includegraphics[width=44pt]{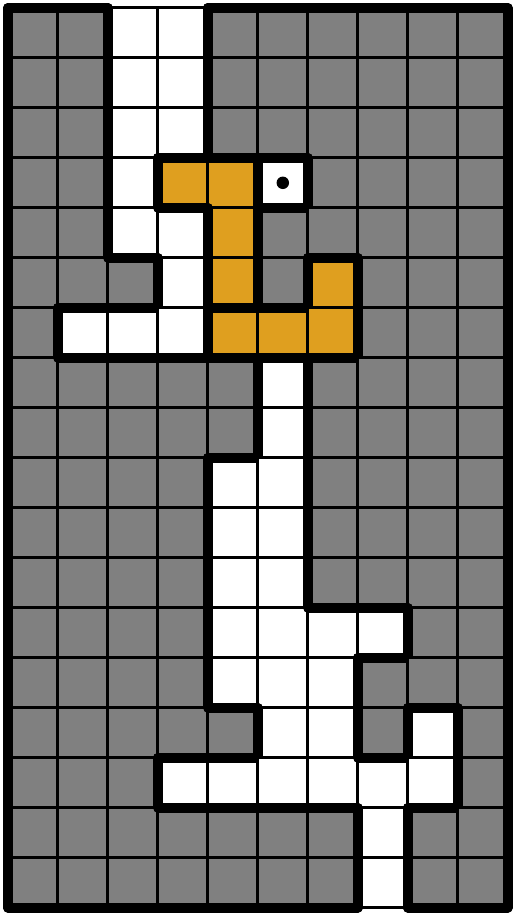}
    \caption{}
  \end{subfigure}
  \begin{subfigure}[b]{0.11\textwidth}
    \centering
    \includegraphics[width=44pt]{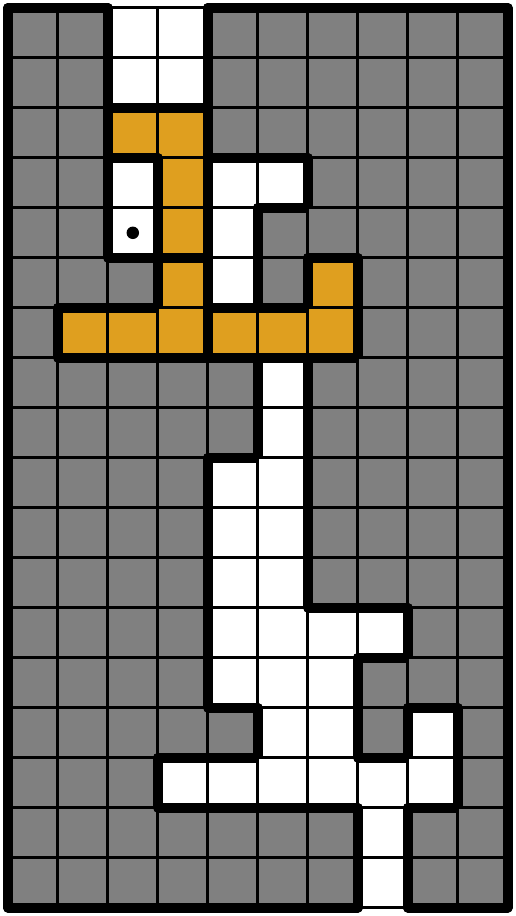}
    \caption{}
  \end{subfigure}
  \begin{subfigure}[b]{0.11\textwidth}
    \centering
    \includegraphics[width=44pt]{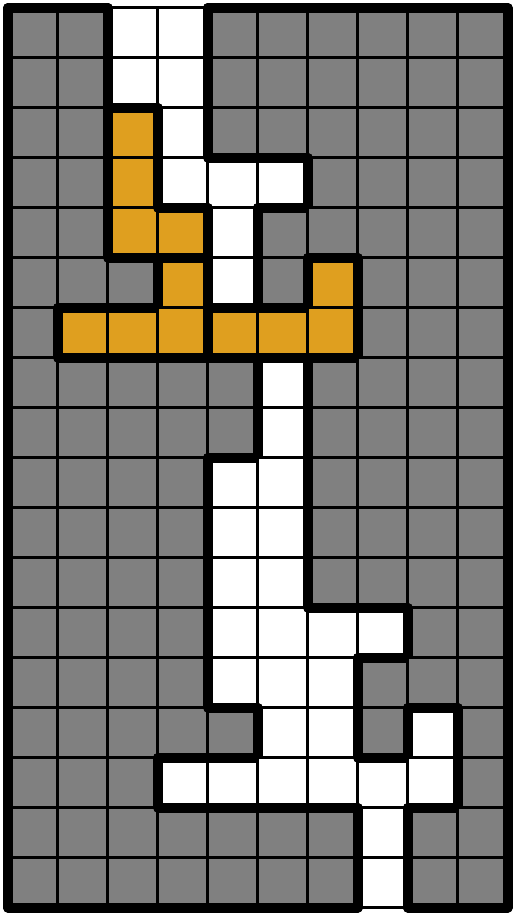}
    \caption{}
  \end{subfigure}
  \caption{Clogs involving $\LL$ pieces from priming sequences in the $\{\JJ, \LL\}$ setup}
  \label{JLLClogs}
\end{figure}

The improper piece placements for $\JJ$ pieces in the filling sequences are shown in Figure \ref{JLJClogs}. Most of the piece placements cause at least one empty square (indicated with a dot) to become inaccessible without overflowing the bottle, either due to being blocked off or due to pieces not being able to rotate in to fill the square or to reach the rest of the bottle (like in (c), (e), and (g)-(i)). The piece placement in (l) prevents any piece from rotating in to cover the empty square with a dot without overflow, as $\JJ$ pieces cannot normally rotate in to cover that empty square and $\LL$ pieces require one of the squares covered by the $\JJ$ piece to be empty in order to rotate in to cover that square (in addition, it is not possible to clear the lines above the $\JJ$ piece without overflow as the number of empty squares is $2\pmod{4}$ and all tetrominoes fill $4$ squares).

\begin{figure}[!ht]
  \centering
  \begin{subfigure}[b]{0.11\textwidth}
    \centering
    \includegraphics[width=44pt]{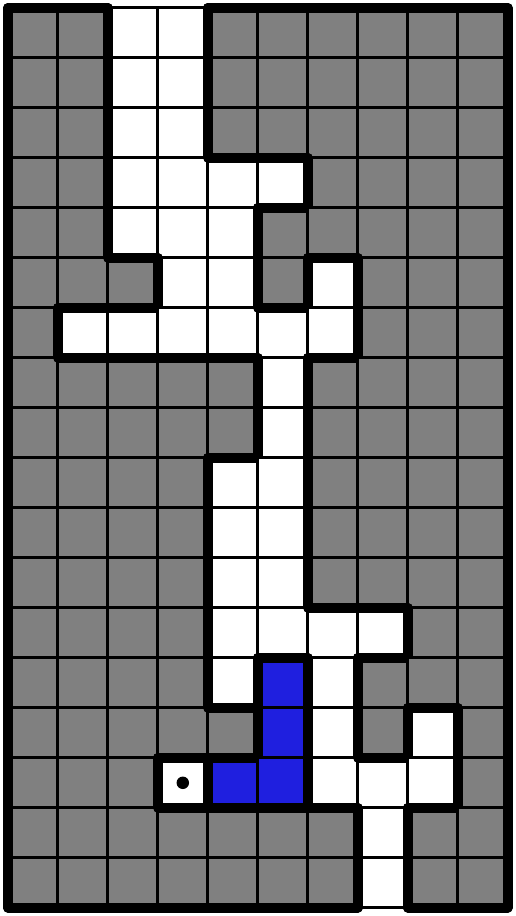}
    \caption{}
  \end{subfigure}
  \begin{subfigure}[b]{0.11\textwidth}
    \centering
    \includegraphics[width=44pt]{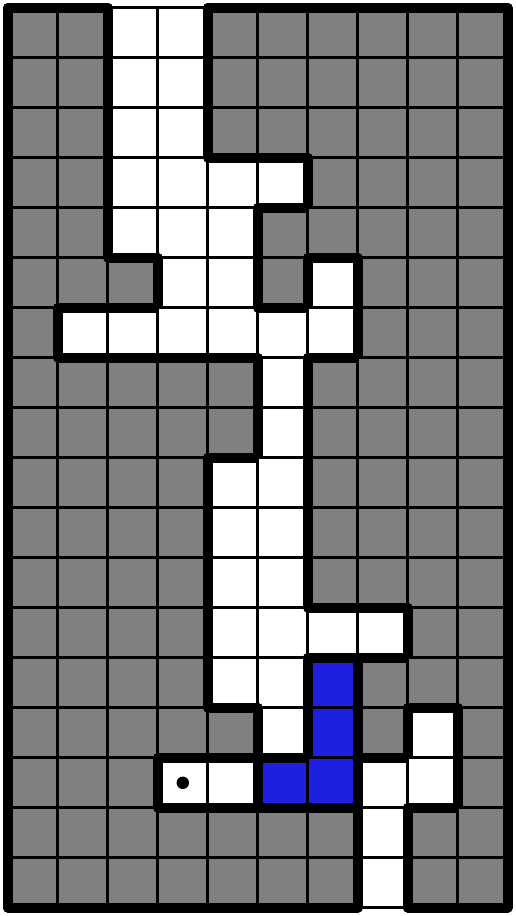}
    \caption{}
  \end{subfigure}
  \begin{subfigure}[b]{0.11\textwidth}
    \centering
    \includegraphics[width=44pt]{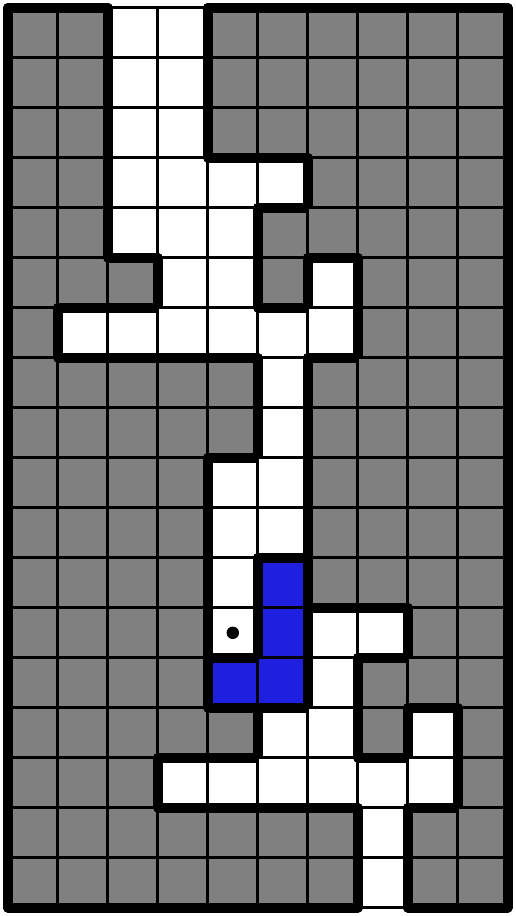}
    \caption{}
  \end{subfigure}
  \begin{subfigure}[b]{0.11\textwidth}
    \centering
    \includegraphics[width=44pt]{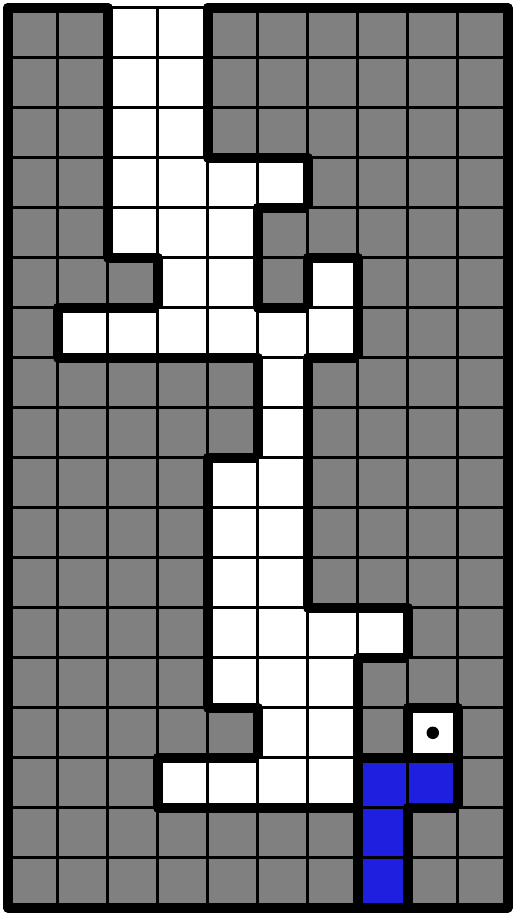}
    \caption{}
  \end{subfigure}
  \begin{subfigure}[b]{0.11\textwidth}
    \centering
    \includegraphics[width=44pt]{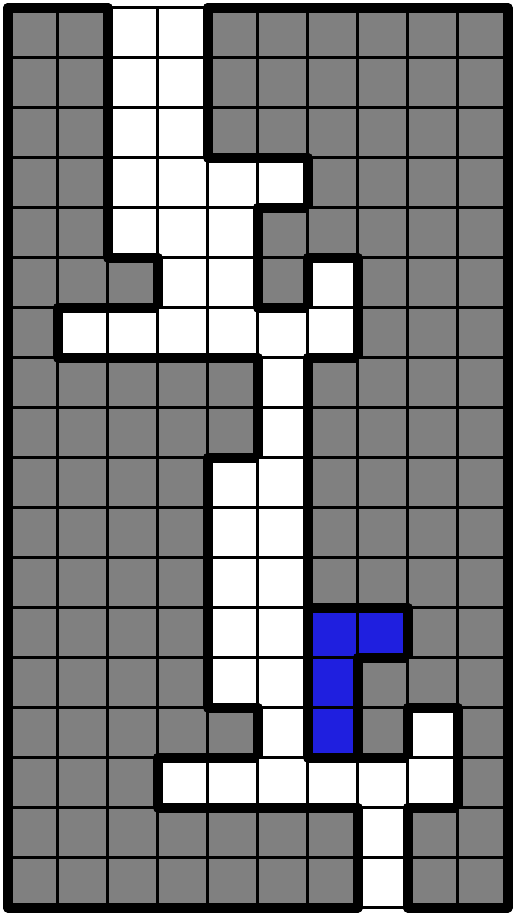}
    \caption{}
  \end{subfigure}
  \begin{subfigure}[b]{0.11\textwidth}
    \centering
    \includegraphics[width=44pt]{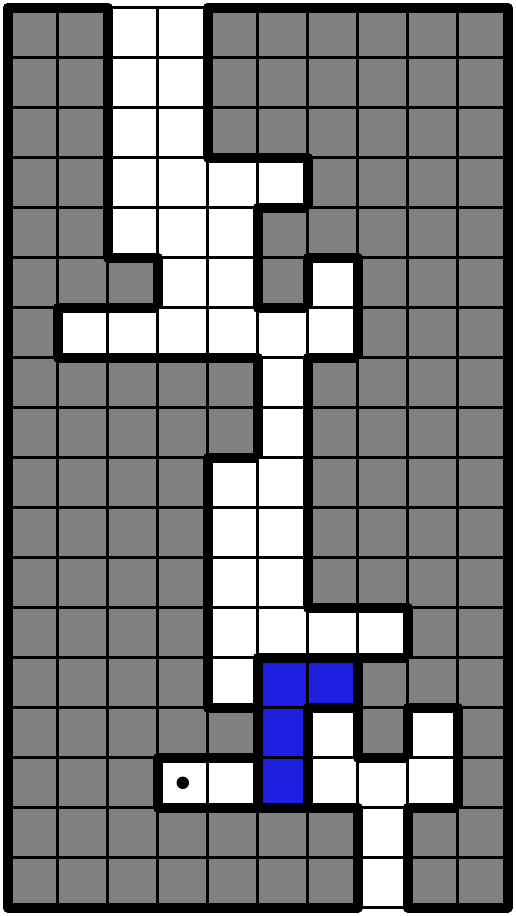}
    \caption{}
  \end{subfigure}
  \begin{subfigure}[b]{0.11\textwidth}
    \centering
    \includegraphics[width=44pt]{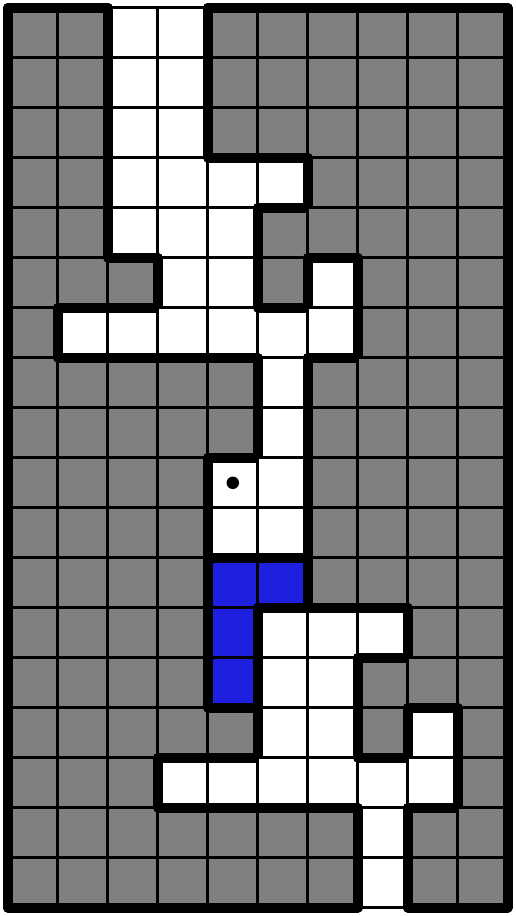}
    \caption{}
  \end{subfigure}
  \begin{subfigure}[b]{0.11\textwidth}
    \centering
    \includegraphics[width=44pt]{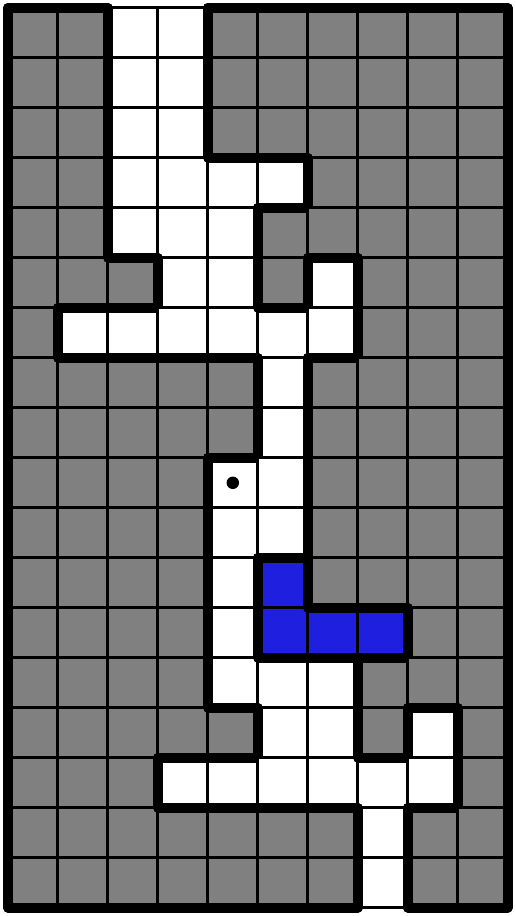}
    \caption{}
  \end{subfigure}
  \begin{subfigure}[b]{0.11\textwidth}
    \centering
    \includegraphics[width=44pt]{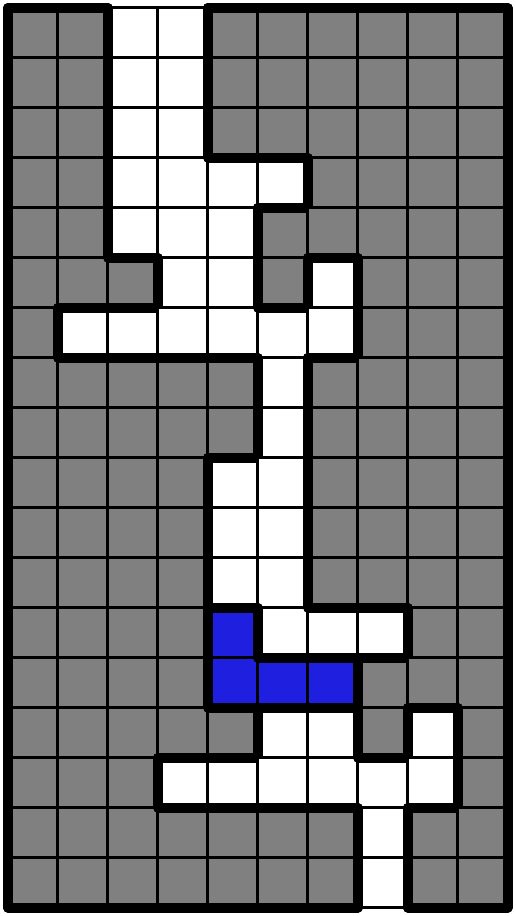}
    \caption{}
  \end{subfigure}
  \begin{subfigure}[b]{0.11\textwidth}
    \centering
    \includegraphics[width=44pt]{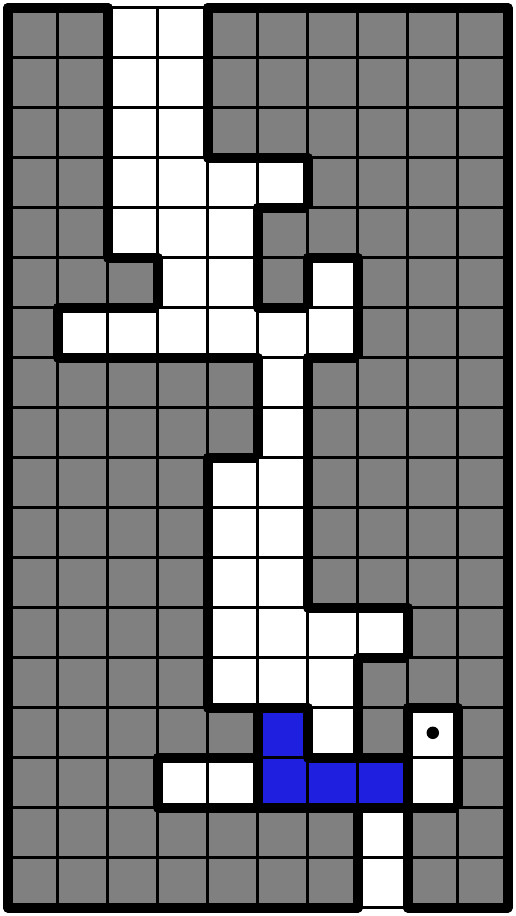}
    \caption{}
  \end{subfigure}
  \begin{subfigure}[b]{0.11\textwidth}
    \centering
    \includegraphics[width=44pt]{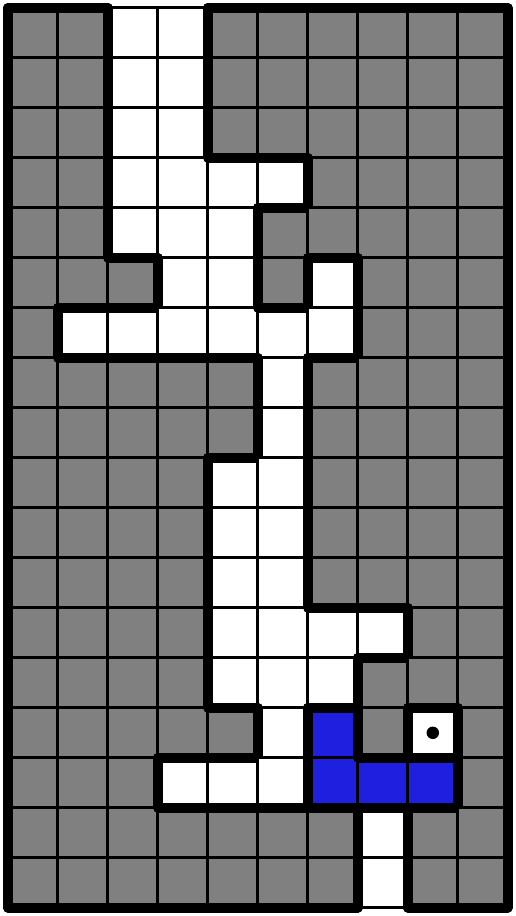}
    \caption{}
  \end{subfigure}
  \begin{subfigure}[b]{0.11\textwidth}
    \centering
    \includegraphics[width=44pt]{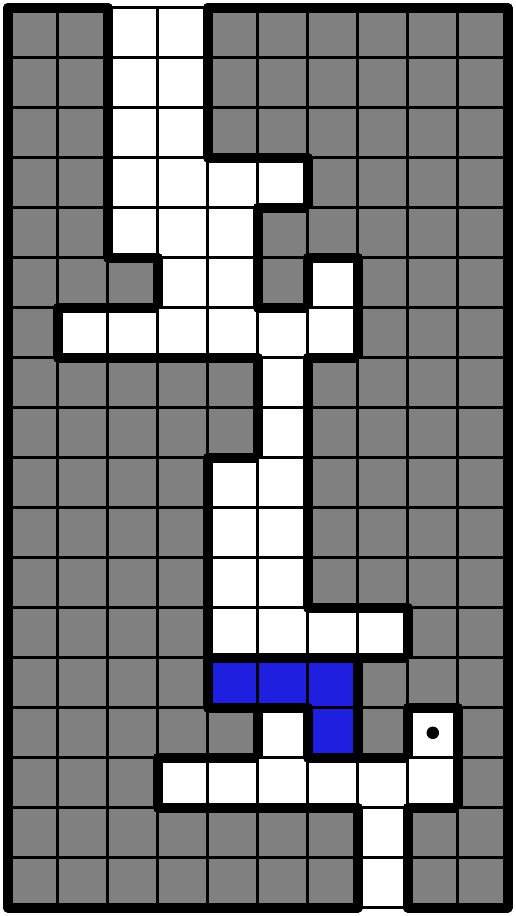}
    \caption{}
  \end{subfigure}
  \begin{subfigure}[b]{0.11\textwidth}
    \centering
    \includegraphics[width=44pt]{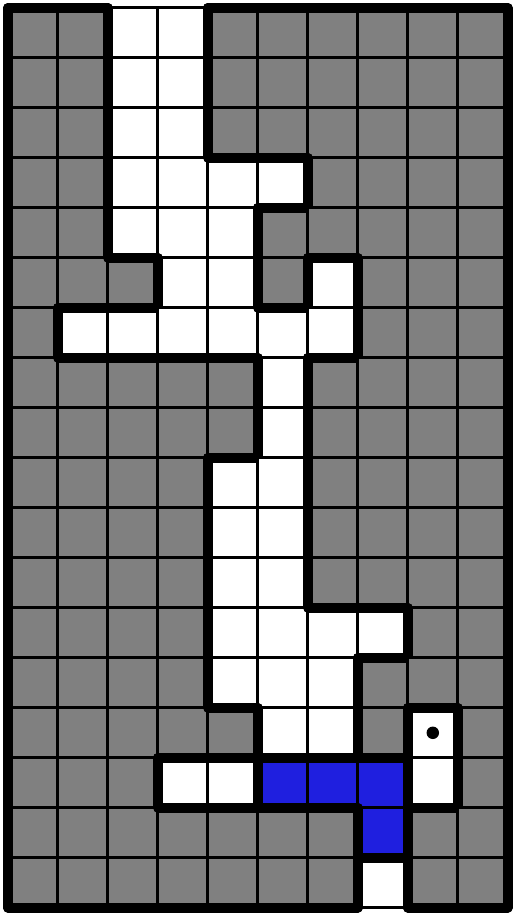}
    \caption{}
  \end{subfigure}
  \caption{Clogs involving $\JJ$ pieces from filling sequences in the $\{\JJ, \LL\}$ setup}
  \label{JLJClogs}
\end{figure}

The improper piece placements for $\JJ$ pieces and $\LL$ pieces in the closing sequences are shown in Figure \ref{JLCloseClogs}; all of them yield scenarios in which it is impossible to re-open the top of the bottle without overflow.

\begin{figure}[!ht]
  \centering
  \begin{subfigure}[b]{0.12\textwidth}
    \centering
    \includegraphics[width=45pt]{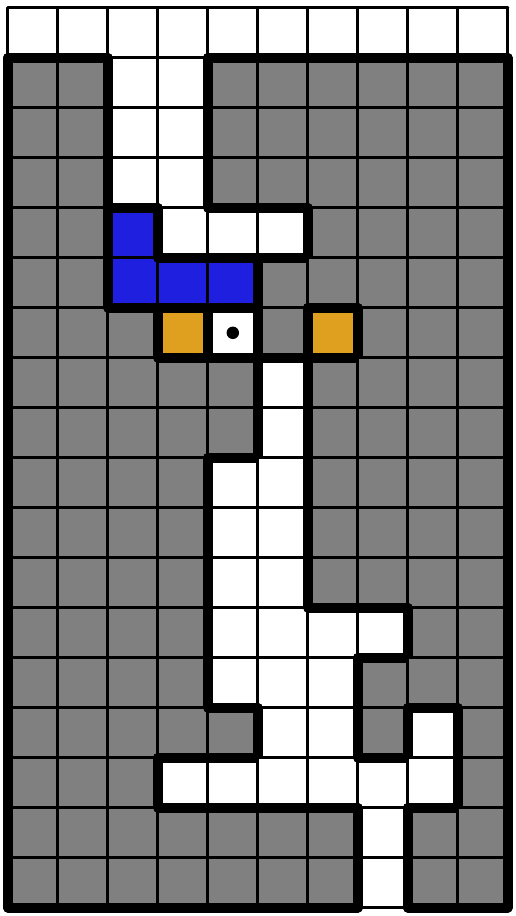}
    \caption{}
  \end{subfigure}
  \begin{subfigure}[b]{0.12\textwidth}
    \centering
    \includegraphics[width=45pt]{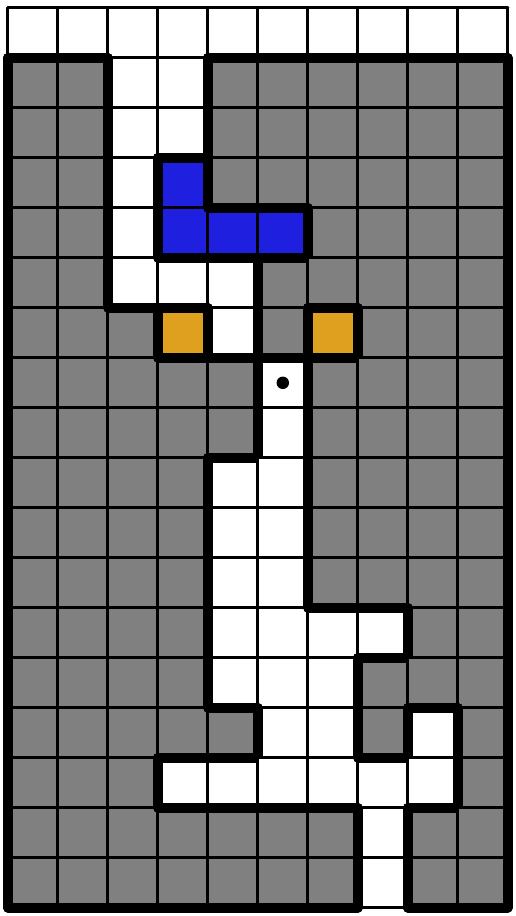}
    \caption{}
  \end{subfigure}
  \begin{subfigure}[b]{0.12\textwidth}
    \centering
    \includegraphics[width=45pt]{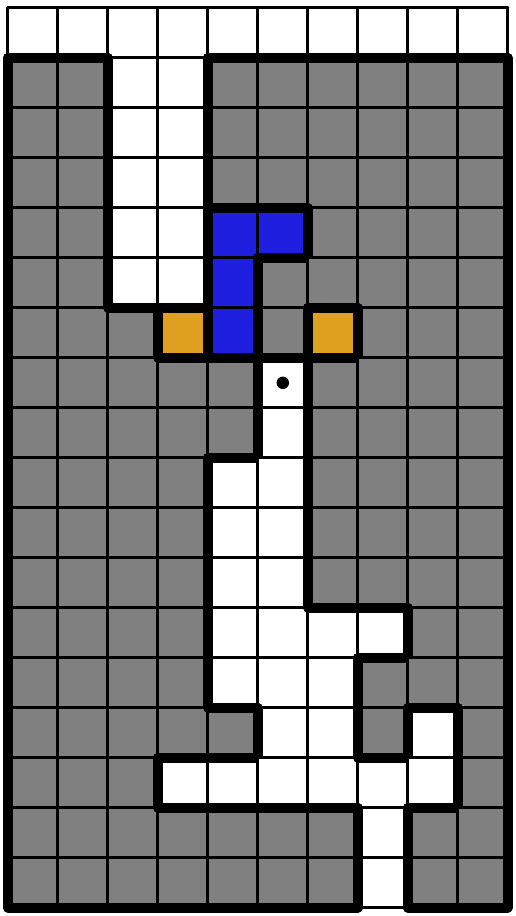}
    \caption{}
  \end{subfigure}
  \begin{subfigure}[b]{0.12\textwidth}
    \centering
    \includegraphics[width=45pt]{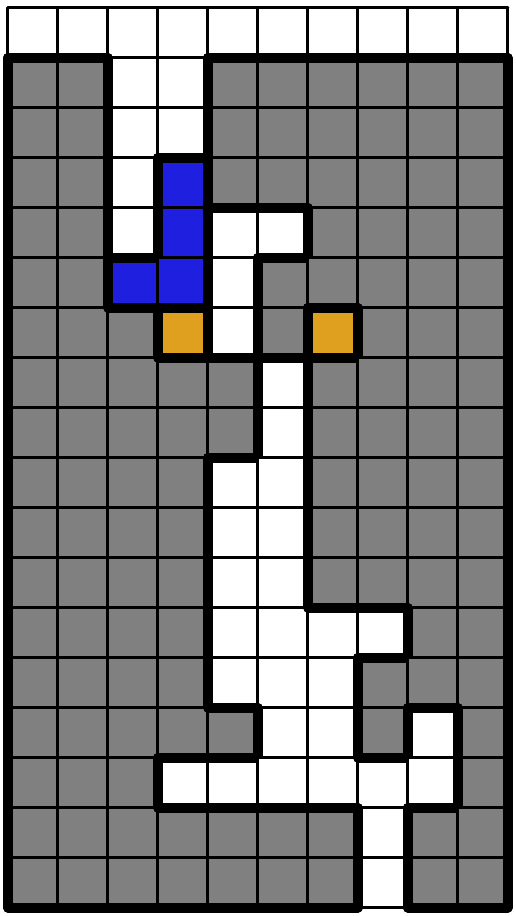}
    \caption{}
  \end{subfigure}
  \begin{subfigure}[b]{0.12\textwidth}
    \centering
    \includegraphics[width=45pt]{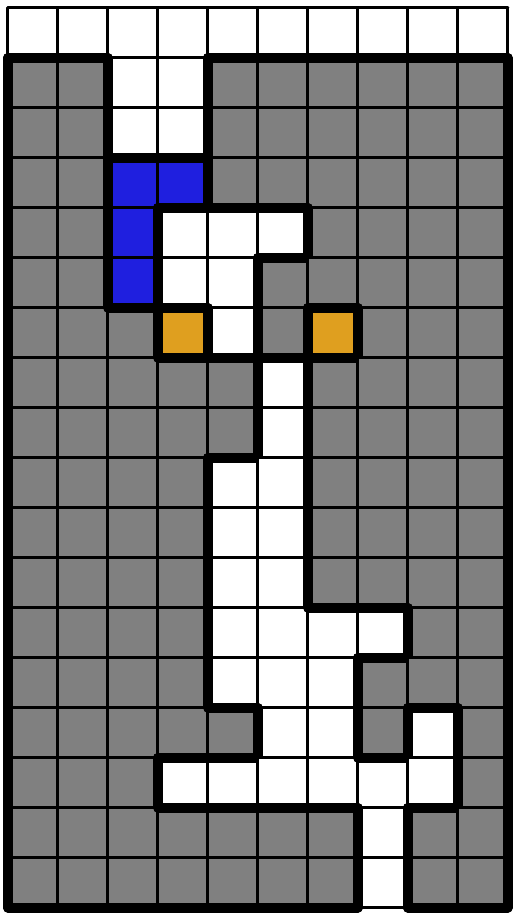}
    \caption{}
  \end{subfigure}
  \begin{subfigure}[b]{0.12\textwidth}
    \centering
    \includegraphics[width=45pt]{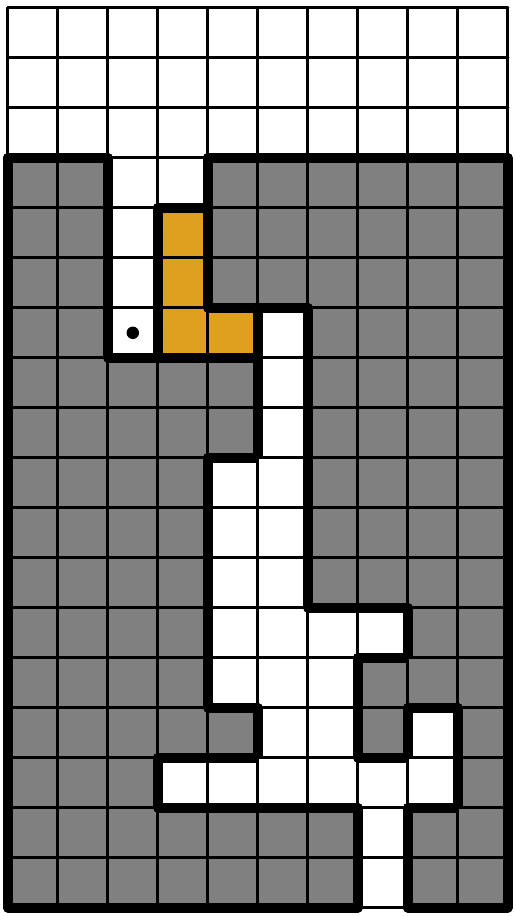}
    \caption{}
  \end{subfigure}
  \begin{subfigure}[b]{0.12\textwidth}
    \centering
    \includegraphics[width=45pt]{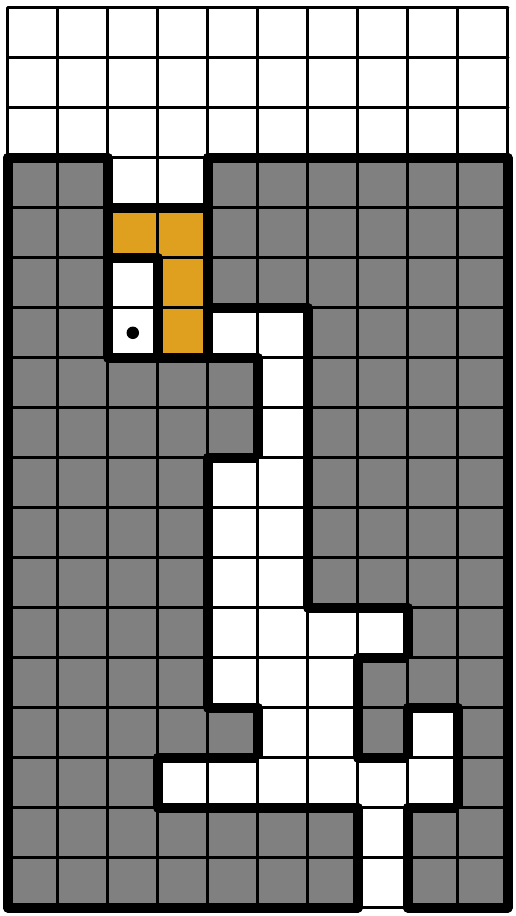}
    \caption{}
  \end{subfigure}
  \caption{Clogs involving pieces from closing sequences in the $\{\JJ, \LL\}$ setup}
  \label{JLCloseClogs}
\end{figure}

$$ $$

\subsection{Clogs for the ASP-Completeness Setup}\label{appendix:aspclogs}

The improper piece placements for $\II$ pieces are shown in Figure \ref{ASPIClogs}. The placements in (a) and (b) block the empty square indicated with a dot. The placement in (c) forces the empty square indicated with a dot to be filled with an $\II$ piece, which blocks off other empty squares. The placement in (d) causes the empty square indicated with a dot to become inaccessible, as only $\II$ pieces can drop below the hole of size $1$ next to the $\II$ piece, and $\II$ pieces cannot rotate in to fill that empty square.

\begin{figure}[!ht]
  \centering
  \begin{subfigure}[b]{0.15\textwidth}
    \centering
    \includegraphics[width=60pt]{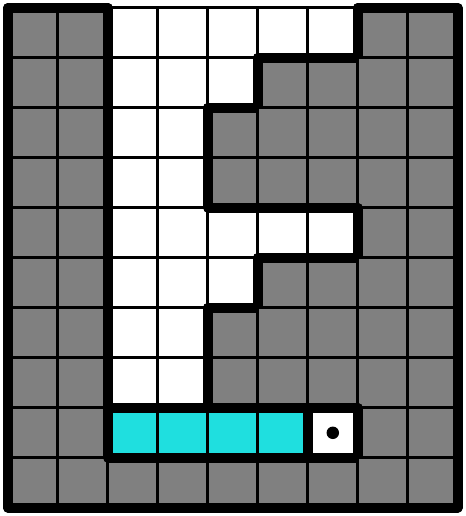}
    \caption{}
  \end{subfigure}
  \begin{subfigure}[b]{0.15\textwidth}
    \centering
    \includegraphics[width=60pt]{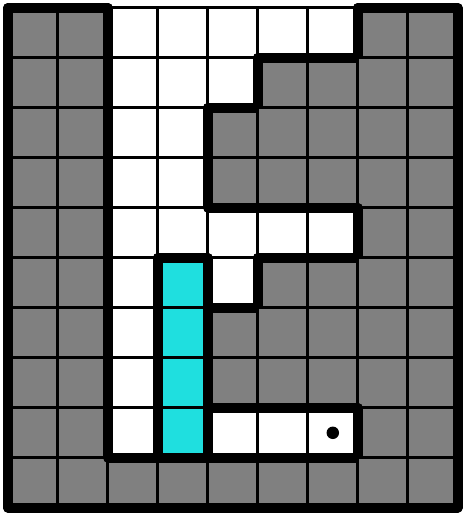}
    \caption{}
  \end{subfigure}
  \begin{subfigure}[b]{0.15\textwidth}
    \centering
    \includegraphics[width=60pt]{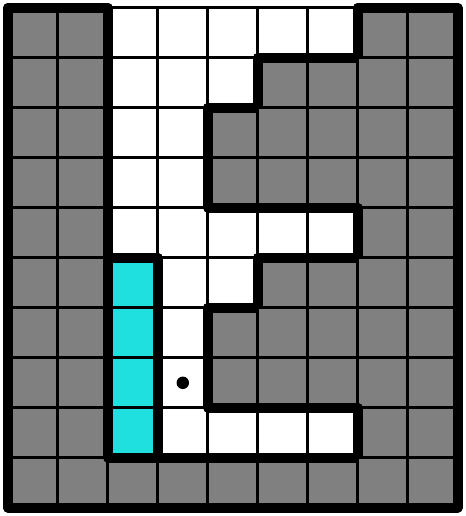}
    \caption{}
  \end{subfigure}
  \begin{subfigure}[b]{0.15\textwidth}
    \centering
    \includegraphics[width=60pt]{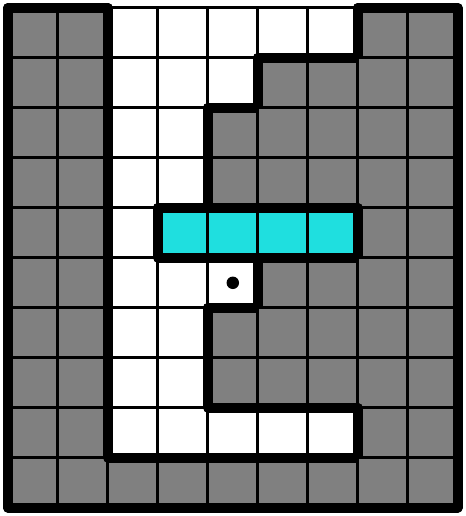}
    \caption{}
  \end{subfigure}
  \caption{Clogs involving $\II$ pieces in the setup for ASP-completeness of $\{\II, \TT, \LL\}$}
  \label{ASPIClogs}
\end{figure}

The improper piece placements for $\TT$ pieces and $\LL$ pieces are shown in Figures \ref{ASPTClogs} and \ref{ASPLClogs}; all of them cause an empty square (indicated with a dot) to become inaccessible.

\begin{figure}[!ht]
  \centering
  \begin{subfigure}[b]{0.15\textwidth}
    \centering
    \includegraphics[width=60pt]{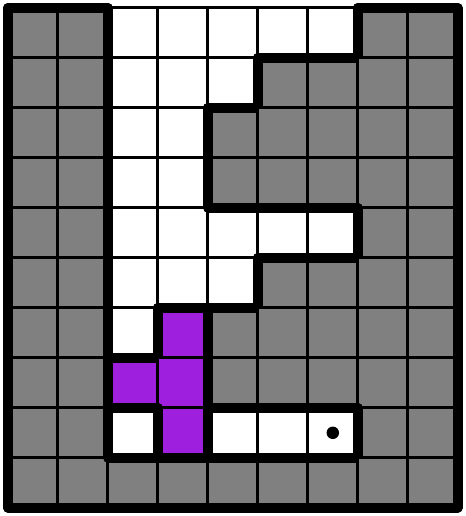}
    \caption{}
  \end{subfigure}
  \begin{subfigure}[b]{0.15\textwidth}
    \centering
    \includegraphics[width=60pt]{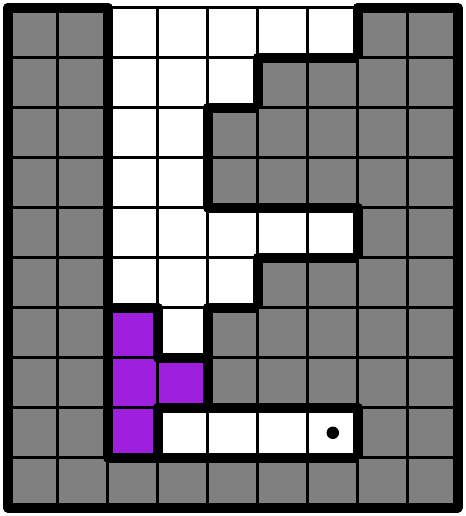}
    \caption{}
  \end{subfigure}
  \begin{subfigure}[b]{0.15\textwidth}
    \centering
    \includegraphics[width=60pt]{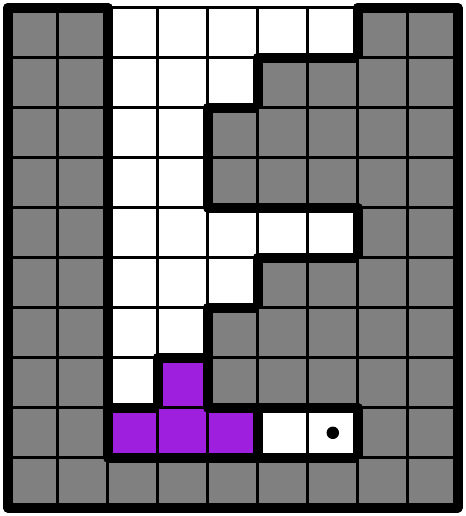}
    \caption{}
  \end{subfigure}
  \begin{subfigure}[b]{0.15\textwidth}
    \centering
    \includegraphics[width=60pt]{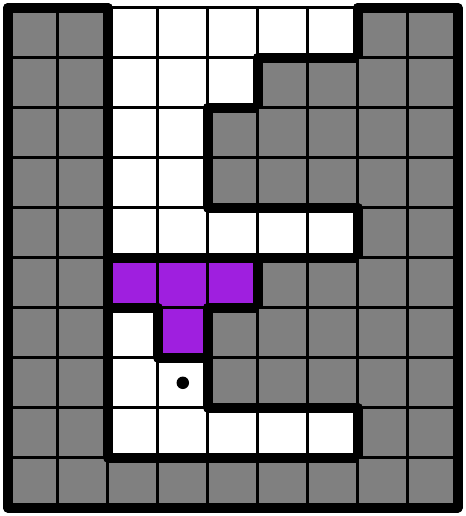}
    \caption{}
  \end{subfigure}
  \begin{subfigure}[b]{0.15\textwidth}
    \centering
    \includegraphics[width=60pt]{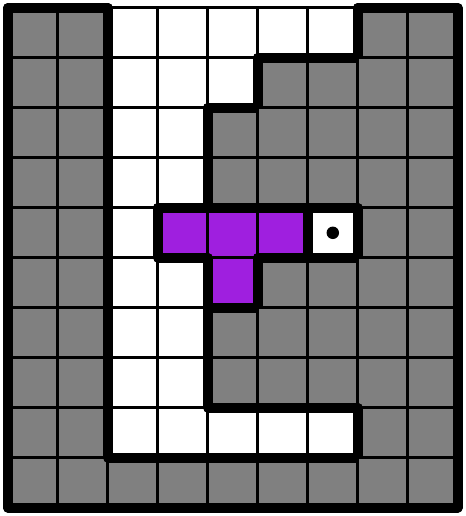}
    \caption{}
  \end{subfigure}
  \begin{subfigure}[b]{0.15\textwidth}
    \centering
    \includegraphics[width=60pt]{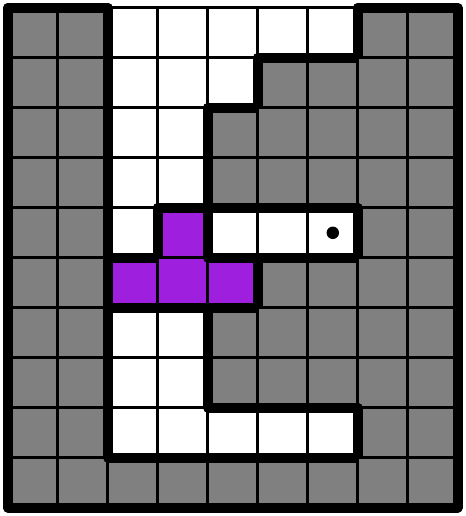}
    \caption{}
  \end{subfigure}
  \begin{subfigure}[b]{0.15\textwidth}
    \centering
    \includegraphics[width=60pt]{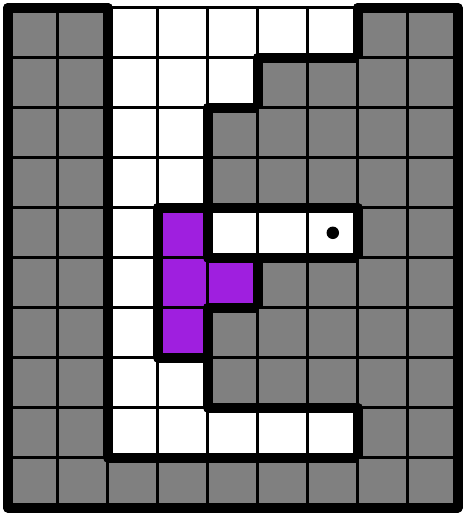}
    \caption{}
  \end{subfigure}
  \begin{subfigure}[b]{0.15\textwidth}
    \centering
    \includegraphics[width=60pt]{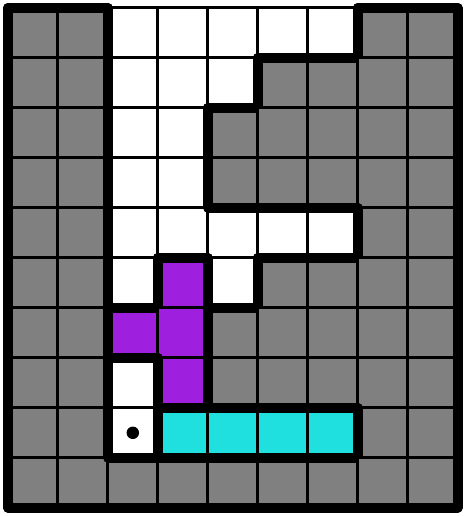}
    \caption{}
  \end{subfigure}
  \begin{subfigure}[b]{0.15\textwidth}
    \centering
    \includegraphics[width=60pt]{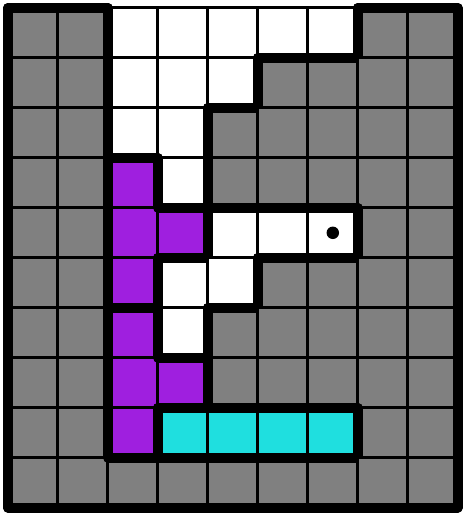}
    \caption{}
  \end{subfigure}
  \begin{subfigure}[b]{0.15\textwidth}
    \centering
    \includegraphics[width=60pt]{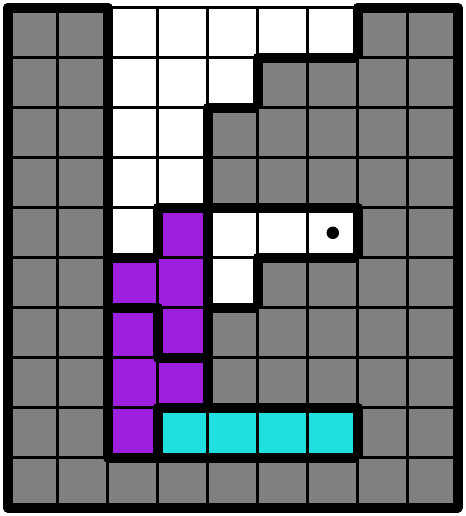}
    \caption{}
  \end{subfigure}
  \caption{Clogs involving $\TT$ pieces in the setup for ASP-completeness of $\{\II, \TT, \LL\}$}
  \label{ASPTClogs}
\end{figure}

\begin{figure}[!ht]
  \centering
  \begin{subfigure}[b]{0.15\textwidth}
    \centering
    \includegraphics[width=60pt]{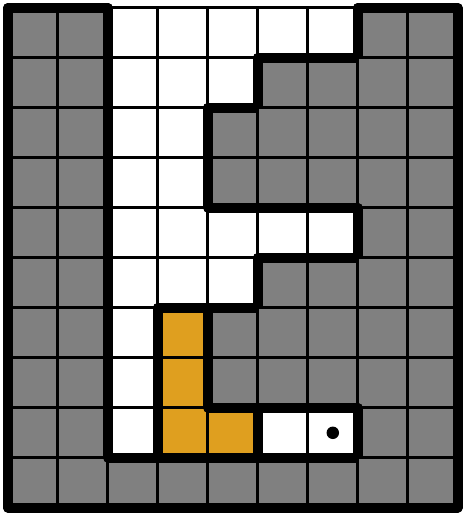}
    \caption{}
  \end{subfigure}
  \begin{subfigure}[b]{0.15\textwidth}
    \centering
    \includegraphics[width=60pt]{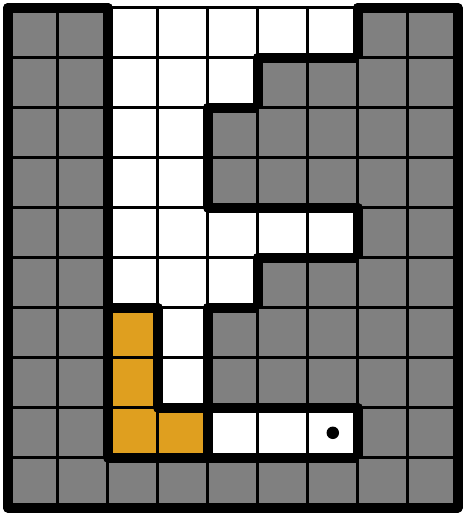}
    \caption{}
  \end{subfigure}
  \begin{subfigure}[b]{0.15\textwidth}
    \centering
    \includegraphics[width=60pt]{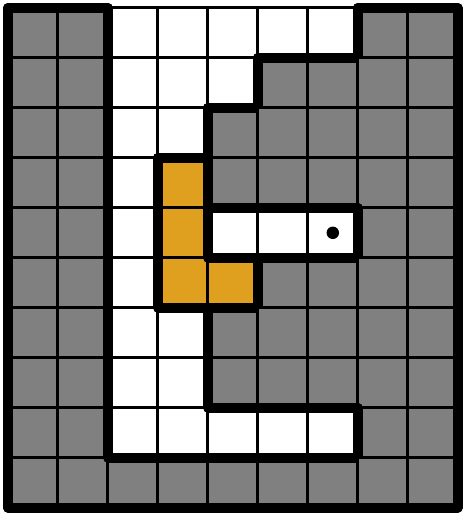}
    \caption{}
  \end{subfigure}
  \begin{subfigure}[b]{0.15\textwidth}
    \centering
    \includegraphics[width=60pt]{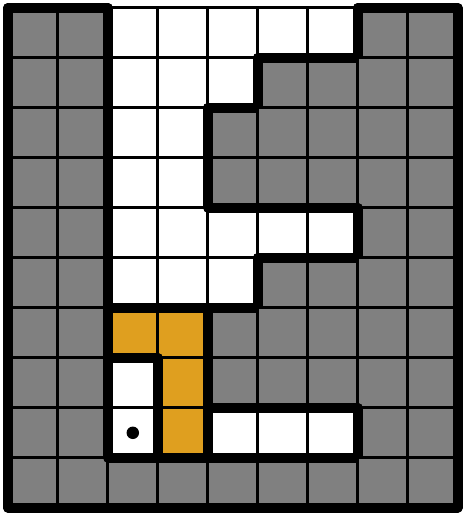}
    \caption{}
  \end{subfigure}
  \begin{subfigure}[b]{0.15\textwidth}
    \centering
    \includegraphics[width=60pt]{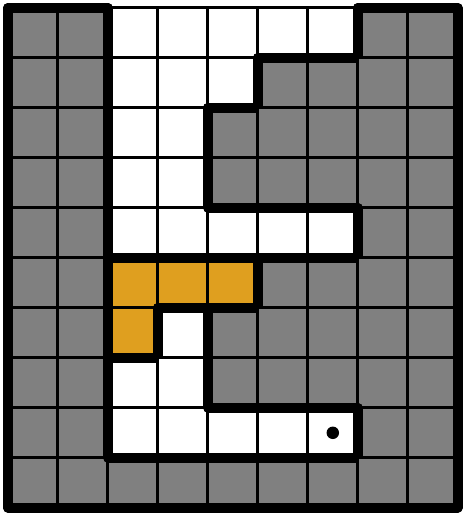}
    \caption{}
  \end{subfigure}
  \begin{subfigure}[b]{0.15\textwidth}
    \centering
    \includegraphics[width=60pt]{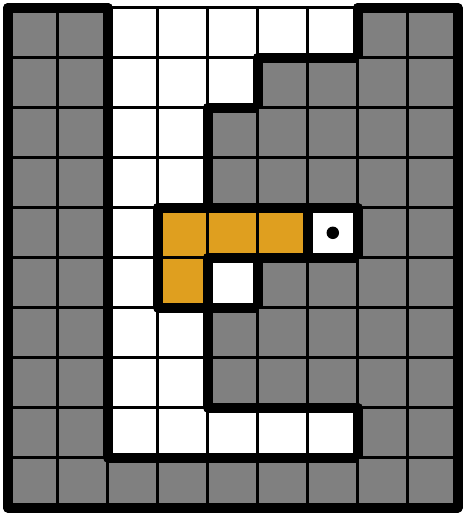}
    \caption{}
  \end{subfigure}
  \begin{subfigure}[b]{0.15\textwidth}
    \centering
    \includegraphics[width=60pt]{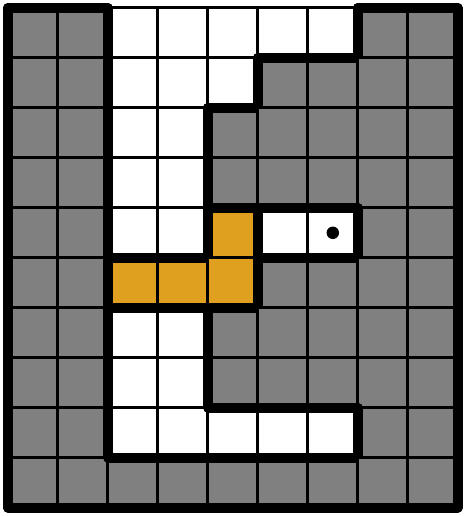}
    \caption{}
  \end{subfigure}
  \begin{subfigure}[b]{0.15\textwidth}
    \centering
    \includegraphics[width=60pt]{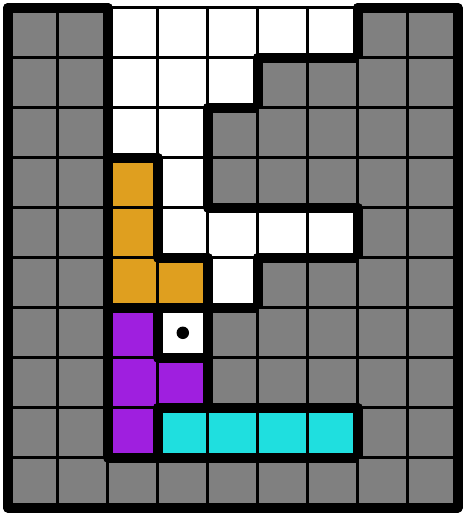}
    \caption{}
  \end{subfigure}
  \begin{subfigure}[b]{0.15\textwidth}
    \centering
    \includegraphics[width=60pt]{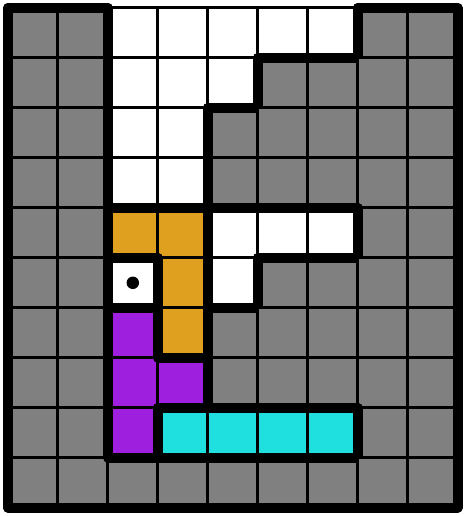}
    \caption{}
  \end{subfigure}
  \begin{subfigure}[b]{0.15\textwidth}
    \centering
    \includegraphics[width=60pt]{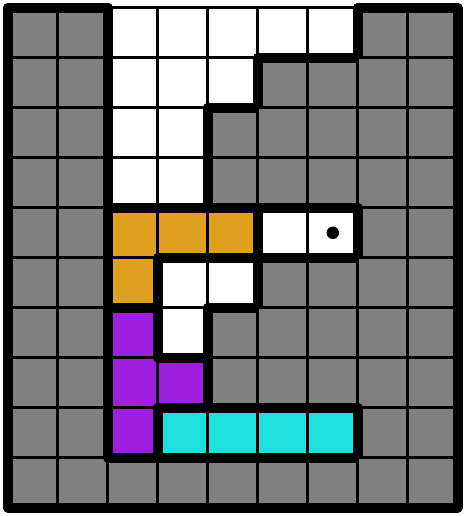}
    \caption{}
  \end{subfigure}
  \caption{Clogs involving $\LL$ pieces in the setup for ASP-completeness of $\{\II, \TT, \LL\}$}
  \label{ASPLClogs}
\end{figure}